%% file: thesis.tex
\documentclass[a4paper, 10pt, twoside]{book}

\newcommand\draft[1]{}
\newcommand\release[1]{#1}
\draft{\newcommand\texorpdfstring[2]{#1}}

\draft{\newcommand\TeXBug[1]{}}
\release{\newcommand\TeXBug[1]{#1}} % Turns on all available TeX bugs

\newcommand\bigone[1]{#1}
\newcommand\smallone[1]{}

\input{\relativepath _MainMacros.tex}
\input{\relativepath _MyAlgorithm.tex}

\usepackage[headings]{fullpage}
\usepackage{wrapfig}
\usepackage{array}

\newcommand{\cA}{{\cal A}}
\newcommand{\cB}{{\cal B}}
\newcommand{\cC}{{\cal C}}
\newcommand{\cD}{{\cal D}}
\newcommand{\cG}{{\cal G}}
\newcommand{\cH}{{\cal H}}
\newcommand{\cI}{{\cal I}}

\newcommand{\cO}{{\cal O}}
\newcommand{\cP}{{\cal P}}
\newcommand{\cZ}{{\cal Z}}

\newcommand{\Adv}{\mathop{\mathrm{ADV}^\pm}}
\newcommand{\pAdv}{\mathop{\mathrm{ADV}}}

\newcommand{\im}{\mathop{\mathrm{im}}}

\newcommand{\bra}[1]{{e^*_{#1}}}

\usepackage{euscript}

\newcommand{\reg}[1]{{\mathsf{#1}}}
\newcommand{\qreg}[1]{\mathsf{\underline{#1}}}
\newcommand{\hilbert}[1]{H_{\reg{#1}}}
\newcommand{\proj}[2]{\Pi_{\reg{#1}}^{(#2)}}
\newcommand{\circuit}[1]{\EuScript{#1}}
\newcommand{\qgets}{\stackrel{+}{\longleftarrow}}
\newcommand{\qungets}{\stackrel{-}{\longleftarrow}}

\DeclareMathOperator{\diag}{diag}
\DeclareMathOperator{\spn}{span}

\input{xy}
\xyoption{all}
\xyoption{poly}
\release{\xyoption{dvips}}

\def \kett|#1>{|#1\rangle}

\usepackage[Lenny]{fncychap}
\usepackage{fancyhdr}
\usepackage{emptypage}

\begin{document}
%---------------------------------------------------------------
% Title page
%---------------------------------------------------------------

\pagestyle{plain}
\thispagestyle{empty}

\begin{center}
\strut
\vskip 1cm

{\large University of Latvia}
\vskip 0.1cm
{\large Faculty of Computing}
\vskip 2.5cm
{\Large Aleksandrs Belovs}
\vskip 2cm
{\LARGE \bf Applications of the Adversary Method in}
\vskip 0.1cm
{\LARGE \bf Quantum Query Algorithms}
\vskip 2cm
{\large Doctoral Thesis}
\vskip 1cm
Area: Computer Science \\
Sub-Area: Mathematical Foundations of Computer Science
\end{center}
\vskip 2cm
\begin{flushright}
Scientific Advisor: \\
Dr. Comp. Sci., Prof. Andris Ambainis
\end{flushright}
\vskip 4cm
\centerline{\large Riga, 2013}

\pagebreak

%---------------------------------------------------------------
% ESF page
%---------------------------------------------------------------

\strut

\begin{quotation}
\begin{figure}[h]
\release{\includegraphics[width=10cm]{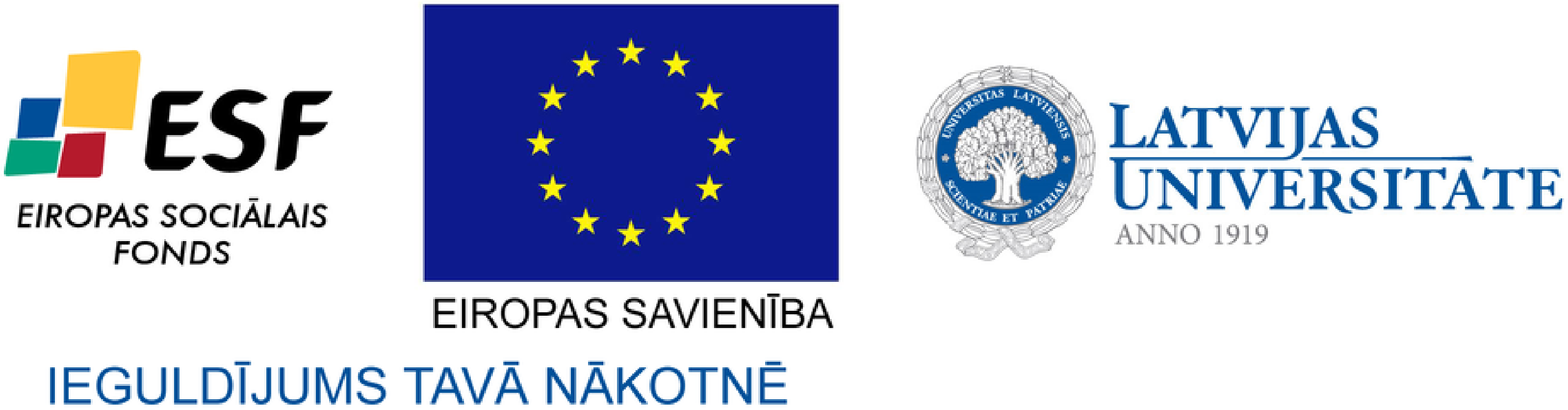}}
\end{figure}

\noindent
This work has been supported by the European Social Fund\\
within the project ``Support for Doctoral Studies at University of Latvia",\\

\noindent
and the FET-Open project ``Quantum Computer Science'' (QCS).
\end{quotation}

\clearpage
\strut
\vspace{2cm}

\section*{\centering \begin{normalsize}Abstract\end{normalsize}}
\begin{quotation}
In the thesis, we use a recently developed tight characterisation of quantum query complexity, the adversary bound, to develop new quantum algorithms and lower bounds.  Our results are as follows:
\begin{itemize}
\item
We develop a new technique for the construction of quantum algorithms: learning graphs.
\item
We use learning graphs to improve quantum query complexity of the triangle detection and the $k$-distinctness problems.
\item
We prove tight lower bounds for the $k$-sum and the triangle sum problems.
\item
We construct quantum algorithms for some subgraph-finding problems that are optimal in terms of query, time and space complexities.
\item
We develop a generalisation of quantum walks that connects electrical properties of a graph and its quantum hitting time.  We use it to construct a time-efficient quantum algorithm for 3-distinctness.
\end{itemize}
\end{quotation}

\clearpage
\strut
\vspace{2cm}

\section*{\centering \begin{normalsize} Acknowledgements \end{normalsize}}
\begin{quotation}
First of all, I want to thank my thesis advisor, Andris Ambainis, for being a source of many interesting and important problems, and for continuous support during my years of PhD.

I am thankful to J\'er\'emie Roland, Juris V\=\i ksna, and Ronald de Wolf for being my thesis referees and their consequent suggestions on the improvements in the text of the dissertation.

I am grateful to the co-authors of the papers and preprints constituting this thesis: 
Andrew Childs, 
Stacey Jeffery, 
Robin Kothari, 
Troy Lee,
Fr\'ed\'eric Magniez, 
Ben Reichardt, 
Ansis Rosmanis, and
Robert \v Spalek.
Additionally, I would like to thank
Dmitry Gavinsky,
Tsuyoshi Ito,
Rajat Mittal, 
Martin R\"otteler,
Miklos Santha, and
Ronald de Wolf
for many fruitful discussions on the subject.

During these years, I made academic visits to CWI, Amsterdam; IQC, Waterloo; NEC Laboratories, Princeton; and NUS, Singapore.  I would like to thank Ronald de Wolf, Harry Buhrman, Andrew Childs, Martin R\"otteler, Miklos Santha, and Troy Lee for hospitality.

I want to thank Abuzer Yakaryilmaz, Laura Man\v cinska, and Alexander Rivosh for closely reading parts of the thesis and suggesting many improvements.

Finally, I would like to thank my parents and friends for many things not directly related to the thesis.
\end{quotation}

\clearpage

\tableofcontents

%\chapter{Introduction}
\chapter*{Introduction}
\pagestyle{fancy}
\addcontentsline{toc}{chapter}{Introduction}
\input{_intro}

\part{Preliminaries}

\chapter{The Model}
\label{chp:model}
\input{_model}

\chapter{Quantum-Walk-Based Algorithms}
\label{chp:walk}
\input{_walk}

\chapter{Lower Bounds for Quantum Query Complexity}
\label{chp:lower}
\input{_lowerBounds}

\part{Query Algorithms}
%\section*{Overview of Part II}
%\addcontentsline{toc}{section}{Overview of Part II}

\chapter*{Overview of Part II}
As we saw in \rf(chp:lower), it is possible to characterise the quantum query complexity of a function by a relatively simple semi-definite program (SDP): the adversary bound given by~\rf(eqn:advPrimal) and~\rf(eqn:advDual).
Unfortunately, for many functions, even this SDP is too hard to solve.  Prior to this work, no explicit negative-weight adversary lower bound was known except for the composed functions described in \rf(sec:advExamples).  Upper bounds also did not go much beyond what we covered in \rf(chp:lower).

In this part of the thesis, we give a number of applications of the adversary bound SDP for explicit non-iterated functions.  As mentioned in the introduction, there are some advantages in constructing a feasible solution to the dual adversary SDP comparison to the development of an explicit quantum query algorithm.  We are not interested in sometimes cumbersome details of the internal organisation of the algorithm: we are only interested in the feasibility of the solution and its objective value.  The construction of a feasible solution to the dual adversary SDP requires other techniques, and may open a completely different perspective on the problem being solved.  Also, the dual adversary SDP is tight, so, in principle, we lose nothing with this transition.  Finally, a solution close to optimal may help in constructing a {\em primal} adversary SDP via semi-definite duality, thus, giving a lower bound.

Using the adversary SDP, we manage to solve (or improve on) a number of long-standing open problems.  In \rf(chp:cert), we reduce the quantum query complexity of the triangle detection from $O(n^{13/10})$ that we saw in \rf(thm:walkTriangle), and prove that the quantum query complexity of the $k$-sum problem is $\Omega(\vars^{k/(k+1)})$.  In \rf(chp:kdist), we reduce the quantum query complexity of the $k$-distinctness problem from $O(\vars^{k/(k+1)})$, that was obtained in \rf(cor:walkKDist), to $o(\vars^{3/4})$.

In \rf(chp:cert), we give a unified approach based on the new notion of certificate structures.  
In this settings, we only consider possible certificates of the function, and ignore everything else about the function.
We obtain results of a similar flavour as in \rf(chp:lower):  We formulate a primal and the corresponding dual optimisation problems for the quantum query complexity of a certificate structure.  We apply them for some problems, like $k$-sum and triangle-sum.
\rf(chp:kdist) features results that use similar techniques as in \rf(chp:cert) but do not fall into the framework of certificate structures.

\chapter{Certificate Structures}
\label{chp:cert}
\input{_certificates}

\chapter{Further Applications of Learning Graphs}
\label{chp:kdist}
\input{_kdist}

\part{Time-Efficient Implementations}
\chapter*{Overview of Part III}
\markboth{Overview of Part III}{Overview of Part III}
%\addcontentsline{toc}{section}{Overview of Part III}
In Part II, we used the algorithm from \rf(thm:advAlgorithm) to come up with a number of quantum query algorithms.  
Their time complexity was left out of our consideration.
In this part of the thesis, we analyse whether similar tools can be used in development of time-efficient algorithms.

One of the drawbacks of the algorithm from \rf(thm:advAlgorithm) is that it is notoriously hard to implement time-efficiently.  Thus, we require different techniques.  We already know one of them from \rf(sec:spanPrograms):  These are span programs.  Recall that in \rf(sec:advAlgorithms) we introduced two algorithms: the aforementioned algorithm based on the dual adversary SDP (\rf(thm:advAlgorithm)), but also an algorithm based on span programs (\rf(thm:spanAlgorithm)).  Taking \rf(thm:spanCanonical) into account, we can represent the relation between these two algorithms as in \rf(fig:relationSpanAndAdv).

{
\renewcommand{\thefigure}{\arabic{outsidefigures}}
\begin{figure}[h]
\[
\xygraph{!{0;<6.5pc,0pc>:<0pc,4pc>::}
(:[u(3.3)] [l(.4)]*\txt{Generality},
:[r(4)] [d(.4)]*{\txt{Size of the \\input alphabet}},
[r]([u(.1)]-[d(.2)] [d(.2)]*\txt{Binary})
[rr]([u(.1)]-[d(.2)] [d(.2)]*\txt{Arbitrary})
)
[ur] *\txt{Canonical \\Span Programs\\ = \\ Dual Adversary for\\ Boolean function}
([u(2)] *\txt{Span Programs},
[r(2)] *\txt{Dual Adversary},
[u(.8)] *\xycircle(.8, 1.7){},
[r(.8)] *\xycircle(1.7, .9){},
[ur(2)] *{\bullet}
)}
\]
\vspace{-1.5cm}
\refstepcounter{outsidefigures}
\caption{Relation between span programs and the dual adversary SDP.}%
\label{fig:relationSpanAndAdv}%
\end{figure}
}

In \rf(chp:claws), we only consider Boolean functions.  For them, we use the technique of span programs.  We obtain improvements both in the query and the time complexities.  Span programs also provide an intuitive way of formulating and analysing the algorithm.  The corresponding adversary SDP, obtained using \rf(thm:spanCanonical), would be much less intuitive.

By a careful examination of the proofs of Theorems~\ref{thm:spanAlgorithm}	 and~\ref{thm:advAlgorithm}, one can formulate a computational model combining both flexibility of span programs and ability to work with non-Boolean alphabet.  This would correspond to the place in \rf(fig:relationSpanAndAdv) marked by the black dot.  However, this is not the way we proceed.  We find it more convenient to base algorithms on direct application of the Effective Spectral Gap \rf(lem:effective).  

In \rf(chp:electric), we obtain a new variant of a Szegedy-type quantum walk.  Using it, we are able to implement learning graphs from \rf(chp:cert) directly as a quantum walk, bypassing the dual adversary SDP and span programs.  The walk is very similar to the original quantum walk by Ambainis (\rf(prp:distinctness)).
It turns out to be very intuitive, so we managed to obtain a quantum algorithm for the 3-distinctness problem having time complexity equal to the query complexity of the algorithm from \rf(thm:kdist), up to logarithmic factors.

\chapter{Graph Properties}
\label{chp:claws}
\input{_claws}

\chapter{Electric Networks and Quantum Walks}
\label{chp:electric}
\input{_electric}

\chapter*{Conclusion}
\addcontentsline{toc}{chapter}{Conclusion}
\input{_conclusion}

\clearpage
\addcontentsline{toc}{chapter}{Bibliography}

\pagestyle{plain}
\bibliographystyle{../habbrvN}
\bibliography{../bib}

\appendix
\chapter{Technical Results}
\input{appendix}

\end{document}

%% file: _MainMacros.tex
\draft{\usepackage{easy-todo}}
%\release{\newcommand{\todo}[1]{}}

% ????
%\usepackage{todonotes}
%\newcommand{\bigtodo}[1]{\todo[inline]{#1}}
%\usepackage[notref,notcite]{showkeys}

\draft{\usepackage[notref,notcite]{showkeys}}

\usepackage{graphicx}

\usepackage{amsfonts}
\usepackage{amssymb}
\usepackage{amsmath}
\usepackage{amsthm}
\usepackage{latexsym}

\bigone{\newtheorem{thm}{Theorem}[chapter]}
\smallone{\newtheorem{thm}{Theorem}}
\newtheorem{lem}[thm]{Lemma}
\newtheorem{prp}[thm]{Proposition}
\newtheorem{cor}[thm]{Corollary}
\newtheorem{clm}[thm]{Claim}

\theoremstyle{definition}
\newtheorem{rem}[thm]{Remark}
\newtheorem{obs}[thm]{Observation}
\newtheorem{defn}[thm]{Definition}
\newtheorem{exm}[thm]{Example}

\newcommand{\refthm}[1]{Theorem~\ref{thm:#1}}
\newcommand{\reflem}[1]{Lemma~\ref{lem:#1}}
\newcommand{\refprp}[1]{Proposition~\ref{prp:#1}}

\newcommand{\refclm}[1]{Claim~\ref{clm:#1}}

\newcommand{\refrem}[1]{Remark~\ref{rem:#1}}

\newcommand{\refdefn}[1]{Definition~\ref{defn:#1}}

\newcommand{\refsec}[1]{Section~\ref{sec:#1}}

\newcommand{\reffig}[1]{Figure~\ref{fig:#1}}
\newcommand{\refeqn}[1]{(\ref{eqn:#1})}
\newcommand{\reftbl}[1]{Table~\ref{tbl:#1}}
\newcommand{\refalg}[1]{Algorithm~\ref{alg:#1}}

\newcommand{\refline}[1]{Line~\ref{line:#1}}

\def \rf(#1:#2){\csname ref#1\endcsname{#2}}

\newcommand{\pfstart}{\begin{proof}} 
\newcommand{\pfsketch}{\begin{proof}[Proof sketch]}
\newcommand{\pfend}{\end{proof}} 
\newcommand{\itemstart}{\begin{itemize}\itemsep0pt}
\newcommand{\itemend}{\end{itemize}}
\newcommand{\descrstart}{\begin{description}\itemsep0pt}
\newcommand{\descrend}{\end{description}}
\newcommand{\enumstart}{\begin{enumerate}\itemsep0pt}
\newcommand{\enumend}{\end{enumerate}}

\newcommand{\C}{{\mathbb{C}}}
\newcommand{\N}{{\mathbb{N}}}
\newcommand{\R}{{\mathbb{R}}}
\newcommand{\Z}{{\mathbb{Z}}}
\renewcommand{\S}{{\mathbb{S}}}

% \ii = sqrt{-1}
\newcommand{\ii}{\mathsf{i}}
% e = 2.7....
\newcommand{\ee}{\mathsf{e}}
\newcommand{\eps}{\varepsilon}

\usepackage{xspace} 
\newcommand{\etal}{{\em et al.}\xspace}

\newcommand{\ignore}[1]{}

\def \s[#1]{\left(#1\right)}
\def \sA[#1]{\bigl(#1\bigr)}
\def \sB[#1]{\Bigl(#1\Bigr)}
\def \sC[#1]{\biggl(#1\biggr)}
\def \sD[#1]{\Biggl(#1\Biggr)}

\def \sk[#1]{\left[#1\right]}
\def \skA[#1]{\bigl[#1\bigr]}
\def \skB[#1]{\Bigl[#1\Bigr]}
\def \skC[#1]{\biggl[#1\biggr]}
\def \skD[#1]{\Biggl[#1\Biggr]}

\def \abs|#1|{\left| #1\right|}
\def \absA|#1|{\bigl|#1\bigr|}
\def \absB|#1|{\Bigl|#1\Bigr|}
\def \absC|#1|{\biggl|#1\biggr|}
\def \absD|#1|{\Biggl|#1\Biggr|}

\def \norm|#1|{\left\| #1\right\|}
\def \normA|#1|{\bigl\| #1\bigr\|}
\def \normB|#1|{\Bigl\| #1\Bigr\|}
\def \normC|#1|{\biggl\| #1\biggr\|}
\def \normD|#1|{\Biggl\| #1\Biggr\|}
\def \normS|#1|{\| #1\|}

\def \normFrob|#1|{\norm|#1|_{\mathrm{F}}}
\def \normtr|#1|{\norm|#1|_{\mathrm{tr}}}

\def \sfig#1{\left\{#1\right\}}
\def \sfigA#1{\bigl\{#1\bigr\}}
\def \sfigB#1{\Bigl\{#1\Bigr\}}

\def \floor[#1]{\lfloor #1 \rfloor}
\def \ceil[#1]{\lceil #1 \rceil}

\def \elem[#1]{[\![#1]\!]}

\def \ket#1|#2>%
{\ifx&#1&%#1 is empty
|#2\rangle
\else
|#2\rangle_{\mathsf #1}
\fi}

\def \ketA#1|#2>%
{\ifx&#1&%#1 is empty
\bigl|#2\bigr\rangle
\else
\bigl|#2\bigr\rangle_{\mathsf #1}
\fi}

\def \ip<#1,#2>{\langle #1, #2\rangle}
\def \ipA<#1,#2>{\bigl\langle #1, #2\bigr\rangle}
\def \ipB<#1,#2>{\Bigl\langle #1, #2\Bigr\rangle}
\def \ipC<#1,#2>{\biggl\langle #1, #2\biggr\rangle}
\def \ipD<#1,#2>{\Biggl\langle #1, #2\Biggr\rangle}

\newcommand{\polylog}{\mathop{\mathrm{polylog}}}
\newcommand{\tr}{\mathop{\mathrm{tr}}}

\usepackage{silence}
\WarningFilter{latex}{Marginpar on page}
\WarningFilter{pdftex}{}

% My commands
%\def\mycommand#1#2{%
%\draft{\marginpar{\fbox{\sf{#1:} $#2$}}}%
%\expandafter\def\csname#1\endcsname{#2}%
%}

\def\mycommand#1#2{%
\draft{\marginpar{\fbox{\sf{#1:} $#2$}}}%
\expandafter\newcommand \csname#1\endcsname {#2}%
}
\def\remycommand#1#2{%
\draft{\marginpar{\fbox{\sf{#1:} $#2$}}}%
\expandafter\renewcommand \csname#1\endcsname {#2}%
}

\newcommand{\myincludegraphics}[2]{\release{\includegraphics[width=#1]{#2}}}

%Stuff for the bibliography
\draft{}
\release{}

%% file: _MyAlgorithm.tex
\smallone{\usepackage{algorithm}}
\bigone{\usepackage[chapter]{algorithm}}

\newcounter{asdf}
\newcounter{nomer}
\newlength{\otstup}
\newlength{\ostatok}
\newlength{\nomerwidth}
\otstup = .7cm
\settowidth{\nomerwidth}{99}
\newcommand{\algbegin}{\setcounter{asdf}{0}\setcounter{nomer}{0}}
\newcommand{\algend}{\vspace{-2ex }}

%for hyperref
\release{}

\newcommand{\state}[1]{%
\refstepcounter{nomer}%
\makebox[\nomerwidth][r]{\footnotesize\arabic{nomer}:}\hspace{\value{asdf}\otstup}\hspace{1mm}%
\ostatok = \textwidth%
\addtolength{\ostatok}{-\value{asdf}\otstup}%
\addtolength{\ostatok}{-2\nomerwidth}%
\parbox[t]{\ostatok}{#1\hspace{\stretch{1}}}%
\\%
}

\newcommand{\tab}{\addtocounter{asdf}{1}}
\newcommand{\untab}{\addtocounter{asdf}{-1}}

%% file: _intro.tex
\markboth{\bf Introduction}{\bf Introduction}

\paragraph{Quantum computing}
The Church-Turing thesis asserts that any physically admissible computational device can be simulated by the Turing machine.  
This is a kind of statement that is hard to prove, because the notion of a physically admissible computational device is not even well-defined.  Nonetheless, the Church-Turing thesis is widely believed to be true.
In practice, this means that one device suffices to solve all computational problems: a general-purpose computer.

A stronger form of the Church-Turing thesis asserts that this simulation is efficient: If the computation requires $n$ elementary operations on some hypothetical physically admissible device, then it can be performed by the Turing machine in time polynomial in $n$.  The stronger form of the thesis seems plausible for computational devices based on classical laws of physics, but it seems to fail for quantum mechanics.

In 1980, Feynman~\cite{feynman:quantumComputer} proposed a general-purpose quantum computer as a tool for simulating quantum physics.  With such a computer at hand, it would be possible to efficiently simulate all processes in quantum mechanics regardless their nature.  Clearly, this device would have a vast range of applications.

But this situation can be observed from a different perspective.  If a general-purpose quantum computer outperforms classical computers in the task of simulating quantum physics, can it be more efficient for other computational problems as well?  The field of quantum computation deals with this question.  Initially, this was a very narrow area of research, featuring speed-ups for some esoteric problems like the Deutsch-Jozsa problem~\cite{deutsch:jozsa} and the Simon's problem~\cite{simon:simon}.  The situation changed dramatically after the discovery of polynomial (in the number of bits) quantum algorithms~\cite{shor:factoring} for integer factorisation and discrete logarithm by Shor in 1994.  For comparison, the best known classical algorithm is the general number field sieve~\cite{lenstra:nfs} that has complexity $\ee^{\tilde O(n^{1/3})}$, where $n$ is the number of bits of the number being factorised.  One year later, in 1995, Grover discovered a quantum algorithm for the OR function~\cite{grover:search}.  Although Grover's algorithm gives a mere quadratic speed-up, the scope of possible applications is much broader.

Understanding the power and limitations of quantum computation is a task of great practical importance.  Quantum computing seems to be at the very boundary of what Nature, as we understand it now, allows us to compute efficiently.  Discovery of new algorithms may result in practical tasks performed more efficiently.  What is even more important, computational problems that are infeasible for quantum computers can serve as a solid cornerstone of future cryptography.  
Most of modern cryptography is based on the RSA and elliptic curve algorithms and is vulnerable to quantum computers.

\begin{wrapfigure}{R}{0pt}
\draft{asdf}
\release{\includegraphics[width=5cm]{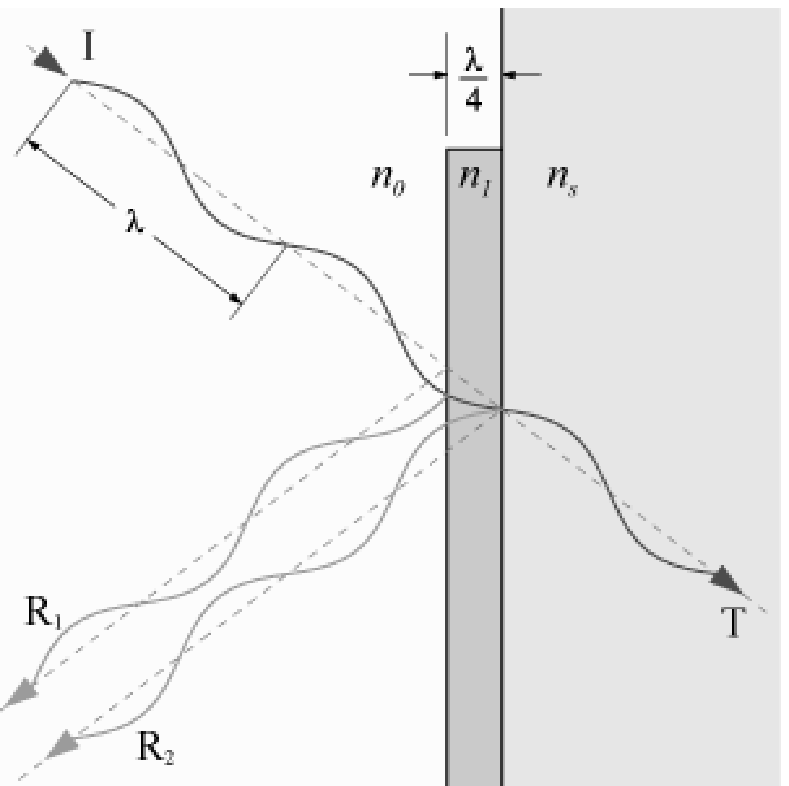}}
\caption{Anti-reflective coating. (Figure courtesy of Wikipedia.)}
\label{fig:coating}
\end{wrapfigure}

\paragraph{On quantum speed-ups}
Quantum computation is similar to randomised computation.  The difference is that the probabilities are replaced by {\em amplitudes}.  The amplitude is a complex number, and the square of its absolute value gives the probability.  While the randomised computation is linear in the probabilities, the quantum computation is linear in the amplitudes.  This gives two main sources of speed-ups.

Firstly, amplitudes may be negative.  Combined with a positive amplitude of the same absolute value, they annihilate, resulting in the zero probability of the corresponding outcome.  This kind of effect is widely used in practice.  For example, consider a glass lens in an optical system.  When a ray of light enters the lens, part of the light is reflected back.  This reflection, described by Fresnel equations, is inevitable.  However, by putting a thin film on the top of the lens, it is possible to achieve that the light reflected from the surface of the film has the amplitude opposite to the light reflected from the surface of the glass (see \rf(fig:coating)).  Then, both rays of light cancel out, and, as a result, more light passes through the lens.  This technology is known as anti-reflective coating, and is used in every camera nowadays.

Quantum algorithms aim to use similar cancelling techniques:  The amplitudes of outcomes that are not interesting are mutually cancelled, that boosts the probability of observing the outcome of interest.  At one step, the success probability may be boosted by a constant factor, that gives an exponential speed-up.  Main examples of such speed-ups are given by the quantum Fourier transform, and the algorithms based on it.  This includes the aforementioned Simon's and Shor's algorithms.  However, in order to ensure proper cancelling, the problem must have a lot of structure.

The second reason for quantum speed-ups is that quantum computation is linear in the square roots of probabilities.  For instance, suppose we have $n$ elements, one of which is marked.  Our goal is to find the marked element.  We assign equal amplitudes of $1/\sqrt{n}$ to all of them.  When observed, this still gives the probability $1/n$ for each element.  But assume that we can add the average amplitude of the states to the marked element.  This is a linear transformation.  After $\sqrt{n}$ additions, we obtain the amplitude approximately equal to 
\[
\overbrace{\frac{1}{\sqrt{n}} + \frac{1}{\sqrt{n}} + \cdots + \frac{1}{\sqrt{n}}}^{\mbox{$\sqrt{n}$ times}} = 1.
\]
This gives a very vague explanation of why Grover's algorithm attains quadratic speed-up.  We will cover it in more detail in \rf(sec:amplitudeAmplification).

In the thesis, we mostly consider the second type of quantum speed-ups, analysing how the linearity in the square roots of probabilities may be translated into a computational speed-up.  Compared to the speed-ups of the first type, this research is mostly theoretical. If modestly-sized quantum computers are constructed in the near future, it is unlikely they will be able to outperform classical computers using only a quadratic speed-up.
However, the field of quantum computation in general may benefit from this type of research.  By studying quantum computation for general, unstructured problems, we understand its fundamental properties.  And this may help in the construction of quantum algorithms for structured problems.  Another reason is that we can prove strong lower bounds in these settings.
%It is of fundamental importance even now:  If we know that the computational problems behind our cryptographic protocols are infeasible even for quantum computers, we may not worry that they suddenly become vulnerable after the development of a working quantum computer.

\paragraph{Query algorithms}
Query complexity is the main point of interest in the thesis.  That is, we assume that all computational operations besides accessing the input string are free of charge.  A large part of quantum algorithms are developed in these settings.  This may seem as an arbitrary assumption at first, but there are a number of reasons to study query complexity.

Firstly, proving strong unconditional lower bounds on computational problems is very hard even in the deterministic settings.  For instance, there is still no super-linear circuit lower bound known for any problem in $\mathsf{NP}$.  Considering only accesses to the input string makes the problem more accessible, and we are able to prove some tight lower bounds.  As is common in science, there is a hope that intuition gathered in the study of this simplified problem will be of help for the general problem.

Secondly, query problems that are feasible for quantum computers but are not for the classical ones, indicate the potential source of quantum speed-ups.  The corresponding problems may be either implemented time-efficiently, as we do it in Part III of the thesis, or serve as a subroutine for other computational problems.  For instance, the query problem of period finding serves as the main building block in the time-efficient Shor's factoring algorithm; Grover's algorithm for the OR function provides a generic quadratic speed-up for any algorithm based on exhaustive search.

Quantum query complexity is a popular area of research.  As we will see in \rf(chp:lower), in some aspects, it is studied better than randomised query complexity.  In particular, Reichardt~\cite{reichardt:spanPrograms, reichardt:advTight} has shown that a relatively simple optimisation problem, the (general) adversary bound~\cite{hoyer:advNegative}, gives a tight (up to a constant factor) characterisation of quantum query complexity.  The adversary bound is a semi-definite optimisation problem that comes in two forms: the primal and the dual.  Any feasible solution to the primal problem yields a lower bound on the quantum query complexity, whereas any feasible solution to the dual problem can be converted into a quantum query algorithm.  Strong semi-definite duality implies that the optimal values of the two problems coincide.  
For instance, using this technique, it is possible to obtain tight quantum query algorithms for iterated functions (see \rf(sec:advExamples)).  No such result is known for randomised query complexity.

\paragraph{Brief description of the results}
The main problem of the thesis can be formulated as follows:
%\begin{quote}
Is it possible to construct new quantum query algorithms using the dual adversary bound?
%\end{quote}
By the results of Reichardt, we know that any quantum query algorithm admits a description in terms of the dual adversary bound.  But few explicit examples were known at the time we started our work on this problem.

%The motivation behind this research is as follows.  The adversary bound is a semi-definite optimisation problem that can be approached using techniques of semi-definite optimisation.  Next, we need not write down the algorithm explicitly and study its internal organisation: we are only interested in the value of the objective function.  Finally, a different description of even a known quantum algorithm may provide more intuition and reveal new sources of possible optimisation.

Most quantum algorithm utilising the second sort of quantum speed-up (as described above) are based on quantum walks, and, in particular, on quantum walks on the Johnson graph that we describe in \rf(chp:walk).  Based on the adversary bound, we developed a new computational framework of learning graphs.  Learning graphs are more flexible than quantum walks on the Johnson graph.  They require simple combinatorial reasoning, in contrast to general quantum walks that require spectral analysis of the underlying graph.  Using learning graphs, we improved quantum query complexity of various problems such as triangle detection and $k$-distinctness.  Using similar techniques, we developed optimal quantum query algorithms for detecting whether the input graph contains a path or a subdivision of the claw of fixed size.

In process of this research, two additional lines of research emerged.  The first one is to convert the dual adversary bounds we have constructed into the corresponding primal adversary bounds, and, hence, lower bounds on quantum query complexity.  This task can be approached using semi-definite duality, but it is not trivial, because the dual solutions we have obtained are not optimal: They are off by a constant factor.  In this way, we managed to obtain tight lower bounds for the $k$-sum and the triangle-sum problems.

Another task is to obtain time-efficient implementations of the constructed algorithms.  The dual adversary bound only gives a query-efficient algorithm.  But, if the solution is sufficiently uniform, it is sometimes possible to implement the algorithm time-efficiently.
In order to do this, we had to drop the adversary SDP and use related techniques such as span programs and the effective spectral gap lemma.  We obtained time-efficient implementations for the path- and claw-detection problems, as well as for 3-distinctness.

\paragraph{Publications}
The results of the research within the thesis are reflected in the following publications:
\begin{enumerate}
\item 
Aleksandrs Belovs.\\
Span Programs for functions with constant-sized 1-certificates.\\
In {\em Proceedings of the 44th ACM Symposium on Theory of Computing (STOC 2012)},
pages 77---84, 2012.

\item
Aleksandrs Belovs and Ben W. Reichardt.\\
Span programs and quantum algorithms for $st$-connectivity and claw detection.\\
In {\em Proceedings of the 20th Annual European Symposium on Algorithms (ESA 2012)},
volume 7501 of {\em Lecture Notes in Computer Science},
pages 193---204, Springer, 2012.

\item
Aleksandrs Belovs.\\
Learning-graph-based quantum algorithm for $k$-distinctness.\\
In {\em Proceedings of the 53rd Annual Symposium on Foundations of Computer Science (FOCS 2012)},
pages 207---216, 2012.

\item
Aleksandrs Belovs and Robert \v Spalek.\\
Adversary lower bound for the $k$-sum problem.\\
In {\em Proceedings of the 4th Innovations in Theoretical Computer Science conference (ITCS 2013)},
pages 323---328, 2013.

\item
Aleksandrs Belovs and Ansis Rosmanis.\\
On the power of non-adaptive learning graphs.\\
In {\em Proceedings of the 28th IEEE Conference on Computational Complexity (CCC 2013)},
pages 44---55, 2013.\\
Best student paper award.

\item
Aleksandrs Belovs, Andrew M. Childs, Stacey Jeffery, Robin Kothari and Fr\'ed\'eric Magniez.\\
Time-efficient quantum walks for 3-distinctness.\\
In {\em Proceedings of the 40th International Colloquium on Automata, Languages and Programming (ICALP 2013), Part I},
volume 7965 of {\em Lecture Notes in Computer Science},
pages 105---122, Springer, 2013.
\end{enumerate}

\noindent
The results were presented by the author at the following international conferences and workshops:
\begin{enumerate}
\item Quantum Computer Science (QCS) Project Workshop, Riga, Latvia, May 2011.\\
Presentation: {\em Span programs for functions with constant-sized 1-cer\-ti\-fi\-ca\-tes.}

\item The 15th Quantum Information Processing workshop (QIP 2012), Montreal, Canada, December 2011.\\
Plenary Lecture: {\em Span programs for functions with constant-sized 1-cer\-ti\-fi\-ca\-tes.}

\item Recent Progress in Quantum Algorithms Workshop, Waterloo, Canada, April 2012.\\
Presentation: {\em Quantum algorithms for the $k$-distinctness problem.}

\item The 44th ACM Symposium on Theory of Computing (STOC 2012), New York, USA, May 2012.\\
Presentation: {\em Span programs for functions with constant-sized 1-cer\-ti\-fi\-ca\-tes.}

\item Device-Independent Quantum Information Processing (DIQIP) \& Quantum Computer Science (QCS) Joint Meeting, Castelldefels, Spain, June 2012.\\
Presentation: {\em Adversary lower bound for the k-sum problem.}

\item The 20th Annual European Symposium on Algorithms (ESA 2012), Ljubljana, Slovenia, September 2012.\\
Presentation: {\em Span programs and quantum algorithms for $st$-con\-nec\-ti\-vity and claw detection.}

\item The 2nd Joint Estonian-Latvian theory days, Medz\=abaki, Latvia, October 2012.\\
Presentation: {\em Learning graphs and quantum query algorithms.}

\item The 53rd Annual Symposium on Foundations of Computer Science (FOCS 2012), New Brunswick, USA, October 2012.\\
Presentation: {\em Learning-graph-based quantum algorithm for $k$-dis\-tinct\-ness.}

\item The 16th Quantum Information Processing workshop (QIP 2013), Beijing, China, January 2013.\\
Presentation: {\em Adversary lower bound for the $k$-Sum problem.}

\item The 16th Quantum Information Processing workshop (QIP 2013), Beijing, China, January 2013.\\
Presentation: {\em Learning-graph-based quantum algorithm for $k$-dis\-tinct\-ness.}

\item Quantum and Crypto Day 2013, Riga, Latvia, April 2013.\\
Presentation: {\em Proving lower bounds for quantum algorithms.}

\item Quantum Computer Science (QCS), Device-Independent Quantum Information Processing (DIQIP) \& Quantum Algorithmics (QAlgo) Joint Meeting, Paris, France, May 2013.\\
Presentation: {\em Negative-weight adversaries for element distinctness and beyond.}

\item The 40th International Colloquium on Automata, Languages and Programming (ICALP 2013), Riga, Latvia, July 2013.\\
Presentation: {\em Time-efficient quantum walks for 3-dis\-tinct\-ness.}

\end{enumerate}

\paragraph{Organisation of the thesis}
The thesis is divided into three parts.  In Part I, we describe the previous results our thesis is built on.  In \rf(chp:model), we define the model of quantum computation, develop some basic tools, and define the important notion of query complexity.  In \rf(chp:walk), we overview quantum algorithms based on quantum walks.  The chapter gives some tools used later the thesis, like the quantum phase detection subroutine, but mainly describes previous algorithms we improve on in the next parts of the thesis.  \rf(chp:lower) is the main chapter of the first part.  In this chapter, we overview main techniques for proving lower bounds on quantum query complexity: the polynomial and the adversary methods.
The adversary method is the main technical tool used in the second part of the thesis, and many ideas from \rf(chp:lower) will be used in Part III.

Parts II and III of the thesis contain original research.  
In Part II, we grouped algorithms that attain improvement only in the query complexity settings.
Part III features algorithms for which we are able to develop time-efficient implementations.

Part II contains Chapters~\ref{chp:cert} and~\ref{chp:kdist}.  In \rf(chp:cert), we consider algorithms based only on the certificate structure of the problem.  We define the computational model of a learning graph that can be converted into a quantum query algorithm and apply it to problems like triangle detection and associativity testing.  Also, we prove tight lower bounds for the $k$-sum and the triangle-sum problems.  In \rf(chp:kdist), we describe algorithms beyond the certificate structure framework.  We obtain new quantum query algorithms for the $k$-distinctness problem and a special case of the graph collision problem.

Part III contains Chapters~\ref{chp:claws} and~\ref{chp:electric}.  In \rf(chp:claws), we use span programs to develop optimal quantum algorithms for the $st$-connectivity problem, as well as path and claw detection.  In \rf(chp:electric), we develop a more flexible variant of quantum walks and use it to construct a time-efficient quantum algorithm for the 3-distinctness problem.

Finally, Appendix lists some basic technical results from linear algebra and semi-definite optimisation we use in the thesis.

\newcounter{outsidefigures}
\newcounter{tempfigures}
\setcounter{outsidefigures}{\value{figure}}

%% file: _model.tex
\mycommand{phase}{\mathsf{phase}}

In this chapter, we describe our quantum computational model used throughout the thesis.  It is obtained by ``quantisation'' of the corresponding deterministic model, similarly to the randomised model that is obtained by randomisation.  Throughout the chapter, we consider all these three models, stressing similarities and differences between them.  We assume familiarity with deterministic and randomised computation.

In \rf(sec:computationalDevices), we describe the mathematical model of the state of a computational device, and, in \rf(sec:evolution), we describe what kind of operations can be applied to these states.  This is enough for the large part of the thesis dealing with query complexity.  However, for some tasks, we will be interested in operations that can be performed efficiently.  Sections~\ref{sec:circuits} and~\ref{sec:RAM} deal with this question: The first one describes a popular model of quantum circuits, and the second one deals with a more sophisticated model of quantum RAM.  Also, in \rf(sec:RAM), we introduce the pseudocode notation we use for the description of quantum algorithms.  \rf(sec:specification) deals with the computational process from another perspective: the task being solved.  We describe what it means for a function to be computed by a quantum computer.  Also, we list different variants of quantum subroutines, their specifications, and mutual relationships.  

\rf(sec:query) introduces the notion of query complexity that is of fundamental importance for the whole thesis.  This is a simplified notion of computational complexity where only the number of accesses to the input string is counted, whereas all other operations are considered free.  Also, the section defines the important notion of certificate complexity and defines a list of functions we will use later in the thesis.

\section{Computational Devices}
\label{sec:computationalDevices}

%\subsection{Registers}
%\label{sec:state}
A computational device is modelled by a {\em register}.  The register stores one of the possible {\em states} of the device.  The description of the state depends on the model of computation.  In the {\em deterministic} settings, the register stores one of a {\em finite} number of states.  For simplicity, we label them by consecutive integers from the set
\begin{equation}
\label{eqn:stateDet}
[n] = \{1,2,\dots,n\}.
\end{equation}
This model can be used to describe usual computational devices like abacus or computers.
We denote registers by Latin letters in sans serif font like $\reg A$ or $\reg i$.

Let $\reg A$ be a register storing one of the $n$ elements from~\refeqn{stateDet}.  In the corresponding {\em randomised} or {\em quantum} register, the state is described by a vector in the $n$-dimensional inner product space having the elements in~\refeqn{stateDet} as its orthonormal basis.  We use $e_i$ to denote the element of the basis corresponding to the element $i$.  The vector space is called the {\em state space} and denoted by $\hilbert A$.  The set of elements in~\refeqn{stateDet} is called the {\em computational basis}.  Sometimes, we also call them {\em classical states}.

If $\psi$ is a vector in $\hilbert A$, we use $\ket A|\psi>$ to denote that $\reg A$ stores the state $\psi$.  
In the current and the next sections, we use this notation for both randomised and quantum states, to emphasise the similarity between the two models.  Everywhere else, this notation is strictly reserved for quantum states.
If the subindex is clear from the context, we omit it.
As we will see in \rf(sec:evolution), many available transformations are quantisations of the corresponding deterministic operations performed to the classical states.  Thus, in contrary to customary linear algebra, elements of the standard (computational) basis are of special importance.  Because of that, we use a short-hand of $\ket|i>$, instead of, say, $\ket|e_i>$, to denote the $i$-th element of the computational basis.  By linearity, a state of $\reg A$ can be written as a {\em superposition} of the classical states:
\begin{equation}
\label{eqn:state}
\ket|\psi> = \alpha_1\ket|1> + \alpha_2\ket|2> + \cdots + \alpha_n\ket|n>.
\end{equation}

The main difference between the randomised and quantum models is the requirements on $\alpha_i$s.  In the randomised case, $\alpha_i$ is the {\em probability} of $\reg A$ being in state $i$. Thus, they are real and satisfy the following condition:
\begin{equation}
\label{eqn:requirementProb}
\alpha_1 + \alpha_2 + \cdots + \alpha_n = 1,\qquad\text{and}\qquad \alpha_i\ge 0\quad\mbox{for all $i\in [n]$}.
\end{equation}
In the quantum case, $\alpha_i$ is the {\em amplitude} of the state $\ket|i>$.
Each $\alpha_i$ is a complex number, and they satisfy the following condition:
\begin{equation}
\label{eqn:requirementQuant}
|\alpha_1|^2 + |\alpha_2|^2 + \cdots + |\alpha_n|^2 = 1.
\end{equation}

A quantum register having $[n]$ as its computational basis is called an {\em $n$-qudit}.  As a special case, $2$-qudits are called {\em qubits}.

Many computational devices, like the Turing machine, assume potential infinity in the number of states, thus
assuring the ability to solve problems of arbitrary unbounded size.  
We, however, stick to the point of view adopted in the circuit model: For each size of the problem, its own computational device of finite size is constructed.
Some notion of {\em uniformity} is then required: There must exist an algorithm that, given the size of the problem, produces the description of the computational device.  For most of the upper bounds in the thesis, this is the case.  Lower bounds, however, will be proven without any uniformity assumptions.

\paragraph{Composed system}
Until now, we focused on a single computational system.  However, a system may have more complicated structure.
The state may be composed from the states of some number of subsystems.  In this part of the section, we study the relation between the state of a composite system, and the states of its subsystems.  

Without loss of generality, it is enough to consider the case of two subsystems.
Let $\reg A$ and $\reg B$ be two deterministic registers with the set of states $[n]$ and $[m]$, respectively.
The system composed by $\reg A$ and $\reg B$ is denoted by $\reg{AB}$.  It is easy to see that the states of $\reg{AB}$ are
\begin{equation}
\label{eqn:stateComposed}
\sfig{ (i,j) \mid i\in [n],\; j\in [m] \strut}.
\end{equation}

The randomised and the quantum cases may be obtained using the general transformation rule.
Applying it here, we get that the state space $\reg{AB}$ has the elements in~\refeqn{stateComposed} as its standard basis.  In other words, $\hilbert {AB} = \hilbert A\otimes \hilbert B$.  If $\ket A|\psi>$ and $\ket B|\phi>$ are the states of the individual systems, we write $\ket A|\psi>\ket B|\phi>$ instead of $\ket{{AB}}|\psi\otimes \phi>$.

Randomised and quantum composed systems exhibit a new property compared to the deterministic ones.  While a state of the composed deterministic system~\refeqn{stateComposed} can be always represented as a composition of the states of the individual subsystems, this is not always the case in the randomised or quantum settings.  For example, the following quantum state
\begin{equation}
\label{eqn:entangle}
\frac{1}{\sqrt{2}} \ket A|1>\ket B|1> + \frac{1}{\sqrt{2}} \ket A|2>\ket B|2>.
\end{equation}
cannot be decomposed into $\ket A|\psi>\ket B|\phi>$ for any $\psi\in\hilbert A$ and $\phi\in\hilbert B$.  Such quantum states are called {\em entangled}, whereas states of the form $\ket A|\psi>\ket B|\phi>$ are called {\em separable}.	 This is akin to random variables being dependent or independent.

\section{Evolution}
\label{sec:evolution}
In the previous section, we considered a static point of view, and described how a state of a computational device looks like.  In this section, we adopt a dynamic point of view, and describe how a state can change in time. There are two possibilities:  Either the system undergoes a transformation, or the state of the system is observed.

\paragraph{Transformations}
Let the system be like in \rf(sec:computationalDevices).  In the deterministic case, any function from $[n]$ to $[n]$ is a valid transformation.  In the randomised or quantum case, when a computational system evolves {\em without interaction with other systems}, the state of the system~\refeqn{state} undergoes a linear transformation $\psi\mapsto A\psi$ that preserves the corresponding requirements~\refeqn{requirementProb} or~\refeqn{requirementQuant}.  

The classes of transformations preserving these properties are well-known.  In the randomised case, the requirements are preserved if and only if $A$ is a {\em stochastic} matrix, i.e., a matrix with non-negative entries with each column summing up to 1.  In the quantum case, this happens if and only if $A$ is {\em unitary}, i.e., if $AA^* = I$, where $A^*$ is the complex-conjugated transposed matrix of $A$ and $I$ is the identity matrix.  
Since all unitary operations are invertible, we get the following
\begin{obs}
\label{obs:reversibility}
Any quantum transformation can be reversed.
\end{obs}
This means that quantum analogues of deterministic operations must be revised to include the reversibility property.

The condition on the system not to interact with other systems is crucial in \rf(obs:reversibility).  For instance, consider the following transformation of a composed quantum system $\reg {AB}$:
\[
\ket A|1>\ket B|1> \mapsto \ket A|1>\ket B|1>,\qquad
\ket A|1>\ket B|2> \mapsto \ket A|2>\ket B|1>,\qquad
\ket A|2>\ket B|1> \mapsto \ket A|1>\ket B|2>,\qquad
\ket A|2>\ket B|2> \mapsto \ket A|2>\ket B|2>.
\]
This is a perfectly valid quantum operation that swaps the states of registers $\reg A$ and $\reg B$.  However, both $\ket B|1>$ and $\ket B|2>$ are mapped to $\ket B|1>$ in the first two cases.  This cannot be achieved by a unitary operator if the the register $\reg A$ is not affected.

\begin{obs}
\label{obs:copy}
In general, a randomised or quantum state cannot be copied.  This means, the transformation $\ket A|\psi>\ket B|0> \mapsto \ket A|\psi>\ket B|\psi>$ that works for any $\psi\in\hilbert A$ cannot be implemented.  Here, $\ket B|0>$ is some fixed state.
\end{obs}

\pfstart
Assume we have such a transformation that copies a randomised state $\psi$, and assume $\reg A$ is at least 2-dimensional with 1 and 2 being two possible deterministic states.  The transformation must satisfy $\ket A|1>\ket B|0>\mapsto \ket A|1>\ket B|1>$ and $\ket A|2>\ket B|0>\mapsto \ket A|2>\ket B|2>$.  Let $\psi = \frac12(e_1+e_2)$.  Then, by linearity, the state $\ket A|\psi>\ket B|0>$ gets mapped to $\frac12 (\ket A|1>\ket B|1> + \ket A|2>\ket B|2>)$ that is different from $\ket A|\psi>\ket B|\psi>$.  This is a contradiction.  The quantum case is similar.
\pfend

In the case of composed systems, we sometimes use subindices to denote the register to which the transformation is applied.  For instance, for a composed register $\reg{AB}$, we write $A_{\reg A}$ for $A\otimes I = A_{\reg A}\otimes I_{\reg B}$, where $I_{\reg B}$ is the identity operator on $\reg B$.  Similarly, $A_{\reg B} = I\otimes A = I_{\reg A}\otimes A_{\reg B}$.

\begin{rem}
\label{rem:realEntries}
As any complex number can be interpreted as a vector in a real linear space of dimension 2, we may assume, if needed, that the amplitudes in \rf(eqn:requirementQuant) are real, and that all transformation are real unitary matrices (also known as orthogonal matrices).
\end{rem}

\paragraph{Measurement}
%After the device has performed some computations, the executor usually wants to get some feedback.  In our settings, this is performed by the {\em measurement} of the state of the system.  
In the deterministic case, the measurement of the state does not affect the system.  In the randomised and quantum cases, the situation is different.
Measurement falls under the scope of interactions of a system with other systems, and, hence, the requirements from the previous paragraph need not to be obeyed.  In particular, measurements are {\em not} reversible.

We consider a general case of measuring a {\em part} of the system.  Let the device be represented as a composition of two registers $\reg A$ and $\reg B$ as in \rf(sec:computationalDevices), 
where the second one is measured, and the first one is not.
Thus, prior to the measurement, the state of the system has the following form:
\[
\ket |\psi> = \sum_{i=1}^n \sum_{j=1}^m \alpha_{i,j} \ket A|i>\ket B|j>.
\]

In the randomised case, the second register is observed in the state $\ket B|j>$ with probability $p_j = \sum_{i=1}^n \alpha_{i,j}$.  Then, according to the law of conditioned probability, the system {\em collapses} to the state
\[
\frac{1}{p_j} \sum_{i=1}^n \alpha_{i,j} \ket A|i>\ket B|j>.
\]
Informally, the part of the system inconsistent with the output is removed, and the remaining part is scaled to satisfy the requirement~\refeqn{requirementProb}.

The quantum case is similar, and Condition~\refeqn{requirementQuant} gives a strong clue what the probability is:  The probability of observing the second register in the state $\ket B|j>$ is $p_j = \sum_{i=1}^n |\alpha_{i,j}|^2$.  If the outcome is $j$, the system collapses to the state
\begin{equation}
\label{eqn:afterMeasurement}
\frac{1}{\sqrt{p_j}} \sum_{i=1}^n \alpha_{i,j} \ket A|i> \ket B|j>.
\end{equation}

%As the measurement operation is an important part of quantum computation, we introduce additional notations. 
For the situation as above, let $\proj Bi$ denote the orthogonal projector $I_{\reg A}\otimes (e_ie_i^*)$ where $I_{\reg A}$ is the identity operator on $\reg A$.  In these notations, the probability of observing $i$ in register $\reg B$ is $p_i = \|\proj Bi \psi\|^2$, and the state~\rf(eqn:afterMeasurement) becomes $\proj Bi \psi / \|\proj Bi \psi\|$ after the collapse.

However, there is a crucial difference between the measurement of the state in the quantum and randomised settings.  In the randomised case, the state prior to the measurement can be easily reconstructed by merely ignoring the output.  This is no longer true in the quantum case:  If one performs a measurement and ignores the outcome, he will get a probabilistic mixture of the quantum states, but not the original state.  Such probabilistic mixtures are called {\em mixed states}.  We will not develop the corresponding formalism here.

Note that the same thing happens if the system is observed by anyone, in particular, by the environment.  This effect, known as {\em decoherence}, is the main source of noise in quantum computation and the main obstacle in the construction of large scale quantum computers.

\section{Circuits}
\label{sec:circuits}
In the previous section, we described which operations can be performed on a computational device {\em in principle}.  This is sufficient for the most part of the thesis, which deals with query complexity.  But, we will occasionally consider time complexity as well.  Therefore, we need to clarify the operations that can be performed on a quantum computer {\em efficiently}.  This section and the next one are devoted to this issue.  The {\em circuit model} we consider here is one of the most popular.  
%However, it is important to be aware of the existence of other models.

A {\em gate} $U$ is a unitary transformation applied to a register $\reg X$ of small size.  The register $\reg X$ usually is a composition of several qudits.  Assume we have fixed a set of {\em elementary gates} that is independent of the problem being solved.   

Let $R$ be the set of the registers of a computational device, and let $\reg R$ be the composition of all the register in $R$ as described in \rf(sec:computationalDevices).
An {\em application} of the gate $U$ is defined as the unitary $U_{\reg A}$, where $\reg A$ is the composition of a number of registers in $R$ equal to $\reg X$ (consisting of qudits of the same sizes and in the same order).

Circuits will be usually denoted by calligraphic capital Latin letters.
The description of a {\em circuit} $\circuit C$ consists of a set of registers $R$, and a sequence of applications of gates $P$.  The circuit $\circuit C$ defines a unitary transformation on $\reg R$ as the product of all the gate applications in $P$.  We use the same letter $\circuit C$ to denote this transformation.  The {\em size} of the circuit is the length of the sequence $P$.  We say that a unitary $U$ on $\reg R$ can be {\em implemented} by a circuit of size $s$ if there exists a circuit $\circuit C$ of size $s$ such that $\circuit C = U\otimes I$ where $I$ is the identity operation on some additional temporary register.

Let $(U_n)$ be a family of unitaries  that depend on some parameter $n=1,2,\dots$, and let $(\circuit C_n)$ be a family of circuits such that $\circuit C_n$ implements $U_n$ for all $n$.  We say that the family $(\circuit C_n)$ is {\em uniform} if there exists a deterministic Turing machine that, given $n$, outputs the description of $\circuit C_n$ in time polynomial in $n$.

We will describe two circuit models: the {\em low level} and the {\em high level}.
In the low level model, all the registers in $R$ are qubits.
In the high level model, the registers are qudits, and elementary gates are the ones that can be decomposed into a small number of gates acting on qubits.
Any high level model can be simulated by the low level model with some expenses.  High level models help to think algorithmically.

\paragraph{Low Level}
At the low level, a deterministic device is considered as consisting of {\em bits}---registers having only two possible states: 0 and 1.  The bits can be combined to store an arbitrary large finite number of states as described in \rf(sec:computationalDevices).  Elementary gates are the operations that act on a bounded number of bits.  It is known that any transformation can be decomposed into a sequence of gates acting on two bits only.  In fact, the NAND gate alone suffices.

The randomised case is similar to the deterministic one with an additional operation of {\em tossing a coin} that produces a random bit.  The coin may be {\em unbiased}, and produce both 0 and 1 with probability 1/2.  Or, it may have a {\em bias}, and produce 1 with probability $p$, and 0 with probability $1-p$ for some real $p$ between 0 and 1.  A coin of any bias may be {\em approximated} to arbitrary precision using unbiased coins.

\begin{table}[bth]
\[\begin{tabular}{|ccc|}%m{5.5cm}|}
\hline &&\\
Controlled NOT & 
{$\begin{pmatrix}
1 & 0 & 0 & 0 \\
0 & 1 & 0 & 0 \\
0 & 0 & 0 & 1 \\
0 & 0 & 1 & 0 
\end{pmatrix}$} & 
{
$\begin{aligned}
\mbox{$\kett|0>\kett|0>\mapsto \kett|0>\kett|0>,\quad \kett|0>\kett|1>\mapsto \kett|0>\kett|1>$}\\
\mbox{$\kett|1>\kett|0>\mapsto \kett|1>\kett|1>,\quad \kett|1>\kett|1>\mapsto \kett|1>\kett|0>$}
\end{aligned}$
} \\
 && \\ 
Hadamard & 
{$\displaystyle \frac1{\sqrt2}\begin{pmatrix}
1 & -1 \\ 1 & 1
\end{pmatrix}$} &
{
$\begin{aligned}
\kett|0>&\mapsto \mbox{$\frac{1}{\sqrt{2}} \kett|0> + \frac{1}{\sqrt{2}} \kett|1>$}\\
\kett|1>&\mapsto \mbox{$-\frac{1}{\sqrt{2}} \kett|0> + \frac{1}{\sqrt{2}}\kett|1>$}
\end{aligned}$
} \\
 && \\ 
$\pi/8$-gate & 
{$\begin{pmatrix}
1 & 0 \\ 0 & \ee^{\pi\ii/4}
\end{pmatrix}$} &
{
$\begin{aligned}
\kett|0> &\mapsto \mbox{$\kett|0>$}\\
\kett|1> &\mapsto \mbox{$\ee^{\pi\ii/4}\kett|1>$}
\end{aligned}$
}\\
\hline
\end{tabular}\]
\caption{Elementary quantum gates that can be used to approximate any unitary transformation to arbitrary precision.}\label{tbl:elementary}
\end{table}

The situation in the quantum case is similar to the randomised one.

\begin{thm}[Universality~\cite{reck:2qubitUniversal, boykin:universal}]
\label{thm:universality}
Any unitary operator can be represented as a circuit with gates acting on one or two qubits only.  Also, it can be approximated to arbitrary precision using the gates from \reftbl{elementary}.
\end{thm}

We give some well-known quantum gates in \reftbl{oftenGates}.  The biased coin gate is defined for any real $p$ between 0 and 1.  For $p=1/2$, this is the Hadamard transformation from \rf(tbl:elementary).  The SWAP gate can be used to exchange the content of any two registers of equal sizes.

\begin{table}[bth]
\[\begin{tabular}{|ccc|}
\hline &&\\
NOT & 
{$\begin{pmatrix}
0 & 1 \\
1 & 0 \\
\end{pmatrix}$} & 
{
$
\kett|0>\mapsto \kett|1>,\; \kett|1>\mapsto \kett|0>
$
} \\
&& \\
Swap & 
{$\begin{pmatrix}
1 & 0 & 0 & 0 \\
0 & 0 & 1 & 0 \\
0 & 1 & 0 & 0 \\
0 & 0 & 0 & 1 
\end{pmatrix}$} & 
{
$\begin{aligned}
\kett|0>\kett|0>\mapsto \kett|0>\kett|0>,\quad \kett|0>\kett|1>\mapsto \kett|1>\kett|0>\\
\kett|1>\kett|0>\mapsto \kett|0>\kett|1>,\quad \kett|1>\kett|1>\mapsto \kett|1>\kett|1>\\
\end{aligned}$
} \\
&&\\
$p$-BiasedCoin & 
{$\begin{pmatrix}
\sqrt{p} & -\sqrt{1-p}\\
\sqrt{1-p} & \sqrt{p}
\end{pmatrix}$} &
{
$
\begin{aligned}
\kett|0> &\mapsto \mbox{$\sqrt{p}\; \kett|0> + \sqrt{1-p}\; \kett|1>$}\\
\kett|1> &\mapsto \mbox{$-\sqrt{1-p}\; \kett|0> + \sqrt p\;\kett|1>$}
\end{aligned}
$
} \\
&& \\
Conditional $\ee^{\ii\varphi}$-phase & 
{$\begin{pmatrix}
1 & 0\\
0 & \ee^{\ii\varphi}
\end{pmatrix}$} &
{
$
\begin{aligned}
\kett|0> &\mapsto \mbox{$\kett|0>$}\\
\kett|1> &\mapsto \mbox{$\ee^{\ii\varphi}\kett|1>$}
\end{aligned}
$
} \\
\hline
\end{tabular}\]
\caption{Some other quantum gates}\label{tbl:oftenGates}
\end{table}

\paragraph{High Level}
We assume the registers are composed of qudits.  Each qudit can be simulated by a number of qubits:
$\lceil\log_2 q\rceil$ qubits can cimulate a $q$-qudit.  We say that a gate acting on qudits is efficient if it can be implemented as a circuit acting on qubits and having size polynomial in the number of qubits, i.e., polylogarithmic in $q$.

%As a general rule, a reversible operation that is efficient deterministically is also efficient quantumly.  For example, swapping content of two registers, or adding the content of one register to the content of another one (modulo $q$) are efficient quantum operations.  On contrary, replacing the content of a register by zero, or sorting the content of several registers in non-decreasing order are not.

We use the following general principles of constructing quantum circuits.
We refer the reader to~\cite{chuang:quantum}, for a more detailed exposition of these results.

\begin{lem}[Reversibility]
\label{lem:reversibility}
If a unitary operator $U$ can be implemented by a quantum circuit, the inverse operator $U^{-1}$ can be implemented by a quantum circuit of the same size.
\end{lem}

%\pfsketch
%This is a more formal variant of \rf(obs:reversibility).
%The set of elementary gates we consider is closed under taking inverses.  If $P$ is the sequence of gate applications in a circuit implementing $U$, then the sequence $P$ reversed and with each gate replaced by its inverse implements $U^{-1}$.
%\pfend

\begin{lem}[Quantum Simulation of Deterministic Calculation]
\label{lem:quantumOfDet}
Suppose a transformation $f\colon\{0,1\}^n\to\{0,1\}^m$ can be implemented by a deterministic circuit of size $s$ acting on classical bits.  
Let $\reg A$ and $\reg B$ be quantum registers composed of $n$ and $m$ qubits, respectively.  Then, the transformation that maps $\ket A|x>\ket B|0>$ to $\ket A|x>\ket B|f(x)>$ can be implemented by a quantum circuit of size $O(s)$ acting on qubits.
\end{lem}

%\pfsketch
%First, the result is proved for elementary deterministic gates using \rf(thm:universality).  Then, the deterministic circuit is executed by putting the outputs of each gate into fresh qubits.  At the end, the output of the circuit is copied into $\reg B$, and the whole process is reversed using \rf(lem:reversibility).
%\pfend

\rf(lem:quantumOfDet) is very important since it allows us to perform arithmetical operations, comparisons, and other elementary operations on quantum data.  We will use this lemma very often and without an explicit reference.  

Let $\reg B$ be a register, and $\reg A$ stores a qubit.  
The conditional operation is defined as the following transformation:
\[
\ket A|0>\ket B|\psi> \mapsto \ket A|0>\ket B|\psi>,\qquad\text{and}\qquad 
\ket A|1>\ket B|\psi> \mapsto \ket A|1>\ket B|U\psi>,
\]
where $\psi\in \hilbert B$.  The corresponding matrix is
\[
\begin{pmatrix}
I & 0\\
0 & U
\end{pmatrix}.
\]

\begin{lem}[Conditional Operations]
\label{lem:conditional}
Assume a unitary transformation $U$ on $\reg B$ can be implemented by a quantum circuit of size $s$.  
Then, the conditional operation on the composed register $\reg{AB}$ can be implemented by a quantum circuit of size $O(s)$.  
\end{lem}

%\pfsketch
%The lemma is proved for elementary gates using \rf(thm:universality).  For arbitrary circuits, just replace each gate by its conditional version, one by one.
%\pfend

\section{Quantum RAM}
\label{sec:RAM}
This section describes the main model for time-efficient quantum computation we use in this thesis: a classical random access machine (RAM) with the ability of manipulating quantum data.  The set of elementary quantum gates is the same as in \rf(sec:circuits).  
The classical program determines which gates are applied to which registers.
The classical program can get feedback from the quantum data by performing measurements as described in \rf(sec:evolution).  Also, in contrast to the circuit model, the quantum RAM can time-efficiently access elements of arrays (either classical, or quantum).

In this section, we also define the pseudo-code used in the description of quantum algorithms in this thesis.  In many cases, this is just a way of presenting an algorithm, and the same transformation can be performed by a quantum circuit in the same cost.  However, in some cases, we require additional operations provided by the quantum RAM, e.g., quantum arrays.  Then, the complexity of the quantum RAM is smaller than the complexity of the circuit.

We consider three types of resources: time, space and query complexity.
The time complexity is the number of elementary gates applied by the program.
The space complexity is the number of bits and qubits used by the program.  
The query complexity will be discussed in \rf(sec:query).

\paragraph{Registers}
The memory of a quantum RAM is composed of a number of registers and arrays.  Registers were described in \rf(sec:computationalDevices) and arrays will be described later.  The number of registers depends on the problem, but not on the size of the input.  The sizes of the registers, however, may depend on the size of the input.  Similarly, the number of arrays does not depend on the size of the input, but their sizes and the sizes of their elements may depend on it.

A register can be either classical or quantum.  Moreover, for any register, its status (being classical or quantum) may change with time.  The gates applied by the quantum RAM may only depend on the content of the classical registers.  But the content of the quantum registers may influence the quantum RAM indirectly via measurements.  The quantum RAM may perform quantum gates on registers as described in \rf(sec:circuits), measure quantum registers, and perform random memory accesses as described later.

If a quantum register $\reg A$ is measured, the outcome is obtained and the state of the quantum RAM collapses as described in \rf(sec:evolution).  We assume that register $\reg A$ becomes classical, and stores the outcome of the measurement.  And reversely, any classical register $\reg A$ can be made quantum by applying a quantum operation on it.  If the content of the classical register is $a$, the initial state of the quantum register is $\ket A|a>$.  

We adopt the following font usage to denote registers of different types in the pseudo-code.  
Variables typed in italics, e.g. $a$ or $i$, denote registers whose content depends only on the size of the input string, but not on content of the input string.
These registers may be interpreted as a part of the deterministic program proving uniformity of a family of circuits, as in \rf(sec:circuits).  Registers typed in sans serif, e.g. $\reg A$ or $\reg b$, store information that depends on the input.  If a register is in the quantum state, we underline it in the pseudo-code like this: $\qreg A$.

%One of the main differences between a quantum RAM and a quantum circuit is that the gates applied may depend on the results of the measurements performed previously.  We will usually call a quantum algorithm that does not use intermediate measurements a quantum circuit.  Indeed, if a quantum circuit is uniform, it can be simulated by a quantum RAM.  It generates the circuit for the size of the given input, and uses a CRAQM to store the registers of the circuit.  In order to apply a gate, it swaps the necessary elements of the array into temporary working registers, applies the gate, and swaps the elements back.

\paragraph{Arrays}
Recall that the classical random access machine got its name due to the ability to access elements of arrays in one computational step.  In the quantum RAM, we have four types of arrays.  
%The action of quantum operations on other states can be obtained by linearity.

\descrstart
\item[Classical Random Access Classical Memory (CRACM)]  These are conventional classical arrays.

\item[Quantum Random Access Classical Memory (QRACM)]  This is a CRACM with the additional quantum random read-only access as follows.  Assume $\reg A$ is a QRACM of size $m$, $\reg i$ is an $m$-qudit, and $\reg O$ is the output register of the same type as the elements of $\reg A$.  Both $\reg i$ and $\reg O$ are quantum.  Then, the random access operation $\reg O\qgets \reg A\elem[\reg i]$ transforms $\ket i|i>\ket O|x>$ into $\ket i|i>\ket O|x+A\elem[i]>$ where $A\elem[i]$ is the content of $\reg A\elem[i]$.  The addition, as usually, is performed modulo the size of $\reg O$.

\item[Classical Random Access Quantum Memory (CRAQM)]  Given an array $\reg A$ consisting of quantum registers, a quantum register $\reg O$ of the same size, and a classical register $\reg i$, the random access operation applies the swap gate to $\reg A\elem[\reg i]$ and $\reg O$.

\item[Quantum Random Access Quantum Memory (QRAQM)]  This is the same as CRAQM, but the register $\reg i$ may be quantum.  The action of the random access operation is extended by linearity.
\descrend

If $\reg A$ is an array, we use $\reg A\elem[i]$ to denote the $i$th register in the array.  We often use notation $\ket A|\psi_1,\psi_2,\dots,\psi_m>$ instead of $\ket A\elem[1] |\psi_1>\cdots \ket A\elem[m] |\psi_m>$.

We use different kind of arrays because they have different difficulties of implementing in hardware.  Classical arrays are already available.  QRACM is easier to implement than CRAQM because it does not require to store quantum data.  CRAQM are similar to quantum registers.  Finally, QRAQM is the most complicated type of resource.  Luckily, we will seldom require it.

\paragraph{Pseudo-code}
We will use a python-type pseudo-code to describe programs for quantum RAM.  
We use usual classical directives such as {\bf if}, {\bf for}, {\bf while}, {\bf repeat}, {\bf return} and execution of functions and procedures.  
We assume that the reader is familiar with all these directives.
We use $\gets$ to denote assignment and {\bf uses} to list classical registers used by a classical subroutine.
Additionally, we use some quantum directives to manipulate quantum data.  
Here, we give a list of the directives.  Some of them are described here.  For others, we refer to \rf(sec:specification) where we describe quantum subroutines.

\descrstart
\item[quprocedure {\rm $<$name$>$($<$list of arguments$>$) {\bf with} $<$modifiers$>$}\quad ] 
Starts the description of a quantum procedure.  See \rf(sec:procedures).  The optional modifier part includes some additional information like the precision of the quantum procedure.

\item[qufunction {\rm $<$name$>$($<$list of arguments$>$) {\bf with} $<$modifiers$>$}\quad ]  Similar as {\bf quprocedure}, but for quantum functions.  Refer to \rf(sec:Functions) for more details.

\item[attach {\rm $<$type$>$ $<$quantum register$>$ }\quad ] Adds a new quantum register of the specified type in the {\em initial state}.  The initial state is denoted by 0.  It is an easily distinguishable deterministic state of the register.  The initial state of a composed register is the tensor product of the initial states of its subregisters: $\ket {{AB}} |0> = \ket A|0>\ket B|0>$.  The new register is initially separable from all other quantum registers.

\item[detach {\rm $<$quantum register$>$ }\quad ] 
Removes the specified quantum register.  It assumes that the register is separable from the remaining quantum registers.  A quantum procedure is required to detach all quantum registers it has attached.

\item[conditioned on {\rm $<$condition$>$} : {\rm commands}\quad ] 
Applies the commands if the condition is true as described in \rf(lem:conditional).

\item[measure {\rm $<$quantum register$>$ }\quad ]
Measures the specified quantum register as described in \rf(sec:evolution).  The register becomes classical and stores the result of the measurement.  The state of the quantum RAM collapses accordingly.

\item[{\rm $<$gate or procedure$>$($<$parameters$>$) {\bf with} $<$modifiers$>$ }\quad ] 
Applies a quantum gate or a quantum procedure.  The parameters is a list of registers with types matching the definition of the procedure or the gate.  Optional modifiers contain some additional information like the required precision of the subroutine, see \rf(sec:specification).

\item[{\rm $<$gate or procedure$>^{-1}$($<$parameters$>$) }\quad ] 
Applies the reverse of the given quantum gate or a quantum procedure, see \rf(lem:reversibility).  

\item[{\rm $<$quantum register$>\qgets <$expression$>$ {\bf with} $<$modifiers$>$ }\quad ] 
Adds the value of the expression to the content of the quantum register.  Denote the quantum register by $\reg A$.  Assume $\reg A$ is an $\ell$-qudit for some $\ell$, and its state is an element of the computational basis, say $\ket |a>$.  If the value of the expression is $b$, the state of $\reg A$ changes to $\ket |a+b>$, where the addition is performed modulo $\ell$.  For all other states of $\reg A$, the action is defined by linearity.
The expression is not allowed to contain $\reg A$.  In this case, this operation is reversible (cf.~\rf(obs:reversibility)).  If the state of the register is not an element of the computational basis, then the action of this operation is defined by linearity.

The expression may stand for a number of things.  It can be an arithmetical expression or a simple classical function involving quantum or classical registers, in which case \rf(lem:quantumOfDet) is applied.  The expression may be an element of a QRACM array as will be described later.  Finally, it may stand for a quantum function as described in \rf(sec:Functions).

\item[{\rm $<$quantum register$>\qungets <$expression$>$ }\quad ] The reverse of the $\qgets$ operation.

\item[{\rm $\phase\qgets <$expression$>$ }\quad ] Evaluation of the quantum function into the phase.  See \rf(lem:circuitToPhase).

\descrend

In addition to these commands, we use obvious shorthands.  For instance, the condition in the {\bf conditioned on} directive may be an arbitrary expression involving quantum or classical registers distinct from the registers used in commands.  In this case, the value of the condition is calculated into a temporary qubit, the conditioned commands are applied, and the calculation of the condition is reversed.

%\begin{rem}
%\label{rem:quantumSufice}
%Actually, given an algorithm for a quantum RAM, it is possible to convert it into an algorithm that uses no classical memory, and only performs measurement at the very end.  Treat all registers as quantum ones.  Perform all deterministic computation in fresh registers using \reflem{quantumDet}.  If a measurement is required, simply do nothing.  Replace all classical branchings by the corresponding conditional operations.  It can be proven that the final measurement gives the same outcome as the initial program.
%
%There are many reasons to keep classical memory.  At first, it is a cheaper resource than quantum memory.  It can be erased that makes the total amount of memory smaller.  Finally, it is often simply more convenient to work with classical memory.
%\end{rem}

\section{Subroutines}
\label{sec:specification}
A {\em subroutine} is a cornerstone of a programming language.  They allow one to divide a complex computational task into a number of easier ones that can be solved independently.  The behaviour of each subroutine is then described by a relatively simple {\em specification}.

We allow classical functions that are specified by the {\bf function} directive.  The classical function can have classical and quantum arguments.  Classical arguments are given by value: Their value is copied for the subroutine, and a change of the argument in the function does not affect the value of the variable outside it.  Quantum arguments are given by reference: Any change to the argument affects the quantum register outside the subroutine.
%This is reasonable because, due to \rf(obs:copy), it is not possible to copy quantum registers.  The function returns the result of its computation using the {\bf return} directive.

A quantum procedure is a limited version of a classical procedure.  The reason for the limitations is the ability to reverse the procedure.  In particular, it is not allowed to measure quantum registers, or to modify classical registers defined outside the subroutine.  

In this section, we continue the description of the pseudo-code we initiated in the previous section to include the execution of quantum subroutines.  We consider general quantum subroutines, and quantum subroutines evaluating functions coherently and non-coherently.  Additionally, we give a number of simple lemmas dealing with various kinds of subroutines.

\subsection{Procedures}
\label{sec:procedures}
A {\em quantum procedure} is a general quantum subroutine.  
The procedure may have quantum and classical arguments.  It implements some quantum transformations on its quantum arguments.  The transformation may depend on the values of the classical arguments.

The {\em specification} of a quantum procedure $U$ consists of a register $\reg A$ (the composition of the quantum arguments), a system of orthonormal vectors $\Psi = \{\psi_1,\dots,\psi_m\} \subset \hilbert A$, and the action of the procedure on the elements of $\Psi$: $\psi_j\mapsto U\psi_j\in\hilbert A$.  Clearly, the vectors $U\psi_1,\dots,U\psi_m$ must be orthonormal.  Usually, the vectors $\psi_j$ are the elements of the standard basis.  
The action of the subroutine in the span of $\Psi$ is then uniquely determined by linearity.  We call the span of $\Psi$ the {\em input subspace} of the subroutine.  It may be a proper subspace of $\hilbert A$.
In this case, we interpret this as a promise that the initial state of the subroutine belongs to the span of $\Psi$.

Quantum procedures are defined in the pseudo-code using the {\bf quprocedure} directive.  The quantum procedure is allowed to attach quantum registers, but, at the end, it has to detach all the registers it has attached.  This means that the state at the end of the procedure is of the form $\ket A|\omega_\psi>\ket B|\upsilon>$ where $\reg B$ is the composition of all the attached registers and $\upsilon$ does not depend on the initial state $\psi\in\Psi$.  We will make this explicit by the use of the {\bf detach} command at the end of the procedure.

\begin{exm}[Preparation of Uniform Superposition]
\label{exm:uniform}
As an example, consider preparation of the uniform superposition.  This is a very common quantum operation.  Assume we have a register $\reg A$ that stores an integer $i$ between 0 and $n-1$.   At the low level, the register is represented as a CRAQM of $\ell \ge \lceil\log_2 n\rceil$ qubits.  The integer $i$ is stored in binary.  The qubit $\reg A\elem[\ell]$ stores the highest bit of $i$, and the qubit ${\reg A}\elem[1]$ stores the lowest bit of $i$.

The task is to transform $\ket A|0>$ into the uniform superposition of all states $\ket A|i>$ for $i\in \{0,\dots,n-1\}$.  Thus, the input subspace is one-dimensional.  For example, assume $n=5$, and $\ell = 4$.  Using the low level representation, the task is to  transform $\ket A|0,0,0,0>$ into
\[
\frac1{\sqrt5} \ket A|0,0,0,0> +
\frac1{\sqrt5} \ket A|1,0,0,0> +
\frac1{\sqrt5} \ket A|0,1,0,0> +
\frac1{\sqrt5} \ket A|1,1,0,0> +
\frac1{\sqrt5} \ket A|0,0,1,0>.
\]

The pseudo-code of the corresponding procedure is given in \rf(alg:uniform).  
The algorithm works recursively by applying the BiasedCoin gate (\rf(tbl:oftenGates)) to distribute the amplitude between the binary strings that begin with 0 and 1.  The qubit $\reg{flag}$ stores a flag indicating that the binary string must be processed by the subroutine.  On each level of recursion, we attach a new register, and detach it before finishing the procedure.  

The procedure uses classical registers $\ell$ and $n$.
%We write them in italic, because their values depend only on the size of the instance.  
Although the subroutine is quite complex, we call it a quantum procedure because it is quite easy to construct a procedure for the reverse operation.

The number of elementary gates in \rf(alg:uniform) is polynomial in the number of qubits.  We denote the execution of the procedure by UniformSuperposition($\qreg A$).  In this case, the argument $n$ is determined by the size of $\reg A$.
\end{exm}

\begin{algorithm}[tbh]
\caption{Preparation of a Uniform Superposition}
\label{alg:uniform}
\algbegin
\state{{\bf quprocedure} UniformSuperposition({\bf integer} $n$, {\bf CRAQM array} $\qreg A\elem[\ell]$ of {\bf qubits}) {\bf :}}
\tab
	\state{{\bf attach qubit} $\qreg{flag}$}
	\state{NOT($\qreg{flag}$)}
	\state{PreparationSubroutine($\ell$, $n$, $\qreg{flag}$)}
	\state{{\bf detach} $\qreg{flag}$}
\state{~ }
\state{{\bf quprocedure} PreparationSubroutine({\bf integers} $\ell$, $n$, {\bf qubit} $\qreg{flag}$) {\bf :}}
\tab
\state{{\bf if} $\ell \ne 0$ {\bf and} $n\ne0$ {\bf :} }
\tab
\state{{\bf attach qubit} $\qreg{f}$ }
\state{{\bf if} $n < 2^{\ell-1}$ {\bf :}}
\tab 
	\state{{\bf conditioned on} $\qreg{flag}=1$ {\bf :} $\qreg f \qgets 1 - \qreg A\elem[\ell]$}
	\state{PreparationSubroutine($\ell-1$, $n$, $\qreg f$)}
	\state{{\bf conditioned on} $\qreg {flag}=1$ {\bf :} $\qreg f \qungets 1 - \qreg A\elem[\ell]$}
\untab
\state{{\bf elseif} $2^{\ell-1}\le n < 2^{\ell}$ {\bf :}}
\tab
	\state{{\bf conditioned on} $\qreg {flag}=1$ {\bf :}}
	\tab
		\state{$(2^{\ell-1}/n)$-BiasedCoin($\qreg A\elem[\ell]$) }
		\state{$\qreg f\qgets 1 - \qreg A\elem[\ell]$ }
	\untab
	\state{PreparationSubroutine($\ell-1$, $2^{\ell-1}$, $\qreg f$) }
	\state{{\bf conditioned on} $\qreg{flag}=1$ {\bf :} NOT($\qreg f$)}
	\state{PreparationSubroutine($\ell-1$, $n-2^{\ell-1}$, $\qreg f$) }
	\state{{\bf conditioned on} $\qreg{flag}=1$ {\bf :} $\qreg {f}\stackrel{-}{\longleftarrow}  \qreg A\elem[\ell]$}
\untab
\state{{\bf elseif} $n = 2^{\ell}$ {\bf :}}
\tab
	\state{{\bf for} $i = 1,\dots, \ell$ {\bf :}}
	\tab
		\state{{\bf conditioned on} $\qreg{flag}=1$ {\bf :} Hadamard($\qreg A\elem[i]$)}
	\untab\untab
	\state{{\bf detach} $\qreg{f}$ }
\untab\untab\untab
\algend
\end{algorithm}

Quantum procedures that perform their task with small imprecision are also acceptable.  In this case, we consider the description of the operator $U$ in the previous paragraphs as a {\em specification}.  We say that a quantum procedure $\circuit U$ {\em $\delta$-approximates} $U$ (or performs $U$ with {\em precision} $\delta$), if $\|\circuit U\psi - U\psi \|\le \delta$ for all $\psi$ in the input subspace of $U$.  The following lemma describes how imprecision accumulates during the execution of a program.

\begin{lem}
\label{lem:circuitRobust}
Assume we have a quantum procedure $\circuit A$ that applies a number subroutines $U_1,U_2,\dots,U_m$ in this order.  Moreover, for all $k$, the state of $\circuit A$ before applying $U_k$ belongs to the input subspace of $U_k$.
Let $\circuit A'$ denote the quantum procedure $\circuit A$ with each subroutine $U_i$ replaced by a quantum procedure $\circuit U_i$ that $\delta_i$-approximates $U_i$.  Then, $\circuit A'$ $(\sum_i \delta_i)$-approximates $\circuit A$.
\end{lem}

\pfstart
We may assume $\circuit A$ is the composition of the subroutines, i.e., $\circuit A = U_mU_{m-1}\cdots U_1$ by treating all gates between the subroutines as additional subroutines that are evaluated exactly.
Let $\psi$ be a state in the input subspace of $\circuit A$.  Denote by $\phi_k = U_kU_{i-1}\cdots U_1\psi$ and $\psi_k = \circuit U_k\circuit U_{i-1}\cdots \circuit U_1\psi$ the states of $\circuit A$ and $\circuit A'$ after $k$ subroutines are applied.
We have $\phi_0 = \psi_0 = \psi$, and
\[
\|\phi_{k+1} - \psi_{k+1} \| =
\|U_{k+1}\phi_{k} - \circuit U_{k+1}\psi_{k} \| \le 
\|(U_{k+1} - \circuit U_{k+1})\phi_{k}\| + \| \circuit U_{k+1}(\phi_k-\psi_{k}) \| \le
\delta_{k+1} + \|\phi_k - \psi_k\|.
\]
Here, the first inequality follows from the triangle inequality, and the second one holds because $\phi_k$ belongs to the input subspace of $U_{k+1}$.  The result follows by induction on $k$.
\pfend

Usually, we will use capital Latin letters in italic font for specifications of quantum procedures, or circuits that follow the specification exactly.  Calligraphic letters will be used for actual circuits that follow the specification approximately.

\subsection{Functions}
\label{sec:Functions}
\mycommand{vyhod}{\ell}
We are mostly interested in subroutines that calculate functions.  Let $\reg A$ and $\Psi$ be as in the previous section, and let $f\colon \Psi\to [\vyhod]$ be a function, where $\vyhod$ is some integer.  

\paragraph{Coherent Evaluation}
We say that a subroutine $U$ {\em evaluates $f$ coherently} if it performs the transformation 
\(
\ket A |\psi> \ket O|i> \mapsto \ket A |\psi>\ket O|i+f(\psi)>
\)
for all $\psi\in\Psi$ and $i\in[\vyhod]$, where $\reg O$ is an $\vyhod$-qudit.  Recall that addition in $\reg O$ is performed modulo $\vyhod$.  We denote the execution of this subroutine by $\reg O\qgets U(\reg A)$.
As for any quantum procedure, we may consider approximate implementations of coherent function evaluations.

\paragraph{Evaluation in the Phase}
If $f$ has a Boolean output, there is another important variant of coherent evaluation.  We say that a quantum procedure $U$ evaluates $f\colon \Psi\to\{0,1\}$ {\em in the phase} if $U\ket A|\psi> = (-1)^{f(\psi)}\ket A|\psi>$ for all $\psi\in\Psi$.  
We denote the execution of this subroutine by $\phase\qgets U(\reg A)$.  
In order to distinguish this notion, we sometimes say that the standard coherent function evaluation, as in the previous paragraph, evaluates $f$ {\em in the register}.  The following lemma shows that these notions are interchangeable.

\begin{lem}%[Evaluating Function in the Phase]
\label{lem:circuitToPhase}
Assume $U$ is a quantum procedure that evaluates a function $f\colon \Psi\to\{0,1\}$ in the register.  
Then, there exists a quantum procedure that, given $U$, evaluates $f$ in the phase.  The procedure calls $U$ once and uses $O(1)$ elementary gates.
\end{lem}

\pfstart
The description of the subroutine can be found in \rf(alg:circuitToPhase).  The operations in \rf(line:circuitToPhase:hadamard) transforms the initial state of the register $\reg b$ into $\s[\ket b |0>-\ket b|1>\strut ]/\sqrt{2}$.  Assume the register $\reg A$ contains a vector $\psi\in\Psi$.  If $f(\psi)=0$, \rf(line:circuitToPhase:application) does not change the state of $\reg b$, otherwise, it flips its sign.  The sign can be transferred to the content of $\reg A$, thus allowing one to detach $\reg{b}$ afterwards.
\pfend

\begin{algorithm}
\caption{Evaluation of a function in the phase}
\label{alg:circuitToPhase}
\algbegin
\state{{\bf quprocedure} EvaluateInThePhase({\bf quprocedure} $U$, {\bf register} $\qreg A$) {\bf :}}
\tab
	\state{{\bf attach qubit} $\qreg{b}$}
	\state{Hadamard(NOT($\qreg b$))} \label{line:circuitToPhase:hadamard}
	\state{$\qreg b\qgets U(\qreg A)$} \label{line:circuitToPhase:application}
	\state{{\bf detach} $\qreg{b}$}
\untab
\algend
\end{algorithm}

\paragraph{Non-Coherent Evaluation}
Although coherent function evaluation is convenient for the executing subroutine, it is not convenient for the executed subroutine.  
%In most cases, this requires some routine actions.  
In order to simplify exposition of algorithms, we describe an alternative specification of function evaluation.  Moreover, this specification allows one to increase the precision of a subroutine exponentially.  We already start with the approximate version of this specification.  Note that it is different from that of a general subroutine.

%
%Usual computational tasks in the thesis will be stated as function evaluation.  Let $f$ be a function with domain $\cD\subseteq [q]^n$ and range $[\ell]$ for some integers $n,q$ and $\ell$.  We usually denote all this by $f\colon [q]^n\supseteq \cD\to[\ell]$.
%
%At the end of the computation, the output $\ell$-qudit $\reg O$ is measured.  It contains the computed value of the function.  There are various modes of computing a function.  The algorithm may compute $f$ {\em exactly}: this means the register $\reg O$ is observed in state $f(x)$ with probability 1.  The algorithm may computer $f$ with a {\em bounded error}.  This means $\reg O$ contains $f(x)$ with probability at least $2/3$ for any valid input string $x\in\cD$.
%
%For Boolean functions, $f\colon [q]^n\supseteq \cD\to\{0,1\}$, we can be more specific.  We say the algorithm computes $f$ with a {\em two-sided} error $\eps<1/2$ if, for any $x\in\cD$, the output register $\reg O$ contains $f(x)$ with probability at least $1-\eps$.  We say the algorithm computes $f$ with a {\em one-sided} error $\eps<1$ if the following holds:  For each negative input $y\in f^{-1}(0)$, the register $\reg O$ contains $0$ with probability 1, and for each positive input $x\in f^{-1}(1)$, the register $\reg O$ contains 1 with probability at least $1-\eps$.

Again, let $f\colon \Psi\to [\vyhod]$ be a function, where $\Psi\subset \hilbert A$ is an orthonormal subset, and let $\reg O$ be the output $\vyhod$-qudit.  The subroutine attaches the {\em working} register $\reg W$.  We say the quantum subroutine $\circuit A$ {\em $\eps$-evaluates} function $f$, or evaluates $f$ {\em with error $\eps$}, if, for all $\psi\in\Psi$,
\begin{equation}
\label{eqn:subroutineEvaluates}
\circuit A\s[\ket A |\psi> \ket W|0>\ket O|0>\strut] = \ket A|\psi> \ket {{WO}} |\omega_\psi>
\qquad\mbox{for some $\omega_\psi$ such that}\qquad
\|\proj O{f(\psi)} \omega_\psi\|^2 \ge 1-\eps,
\end{equation}
where the projector $\proj O{f(\psi)}$ is defined in \rf(sec:evolution).
In particular, if the register $\reg O$ is measured after the application of $\circuit A$, the probability of observing $f(\psi)$ is at least $1-\eps$.
We say that $\circuit A$ evaluates $f$ {\em exactly}, if $\eps = 0$.  
Note that the separability condition in~\rf(eqn:subroutineEvaluates) is only required for the elements of $\Psi$.  In general, the state $\circuit A \s[\ket A |\psi> \ket W|0>\ket O|0>\strut]$ will be entangled over all the three registers.

For functions with Boolean output, we can also define one-sided error.  We say that $\circuit A$ evaluates $f$ with one-sided error $\eps$ if, in addition to~\refeqn{subroutineEvaluates}, it holds that $\proj O1 \omega_\psi = 0$ for all $\psi\in f^{-1}(0)$.

We use the keyword {\bf qufunction} to describe non-coherent function evaluation in the pseudo-code.  The output register $\reg O$ will be listed as the last register in the arguments of the function.  We use the keyword {\bf with} to specify the size of the error and whether it is one-sided.  If nothing is specified, we assume that $\eps=0$.

\paragraph{Classical Functions}
Similar definitions can be made for classical functions.  We say that a classical function evaluates a function $f$ with error $\eps$, if the probability the function outputs $f(\psi)$, given $\ket A|\psi>$, is at least $1-\eps$.  Unlike quantum functions, the classical functions need not be reversible.

\mycommand{vars}{N}
\paragraph{Specification of a Program}
Let us define the specification of the whole program.  Computational tasks in the thesis will be stated as function evaluations.  Let $f$ be a function with domain $\cD\subseteq [q]^\vars$ and range $[\vyhod]$ for some integers $\vars,q$ and $\vyhod$.  In this case, we write $f\colon [q]^\vars\supseteq \cD\to[\vyhod]$.

Assume that the input string is $z \in\cD$.  For the $j$th element of $z$, we use notation $z_j$, or $z\elem[j]$. We call $z_j$ {\em input variable}, or {\em input element}.

Given access to the quantum procedure InputOracle, the program has to evaluate $f(z)$ with error at most $1/3$.
The InputOracle procedure has an $\vars$-qubit $\reg A$ as its argument and coherently evaluates the function $\ket A|j>\mapsto z_j$ for all $j\in [\vars]$.
For example, it could be a read-only QRACM array of length $\vars$ containing $q$-qudits.

We call the function {\em total} if $\cD = [q]^n$.  Otherwise, the function is {\em partial}.  
In most cases, the output of the function $f$ is Boolean.  In this case, we assume the range of $f$ is $\{0,1\}$.  We stick to the notational convention that $x$ denotes a {\em positive input}, and $y$ denotes a {\em negative input}, i.e., $x$ and $y$ are such that $f(x)=1$ and $f(y)=0$.
If $q=2$, we also assume that each $z_j$ takes values from $\{0,1\}$.  If $q=\ell=2$, we call the function Boolean.

\paragraph{Making Function Coherent}
Non-coherent evaluation is convenient for describing a subroutine, but it is usually not sufficient for use in other subroutines, because the elements may fail to interfere due to the working register.  However, there is a general way of converting a non-coherent evaluation of a function into a coherent one.

\begin{lem}
\label{lem:coherent}
Assume a quantum procedure $\circuit A$ $\eps$-evaluates a function $f\colon \Psi\to[\vyhod]$ for some $\eps\ge0$.  Then, there exists a quantum procedure $\circuit B$ that evaluates $f$ coherently with precision $\sqrt{2\eps}$.
The procedure $\circuit B$ uses 2 executions of $\circuit A$ and $O(\log \vyhod)$ 2-qubit gates.
\end{lem}

\begin{algorithm}
\caption{Converting a non-coherent evaluation into a coherent one}
\label{alg:coherent}
\algbegin
\state{{\bf quprocedure} $\qreg O\qgets$ MakeCoherent({\bf qufunction} $\circuit A$, {\bf register} $\qreg A$) :}
\tab
	\state{{\bf attach} the working register $\qreg{W}$ and $\vyhod$-{\bf qudit} $\qreg O'$}
	\state{$\circuit A(\qreg A, \qreg W, \qreg O')$}
	\state{$\qreg O\qgets \qreg O'$ \label{line:coherent:copy} }
	\state{$\circuit A^{-1}(\qreg A, \qreg W, \qreg {O}')$}
	\state{{\bf detach} $\qreg{W}, \qreg{O}'$}
\untab
\algend
\end{algorithm}

\pfstart
The description of $\circuit B$ is given in \rf(alg:coherent).  
Recall that $\Psi$ is a set of orthonormal vectors.
Let $U$ be the specification of evaluating $f$, i.e., a unitary performing the transformation $\ket A|\psi>\ket O|a>\mapsto \ket A|\psi>\ket O|a+f(\psi)>$ for all $\psi\in\Psi$ and $a\in[\ell]$.
Denote by $C$ be the copying circuit in \rf(line:coherent:copy).  Clearly, $C$ requires $O(\log\vyhod)$ 2-qubit gates.

For any $\phi\in \spn\Psi$, we have 
\begin{equation}
\label{eqn:coheration}
\| {\circuit A}^{-1} C\circuit A \s[\strut \ket A|\phi> \ket{{WO'O}}|0>] -
U\s[\strut \ket A|\phi> \ket{{WO'O}}|0> ] \| 
= \| C\circuit A \s[\strut \ket A|\phi> \ket{{WO'O}}|0> ]
 - {\circuit A} U\s[\strut \ket A|\phi> \ket{{WO'O}}|0> ]\|
\end{equation}
because $\circuit A$ is unitary.  Let, at first, $\psi\in\Psi$.  Then,
\[
C\circuit{A}\s[\strut \ket A|\psi> \ket{{WO'O}}|0>] = 
\ket A|\psi>\otimes \sum_{a\in[\vyhod]} \beta_a \ket W|\omega_a>\ket O' |a>\ket O|a>
\]
for some complex numbers $\beta_a$ satisfying $\sum_{a} |\beta_a|^2 = 1$ and some unit vectors $\omega_a$.  Both $\beta_a$ and $\omega_a$ depend on $\psi$.  Moreover, $|\beta_{f(\psi)}|^2\ge 1-\eps$, and,
\[
{\circuit A} U\s[\strut \ket A|\psi> \ket{{WO'O}}|0> ] = 
\ket A|\psi>\otimes \sum_{a\in[\vyhod]} \beta_a \ket W|\omega_a>\ket O' |a>\ket O|f(\psi)>.
\]
Thus,
\[
C\circuit A \s[\strut \ket A|\psi> \ket{{WO'O}}|0> ]
 - {\circuit A} U\s[\strut \ket A|\psi> \ket{{WO'O}}|0> ]
= \ket A|\psi>\otimes \sum_{a\in[\vyhod]\setminus\{f(\psi)\}} \beta_a \ket W|\omega_a>\ket O' |a>
\sA[ \ket O|a> - \ket O|f(\psi)>].
\]
Denote the sum in the last equation by $\upsilon_\psi$.  Since all $\ket O'|a>$ are orthogonal, we get 
\[
\|\upsilon_\psi \|^2 = 2\sum_{a\in[\vyhod]\setminus\{f(\psi)\}} |\beta_a|^2\le 2\eps .
\]
Now let $\phi = \sum_{\psi\in\Psi} \alpha_\psi\psi$ be an arbitrary unit vector in $\spn \Psi$.  Then, $\sum_\psi |\alpha_\psi|^2 = 1$, and,
\[
C\circuit A \s[\strut \ket A|\phi> \ket {{WO'O}}|0> ]
 - {\circuit A} U\s[\strut \ket A|\phi> \ket{{WO'O}}|0> ] 
 = \sum_{\psi\in\Psi} \alpha_\psi \ket A |\psi> \ket{{WO'O}}|\upsilon_{\psi}>. 
\]
The norm of the last vector squared is at most $2\eps$, hence, due to~\rf(eqn:coheration), $\circuit B$ evaluates $f$ coherently with precision $\sqrt{2\eps}$.
\pfend

\paragraph{Precision Amplification}
An important feature of non-coherent function evaluation is that the precision of the subroutine can be amplified.  This is not always true for coherent function evaluation.

\begin{lem}
\label{lem:circuitAmplify}
Assume we have a quantum procedure $\circuit A$ performing one of the following tasks non-coherently:
\itemstart
\item it $c$-evaluates a function $f\colon \Psi\to[\vyhod]$ for some constant $c$ strictly less than $1/2$;
\item it evaluates a function $f\colon \Psi\to\{0,1\}$ with one-sided error $c$, where $c$ is a constant strictly less than 1.
\itemend
Then, for each $\eps>0$, there exists a quantum procedure $\circuit B$ that $\eps$-evaluates $f$.
The procedure $\circuit B$ executes $\circuit A$ $O(\log(1/\eps))$ times.  If $\circuit A$ has a one-sided error, then $\circuit B$ has a one-sided error as well.
\end{lem}

\begin{algorithm}
\caption{Precision Amplification for Non-Coherent Function Evaluation}
\label{alg:circuitAmplify}
\algbegin
\state{{\bf qufunction} PrecisionAmplification(qufunction $\circuit A$, registers $\qreg A$, $\qreg O$) :}
\tab
	\state{{\bf attach} CRAQM array $\qreg O'$ of $\vyhod$-qudits of length $k$}
	\state{{\bf for} $i=1,\dots,k$ {\bf :}}
	\tab
		\state{$\circuit A(\qreg A, \qreg O'\elem[i])$ }
	\untab
	\state{$\qreg O\qgets$ Majority($\qreg O'$)} \label{line:circuitAmplify:majority}
\untab
\algend
\end{algorithm}

\pfstart
Let us consider the two-sided error case first.  The description of $\circuit B$ is given in \rf(alg:circuitAmplify).  Here, $k$ is an integer to be specified later.  In \rf(line:circuitAmplify:majority), the majority of an input $z \in [\vyhod]^k$ is defined as the entry in $[\vyhod]$ that appears among $z_1,\dots,z_k$ most frequently.  If there are several such entries, we take the smallest one.

Let $\psi\in\Psi$.  Before the execution of the majority gate in \rf(line:circuitAmplify:majority), the state of the subroutine is
\[
\ket A|\psi> \ket W\elem[1]\reg O\elem[1] |\omega_\psi>\cdots \ket W\elem[k]\reg O\elem[k] |\omega_\psi>\ket O|0> .
\]
Let $M_i$, for $i\in [k]$, be independent random variables, each of them equal to 1 with probability 
$\proj O{f(\psi)} \omega_\psi\|^2\ge 1-c$, and to 0 otherwise.  (See \rf(sec:evolution) for the definition of $\Pi$.)  After the majority gate, the probability of measuring $\reg O$ in state $f(\psi)$ equals the probability of more than $k/2$ of $M_i$s being equal to 1.  By the Chernoff bound, this probability is at least 
$1 - \ee^{-\Omega(k)}$.  Hence, it is sufficient to take $k=O(\log(1/\eps))$ to assure the error is at most $\eps$.

For the one-sided error, apply the same procedure with the majority gate replaced by the OR gate.  The analysis is similar.
\pfend

We will usually describe a quantum procedure that evaluates a function non-coherently with some constant error.  While executing, we will use the corresponding coherent version with a better precision.  In such a case, we assume that Algorithms~\ref{alg:circuitAmplify} and~\ref{alg:coherent} are used to improve the precision and make the procedure coherent.

\section{Query Complexity}
\label{sec:query}
In this section, we introduce the notion of query complexity that is used throughout the thesis.  The motivation for introducing this notion is as follows.  Assume that we have a function $f$ and we want to estimate the amount of resources required to calculate it.  To prove an upper bound, it is sufficient to come up with an algorithm that calculates $f$.  But, in general, it is very hard to prove strong lower bounds on the time complexity of $f$, since the algorithm can use very sophisticated internal data representations that are difficult to reason about.  Instead of this, we can consider only the accesses of the algorithm to the input string.  Since we understand how the input data is represented, it is much easier to prove lower bounds.  A lower bound on the number of accesses is simultaneously a lower bound on the time complexity.  In general, such a bound can be very loose, but for many functions, it is good enough.

\subsection{Definitions}
We start by defining the notion of query complexity for all three models of computations.

\begin{wrapfigure}{R}{0pt}
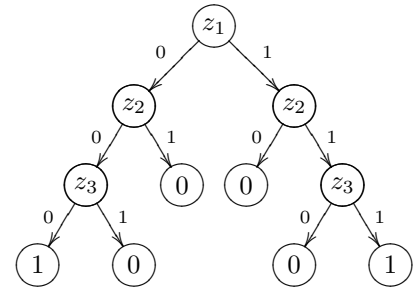

\vspace{-.3in}
\xygraph{!~*{\cir<8pt>{}} !~-{@{->}_{0}} !~:{@{->}^{1}} !{0;<2.5pc,0pc>:}
*{z_1}[]{}
(
	-[dl]{}[]*{z_2}{}
	(
		-[l(.6)d]{}[]*{z_3}{}
		(
			-[l(.6)d]{}[]*{1},
			:[r(.6)d]{}[]*{0}
		),
		:[r(.6)d]{}[]*{0}
	),
	:[dr]{}[]*{z_2}{}
	(
		:[r(.6)d]{}[]*{z_3}{}
		(
			:[r(.6)d]{}[]*{1},
			-[l(.6)d]{}[]*{0}
		),
		-[l(.6)d]{}[]*{0}
	)
)}
\caption{A deterministic decision tree calculating the total Boolean function on 3 variables that evaluates to 1 iff $z_1=z_2=z_3$.}
\label{fig:decisionTree}
\end{wrapfigure}
\paragraph{Deterministic}
Any deterministic algorithm can be described as follows.  It starts its computation.  At some point, it requires the value of an input variable.  Since the algorithm is deterministic, the index of the variable will always be the same.  The algorithm {\em queries} the value of the variable, and, after being told the value, it proceeds with its computation, until it requires the value of another variable.  Again, the index of the requested variable depends solely on the value of the first variable.  And so on.  At the end of the computation, the algorithm returns the value of the function.

By ``contracting'' (in graph-theoretical terms) all intermediate calculations, we get the computational model of a deterministic {\em decision tree}.  It is a rooted $q$-ary tree $T$.  Each internal node is labelled by an index of an input variable, $j\in[\vars]$.  Each leaf of the tree is labelled by an output value in $[\vyhod]$.  

The value $T(z)$ of the decision tree $T$ on input $z \in[q]^\vars$ is defined by induction on the depth of $T$.  If the depth is zero, then $T$ consists of one leaf.  Define $T(z)$ as the value of this leaf.  Now assume that the depth of $T$ is non-zero.  Then, the root is labelled by some $j\in[\vars]$.  Define $T(z)$ as the value of the $z_j$th subtree of the root.  Since the depth of the subtree is less than the depth of $T$, this is a valid inductive definition.  

We say that $T$ evaluates a function $f\colon [q]^\vars\supseteq \cD\to[\vyhod]$ iff $f(z)=T(z)$ for all $z\in\cD$.  The {\em complexity} of the decision tree is defined as its depth, i.e., the number of variables queried for the worst input.  The {\em deterministic query complexity}, $D(f)$, of $f$ is defined as the minimal complexity of a decision tree evaluating $f$.

\rf(fig:decisionTree) shows an example of a decision tree that evaluates the total function $f\colon \{0,1\}^3\to\{0,1\}$ defined by $f(z_1,z_2,z_3)=1$ if and only if $z_1=z_2=z_3$.  The complexity of the tree is 3, and it is tight.

\paragraph{Randomised}
In the randomised case, the index of the variable being queried at some point in the algorithm is given by a probability distribution that depends only on the values of the variables queried previously.  All these probability distributions may be combined into one huge probability distribution from which we sample at the very beginning.
Thus, we define a {\em randomised decision tree} as a probability distribution $\mu$ over deterministic decision trees as defined above.  We say $\mu$ evaluates $f$ if, for each $z\in\cD$, $T(z) = f(z)$ with probability at least $2/3$ when $T$ is sampled from $\mu$.

The complexity of $\mu$ is defined as the largest complexity of a deterministic decision tree $T$ having non-zero probability in $\mu$.  The {\em randomised query complexity} $R(f)$ of a function $f$ is defined as the minimal complexity of a randomised decision tree evaluating $f$.  The constant $2/3$ can be replaced by any constant strictly between $1/2$ and $1$.  This changes the randomised query complexity by at most a constant factor.

\begin{exm}
\label{exm:balancedClassical}
Let $n$ be an integer, $\vars=2n$, and $f$ be a partial function from $\{0,1\}^\vars$ to $\{0,1\}$ defined as follows.  The function evaluates to 0 if all input variables have the same value.  The function evaluates to 1 if exactly $n$ of the input variables are equal to 0, and the remaining $n$ input variables equal 1.  For all other input strings, the function is not defined.

The randomised query complexity of this function is 3.  Indeed, define the randomised decision tree $\mu$ as the probability distribution over decision trees as in \rf(fig:decisionTree) with $\{z_1,z_2,z_3\}$ replaced by $\{z_a,z_b,z_c\}$ where $\{a,b,c\}$ is a 3-subset of $[\vars]$ chosen uniformly at random, i.e., with probability ${\vars\choose 3}^{-1}$.  If $f(z)=0$, $\mu$ does not err.  If $x$ is a positive input and $n$ is large, we may approximately assume that each $x_a,x_b,x_c$, chosen by $\mu$, takes values in $\{0,1\}$ independently and uniformly at random.  
The algorithm $\mu$ errs when $x_a=x_b=x_c$, and this happens with probability approximately $2/8 = 1/4<1/3$.

On the other hand, the deterministic query complexity of $f$ is $n+1 = \vars/2+1$.  
For the upper bound, consider the decision tree that queries the first $n+1$ variables, and returns 1 iff all of them are equal.
For the lower bound, we reason as follows.
Assume $T$ is a deterministic decision tree of depth at most $n$ that evaluates $f$.  Consider the path $P$ that starts at the root of $T$ and follows the arcs labelled by 0.  Let $j_1,j_2,\dots,j_t$ be the indices of the variables in the vertices along $P$.  Due to our assumption on the depth of $T$, $t\le n$.  Let $A$ be any set of size $n$ such that $\{j_1,\dots,j_t\}\subseteq A\subset [\vars]$.  Let $y$ be the all-0 input, and $x$ be the positive input defined by $x_j = 0$ if $j\in A$, and $x_j=1$ otherwise.  We have $f(x)=1$, and $f(y)=0$, but $T$ outputs the same value (contained in the end-vertex of $P$) on both of them.  It is a contradiction, hence, $D(f)=n+1$.
\end{exm}

\paragraph{Quantum}
The {\em query complexity of a quantum algorithm} $\circuit C$ evaluating a function $f\colon [q]^\vars\supseteq \cD\to[\vyhod]$ {\em on an input $z\in\cD$} is defined as the largest possible number of times the algorithm invokes the InputOracle subroutine when executed on the input $z$.  The {\em query complexity of a quantum algorithm} is the maximum of the query complexity over all inputs in $\cD$.  The {\em quantum query complexity} $Q(f)$ of the function $f$ is the smallest query complexity of a quantum algorithm evaluating $f$.

The definition above is nice for upper bounds, but it is not well-suited for lower bounds.  For the latter,
we transform the program into a more restricted form.
At first, we use \rf(lem:quantumOfDet) to transform all classical computations performed by the quantum algorithm $\circuit C$ into the quantum form.  It is possible to check that the transformation preserves the number of times the input oracle is invoked.
This gives the following definition.
A {\em quantum query algorithm} uses 3 registers: index $\reg j$, value $\reg v$, and workspace $\reg W$.
The index register has basis elements $0,1,\dots,\vars$.
The value register is a $q$-qudit.
The workspace register can be arbitrary, but it contains the output $\vyhod$-qudit $\reg O$ as its component.
The initial state is $\ket j|0> \ket v|0> \ket W|0>$.  The computation is modelled as a sequence of unitary transformations the device performs on its own, interchanged with some number of queries to the input oracle:
\begin{equation}
\label{eqn:sequence}
U_0\to O_z\to U_1\to O_z \to \cdots \to U_{T-1} \to O_z\to U_T.
\end{equation}
Here $U_i$s are arbitrary transformations, as described in \refsec{evolution}.  The index $i$ indicates that they may be different at different positions, but they are independent of the input.  The transformation $O_z$, on contrary, is the same in all places, but it depends on the input $z\in\cD$.  
Assume for notational convenience that each input strings $z$ is extended with an additional input variable $z_0 = 0$.  Then, $O_z$ can be decomposed as $O_z = \bigoplus_{j=0}^{\vars} O_{z_j}$, where $O_a$, for $a\in[q]$, is a unitary in $\hilbert v\otimes \hilbert W$.  
We assume, as in \rf(sec:RAM), that $O_a$ is given by $\ket v|i>\mapsto \ket v|i+a>$.  
The part of the state with $j=0$, thus, does not change during the input oracle execution.
We say that the algorithm in~\rf(eqn:sequence) makes $T$ quantum queries.

After all the transformations in~\refeqn{sequence} are performed, the output register $\reg O$ of the program is measured.  We say that the quantum query algorithm {\em evaluates} $f$ (with bounded error) if, for any $z\in\cD$, the register $\reg O$ contains $f(z)$ with probability at least $2/3$.  The {\em quantum query complexity} $Q(f)$ of the function $f$ is the smallest possible number of queries made by a quantum algorithm that evaluates $f$.

\begin{exm}[Deutsch-Jozsa~\cite{deutsch:jozsa}]
\label{exm:balancedQuantum}
Consider the same function as in \rf(exm:balancedClassical).  As we have seen, an exact (deterministic) classical algorithm requires $\vars/2+1$ queries.  Now we show that it can be evaluated by a quantum \rf(alg:deutsch) in one query without error.

\begin{algorithm}
\caption{Quantum Algorithm for Deutsch-Jozsa Problem}
\label{alg:deutsch}
\algbegin
\state {{\bf function} DeutschJozsa({\bf quprocedure} InputOracle) {\bf:}}
\tab
  \state {{\bf attach} $\vars$-qudit $\qreg j$ }
  \state {UniformSuperposition($\qreg j$) \label{line:deutsch:uniform} }
  \state {$\phase \qgets$ InputOracle$(\qreg j)$ \label{line:deutsch:oracle}}
  \state {UniformSuperposition$^{-1}(\qreg j)$ \label{line:deutsch:reverse}}
  \state {{\bf measure} $\qreg j$ \label{line:deutsch:measure}}
  \state {{\bf if} $\reg j=0$ {\bf : return} 0 }
  \state {{\bf else : return} 1}
\untab
\algend
\end{algorithm}

The algorithm uses the UniformSuperposition procedure from \rf(exm:uniform).  
Let $\ket j|\psi>$ be the state generated in \rf(line:deutsch:uniform).  Consider the transformation in \rf(line:deutsch:oracle).  If all $z_j$ equal 0, it does not change the state.  If all $z_j$ equal 1, the state changes to $-\ket j|\psi>$.  Otherwise, exactly half of the amplitudes of $\ket j|\psi>$ change sign, hence, the state becomes orthogonal to $\psi$.  Thus, after \rf(line:deutsch:reverse), the state of the algorithm is $\pm\ket j|0>$ if the input is negative, and it is orthogonal to $\ket j|0>$ otherwise.  In the latter case, the probability of obtaining outcome 0 during the measurement in \rf(line:deutsch:measure) is zero.  Hence, the algorithm never errs.  And, it uses only one quantum query.
\end{exm}

\subsection{Related Notions}
\label{sec:queryRelated}
In this section, we briefly consider relations between the various query complexity notions defined in the previous section.  Firstly, we have
\[
1\le Q(f) \le R(f) \le D(f) \le \vars
\]
for any non-constant function $f\colon [q]^\vars\supseteq \cD\to[\vyhod]$, because a deterministic decision tree is a special case of a randomised one, and any randomised computation can be simulated by a quantum computation.  Finally, a deterministic algorithm can query all $\vars$ variables and thus detect the value of the function.

\paragraph{Certificate Complexity}
%In the non-deterministic settings, one party, having perfect knowledge of the input $z$, has to convince another party in the value $f(z)$ of the function.  For query complexity, this is formalised by the notion of certificate complexity.
An {\em assignment} on $\vars$ variables is a function $\alpha\colon S\to [q]$ with $S\subseteq [\vars]$.  The {\em size} of $\alpha$ is $|S|$.  We say an input $z\in [q]^\vars$ {\em satisfies} assignment $\alpha$ iff $\alpha(j) = z_j$ for all $j\in S$.  For each subset $S\subseteq[\vars]$, there is unique assignment $z_S\colon S\to[q]$ that is satisfied by $z$.  We say inputs $x$ and $y$ {\em agree} on $S$ if $x_S = y_S$.

An assignment $\alpha$ is called a {\em $b$-certificate} for $f$, with $b\in[\vyhod]$, iff $f(z)=b$ for any $z\in \cD$ satisfying $\alpha$.  For a fixed $z\in\cD$, we call a subset $S\subseteq[\vars]$ a certificate for $z$ if $z_S$ is a certificate for $f$.

The {\em certificate complexity} $C_z(f)$ of $f$ on $z\in\cD$ is defined as the minimal size of a certificate for $f$ that $z$ satisfies.  The certificate complexity $C(f)$ of the function $f$ is defined as the maximum of $C_z(f)$ over all $z\in\cD$.  For $b\in[\vyhod]$, we define the $b$-certificate complexity $C^{(b)}(f)$ of $f$ as $\max_{z\in f^{-1}(b)} C_z(f)$.  

Almost all functions considered in the thesis are with Boolean output and with bounded 1-certificate complexity.  For the convenience, we list them here.

\begin{defn}[$k$-threshold]
\label{defn:threshold}
The input to the {\em $k$-threshold function} is a binary string $z \in\{0,1\}^\vars$.  The value of the function is 1 iff the Hamming weight of $z$ is at least $k$, i.e., there exist
$1\le a_1<a_2<\dots<a_k\le \vars$ such that $z\elem[a_1]=\cdots=z\elem[a_k]=1$.  The 1-certificate complexity of this function is $k$.  The 1-threshold function is the {\em OR function}.
\end{defn}

\begin{defn}[$k$-distinctness]
\label{defn:kdist}
Let $k\ge 2$ be a fixed integer.  The {\em $k$-distinctness function}, given a string $z\in[q]^\vars$ as its argument, evaluates to 1 iff there is a $k$-tuple of equal elements in the input, i.e., there exist $1\le a_1<a_2<\dots<a_k\le \vars$ such that $z\elem[a_1] = z\elem[a_2] = \cdots = z\elem[a_k]$.  The 1-certificate complexity of this function is $k$.  For the 2-distinctness function, we use the name {\em element distinctness}.
\end{defn}

\begin{defn}[$k$-sum]
\label{defn:ksum}
The {\em $k$-sum function}, given a string $z\in [q]^\vars$, evaluates to 1 iff there exist indices $1\le a_1<a_2<\dots<a_k\le \vars$ such that $z\elem[a_1] + \cdots + z\elem[a_k]$ is divisible by $q$.  The $k$-sum problem has 1-certificate complexity $k$.
\end{defn}

\begin{defn}[Graph collision]
\label{defn:graphCollision}
Let $G$ be a fixed graph with vertices labelled by integers in $[\vars]$.  The {\em graph collision} function, given a string $z\in\{0,1\}^\vars$, evaluates to 1 iff there exist an edge $ab$ of the graph such that $z_a = z_b = 1$.  This function has 1-certificate complexity 2.
\end{defn}

\mycommand{str}{z}
\begin{defn}[Collision]
\label{defn:collision}
Let $n$ be an integer, and $\vars = 2n$.  
Given an input string $\str\in[q]^\vars$, the task is to distinguish whether $\str$ is 1-to-1 or 2-to-1.  That is, in the negative case, all the elements of $\str$ are distinct.  In the positive case, there exists a decomposition of the input variables
\begin{equation}
\label{eqn:collisionDecomposition}
[\vars]=\{a_1,b_1\}\sqcup\{a_2,b_2\}\sqcup\cdots\sqcup\{a_n,b_n\}
\end{equation}  into $n$ disjoint pairs such that $\str\elem[a_i]=\str\elem[b_i]$ for all $i\in[n]$, but $\str\elem[a_i]\ne \str\elem[a_j]$ for all $i\ne j$.
\end{defn}

\begin{defn}[Set Equality and Hidden Shift]
\label{defn:setEquality}
Both problems are defined as the collision problem with additional promises in the positive case.  In the set equality problem, we are promised that $a_i$ and $b_i$ from~\rf(eqn:collisionDecomposition) satisfy $1\le a_i\le n$ and $n+1\le b_i\le n$.  
In the hidden shift problem, besides that, we are promised that there exists $d\in[n]$ such that $a_i\equiv b_i+d\pmod{n}$ for all $i\in[n]$.
\end{defn}

The 1-certificate complexity of collision, set equality and hidden shift is 2.  However, the problems are much easier to compute than the element distinctness problem because there are much more certificates.

\begin{defn}[Triangle]
\label{defn:triangle}
In the triangle problem on $n$ vertices, the input is a binary string $\str$ of length $\vars={n\choose 2}$.  We index the input variables by $\str_{ij}$ where $1\le i<j\le n$ are integers.  The task is to detect whether there exist indices $1\le a<b<c\le n$ such that $\str_{ab} = \str_{ac} = \str_{bc} = 1$.
Graph-theoretically, given a graph on $n$ vertices encoded by its adjacency matrix, the task is to detect whether it contains a triangle, i.e., a complete subgraph on 3 vertices.  This function has 1-certificate complexity 3.
\end{defn}

\paragraph{Block Sensitivity}
For proving lower bounds on deterministic and randomised query complexities, the following result is quite useful.  Recall that the {\em Hamming weight} of a binary string is the number of occurrences of symbol 1 in it.

\begin{thm}
\label{thm:randomORlower}
Let $\cD\subseteq \{0,1\}^\vars$ consist of all Boolean strings of Hamming weight at most 1.  The function OR on domain $\cD$ has deterministic query complexity $\vars$ and randomised query complexity at least $\vars/3$.
\end{thm}

\mycommand{OR}{\mathop{\mathrm{OR}}}

\pfstart
In the deterministic case, it is possible to define the input so that the answer to the first $\vars-1$ queries is 0.  And after that, the decision tree still does not know what the value of the function is.

The randomised case is similar.
Assume that there exists a randomised query algorithm $\mu$ that evaluates OR on $\cD$ in less than $\vars/3$ queries.  Given the algorithm, we come up with a positive input $x$ such that $\mu$ fails to distinguish it from the all-0 string $y$.

Recall that $\mu$ is a probability distribution on deterministic decision trees $T$ of depth less than $\vars/3$.  On input $y$, the tree $T$ queries less than $\vars/3$ variables.  Thus, the set $S(T)$ of variables not queried by $T$ satisfies $|S(T)|\ge 2n/3$.  For any $j\in S(T)$, we have $T(x^{(j)})=T(y)$, where $x^{(j)}$ is the input with the $j$th bit is set to 1 and all other bits equal to 0.

By the linearity of expectation, there exists $k\in[\vars]$ such that the $k$th input variable is not queried by $\mu$ on the input $y$ with probability greater than $2/3$.
Conditioned on not querying the $k$th input variable, $\mu$ outputs on $y$ one of the values, 0 or 1, with probability at least $1/2$.  This probability is the same for the input $x^{(k)}$ (as the $k$th input variable is not queried).
If $\mu$ outputs 0 with larger probability, then $\mu$ outputs 0 on the input $x^{(k)}$ with probability greater than $1/3$.
Otherwise, $\mu$ outputs 1 on $y$ with probability greater than $1/3$.  Both cases contradict the assumption that $\mu$ calculates the OR function.
\pfend

\mycommand{bs}{\mathop{\mathrm{bs}}}
\mycommand{sen}{\mathop{\mathrm{s}}}
This motivates the following definition.
The {\em block sensitivity} $\bs(f)$ of $f$ is defined as the maximal possible $b$ over all sequences $y,x^{(1)},\dots, x^{(b)}\in\cD$ satisfying the following two properties.  Firstly, for all $i\in[b]$, $f(x^{(i)})\ne f(y)$.  Let $B_i = \{j\in[\vars]\mid x^{(i)}\elem[j] \ne y\elem[j]\}$.  (Recall that $x\elem[j]$ stands for the $j$th symbol of $x$.)  The second property is that the subsets $B_1,\dots,B_b$ are pairwise disjoint.  The {\em sensitivity} $\sen(f)$ is defined similarly with the additional requirement that $|B_i|=1$ for all $i\in[b]$.

\begin{thm}
\label{thm:R2bs}
For any function $f$, $D(f)\ge\bs(f)$ and $R(f)\ge\bs(f)/3$.
\end{thm}

\pfstart
Let $b=\bs(f)$.  We reduce the calculation of OR over $b$ bits to the calculation of $f$ on inputs $y,x^{(1)},\dots,x^{(b)}$.  Assume we are given the input oracle $\cO$ to a bit-string $z\in \{0,1\}^b$ of Hamming weight at most 1.  We use it to simulate an oracle $\cO'$ encoding one of $x^{(i)}$ or $y$.  More precisely, if $z\in \{0,1\}^b$ is the all-0 string, then $\cO'$ encodes $y$, otherwise, it encodes $x^{(i)}$, where $i$ is the index of the non-zero entry in $z$.

The simulation is as follows.  Let $j$ be the index queried to $\cO'$.  If $j$ lies outside of all $B_i$, we return $y\elem[j]$.  If $j\in B_i$, then we query $\cO$ for the value of $z_i$.  If $z_i=0$, we return $y\elem[j]$, otherwise, we return $x^{(i)}\elem[j]$.  It is easy to see that $\cO'$ works as intended, and each oracle access to $\cO'$ costs at most one oracle access to $\cO$.  Thus, we can use a query algorithm for $f$ to calculate the OR function on $b$ bits in at most the same number of queries.  Together with \rf(thm:randomORlower), this implies the statement of the theorem.
\pfend

From \rf(exm:balancedClassical), we can see that there can be arbitrarily large gap between the deterministic and the randomised query complexities of a function.  For total functions, we have the following result.
\begin{thm}[\cite{beals:pol}]
\label{thm:D2bs}
For a total Boolean function $f$, its deterministic query complexity satisfies $D(f)\le \bs(f)^3$.
\end{thm}
By combining this with \rf(thm:R2bs), we get $D(f)\le R(f)^3$ for any total Boolean function $f$.

\section{Chapter Notes}
The material in Sections~\ref{sec:computationalDevices}---\ref{sec:circuits} is rather standard.  The reader is advised to consult any textbook on quantum computation, e.g. the book by Nielsen and Chuang~\cite{chuang:quantum} or the book by Kitaev \etal~\cite{kitaev:quantum}, for more details.  Our notation is slightly different from the notation in these books and is partly inspired by the lecture notes by Watrous~\cite{watrous:lectures}.  For a detailed history of development of quantum mechanics and quantum computing, we refer the reader to~\cite{chuang:quantum}.

The first model of quantum computation was the quantum Turing machine introduced by Benioff~\cite{benioff:quantumTuring} in 1980.  
A more modern version is due to Bernstein and Vazirani~\cite{bernstein:quantumComplexity}.  See also~\cite{melkebeek:newQuantumTuring}. Subsequently, it was replaced by a more natural notion of quantum circuits developed by Deutsch~\cite{deutsch:quantumCircuits} in 1989 and proven to be equivalent to quantum Turing machine by Yao~\cite{yao:quantumCircuits}.  A reason that quantum circuits are more popular than quantum Turing machines is that their uniformity can be shown by a {\em classical} Turing machine.  In the classical case, one is not able to fully replace Turing machines by circuits because the uniformity condition still requires an alternative model of computation.

Quantum RAM machines from \rf(sec:RAM) seem to be less popular than quantum circuits, although their implicit use, as we will see in the next chapters, is quite wide-spread.  For more details on our model and the related topics we refer the reader to the PhD thesis by \"Omer~\cite{omer:quantumRAM}.  The distinction between different variants of arrays is adapted from~\cite{kuperberg:anotherDihedral}.  Low level realisation of quantum random access memory is studied in~\cite{giovannetti:quantumRAM, hong:robustRAM}.  Our pseudo-code notation is mostly based on a technical report by Knill~\cite{knill:pseudocode}.  For a survey on various models of quantum computation see, e.g., \cite{miszczak:models}.

The results in \rf(sec:specification) are well-known.  An analogue of \rf(lem:circuitRobust) can be found in~\cite{bernstein:quantumComplexity}.  

For a more detailed exposition of the topics in \rf(sec:query), we refer the reader to the survey by Buhrman and de Wolf~\cite{buhrman:querySurvey}.  The notion of block sensitivity and the proof of \rf(thm:R2bs) is due to Nisan~\cite{nisan:bs}.

%% file: _walk.tex
A random walk is a randomised algorithm that works in the following way.  The algorithm only keeps track of its current {\em state}.  Given the state, the algorithm checks whether it satisfies some specified properties.  If it does, the algorithm stops and outputs the state (we say that the state is {\em marked} in this case).  Otherwise, the algorithm applies a small random transformation to the state and repeats the same procedure.  It is common to represent the possible states as the vertices of a graph with two vertices adjacent iff one is reachable from the other in one step.

Random walks have been successfully applied to a variety of computational problems.  A random walk can be a suitable choice if some restrictions are present.  One restriction could be that the complete set of possible states is not given in advance: Imagine a robot trying to get out of a labyrinth.  Another possibility is that the set of the states has a very complicated structure.  For instance, many exact algorithms for constraint satisfiability~\cite{schoning:ksat, minton:minimizingConflicts} were constructed using these ideas.  Finally, it is possible that the set of the states is available and simple, but the algorithm does not have enough memory to store all of them.  A famous example is given by the $st$-connectivity algorithm running in logarithmic space~\cite{aleliunas:connectivity}.  (Later this algorithm was successfully derandomised~\cite{reingold:connectivity}.)

In the quantum settings, quantum walks have much greater importance.  As first realised by Grover~\cite{grover:search}, and then more explicitly by Ambainis~\cite{ambainis:distinctness}, quantum walks are effective even if there are no restrictions in accessing, processing or storing the input.

The purpose of this chapter is mostly illustrative.  We describe techniques that were used to obtain quantum algorithms that we improve on in the second part of the thesis.  The chapter is organised as follows.
In \rf(sec:walkClassical), we describe three different models of classical random walks.  In Sections~\ref{sec:amplitudeAmplification} and~\ref{sec:walkQuantum}, we describe the quantum counterparts of two of these random walks, and give examples of their applications.

\mycommand{setup}{T_S}
\mycommand{update}{T_U}
\mycommand{checking}{T_C}
\mycommand{lucky}{\eps}
\mycommand{spectralgap}{\delta}

\section{Classical Random Walks}
\label{sec:walkClassical}
Before tackling quantum walks, it is worth getting acquainted with the random ones.  
The settings we consider in this section are not typical for random walks, but they are the closest analogue of the quantum walks we will consider in the next sections.
We solve the task of finding (alternatively, detecting the presence of) a {\em marked} element in some set.  Formally, the settings are as follows:
\begin{defn}[Detection and search problems]
\label{defn:problem}
Let $X$ be a finite set with $n=|X|$.  We assume we have perfect knowledge of $X$.  Additionally, an unknown set $M\subseteq X$ of {\em marked} elements is fixed.  In the {\em detection} problem, the task is to distinguish whether $M$ is empty or non-empty.  In the {\em search} problem, we are promised that $M$ is non-empty, and the task is to output any element $x\in M$.
\end{defn}

Each element of $X$ has some {\em data} associated with it.  We denote it by $d(x)$, and assume it belongs to some finite set $D$.  This data is required to detect whether the element is marked.

We are given some procedures assisting us in the tasks.  For each of these procedures, some abstract {\em costs} are assigned.  We are interested in minimising the total cost of the algorithm, i.e., the sum of the costs of all the procedures executed by the algorithm.  All operations except executing the procedures are assumed to be costless.
The motivation for this convention is twofold.  Firstly, in many cases, this framework is used for query algorithms, and, since $X$ is given in advance, all operations manipulating the elements of $X$ are indeed costless.  Secondly, even in the time-efficient settings, the wrapping part of the algorithm is so simple that its cost can be neglected.

\mycommand{stationary}{\sigma}

\begin{defn}[Set-up]
\label{defn:setup}
The set-up procedure performs two operations.  At first, it samples an element $x\in X$ according to some known probability distribution $\stationary_x$.  After that, the procedure generates some {\em data} $d(x)$ {\em associated} with the element.  
There is a pre-defined threshold $\lucky>0$ such that, if $M$ is non-empty, the set-up procedure samples a marked element with probability at least $\lucky$, i.e., $\sum_{x\in M} \stationary_x\ge\lucky$.  
The cost of the set-up procedure is denoted by $\setup$.
\end{defn}

\begin{defn}[Check]
\label{defn:check}
Given $x\in X$ together with the associated data $d(x)$, the checking procedure decides whether the element is marked.  The cost of the checking procedure is denoted by $\checking$.
\end{defn}

At first, it may seem unclear why the checking procedure is separated from the set-up procedure, but it will become apparent later, after the introduction of the update operation.

Given these two procedures, one can come up with a simple \refalg{simplest}. It is easy to see that \refalg{simplest} returns a marked element with probability $\Omega(1)$ with the total cost of $O(\setup+\checking)/\lucky$.

\begin{algorithm}
\caption{A Simple Search Algorithm}
\label{alg:simplest}
\algbegin
\state{{\bf Repeat} $\Theta(1/\lucky)$ times :}
\tab
\state{Sample $x$ and construct $d(x)$ using the set-up procedure}
\state{Check if $x$ is marked, and if it is, output $x$ and stop}
\algend
\end{algorithm}

Sometimes, it can be too expensive to set-up an element $x$ from scratch in every iteration of the loop in \rf(alg:simplest).
Instead of that, we would like to transform the data associated with the element of the previous iteration into the data of the new element.  In general, it can be infeasible.  So, for $x\in X$, let $N(x)$ denote the set of $y\in X$ such that $d(y)$ can be easily obtained from $d(x)$.  In each case, we define $N(x)$ explicitly.

\begin{defn}[Update]
\label{defn:update}
Given $x\in X$, the associated data $d(x)$ and an element $y\in N(x)$, the update procedure returns the data $d(y)$ associated with $y$.  The cost of the update procedure is denoted by $\update$.
\end{defn}

Thus, given an element $x\in X$, we could potentially move to any element in $N(x)$ using the update procedure.  However, it is still unclear to which element we should move.
%The choice of the particular element is performed by the following procedure.
In order to specify this, we adopt the following convention.  First, assume that $y\in N(x)$ if and only if $x\in N(y)$.  Let $G$ be the undirected graph with the vertex set $X$ and vertices $x,y\in X$ connected iff $x\in N(y)$.  We  assume $G$ is connected.  Each edge $xy$ of the graph is assigned a positive {\em weight} $w_{xy}$.  (Since $G$ is undirected, we have $w_{xy}=w_{yx}$.) The next element we proceed to is chosen randomly with probability proportional to the weight.  More formally, if $w_y = \sum_{x\in N(y)} w_{yx}$, then the probability of going from $y$ to $x$ is $p_{xy} = w_{xy}/w_y$.  This is known as the diffusion operation.

Let $P=(p_{xy})$ be the corresponding $X\times X$ matrix.  Matrices constructed in such a way correspond to what are known as {\em reversible random walks}.  If $u$ is a probability distribution on $X$ written as a column-vector, the probability distribution after the diffusion operation is $Pu$.  Any such matrix $P$ is {\em stochastic}: it is {\em non-negative} (all its entries are non-negative) and the sum of each column is 1.  

\begin{prp}
\label{prp:stationary}
The vector $w = (w_x)_{x\in X}$ is an eigenvector of $P$ with eigenvalue 1.
\end{prp}

\pfstart
Indeed, for all $x\in X$,
\[
(Pw)\elem[x] = \sum_{y\in N(x)} p_{xy}w_y = \sum_{y\in N(x)} w_{xy} = w_x.\qedhere
\]
\pfend

Define the $X\times X$ matrix $P'$ by $P'\elem[x,y] = w_{xy}/\sqrt{w_xw_y}$, and let $W$ be the diagonal $X\times X$ matrix given by $W\elem[x,x] = \sqrt{w_x}$.  The following statement is trivial.

\begin{prp}
\label{prp:walkPprime}
We have $P' = W^{-1}PW$.  In particular, $P'$ has the same eigenvalues as $P$ does.  
The vector $(\sqrt{w_x})_{x\in X}$ is a 1-eigenvector of $P'$. 
\end{prp}

From this point on, we assume that $P$ is {\em aperiodic}, i.e., there exists a positive integer $i$ such that all entries of $P^i$ are positive.  The Perron-Frobenius theorem (cf. \rf(sec:linearAlgebra)) implies the following result:

\begin{thm}
\label{thm:perron}
Every stochastic aperiodic matrix $P$ has unique 1-eigenvector.  All other eigenvalues of $P$ are strictly less than 1 in absolute value.
\end{thm}

The 1-eigenvector, normalised so that the sum of its entries is 1, is called the {\em stationary distribution} of $P$.  By \rf(prp:stationary), it is proportional to $w$.
Let $\lambda_1, \lambda_2, \dots, \lambda_n$ be the eigenvalues of $P$ arranged by their absolute values: $|\lambda_1|\ge|\lambda_2|\ge\cdots\ge|\lambda_n|$.
  Hence, $\lambda_1=1$, and $|\lambda_2|<1$.  The value $\spectralgap = 1-|\lambda_2|$ is known as the {\em spectral gap} of $P$.

We assume that we have the following procedure.

\begin{defn}[Diffusion]
\label{defn:diffuse}
Given an element $y\in X$, the diffusion procedure samples an element $x\in N(y)$ according to the probability distribution given by a stochastic aperiodic matrix $P=p_{xy}$ such that the stationary distribution of $P$ equals the probability distribution $(\stationary_x)$ from \rf(defn:setup).  We assume that the diffusion procedure is costless.
\end{defn}

\begin{algorithm}
\caption{Random Walk Algorithm}
\label{alg:walksimple}
\algbegin
\state{Sample $x$ and construct $d(x)$ using the set-up procedure}
\state{{\bf Repeat}:}
\tab
\state{Check if $x$ is marked, and if it is, output $x$ and stop}
\state{Otherwise, diffuse to an element in $N(x)$ and update the data accordingly}
\algend
\end{algorithm}

Given all these operations, one may come up with \refalg{walksimple}.  The average number of iterations of the loop in \refalg{walksimple} performed before a marked element is reached is called the {\em hitting time} of $P$.  It can be estimated as follows.

\begin{prp}
\label{prp:hitting}
The hitting time of \refalg{walksimple} is $O\s[\frac1{\lucky\spectralgap}]$ where $\eps$ is as in \rf(defn:setup) and $\delta$ is the spectral gap of $P$.
\end{prp}

\pfsketch
Let $P'$ and $W$ be as in \rf(prp:walkPprime).  The matrix $P'$ is symmetric, hence, it has an orthonormal set of eigenvectors $v_1,\dots,v_n$ that correspond to the eigenvalues $1=\lambda_1 \ge |\lambda_2|\ge\dots \ge |\lambda_n|$ of $P$.  Also, $v_1\elem[x] = \sqrt{\stationary_x}$ by the assumption on the stationary distribution of $P$.

Assume $C$ is some large constant.  For $i>1$ and $\ell = C/\spectralgap$, we have 
\[
|\lambda_i^{\ell}| \le (1-\spectralgap)^{C/\spectralgap} \le \ee^{-C}\approx 0.
\]
Hence, $(P')^\ell \approx v_1v_1^*$.  
Thus, for $P^\ell = W (P')^{\ell} W^{-1}$, we have $P^\ell\elem[x,y] \approx \sigma_x$.  That is, $P^\ell$ transforms every probability distribution into a distribution close to $\sigma$.  Hence, performing $C/\spectralgap$ steps on the random walk is approximately equivalent to sampling an element using the set-up procedure.  Performing the loop in \refalg{walksimple} $\Theta(1/(\lucky\spectralgap))$ times is approximately as good as executing \refalg{simplest}.
\pfend

\begin{algorithm}
\caption{Forbearing Random Walk Algorithm}
\label{alg:forbearing}
\algbegin
\state{Sample $x$ and construct $d(x)$ using the set-up procedure}
\state{{\bf Repeat} $\Theta(1/\lucky)$ times :}
\tab
\state{Check if $x$ is marked, and if it is, output $x$ and stop}
\state{{\bf Repeat} $\Theta(1/\spectralgap)$ times :}
\tab
\state{Go to an element in $N(x)$ and update the data accordingly}
\algend
\end{algorithm}

If the checking cost $\checking$ is relatively high, it may be a good idea to avoid checking the element after each step of the walk.  Inspired by the proof of \refprp{hitting}, one may come up with \refalg{forbearing}.  
The results of this section are summarised in the following
\begin{thm}[Classical random walks]
\label{thm:walkClassical}
Assume the set-up and checking procedures are as in Definitions~\ref{defn:setup} and~\ref{defn:check}.  Then, \refalg{simplest} finds a marked element with $\Omega(1)$ probability.  Assume, additionally, that the update procedure as in \refdefn{update} is given, where $P$ corresponds to a reversible random walk that is aperiodic and has $(\stationary_x)$ from \refdefn{setup} as its stationary distribution.  Then, Algorithms~\ref{alg:walksimple} and~\ref{alg:forbearing} find a marked element with $\Omega(1)$ probability.  The total costs of the algorithms are given by
\[\begin{tabular}{rl}
\refalg{simplest} & $O(\setup+\checking)/\lucky$ \\
\refalg{walksimple} & $ O(\setup + H(\checking+\update)) = O\s[\setup+\frac{\checking+\update}{\lucky\spectralgap}]$ \\
\refalg{forbearing} & $ O\s[\setup + \frac1{\lucky} {\s[\checking + \frac{\update}{\spectralgap}]}]$
\end{tabular}\]
Here, $\eps$ is as in \rf(defn:setup), $H$ is the hitting time, and $\spectralgap$ is the spectral gap of $P$.
\end{thm}

We will present a quantum analogue of \rf(alg:simplest) in \rf(sec:amplitudeAmplification) and a quantum analogue of \rf(alg:forbearing) in \rf(sec:walkQuantum).

\section{Amplitude Amplification}
\label{sec:amplitudeAmplification}
In this section, we develop a quantum analogue of \rf(alg:simplest).  We define a {\em step of a quantum walk} as a correspondingly chosen unitary operator, and study its behaviour on some initial quantum state.  In \rf(sec:phaseDetection), we describe a technical tool for separating eigenvectors of a unitary operator $U$ based on their eigenvalues.  In \rf(sec:amplification), we apply this tool to the detection of marked elements.  Finally, in \rf(sec:amplificationExamples), we give some applications of the algorithm from \rf(sec:amplification).

\subsection{Quantum Phase Detection}
\label{sec:phaseDetection}
As described in the proof of \rf(prp:hitting), the repeated application of the step of the random walk results in the eigenvectors with eigenvalues smaller than 1 fading away: The probability distribution converges to the stationary one.
In the quantum case, a step of a quantum walk is a unitary operation whose repeated application does not converge.  
In this section, we show how to overcome this by averaging over time.

Let $U$ be a unitary operator acting on a register $\reg A$.  We develop a subroutine that distinguishes the 1-eigenvectors of $U$ from the eigenvectors with other eigenvalues.  Recall that 
$U$ is a normal operator, hence, it has an orthonormal eigenbasis $\Psi = \{\psi_1,\dots, \psi_m\}$ (cf. \rf(sec:linearAlgebra)).  The corresponding eigenvalues are of the form $\ee^{\ii\theta_j}$ with $-\pi < \theta_j\le \pi$.  
In the notations of \rf(sec:Functions), the phase detection procedure evaluates the function $g_U\colon \Psi\to\{0,1\}$ defined by $g_U(\psi_j) = 0$ iff $\theta_j = 0$.  
Intuitively, it is clear that the complexity of this problem has to depend on the {\em phase gap} of $U$ that is the minimal non-zero value of $|\theta_j|$.

%\mycommand{detect}{\circuit D_{U}}
\mycommand{gap}{\delta}
\mycommand{error}{\gamma}

\begin{thm}
\label{thm:detection}
Let $\frac12\ge \error>0$ and $\gap>0$ be fixed real numbers, and let $U$, $\reg A$, and $g_U$ be as above.  Then, there exists a quantum procedure that evaluates the function $g_U$ with one-sided error $\error$ for every unitary $U$ acting on the register $\reg A$ and having phase gap at least $\delta$.  The procedure uses $O\sA[\frac{1}{\gap}\log\frac{1}{\error}]$ controlled applications of $U$.
\end{thm}

\pfstart
We prove the theorem for the special case of $\error=1/2$.  The general case then follows by \reflem{circuitAmplify}.  Let $\reg O$ be the output qubit.  Define $K = \lceil 8/\gap\rceil$, and let $\reg W$ be the working register with $\{0,\dots,K-1\}$ as the computational basis.  
%Denote the initial state of $\reg W$ by $\phasestart$.
The description of the circuit is given in \refalg{phaseDetection}.

\begin{algorithm}
\caption{Quantum Phase Detection}
\label{alg:phaseDetection}
\algbegin
\state{{\bf qufunction} PhaseDetection({\bf quprocedure} $U$, {\bf real} $\gap$, {\bf registers} $\qreg A, \qreg O$) {\bf with} 1-sided error $1/2$ {\bf :}}
\tab
\state{$K \gets \lceil 8/\gap\rceil$ }
\state{{\bf attach} $K$-qudit $\qreg W$ }
\state{UniformSuperposition($\qreg W$)} \label{line:phaseDetection:uniform}
\state{{\bf for} $k=1,\dots,K-1$ {\bf:}} \label{line:phaseDetection:loop}
\tab
\state{{\bf conditioned on} $\qreg W \ge k$ {\bf :} $U(\qreg A)$}
\untab
\state{UniformSuperposition$^{-1}$($\qreg W$)} \label{line:phaseDetection:ununiform}
\state{{\bf conditioned on} $\qreg W \ne 0$ {\bf :} $\qreg O\qgets 1$} \label{line:phaseDetection:final}
\untab
\algend
\end{algorithm}

Assume the initial state of the procedure is $\ket A|\psi>\ket W|0> \ket O|0>$, where $\psi$ is an $\ee^{\ii\theta}$-eigenvector.  After \rf(line:phaseDetection:uniform), the state of the algorithm is
\begin{equation}
\label{eqn:detection1}
\frac{1}{\sqrt{K}}\sum_{k=0}^{K-1} \ket A|\psi>\ket W|k>\ket O|0>.
\end{equation}
The loop in \rf(line:phaseDetection:loop) transforms it into
\begin{equation}
\label{eqn:detection2}
\frac{1}{\sqrt{K}}\sum_{k=0}^{K-1} \ket A|U^k\psi>\ket W|k>\ket O|0> = 
\frac{1}{\sqrt{K}}\sum_{k=0}^{K-1} \ee^{\ii k\theta} \ket A|\psi>\ket W|k>\ket O|0> = 
\ket A|\psi> \otimes\s[\frac{1}{\sqrt{K}}\sum_{k=0}^{K-1} \ee^{\ii k\theta}\ket W|k> ]\otimes \ket O|0>.
\end{equation}
The state of the algorithm after \rf(line:phaseDetection:ununiform) is of the form $\ket A|\psi> \ket W|\upsilon_\psi> \ket O|0>$, and the final state of the procedure is of the form $\ket A|\psi>\ket{{WO}}|\omega_\psi>$ as required by the definition of non-coherent function evaluation.

If $\theta = 0$, then the state on the right hand side of~\refeqn{detection2} equals the state in~\refeqn{detection1}, hence, the state after \rf(line:phaseDetection:ununiform) equals the initial state of the procedure.  Then, \rf(line:phaseDetection:final) has no effect, and the output register $\reg O$ contains 0 with certainty.

Now assume $|\theta|\ge\gap$.  Since the unitary operator in \rf(line:phaseDetection:ununiform) does not affect the inner product, we have 
\[
\abs| \ip<\upsilon_\psi, e_0> | = 
\abs| \ipC<\frac{1}{\sqrt{K}}\sum_{k=0}^{K-1} e_k, \frac{1}{\sqrt{K}}\sum_{k=0}^{K-1} \ee^{\ii k\theta} e_k >| = 
\absC| \frac{1}{K}\sum_{k=0}^{K-1} \ee^{\ii k\theta} | \le 
\abs| \frac{\gap(1-\ee^{\ii K\theta})}{8(1-\ee^{\ii\theta})} | \le \frac12
\]
because $\abs| 1-\ee^{\ii K\theta} |\le2$, and $\abs| 1-\ee^{\ii\theta} | \ge |\theta|/2$ for all $|\theta|\le\pi$. Thus, after \rf(line:phaseDetection:final), the output register contains 1 with probability at least $1/2$.
\pfend

\subsection{Algorithm}
\label{sec:amplification}
\mycommand{start}{\varsigma}
\mycommand{qsetup}{S}
\mycommand{qcheck}{C}
\mycommand{total}{{\psi}}
\def\data#1{d_{#1}}
%In this section, we describe a quantum counterpart of \refalg{simplest}.  
The quantum analogue of \rf(alg:simplest) uses two registers: the {\em index register} $\reg X$ and the {\em data register} $\reg D$.  For the first register, we have $\hilbert X = \C^X$, where $X$ is the same set as in \refsec{walkClassical}.  The second register stores a unit {\em data vector} $\data x\in\hilbert D$ associated with $x\in X$.  The quantum counterparts of the set-up and checking procedures are as follows.

\begin{defn}[Quantum Set-up]
\label{defn:quantumSetup}
A {\em quantum set-up procedure} is a unitary $\qsetup$ that maps the initial state $\ket|0>$ into a state of the form 
\begin{equation}
\label{eqn:quantumSetup}
\ket|\total> =  \sum_{x\in X} \alpha_x \ket  X|x> \ket D|\data x>.
\end{equation}
The probability distribution associated with $\qsetup$ is $\stationary_x = |\alpha_x|^2$.  We again assume that the probability of $x$ being marked is at least $\lucky$, i.e., $\sum_{x\in M} |\alpha_x|^2\ge\lucky$.  The cost of this subroutine is denoted by $\setup$.
\end{defn}

\begin{defn}[Quantum Check]
\label{defn:quantumCheck}
The {\em quantum check function} $\qcheck$, given the element and the associated data, evaluates whether the corresponding element is marked.  Thus, in the notations of \rf(sec:Functions), $\Psi = \{\ket X|x>\ket D|\data x> \mid x\in X\}$, and the function is defined by $\ket X|x>\ket D|\data x> \mapsto 1$ if $x$ is marked, and 0, otherwise.  The cost of the checking subroutine is $\checking$.
%\[
%\qcheck \sA[\ket  X|x> \ket D|\data x>] =
%\begin{cases}
%{-\ket X|x> \ket D|\data x>},& \text{$x$ is marked;}\\
%{\ket X|x> \ket D|\data x>},& \text{$x$ is not marked.}\\
%\end{cases}
%\]
\end{defn}

By \refthm{walkClassical}, classical \refalg{simplest} finds a marked element with cost $O(\setup+\checking)/\lucky$.  
Quantumly, we can get better complexity by appling the ``linearity in the square roots of probabilities'' as mentioned in the introduction.

\begin{thm}
\label{thm:amplification}
Assume the quantum set-up and checking procedures are as in Definitions~\ref{defn:quantumSetup} and~\ref{defn:quantumCheck}.  Then, there exists a quantum procedure that detects the presence of a marked element with one-sided error $1/2$ and the total cost  $O(\setup+\checking)/\sqrt{\lucky}$.
\end{thm}

\mycommand{marked}{\psi_{\mathrm m}}
\mycommand{nonmarked}{\psi_{\mathrm n}}

\def\eket#1{e_{#1}}

\mycommand{step}{V}

\begin{algorithm}
\caption{Quantum Amplitude Amplification}
\label{alg:amplification}
\algbegin
\state {{\bf qufunction} AmplitudeAmplification({\bf quprocedures} Setup, Check, {\bf real} $\eps$, {\bf registers} $\qreg X$, $\qreg D$, $\qreg O$ ) {\bf with} 1-sided error $1/2$ {\bf:}}
\tab
	\state {Setup$(\qreg X, \qreg D)$}
%	\state {{\bf attach} $\qreg W$ and $\qreg O$}
	\state {PhaseDetection(StepOfWalk, $\sqrt{\lucky}$, $\qreg{XD}$, $\qreg O$) {\bf with} 1-sided error $1/2$}
\state {\ }
\state {{\bf quprocedure} StepOfWalk :}
\tab
	\state {$\phase \qgets$ Check$(\qreg X, \qreg D)$ }
	\state {ReflectAbout $\total$}
\untab
\state {\ }
\state {{\bf quprocedure} ReflectAbout $\total$ :}
\tab
	\state {Setup$^{-1}(\qreg X, \qreg D)$}
	\state {{\bf conditioned on} $\qreg{XD} \ne 0$ : apply $(-1)$-phase gate }
	\state {Setup$(\qreg X, \qreg D)$}
\untab \untab
\algend
\end{algorithm}

\pfstart
In general terms, the algorithm is as follows.  We define an input-dependent unitary transformation $\step$ (a step of the quantum walk) and an initial state $\psi$ such that
\itemstart
\item if there is no marked element, then $\psi$ is a 1-eigenvector of $\step$;
\item if there is a marked element, then $\psi$ belongs to the span of the eigenvectors of $\step$ with eigenvalues sufficiently away from 1.
\itemend
Then, we can use the phase detection subroutine to distinguish these two cases.  
This idea will appear repeatedly throughout the thesis.

The procedure is described in \refalg{amplification}.  The step of the quantum walk $\step$ is the composition of two reflections:  The first reflection is the checking subroutine $\qcheck$, and the second one is about $\psi$ as defined in \refeqn{quantumSetup}.  
% The reflection about $\psi$ is performed by applying $\qsetup^{-1}$, the ($-1$)-phase gate conditioned on the state not being equal to $\ket|0>$ and applying $\qsetup$ again.  
One step of the quantum walk costs $2\setup+\checking$.  The initial state of the quantum walk is $\psi$.

\begin{figure}[htb]
\[
\qquad\qquad\qquad
\xygraph{ !{0;<3cm,0pc>:}
(:[u][r(.1)]*{\marked'},
:[r][r(.3)]*{\psi_n' = \qcheck \nonmarked'},
:[r(0.96)u(0.26)][r(.1)]*{\psi},
[r(.4)u(.05)]*{\theta},
:[r(0.96)d(0.26)][r(.15)]*{\qcheck\psi}
)
[rr]
(:[u][r(.1)]*{\marked'},
:[r][r(.1)]*{\nonmarked'},
:[r(0.96)u(0.26)][r(.1)]*{\psi},
[r(.4)u(.05)]*{\theta},
:[r(0.87)u(0.5)][r(.15)]*{\step \nonmarked'},
:[r(0.71)u(0.71)][r(.15)]*{\step\psi}
)
}
\]
\vspace{-2cm}
\caption{An illustration of the step of the quantum walk in \TeXBug{\rf(alg:amplification)}.  Here it is assumed that both $\marked$ and $\nonmarked$ are non-zero, and $\marked' = \marked/\|\marked\|$ and $\nonmarked' = \nonmarked /\norm|\nonmarked|$ are the respective normalised vectors.
On the left, the images of $\psi$ and $\nonmarked'$ after the application of $\qcheck$ are shown.  On the right, the images of the same vectors under $\step$ are shown.  It is easy to see that $\step$ rotates both vectors $\psi$ and $\nonmarked'$ by the angle $2\theta$.  Since they span $H$, the step $\step$ of the quantum walk acts as the rotation by $2\theta$ in $H$.}
\label{fig:amplification}
\end{figure}

The initial state belongs to the invariant subspace $H$ of $\step$ spanned by (not normalised) vectors
\[
\nonmarked = \sum_{x\in X\setminus M}  \alpha_x \ket  X|x> \ket D|\data x>
\qquad\text{and}\qquad
\marked = \sum_{x\in M}  \alpha_x \ket  X|x> \ket D|\data x>.
\]
Indeed, $\nonmarked$ and $\marked$ are eigenvectors of $\qcheck$ with eigenvalues 1 and $-1$, respectively.  Also,
$\psi\in H$, hence, $H$ is an invariant subspace of the reflection about $\psi$.  Thus, we may restrict our attention to the spectrum of $\step$ in $H$.

If there are no marked elements, then $\psi = \nonmarked$ is a 1-eigenvector of $\step$.  
If all the elements are marked, then $\psi = \marked$ is a $(-1)$-eigenvector of $\step$.
In both these cases, $\step$ and $\psi$ possess the required properties.

In all other cases, the subspace $H$ is two-dimensional.  The step $\step$, restricted to $H$, is the composition of the reflections about $\nonmarked$ and $\psi$, thus, it is the rotation by $2\theta$, where $\theta$ is the angle between the two vectors (see \rf(fig:amplification)).  For the angle, we have 
\[
\sin\theta = \frac{\ip<\psi, \marked>}{\norm|\marked|} = \sqrt{\sum_{x\in M} |\alpha_x|^2} \ge \sqrt{\lucky}.
\]
Thus, $\pi/2 \ge \theta\ge \sqrt{\lucky}$.  The eigenvalues of the rotation by $2\theta$ are $\ee^{\pm 2\ii\theta}$.  Hence, by \refthm{detection}, we can distinguish the two cases using $O(1/\sqrt{\eps})$ steps of the quantum walk.
\pfend

The quantum walk in \rf(alg:amplification) is quite simple, and it is 
possible to use it to find marked elements, not just to detect their presence.
Recall from \rf(defn:problem) that, in the search problem, we are guaranteed that the set of marked elements $M$ is non-empty, and the task is to find an $x\in M$.

\begin{prp}
\label{prp:amlificationFind}
Assume a quantum set-up procedure, that generates the state $\ket |\psi>$ from~\rf(eqn:quantumSetup) in cost $\setup$, and the checking procedure from \rf(defn:quantumCheck) are available.  Then, there exists a quantum algorithm that finds a marked element with $\Omega(1)$ success probability and the total cost $O(\setup+\checking)/\sqrt{\lucky}$, where $\lucky = \sum_{x\in M} |\alpha_x|^2$.  It is not necessary to know $\lucky$ in advance.
\end{prp}

\pfsketch
The description of the algorithm is given in \rf(alg:amplificationSearch).  Here, StepOfWalk is as defined in \rf(alg:amplification), and $M>1$ is a constant.

\begin{algorithm}
\caption{Search Using Quantum Amplitude Amplification}
\label{alg:amplificationSearch}
\algbegin
\state {{\bf function} FindMarkedElement({\bf quprocedures} Setup, Check, {\bf registers} $\qreg X$, $\qreg D$)  {\bf:}}
\tab
	\state {{\bf real} $i\gets 1$}
	\state {{\bf repeat :}}
	\tab
		\state {{\bf repeat} $\Omega(1)$ times  {\bf:}}
		\tab
			\state {$\reg {XD}\gets 0$}
			\state {Setup$(\qreg X, \qreg D)$}
			\state {{\bf repeat} $\lfloor i\rfloor $ times {\bf:} }
			\tab \state {StepOfWalk} \untab
			\state {\bf{measure} $\qreg X$}
			\state {{\bf if} Check($\reg X$, $\qreg D$) = 1 {\bf:}
			{\bf return} $\reg X$}  
		\untab
	\state {$i\gets i\cdot M$}
	\untab
\untab
\algend
\end{algorithm}

We freely use the notations from the proof of \rf(thm:amplification).  \rf(fig:amplification) suggests that if $t\approx \pi/(4\theta)$, then $\step^t\psi$ is close to $\marked'$.  Thus, measuring the register $\reg X$ gives a marked element with sufficiently large probability. 
Unfortunately, we do not know $\theta$ in advance, and so the required value of $t$.  We seek the correct number of iterations using the geometric series $1,M,M^2,\dots$.  There exists $i$ such that $M^i<t\le M^{i+1}$.  If $M$ is small enough, then $\step^{M^i}\psi$ is also close to $\marked'$, and the measurement of $\reg X$ after $M^i$ iterations of the quantum walk yields a marked element with $\Omega(1)$ probability.  On the other hand, the total number of iterations of the quantum walk is dominated by the last term in the geometric series.
\pfend

%Approximate implementations of the checking procedure are also legit.
Assume that we have a quantum function $\qcheck'$ that, given an element $x\in X$ and the corresponding data vector $d_x$, $(1/3)$-evaluates whether $x$ is marked.  We would like to use it in the quantum amplification algorithm.
One solution is to reduce the error by Lemma~\ref{lem:circuitAmplify}.
The checking subroutine in \rf(alg:amplification) is executed $O(1/\sqrt{\lucky})$ times.  Thus, by \rf(lem:circuitRobust), it is enough to reduce the error to $O(1/\sqrt{\lucky})$ per execution.  This increases the total number of applications of $\qcheck'$ by the multiplicative factor of $O(\log\frac 1\lucky)$.  
However, there exists a way to perform amplitude amplification without increasing the cost.

\begin{thm}[Noisy Amplitude Amplification~\cite{hoyer:noisyGrover}]
\label{thm:noisyGrover}
Assume that we have a quantum set-up operation $\qsetup$ from \rf(defn:quantumSetup) and a subroutine $\qcheck'$ that, given element $x\in X$ and the corresponding data vector $d_x$, $(1/3)$-evaluates whether $x$ is marked.  Then, there exists a quantum circuit that detects the presence of a marked element with error $1/3$ and the total cost $O(\setup+\checking)/\sqrt{\lucky}$.
\end{thm}

\subsection{Applications}
\label{sec:amplificationExamples}
A very simple application of amplitude amplification is exhibited by the Grover search for calculating the OR function.  

\begin{prp}[Grover Search]
\label{prp:Grover}
The OR problem on $\vars$ variables can be solved in $O(\sqrt{\vars})$ quantum queries.  Moreover, an index $j\in [\vars]$, such that $z_j = 1$, can be found in $O(\sqrt{\vars/k})$ quantum queries, where $k=|z|$ is the Hamming weight of $z$.
\end{prp}

\pfstart
The Grover search is a simple special case of amplitude amplification, where there is no data register, $\qsetup$ is given by the Uniform Superposition procedure from \rf(exm:uniform), and $\qcheck$ is the input oracle itself.  The costs are $\setup=0$ and $\checking=1$, respectively.  If there are $k$ ones in the input string, then $\lucky = k/\vars$, hence, \refalg{amplificationSearch} finds a marked element in $O(\sqrt{\vars/k})$ queries.
\pfend

The OR function is conjectured to provide the greatest possible separation between the quantum and the deterministic query complexities for total Boolean functions.

\begin{cor}
\label{cor:findingAllOnes}
There exists a quantum algorithm that, given an oracle access to a string $z\in\{0,1\}^\vars$ outputs the string after $O(\sqrt{\vars k})$ queries, where $k$ is the Hamming weight of $z$.
\end{cor}

\pfsketch
The algorithm works as follows.  An occurrence of a symbol 1 in $z$ can be found in $O(\sqrt{\vars/k})$ quantum queries by \rf(prp:Grover).  Put it aside, and search for another occurrence in $O(\sqrt{(\vars-1)/(k-1)})$ quantum queries.  Repeat this procedure until all $k$ ones are found.  This requires
\[
\sum_{i=0}^{k-1} \sqrt{\frac{\vars-i}{k-i}} \le \sqrt{\vars}\int_{0}^k \frac{\mathrm{d}x}{\sqrt{x}} = 2\sqrt{\vars k}
\]
quantum queries altogether.
\pfend

From \rf(cor:findingAllOnes), and by negating the input string if necessary, we get the following result:

\begin{cor}
\label{cor:kThresholdUpper}
The $k$-threshold function on $\vars$ variables can be evaluated in $O\sA[\sqrt{k(\vars-k+1)}]$ quantum queries.
\end{cor}

As another application of quantum amplitude amplification, we consider the collision problem.
Recall from \rf(defn:collision) that, in this problem, we have to distinguish whether the input string is 1-to-1 or 2-to-1.

\begin{prp}
\label{prp:collision}
The collision problem can be solved with bounded error in $O(\vars^{1/3})$ quantum queries and $\tilde O(\vars^{1/3})$ quantum time.  The algorithm uses a QRACM array of $q$-qudits of size $O(\vars^{1/3})$ and $\tilde O(1)$ other quantum registers.  The $\tilde O$ notation suppresses factors polylogarithmic in $\vars$ and $q$.
\end{prp}

\begin{algorithm}
\caption{Quantum Algorithm for the Collision Problem}
\label{alg:collision}
\algbegin
\state {{\bf function} CollisionProblem({\bf quprocedure} InputOracle) {\bf:}}
\tab
  \state {{\bf uses} {\bf QRACM array} $\reg D\elem[r]$ of $q$-qudits, {\bf integer} $i$}
  \state {{\bf attach} $(\vars-r)$-qudit $\qreg X$ }
  \state {{\bf for} $i = 1,2,\dots,r$ {\bf :} \label{line:collision:loop}}
  \tab
  	\state{$\reg D\elem[i] \gets$ InputOracle$(\vars-i)$}
  \untab
  \state {sort($\reg D$)}
  \state {{\bf if} there are equal elements in $\reg D$ {\bf : return } 1}
  \state {{\bf return} GroverSearch(Check, $r/\vars$, $\qreg X$) \label{line:collision:grover}}
  \state {\ }
  \state {{\bf qufunction} Check($\qreg X, \qreg O$) {\bf with} no error {\bf:}}
	\tab
		\state{{\bf attach} $q$-qudit $\qreg A$}
		\state{$\qreg A\qgets$ InputOracle$(\qreg X)$}
		\state{$\qreg O \qgets$ binarySearch$(\reg D, \qreg A)$ \label{line:collision:binary}}
	\untab
\untab
\algend
\end{algorithm}

\pfstart
The details are given in \rf(alg:collision).  In the beginning, the last $r$ elements of the input string are queried, where $r$ is a parameter to be specified later.  If there are equal elements among them, accept the input string.  Otherwise, search through the remaining elements for an element equal to the one in $\reg D$.  In the positive case, there are exactly $r$ marked elements among the remaining $\vars-r$ elements, hence, the success probability is $\lucky \ge r/\vars$.  

Note that since $\reg X$ is an $(\vars-r)$-qudit, Grover's search in \rf(line:collision:grover) indeed searches only among the elements outside $\reg D$.  The binarySearch subroutine in \rf(line:collision:binary) is a quantum analogue of the classical binary search that returns 1 if the value of $\reg A$ is in the sorted array $\reg D$.  Here, we apply \rf(lem:quantumOfDet).  It is not hard to check that the quantum analogue runs in logarithmic time and uses additional logarithmic space.

The loop in \rf(line:collision:loop) requires $r$ queries to the input string.  The Grover search calls the checking subroutine $\sqrt{1/\lucky} = \sqrt{\vars/r}$ times, and each call costs 2 queries (the oracle is executed in reverse while making the checking subroutine coherent, cf. \rf(lem:coherent)). The total query complexity of the algorithm is
\(
O(r+\sqrt{\vars/r})
\)
that attains its optimal value $O(\vars^{1/3})$ when $r = \vars^{1/3}$.  The time complexity is the same up to logarithmic factors.  The claim on the space usage is trivial.
\pfend

A slight modification of the previous algorithm can be used to solve the element distinctness problem.
In this problem, we are given a string $z\in[q]^\vars$, and the task is to detect whether there are two equal elements in it.  Note, however, that this algorithm is {\em not} the best possible: we will see a better one in \rf(sec:mnrs).  Also, due to simplicity, we restrict our analysis to query complexity.  However, it is possible to implement the algorithm time-efficiently.

\begin{prp}
\label{prp:distinctnessOld}
The element distinctness problem on $\vars$ variables can be solved in $O(\vars^{3/4})$ quantum queries.
\end{prp}

\pfstart
See \rf(alg:distinctnessSetup) for the description.  The algorithm is similar to the one in  \rf(prp:collision), but instead of choosing one fixed subset of $r$ elements, we pick one at random.  Moreover, we do this quantumly, using amplitude amplification.  

\begin{algorithm}
\caption{Set-up and Checking Procedures for the Element Distinctness Problem}
\label{alg:distinctnessSetup}
\algbegin
\state {{\bf function} ElementDistinctness({\bf quprocedure} InputOracle) {\bf:}}
\tab
  \state {{\bf attach} {\bf QRAQM arrays} $\reg X\elem[r]$ of $\vars$-qudits, $\reg D\elem[r]$ of $q$-qudits}
	\state {{\bf return} AmplitudeAmplification(Setup, Check, $r/\vars$, $\qreg X$, $\qreg D$)}
	\state {\ }

\state {{\bf quprocedure} Setup ({\bf registers} $\qreg X, \qreg D$) {\bf:}}
\tab
	\state {{\bf uses integer} $i$}
	\state {transform $\ket X|0>\ket D|0>$ into $\alpha\sum_{S\in X} \ket X|S>\ket D|0>$}
	\state {{\bf for} $i = 1,\dots,r$ :}
	\tab
		\state { $\qreg D\elem[i] \qgets $ InputOracle($\qreg X\elem[i]$) }
	\untab
\untab
\state {\ }
\state {{\bf qufunction} Check({\bf registers} $\qreg X, \qreg D$, $\qreg O$) {\bf with} 1-sided error 1/2 {\bf:}}
\tab
  \state {{\bf attach} $\vars$-qudit $\qreg J$ }
  \state { GroverSearch(InternalCheck, $\qreg J$, $\qreg O$) }
\state {\ }
\state {{\bf qufunction} InternalCheck({\bf registers} $\qreg J$, $\qreg O$) {\bf with} no error {\bf:}}
\tab
  \state { {\bf attach} $q$-qudit $\qreg E$, qubit $\qreg b$  }
  \state { $\qreg b \qgets \bigwedge_{i\in [r]} (\qreg J \ne \qreg X\elem[i] )$ \label{line:distinctnessSetup:wedge} }
  \state { {\bf conditioned on} $\qreg b = 1$ {\bf:}}
  \tab
  	\state { $\qreg E \qgets$ InputOracle($\qreg J$) }
  	\state { $\qreg O \qgets \bigvee_{i\in [r]} (\qreg E = \qreg D\elem[i])$ \label{line:distinctnessSetup:vee} }
  \untab
\untab \untab\untab
\algend
\end{algorithm}

Formally, we have two QRAQM arrays $\reg X$ and $\reg D$ of length $r$.  The elements of $\reg X$ are $\vars$-qubits, while $\reg D$ stores $q$-qubits.  The set $X$, on which amplitude amplification is performed, consists of all $r$-subsets of $[\vars]$.  For a subset $S\subseteq[\vars]$ of size $r$, let $\ket X|S>$ be some fixed representation of the indices in $S$, and $\ket D|\str_S>$ be some fixed representation of the corresponding elements in the input string $\str$.
We choose the following representation: $\ket X|S> = \ket X|s_1,s_2,\dots, s_r>$ where $s_1<s_2<\cdots<s_r$ are the elements of $S$, and $\ket D|\str_S> = \ket D|\str\elem[s_1], \dots ,\str\elem[s_r]>$.

The set-up procedure prepares the initial state $\psi$ from~\refeqn{quantumSetup}.  All amplitudes in the initial state equal $\alpha = {\vars\choose r}^{-1/2}$.  A subset $S\in X$ is marked iff it contains exactly one element from a pair of equal elements.  This is tested using the Grover search on all elements of $[\vars]$, where an element $j\in[\vars]$ is marked if $j\notin S$ but $\str_j = \str_i$ for some $i\in S$.  The latter condition is tested using the InternalCheck subroutine.  The qubit $\reg b$ on \rf(line:distinctnessSetup:wedge) is set to 1 iff the index $j$ in $\reg J$ is different from all the elements of $S$ stored in $\reg X$.  Similarly, the output bit is set to 1 in \rf(line:distinctnessSetup:vee) iff $\reg b$ is set to 1 and the content of $\str_j$ is equal to an element from $\reg D$.  Both these checks are made quantumly using \rf(lem:quantumOfDet).

The InternalCheck procedure uses 2 queries, hence, the Check subroutine uses $O(\sqrt{\vars})$ queries by \rf(prp:Grover).  The Setup procedure uses $r$ queries.  Let us calculate the fraction $\eps$ of marked elements in amplitude amplification.  Let $\{a,b\}$ be a pair of equal elements in a positive input.  Thus, $S\in X$ is marked if it contains $a$ and does not contain $b$ (there can be other possibilities as well).  It is easy to see that an $\Omega(r/\vars)$ fraction of all the $r$-subsets of $[\vars]$ satisfy this condition.  By \rf(thm:noisyGrover), the total query complexity of the algorithm is $O(r+\sqrt{\vars})\sqrt{\vars/r}$.  It is optimised to $O(\vars^{3/4})$ when $r=\sqrt{\vars}$.
\pfend

\section{Quantum Walks}
\label{sec:walkQuantum}
The purpose of this section is to develop a quantum analogue of \rf(alg:forbearing).  In \rf(sec:szegedy), we prove a theorem about composition of two reflections that will be of importance later in the thesis.  In \rf(sec:mnrs), we describe this algorithm, and in \rf(sec:mnrsApplications), we give some applications.

\subsection{Composition of two Reflections}
\label{sec:szegedy}
A step of the quantum walk in the proof of \refthm{amplification} is defined as a composition of two reflections.  In this section, we study such compositions in full generality.  Assume $A$ and $B$ are matrices with the same number of rows, and each having orthonormal columns.  Let $\Pi_A = A A^*$ and $\Pi_B = B B^*$ be the projectors onto $\im(A)$ and $\im(B)$, respectively.  $R_A = 2\Pi_A-I$ and $R_B = 2\Pi_B-I$ are the reflections about the corresponding subspaces, and let $U = R_B R_A$ be their composition.  
Finally, let $D = D(A, B) = A^* B$.  This matrix is known as the {\em discriminant matrix}.

\begin{lem}[Spectral Lemma] \label{lem:szegedy}
In the above notations, all the singular values of $D$ are at most~$1$.  Let $\cos\theta_1, \ldots, \cos\theta_\ell$ be all the singular values of $D$ lying in the open interval $(0,1)$ counted with their multiplicity.  Then, the following is a complete list of the eigenspaces and eigenvalues of $U$: 
\itemstart
\item The $1$-eigenspace is $(\im(A)\cap \im(B))\oplus(\im(A)^\perp \cap \im(B)^\perp)$.  Moreover, $\im(A)\cap \im(B)$ coincides with the image, under the action of $A$, of the set of left singular vectors of $D$ with singular value 1.  Also, $\im(A)\cap \im(B)$ coincides with the image, under the action of $B$, of the set of right singular vectors of $D$ with singular value 1.
\item The $(-1)$-eigenspace is $(\im(A)\cap \im(B)^\perp)\oplus(\im(A)^\perp \cap \im(B))$.  Moreover, $\im(A)\cap \im(B)^\perp = A(\ker(D^*))$ and $\im(A)^\perp \cap \im(B) = B(\ker(D))$.
\item The orthogonal complement of the above subspaces is decomposable into a direct sum of $\ell$ two-dimensional pairwise orthogonal invariant subspaces $\{S_j\}_{j\in[\ell]}$ of $U$.  For all $j\in[\ell]$, both $S_j\cap \im(A)$ and $S_j\cap\im(B)$ are one-dimensional with the angle $\theta_j$ between them, and the eigenvalues of $U$ in $S_j$ are $\ee^{\pm2 \ii \theta_j}$.
\itemend
\end{lem}

\pfstart
By definition, the set of singular values of $D$ is equal to the set of non-zero eigenvalues of $D^*D = B^*AA^*B$, that is equal to the set of non-zero eigenvalues of $AA^*BB^* = \Pi_A\Pi_B$ (cf. \rf(sec:linearAlgebra)).  It is obvious that a composition of two projectors cannot have an eigenvalue greater than 1 in absolute value.  This proves the first statement.

A vector $Bw$ is in $\im(A)^\perp$ if and only if $A^*Bw=0$, hence, $\im(A)^\perp\cap\im(B) = B(\ker(D))$.  Similarly, $\im(A)\cap\im(B)^\perp = A(\ker(D^*))$.  It is easy to see that $\im(A)^\perp\cap\im(B)$ and $\im(A)\cap\im(B)^\perp$ are $(-1)$-eigenspaces of $U$, and that $\im(A)^\perp\cap\im(B)^\perp$ is an eigenvalue~1 eigenspace of $U$.  This covers the second bullet in the statement of the lemma, and a half of the first one.  It remains to detect the eigenvalues of $U$ inside $A(\ker(D^*)^\perp)+B(\ker(D)^\perp)$, that equals to the space spanned by the images of the left and right singular vectors of $D$ under the action of $A$ and $B$, respectively.

Let $v$, $w$ be a pair of left and right singular vectors of $D$ with singular value $\sigma$.  Thus, 
$A^*Bw = \sigma v$, and $B^*Av = \sigma w$.  From this, we get
\begin{equation}
\label{eqn:szegedy1}
\Pi_A Bw = AA^*Bw = \sigma Av\qquad\text{and}\qquad \Pi_BAv = BB^* Av = \sigma Bw.
\end{equation}
Hence,
\begin{equation}
\label{eqn:szegedy2}
\Pi_B\Pi_A Bw = \sigma \Pi_B Av =  \sigma^2 Bw.
\end{equation}
If $\sigma = 1$, then $Bw = Av\in \im(A)\cap\im (B)$ is a 1-eigenvector of $U$.  This corresponds to the second half of the first bullet in the statement of the lemma.  

Now, we prove the statement in the third bullet.  Let $v_1,\dots,v_\ell$ be an orthonormal set of the left singular vectors of $D$ with $v_j$ having singular value $\sigma_j = \cos\theta_j$.  Let $w_1,\dots,w_\ell$ be the corresponding right singular vectors.  
Let $S_j$ be the subspace spanned by $Av_j$ and $Bw_j$.  Using~\refeqn{szegedy2} and that $\sigma_j<1$, we have that $Bw_j\notin \im(A)$, hence $S_j$ is two-dimensional.  Since $A$ and $B$ are isometries, $Av_i\perp Av_j$ and $Bw_i\perp Bw_j$ for all $i\ne j$.  Also,
\[
\ip<Av_i, Bw_j> = \ip<v_i, A^*Bw_j> = \ip<v_i, \sigma_j v_j> = 0.
\]
Hence, $S_j$ are pairwise orthogonal.  Finally, from~\refeqn{szegedy1}, we can see that $S_j$ is invariant for $U = (2\Pi_B-I)(2\Pi_A-I)$.

Because of the orthogonality of $S_j$, operators $\Pi_A$ and $\Pi_B$, restricted to $S_j$, coincide with the projectors on $Av_j$ and $Bw_j$, respectively.  Hence, by~\refeqn{szegedy2}, we get that the angle between $Av$ and $Bw$ is $\arccos\sqrt{\sigma^2} = \theta$.  Thus, $U$ acts in $S$ as the rotation on the double angle $2\theta$. And, this operation has eigenvalues $\ee^{\pm 2\ii\theta}$.
\pfend

\subsection{Algorithm}
\label{sec:mnrs}
\mycommand{smallstart}{\start} %{\start^{\phantom{*}}_{\mathrm{d}}}
\newcommand{\oket}[1]{\bar{e}^{\phantom{*}}_{#1}}
\mycommand{totalspace}{{\widetilde{\cH}}}
The update operation and the matrix $P$ were used in \rf(alg:forbearing) to avoid execution of the set-up procedure on each iteration of the loop.  Similarly, we use quantum analogues of the update and diffusion operations to avoid execution of the quantum set-up procedure at every step of the quantum walk.  That is, we implement the reflection about $\psi$ from~\rf(eqn:quantumSetup) in a different way.

Similarly to \rf(defn:update), the reflection is defined using some stochastic matrix $P$ such that $(\stationary_x)$ from \rf(defn:quantumSetup) is its stationary distribution.
We use \refthm{detection} to detect the eigenvector and perform the reflection about it.  This theorem, however, is not directly applicable to $P$, because $P$ is not unitary.  This section is mostly devoted to developing a unitary counterpart of $P$.

%We solve this complication by deriving an analogue of $P$ as a product of two reflections.

\mycommand{tP}{\tilde P}

It has become more conventional to consider quantum walk on a bipartite graph $G$ with parts $X$ and $Y$.  Recall from \rf(sec:walkClassical) that $w_{xy}$ are the weights of the edges of $G$, and the probability of going from $y$ to $x$ is given by $\tilde p_{xy} = w_{xy}/w_y$, where $w_y = \sum_{x\in N(y)} w_{xy}$.  Since the graph is bipartite, the corresponding matrix $\tP= (\tilde p_{xy})$ looks like
\begin{equation}
\label{eqn:tP}
\tP = \begin{pmatrix}
0 & Q_2 \\
Q_1 & 0
\end{pmatrix},
\end{equation}
where $Q_1 = \tP\elem[Y,X]$ is the part of the matrix representing the transitions from $X$ to $Y$, and $Q_2 = \tP\elem[X,Y]$ is representing the transitions in the opposite direction.  The second iteration of the walk, $\tP^2$, breaks down into a random walk on $X$ given by $Q_2Q_1$, and a random walk on $Y$ given by $Q_1Q_2$.  {\em We identify the walk $P$ from \rf(sec:walkClassical) with $Q_2Q_1$.}

Informally, $X$ is the ``main'' set, and $Y$ is a ``supplementary'' set required to implement the walk $P$ on $X$.  The definitions of the quantum set-up and check procedures from \rf(sec:amplification) carry over to this case as applied to the elements of $X$ only.  We again require that $P$ is aperiodic, and $(\sigma_x)$ from \rf(defn:quantumSetup) is its stationary distribution.
The condition on $P$ to be represented by $Q_2Q_1$ is not restrictive.  One may always define $Y=X$ and $Q_2=Q_1=P$.  Thus, a step of the walk becomes $P^2$, and, for an aperiodic random walk, that merely reduces the number of iterations by a factor of 2.

%We introduce additional {\em coin register} \reg C, whose content may be interpreted as labels of the edges of the graph \todo{more on graph}.  We assume that all edges leaving any vertex in $X\cup Y$, are labeled with different elements of \reg C.

Let us consider the diffusion operation.  For $x\in X$, define the unit vector $\phi_x\in \hilbert Y$; and, for $y\in Y$, define $\phi_y\in \hilbert X$ as follows:
\begin{equation}
\label{eqn:phixandy}
\phi_x = \sum_{y\in N(x)} \sqrt{Q_1\elem[y,x]}\; \ket Y|y>,
\qquad\text{and}\qquad
\phi_y = \sum_{x\in N(y)} \sqrt{Q_2\elem[x,y]}\; \ket X|x>.
\end{equation}
Since $\tP$ is a stochastic matrix, both $\phi_x$ and $\phi_y$ are unit vectors.

\mycommand{qdiffuse}{\circuit D}
\mycommand{qupdate}{\circuit U}

\begin{defn}[Quantum diffusion]
There are two quantum diffusion operations: from $X$ and from $Y$.  The first one, denoted $\qdiffuse_1$, transforms $\ket X|x>\ket Y|0>$ into $\ket X|x>\ket Y|\phi_x>$ for all $x\in X$.  The second one, $\qdiffuse_2$, transforms $\ket X|0>\ket Y|y>$ into $\ket X|\phi_y>\ket Y|y>$ for all $y\in Y$.
\end{defn}

For the quantum update operation, we take into account the non-cloning theorem (\rf(obs:copy)), and keep only one copy of the data register.

\begin{defn}[Quantum Update]
\label{defn:quantumUpdate}
The quantum update operation is a unitary $\qupdate$ that transforms $\ket X|x>\ket Y|y> \ket D|d_x>$ into $\ket X|x>\ket Y|y> \ket D|d_y>$ for all $x\in X$ and $y\in N(x)$.  The cost of the operation is $\update$.
\end{defn}

We only need one instance of the update operation, because the reverse transformation can be performed by $\qupdate^{-1}$.

There are two different algorithms for implementing quantum walks, that correspond to Algorithms~\ref{alg:walksimple} and~\ref{alg:forbearing}, respectively.

\begin{thm}[Szegedy quantum walk]
\label{thm:walkSzegedy}
Assume $X$, $Y$, $Q_1$, and $Q_2$ are as above, and $P=Q_2Q_1$ is an aperiodic reversible random walk.  Let the quantum set-up and checking procedures be as in \rf(sec:amplification), and the quantum diffusion and update operations be as above.
Then, there exists a quantum algorithm detecting the presence of a marked element with error $1/3$ and the total cost $O\s[\setup + \sqrt{H} (\checking + \update)]$, where $H$ is the hitting time of the random walk corresponding to $P$.
\end{thm}

We do not give the proof, since we do not use this result in the thesis.

\begin{thm}[MNRS quantum walk]
\label{thm:mnrs}
In the assumptions of \rf(thm:walkSzegedy), there exists a quantum algorithm finding marked elements with probability $\Omega(1)$ and the total cost $O\s[\setup + \frac1{\sqrt{\eps}} {\s[\checking + \frac{\update}{\sqrt{\delta}}]} ]$, where $\eps$ is the probability of measuring a marked element in the initial state (cf. \rf(defn:quantumSetup)) and $\spectralgap$ is the spectral gap of $P$.
\end{thm}

\pfstart
The algorithm is the same as \refalg{amplification} (or, alternatively, \rf(alg:amplificationSearch)) with the replaced procedure for the reflection about $\psi$.  The algorithm uses the register $\reg X$ to store elements of $X$, the register $\reg Y$ for the elements of $Y$, and $\reg D$ to store data.  
The new reflection subroutine is described in \refalg{quantumWalk} and it is performed with precision $\gamma$ that is defined so that the total precision is sufficient.

\begin{algorithm}
\caption{Alternative Implementation of the Reflection About $\total$}
\label{alg:quantumWalk}
\algbegin
\state {{\bf quprocedure} ReflectionAbout $\total$ ({\bf  quprocedures} $\qdiffuse_1$, $\qdiffuse_2$, $\qupdate$, {\bf registers} $\qreg X$, $\qreg D$) {\bf with} precision $\gamma$ {\bf :}}
\tab
  \state {{\bf attach} $Y$-qudit $\qreg Y$}
	\state {$\qdiffuse_1(\qreg X, \qreg Y)$}
	\state {phase$\gets$ PhaseDetection(StepOfInternalQuantumWalk, $\sqrt{\delta}$, $\qreg{XYD}$) {\bf with} precision $\gamma$	\label{line:walk:detection} }
	\state {$\qdiffuse_1^{-1}(\qreg X, \qreg Y)$}
  \state {{\bf detach} $\qreg Y$  \label{line:walk:detach} }
\state {\ }
\state {{\bf quprocedure} StepOfInternalQuantumWalk {\bf:} \label{line:walk:internalstep}}
\tab
	\state {ReflectionAbout $\im(A)$ \label{line:walk:step1} }
	\state {$\qupdate(\qreg X, \qreg Y, \qreg D)$ \label{line:walk:step2} }
	\state {ReflectionAbout $\im(B)$ \label{line:walk:step3} }
	\state {$\qupdate^{-1}(\qreg X, \qreg Y, \qreg D)$ \label{line:walk:step4} }
\state {\ }
\state {{\bf quprocedure} ReflectionAbout $\im(A)$ {\bf:}}
\tab
	\state {$\qdiffuse_1^{-1}(\qreg X, \qreg Y)$}
	\state {{\bf conditioned on} $\qreg{Y} \ne 0$ : apply $(-1)$-phase gate}
	\state {$\qdiffuse_1(\qreg X, \qreg Y)$}
\untab
\state {\ }
\state {{\bf quprocedure} ReflectionAbout $\im(B)$ {\bf:}}
\tab
	\state {$\qdiffuse_2^{-1}(\qreg X, \qreg Y)$}
	\state {{\bf conditioned on} $\qreg{X} \ne 0$ : apply $(-1)$-phase gate}
	\state {$\qdiffuse_2(\qreg X, \qreg Y)$}
\untab
\untab
\untab
\algend
\end{algorithm}

Let us at first calculate the cost of the algorithm.  Each step of the internal quantum walk, the procedure in \rf(line:walk:internalstep), uses two update operations and all other operations are costless.  By \rf(thm:detection), the phase detection subroutine uses $O(\frac1{\sqrt{\delta}}\log\frac1\gamma)$ steps of the internal quantum walk.  \rf(alg:amplification) uses $O(1/\sqrt{\lucky})$ steps of the outer walk, each involving the checking subroutine and the new reflection subroutine.  Also, there is one call to the set-up subroutine at the very beginning of \rf(alg:amplification).  

Since the reflection in \rf(alg:amplification) is executed $O(1/\sqrt{\eps})$ times, it is sufficient, by \rf(lem:circuitRobust), if $\gamma \le c\sqrt{\eps}$ for small enough constant $c$.  Thus, the total cost of the algorithm is
\[
O\s[\setup + \frac1{\sqrt{\eps}} {\s[\checking + \frac{\update\log(1/\lucky)}{\sqrt{\delta}}]} ].
\]
This is the value of the cost we actually prove.  It has an extra logarithmic factor compared to the claimed one.  This factor can be removed with techniques similar to \rf(thm:noisyGrover).  We refer the reader to~\cite{magniez:walkSearch} for the details.

It remains to prove the correctness of the algorithm.  At first, we would like to remove the data register from our analysis.  For that, we prove some assertions on its content.  As described in \rf(sec:procedures), we may assume in our analysis that the phase detection subroutine in \refline{walk:detection} is implemented perfectly.

We claim that the state of $\reg{XD}$ before Lines \ref{line:walk:detection}, \ref{line:walk:step1}, and \ref{line:walk:step2} is a linear combination of the vectors $\sfig{\ket X|x>\ket D|d_x>}_{x\in X}$.  Similarly, the state of $\reg{YD}$ before Lines \ref{line:walk:step3} and \ref{line:walk:step4} is the linear combination of $\sfig{\ket Y|y>\ket D|d_y>}_{y\in Y}$.  
These claims can be verified by going through the algorithm line by line, and recalling that the phase detection subroutine only repeatedly applies the step of the internal quantum walk.
From now on, we ignore the register $\reg D$ since its content is uniquely determined by the content of $\reg X$ or $\reg Y$ in dependence on the place in the algorithm.

%It is unclear how to perform the reflection about the uniform superposition $\total$.  However, we can perform the reflection ``locally'' using the quantum diffusion operation and \refsub{reflect}.  Then, we glue this reflections together.

\mycommand{tv}{\tilde v}
\mycommand{tw}{\tilde w}
\mycommand{tW}{\widetilde W}

Let $\{e_x\}_{x\in X}$ and $\{h_y\}_{y\in Y}$ denote the computational bases of registers $\reg X$ and $\reg Y$, respectively.
Let $A$ and $B$ be matrices defined by the action of $\qdiffuse_1$ and $\qdiffuse_2$, respectively:
\begin{equation}
\label{eqn:mnrsAandB}
A = \sum_{x\in X} (\eket{x}\otimes \phi_x)\bra{x}
\qquad\text{and}\qquad
B = \sum_{y\in Y} (\phi_y\otimes \eket{y})\bra{y},
\end{equation}
where $\phi_x$ and $\phi_y$ are defined in~\rf(eqn:phixandy).
The matrix $A$ has its columns labelled by the elements of $X$, and the matrix $B$---by the elements of $Y$.  Both of them have columns in $\C^{X\times Y}$.  The columns of $A$ are orthonormal, as well as those of $B$.  
The procedures ReflectionAbout $\im(A)$ and ReflectionAbout $\im(B)$ implement reflections about $\im(A)$ and $\im(B)$, respectively.  Also, it is not hard to check that the state of $\reg{XY}$ during the algorithm always belongs to $\im(A)+\im(B)$.

The matrices $A$ and $B$ satisfy the promise of \reflem{szegedy}.  Let $D = A^*B$ be the corresponding discriminant matrix.  We have
\[
D\elem[x,y] = \phi_x\elem[y] \phi_y\elem[x] = \sqrt{Q_2\elem[x,y]Q_1\elem[y,x]} = \frac{w_{xy}}{\sqrt{w_xw_y}}
\]
for all $x\in X$ and $y\in Y$.  Here $w_{xy}$ are the weights of the edges of the bipartite graph $G$ and $w_x = \sum_{y\in N(x)} w_{xy}$.

By \rf(prp:stationary), the vector $\tw = (w_z)_{z\in X\cup Y}$ is a 1-eigenvector of $\tP$ from~\rf(eqn:tP).  By \rf(thm:perron), $P$ has an unique eigenvalue-1 eigenvector.  As $P =\tP^2\elem[X,X]$, this eigenvector must coincide with $w = \tw\elem[X]$.

Let $P'$ and $W$ be as in \rf(prp:walkPprime), and let $\tP'$ and $\tW$ be defined as $P'$ and $W$ from the same proposition for $P$ equal to $\tP$.  Then, $D = \tP'\elem[X,Y]$ and
\begin{equation}
\label{eqn:mnrs1}
DD^* = (\tP')^2\elem[X,X] = (\tW^{-1} \tP^2 \tW)\elem[X,X].
\end{equation}
Since $w=\tw\elem[X]$, we have that $W$ is proportional to $\tW\elem[X,X]$.  From~\rf(eqn:mnrs1), we get $P' = DD^*$.  Also, as in the proof of \rf(prp:hitting), the unique normalised 1-eigenvector $v_1$ of $P'$ satisfies $v_1\elem[x] = \alpha_x$, where $\alpha_x = \sqrt{\stationary_x}$ are as in \rf(defn:quantumSetup).

\mycommand{tpsi}{\widetilde\psi}

Let $1 = \lambda_1 \ge |\lambda_2|\ge\cdots\ge|\lambda_n|\ge 0$ be the common list of eigenvalues of $P$ and $DD^*$, and $\cos\theta_i = \sqrt{\lambda_i}$ be the corresponding singular values of $D$, where $i\in[n]$.  
Let $U$ denote the step of the internal quantum walk.
By \rf(lem:szegedy), $U$ has eigenvalues $\pm1$ and $\ee^{\pm 2\ii\theta_j}$.  The smallest non-zero phase is
\[
2\theta_2 \ge |1 - \ee^{2\ii\theta_2}| = 2\sqrt{1-\cos^2\theta_2} \ge 2\sqrt{\cos\theta_2} = 2\sqrt{\delta}.
\]
Let $R$ denote the reflection in \rf(line:walk:detection).  As the argument $\sqrt{\delta}$ in the phase detection subroutine is small enough, $R$ performs the reflection about the 1-eigenspace of $U$.  By \rf(lem:szegedy), it is given by the span of
\[
\tpsi = A v_1 = \sum_{x\in X} \alpha_x \ket X|x> \ket Y|\phi_x>
\]
and $\im(A)^\perp\cap \im(B)^\perp$.  By our claim, the state of the algorithm always resides in $\im(A)+\im(B)$.  Hence, we may assume that $R$ reflects about $\tpsi$.  Then, $\qdiffuse_1^{-1} R\qdiffuse_1$ performs the reflection about $v_1$ in the register $\reg X$, and leaves the register $\reg Y$ intact.
By putting back the data register, we get that the procedure reflects about $\psi$.
\pfend

\newcommand{\bket}[1]{|#1\rangle}
\subsection{Applications of MNRS Quantum Walk}
\label{sec:mnrsApplications}
The MNRS quantum walk is most often applied for the Johnson graph.  In this section, we describe such an application for the element distinctness problem.  After that, we provide some other examples.

Recall that we have already seen a quantum query algorithm for the element distinctness problem in \rf(sec:amplificationExamples).  As it has been mentioned, the main purpose of the MNRS quantum walk is to avoid execution of the set-up procedure on every step of the quantum walk.  The Setup procedure in \rf(alg:distinctnessSetup) prepares a superposition over subsets of $[\vars]$ of size $r$.  The cheapest possible update operation is to replace an element in $S\subseteq[\vars]$ with an element outside $S$ resulting in a subset $T$ such that $|S\cap T|=r-1$.  This graph is known as the {\em Johnson graph}.

%\mycommand{str}{z}
%\newcommand{\data}[1]{d_{#1}}
%\newcommand{\bket}[1]{|#1\rangle}

%Using the MNRS quantum walk, it is possible to obtain a better algorithm.
%
%\begin{thm}
%\label{thm:distinctness}
%The element distinctness problem can be solved in $O(n^{2/3})$ quantum queries with bounded error.
%\end{thm}
%
%In order to apply the MNRS quantum walk, we have to better understand the spectral properties of the underlying graph.

\begin{defn}[Johnson graph]
\label{defn:johnson}
The Johnson graph $J(n,r)$ has the set of all $r$-subsets of $[n]$ as its vertex set.  Two vertices $S$ and $T$ are connected iff $|S\cap T|=r-1$.  All edges of the graph have weight 1.  We assume that $r\le n/2$.
\end{defn}

Since quantum walks are performed on bipartite graphs, we also consider the following family of bipartite graphs.  For $r\ge 1$, let $G(n,r)$ be the bipartite graph with parts $X$ and $Y$ that consist of all subsets of $[n]$ of sizes $r$ and $r-1$, respectively.  A vertex $S\in X$ is connected to $T\in Y$ iff $S\supseteq T$.  Let $G_{n,r}$ be the biadjacency matrix of $G(n,r)$ with rows in $X$ and columns in $Y$, and let $J_{n,r}$ be the adjacency matrix of the Johnson graph $J(n,r)$ (it is {\em not} the all-1 matrix).

By applying the construction of \rf(sec:mnrs), we get that $Q_1 = G_{n,r}^*/r$ and $Q_2 = G_{n,r}/(n-r+1)$.  The resulting random walk matrix $P$ on $X$ is 
\begin{equation}
\label{eqn:JohnsonStep}
Q_2Q_1 = \frac{G_{n,r}G_{n,r}^*}{r(n-r+1)} = \frac{1}{r(n-r+1)} \sA[J_{n,r} + rI_{{n\choose r}}],
\end{equation}
where $I_m$ stands for the $m\times m$ identity matrix.  
Indeed, if $S\ne S'$ satisfy $J_{n,r}\elem[S,S']=0$, there is no way to get from $S$ to $S'$ in two steps on $G(n,r)$.  If $J_{n,r}\elem[S,S']=1$, there is exactly one such path: through $S\cap S'$.  Finally, there are exactly $r$ paths of length 2 from $S$ to itself: through any vertex in $Y$ labelled by a subsets of $S$.  In the following, we estimate the spectral gap of the matrix in~\refeqn{JohnsonStep}.

\begin{lem}
\label{lem:johnson}
If $1<r<n/2$, the spectral gap of the matrix $Q_2Q_1$ from~\rf(eqn:JohnsonStep) is $\Omega(1/r)$.
\end{lem}

\pfstart
Fix a value of $n$, and calculate the eigenvalues of $J_{n,r}$ by induction on $r$.  The induction base is $J_{n,0}$ that is the zero $1\times 1$ matrix.  For the induction step, we have the following identities
\begin{equation}
\label{eqn:johnson1}
G_{n,r}^*G_{n,r} = J_{n,r} + rI_{{n\choose r-1}}
\qquad  \mbox{and}\qquad
G_{n,r}G_{n,r}^*= J_{n, r-1}+(n-r+1)I_{{n\choose r}}
\end{equation}
that hold for all $r\ge 1$.  They can be proved similarly to~\rf(eqn:JohnsonStep).
The matrices $G_{n,r}G_{n,r}^*$ and $G_{n,r}^*G_{n,r}$ have the same set of non-zero eigenvalues.
Thus, the eigenvalues of $J_{n,r}$ are $-r$ with multiplicity ${n\choose r}-{n\choose r-1}$, and $\{\lambda_j + n-2r+1\}$, where $\{\lambda_j\}$ are the eigenvalues of $J_{n,r-1}$.

The largest eigenvalue of $J_{n,r}$ is $\sum_{i=1}^r (n-2i+1) = r(n-r)$.  The second largest (in absolute value) eigenvalue is
\(
-1 + \sum_{i=2}^r (n-2i + 1) = r(n-r) - n.
\)
Hence, the spectral gap of $Q_2Q_1 = (J_{n,r}+r I)/(r(n-r+1)$ is $n/(r(n-r+1)) = \Omega(1/r)$ under our assumption $r<n/2$.
\pfend

Now we are able to improve the algorithm for element distinctness from \rf(prp:distinctnessOld).

\begin{prp}
\label{prp:distinctness}
The element distinctness problem on $\vars$ variables can be solved in $O(\vars^{2/3})$ quantum queries.
\end{prp}

\pfstart
We apply the MNRS quantum walk on the graph $G(\vars,r)$ as described above.  
The set-up and check operations were already described in \rf(alg:distinctnessSetup).  The diffusion operation is a unitary that does not require any oracle query.  We will not describe it.

It remains to describe the update operation.  In this operation, we are given $\ket X|S>\ket Y|T>\ket D|\str_S>$ and the task is to transform it into $\ket X|S>\ket Y|T> \ket D|\str_T>$.  Let $s_1<s_2<\dots <s_r$ be the elements of $S$.  The subset $T$ has the same elements with one removed.  As in \rf(prp:distinctnessOld), the data register has the form $\ket D|\strut\str\elem[s_1], \dots  \dots \str\elem[s_r]>$, appended with zeroes if necessary.
We proceed as follows.  Attach an $r$-qubit $\reg i$, and transform the state
\begin{equation}
\label{eqn:distinctnessAlg1}
\ket X|S>\ket Y|T>\ket D|\str_S>\ket i|0> \mapsto \ket X|S>\ket Y|T>\ket D|\str_S>\ket i|i>,
\end{equation}
where $i$ is the unique element of $[r]$ satisfying $s_i\notin T$.  Then perform \mbox{$\qreg D\elem[i]\qungets$ InputOracle$(\qreg X\elem[i])$}.
This transforms the data register into
\begin{equation}
\label{eqn:distinctnessAlg2}
\ket D |\strut\str\elem[s_1], \dots, \str\elem[s_{i-1}], 0,  \str\elem[s_{i+1}],  \dots, \str\elem[s_r]> .
\end{equation}
Now, apply the unitary transformation mapping the state in \rf(eqn:distinctnessAlg2) and $\ket i|i>$ into
\[
\ket D |\strut\str\elem[s_1], \dots,  \str\elem[s_{i-1}], \str\elem[s_{i+1}],  \dots, \str\elem[s_r], 0> \ket i|i> .
\]
Finally, undo the transformation in \rf(eqn:distinctnessAlg1).  We have the required state.  This operation costs one query.

By \rf(thm:mnrs), the total cost of the algorithm is
\[
O\s[{ r+\sqrt{\frac nr}\sA[\sqrt{\vars} + \sqrt{r}] }] = O\s[r + \vars/\sqrt{r}].
\]
That is optimised to $O(\vars^{2/3})$ when $r = \vars^{2/3}$.
\pfend

Note that only the first sampling of an element in $X$ costs $r$ queries.  All others cost only $\sqrt{r}$.  This is the source of the speed-up when comparing to the algorithm given in \rf(prp:distinctnessOld).  By using more sophisticated data structures and QRAQM arrays, it is possible to implement this algorithm time-efficiently~\cite{ambainis:distinctness}.

Recall the definition of the 1-certificate complexity from \rf(sec:queryRelated): It is the maximum, over $x\in f^{-1}(1)$, of the smallest subset $S\subseteq[\vars]$ such that $f(z)=1$ for all $z$ agreeing with $x$ on $S$.  The algorithm for element distinctness can be generalised as follows.

\begin{thm}
\label{thm:walkKDistinctness}
Let $f: [q]^\vars\supseteq \cD\to\{0,1\}$ be any function with 1-certificate complexity $k=O(1)$.  The quantum query complexity of $f$ is $O(\vars^{k/(k+1)})$.
\end{thm}

\pfstart
We keep the same graph on parts $X$ and $Y$, the same set-up, diffusion, and update operation as in the previous proof.  We change the set of marked vertices and the checking procedure.

We say $S\in X$ is marked iff it contains a 1-certificate for $f$ on the input string $\str$.  The checking operation is a unitary transformation that does not require any query: the information in $\ket X |S>\ket D|z_S>$ is enough to detect whether $S$ is marked.

The probability an element is marked is $\eps = \Omega(r^k/\vars^k)$.  The spectral gap still is $\Omega(1/r)$.  Thus, the query complexity of the algorithm is
\(
O\s[{ r+\s[\frac \vars r]^{k/2} \sqrt{r} }]
\)
that is optimised to $\vars^{k/(k+1)}$ when $r = \vars^{k/(k+1)}$.
\pfend

\begin{cor}
\label{cor:walkKDist}
The $k$-distinctness and the $k$-sum problems from Definitions~\ref{defn:kdist} and~\ref{defn:ksum} can be solved in $O(\vars^{k/(k+1)})$ quantum queries, where $\vars$ is the number of input variables.
The graph collision problem from \rf(defn:graphCollision) can be solved in $O(\vars^{2/3})$ quantum queries.
\end{cor}

Consider the triangle problem from \rf(defn:triangle).  It is possible to apply \rf(thm:walkKDistinctness) here, and get a quantum query algorithm with complexity $O(\vars^{3/4}) = O(n^{3/2})$.  However, using the structure of the problem, it is possible to do better.

\begin{thm}
\label{thm:walkTriangle}
The quantum query complexity of the triangle problem on $n$ vertices is $O(n^{13/10}) = O(\vars^{13/20})$.
\end{thm}

\pfstart
The quantum walk is again on the graph $G(n,r)$.
In order to avoid possible confusion with the input graph, we call the vertices of $G(n,r)$ {\em elements}.
An element is marked iff it contains exactly 2 vertices of a triangle.  Thus, the fraction of marked elements is $\lucky = \Omega(r^2/n^2)$.  Let us denote, for simplicity, $\str_{ab} = \str_{ba}$ even if $a>b$.

The data register $\reg D$ now contains the values of $\str_{ab}$ for $a,b\in S$.  The set-up cost is $O(r^2)$ queries, and the update cost is $O(r)$.  For the check subroutine, perform the Grover search for the third node of the triangle.  The check subroutine of the Grover search, in its turn, performs the graph collision algorithm on the subgraph induced by $S$ as follows.  Let $c$ be the potential third node of the triangle.  For each vertex $v\in S$, define $g(v) = 1$ iff $\str_{vc}=1$.  We can execute the graph collision algorithm here because we know all the edges connecting the vertices in $S$.  There is a graph collision, if and only if there is a triangle in $G$ with two vertices in $S$, and the third one being $c$.  The total checking cost is $\tilde O(\sqrt{n}r^{2/3})$, where the logarithmic factors are again due to the applications of \rf(lem:circuitRobust) in order to reduce error.

The total complexity of the algorithm is
\[
\tilde O\s[r^2 + \frac n r{\sA[ \sqrt{n}r^{2/3} + \sqrt{r}\cdot r ]} ] = \tilde O\s[r^2 + \frac{n^{3/2}}{r^{1/3}} + n\sqrt{r} ]  .
\]
This expression is optimised to $\tilde O(n^{13/10})$ when $r = n^{3/5}$.  As usually, the logarithmic factors can be removed.
\pfend

We will present a better quantum query algorithm for this problem in \rf(chp:cert).  Apart from the mentioned applications, the quantum walk on the Johnson graph can be used in matrix product verification~\cite{buhrman:productVerification}, restricted range associativity testing~\cite{dorn:associativity}, and other problems.

\section{Chapter Notes}
Quantum walks is a very broad area, with a variety of surveys~\cite{ambainis:walkApplications, kempe:walkOverview, santha:walkBasedAlgorithms}.  We mostly followed the latter in this chapter.  We only consider algorithmic applications of quantum walks, in among them, we only consider discrete-time quantum walks, ignoring the continuous-time quantum walks.  For other types of quantum walks, refer to the above-mentioned surveys.

The oldest quantum algorithm presented in this chapter is the Grover search~\cite{grover:search}, although at the time of discovery, its relation to random walks was not noticed.  Very soon, it was generalized to quantum amplitude amplification by Brassard and H\o yer~\cite{brassard:exactSimon}, and, independently, by Grover~\cite{grover:amplification}.  The name comes from~\cite{brassard:amplification}.  
\rf(cor:findingAllOnes) is well-known, see, e.g.,~\cite{ambainis:newLowerBoundMethod}.  Our proof is from~\cite{childs:graphProperties}.
The application to the collision problem, \rf(prp:collision), is due to Brassard, H\o yer and Tapp~\cite{brassard:collision}.

The quantum phase detection procedure from \rf(sec:phaseDetection) is usually presented in the form of quantum phase estimation~\cite{kitaev:phaseEstimation, cleve:phaseEstimation}.  The latter additionally uses quantum Fourier transform to estimate the value of the phase, not just detect whether it is non-zero.  For instance, this can be used in quantum counting~\cite{brassard:counting}.  %We find the exposition only featuring phase detection more transparent.

Most of the first applications of quantum walks followed similar restricted settings as we outlined for random walks in the preamble.  Watrous gives quantum analogues of random walks with low space complexity~\cite{watrous:walkSimulations}.  Quantum walks on various types of graphs also have been studied extensively.  This includes walks on the line~\cite{ambainis:walkOneDimensional} motivated by research on  quantum cellular automata~\cite{meyer:quantumCellular}, and $k$-dimensional torus~\cite{ambainis:coinsMakeFaster}.  In some specific cases, like for the opposite vertices of a hypercube~\cite{kempe:walkHitFaster} or the roots of two glued binary trees~\cite{childs:walkExponentialSeparation}, it is possible to obtain exponential advantage over the corresponding classical hitting times.  In~\cite{childs:walkExponentialSeparation}, the separation is demonstrated over an arbitrary classical algorithm (not necessary one based on random walks).

The results in~\rf(sec:szegedy) are due to Szegedy~\cite{szegedy:walk}.  Apparently, they can be also deduced from results by Camille Jordan from the 19th century~\cite{jordan:essai}.  The paper by Szegedy also describes the algorithm from \rf(thm:walkSzegedy).  Our main quantum walk algorithm in \rf(sec:mnrs) is due to Magniez, Nayak, Roland and Santha~\cite{magniez:walkSearch}.

The amplitude-amplification algorithm for element distinctness from \rf(sec:amplificationExamples) is due to Buhrman \etal~\cite{buhrman:distinctness}.  The optimal algorithm in \rf(sec:mnrsApplications) is due to Ambainis~\cite{ambainis:distinctness}.  Ambainis gives a more complicated proof that was later generalised to the MNRS quantum walks, yielding the current presentation.  The generalization in \rf(thm:walkKDistinctness) is due to Childs and Eisenberg~\cite{childs:subsetFinding}.

The triangle problem is interesting classically because of its connection to matrix multiplication~\cite{alon:triangle}.  Quantumly, the first algorithm was due to Buhrman \etal~\cite{buhrman:distinctness}.  The algorithm from \rf(thm:walkTriangle) is by Magniez, Santha and Szegedy~\cite{magniez:triangle}.

%% file: _lowerBounds.tex
In this chapter, we describe some known techniques for proving lower bounds on quantum query complexity.  
We consider the two main techniques: the polynomial method, and the adversary method.  Sections \ref{sec:adv}--\ref{sec:advAlgorithms} contain the majority of technical tools we will use in the later chapters.  The main result is \rf(thm:adv) that relates the quantum query complexity of a function to a relatively simple semi-definite optimisation problem.

The chapter is organised as follows.  In \rf(sec:pol), we describe the polynomial method and apply it to the collision and the element distinctness problems.  In \rf(sec:adv), we define the adversary bound, and in \rf(sec:advDuality), we define the dual adversary bound and a closely related notion of the span program.  In \rf(sec:advAlgorithms), we prove that the dual adversary bound provides an upper bound on the quantum query complexity.

\section{Polynomial Method}
\label{sec:pol}
In this section we describe the polynomial method for proving lower bounds on quantum query complexity.  This method is most notable for providing a lower bound on the collision and the element distinctness problems, the result we describe in \rf(sec:collisionLower).  Before that, in \rf(sec:polandAlgorithms), we explain the relation between polynomials and quantum query algorithms, and in \rf(sec:polLowerBounds), we show how the relation can be used to prove lower bounds.

The results of this section are illustrative.  We will not use the polynomial method further in the thesis.  However, we will reprove \rf(cor:distLower) in \rf(chp:cert) using the adversary method.

\subsection{Polynomials and Quantum Query Algorithms}
\label{sec:polandAlgorithms}
\mycommand{indexset}{I}
Let $\indexset$ be some finite set of {\em indices}.  A real (respectively, complex) {\em multilinear polynomial} in variables $(x_j)_{j\in\indexset}$ is an expression of the form
\begin{equation}
\label{eqn:polynomial}
P = \sum_{S\subseteq\indexset} a_S \prod_{j\in S} x_j,
\end{equation}
where $a_S$ are real (respectively, complex) numbers.  The {\em degree} of the polynomial, denoted $\deg P$, is the maximum of $|S|$ over all $S$ such that $a_S\ne 0$.

A real (complex) polynomial $P$ can be considered as a function $P\colon \R^\indexset\to\R$ (respectively, $P\colon \C^\indexset\to\C$).  Its value $P(x)$ on a sequence $x=(x_j)$ in $\R^\indexset$ (respectively, $\C^\indexset$) is defined by plugging the values of $x_j$ into the right hand side of~\refeqn{polynomial}.

Recall from \rf(sec:query) that a quantum query algorithm is a sequence
\[
U_0\to O_z\to U_1\to O_z \to \cdots \to U_{T-1} \to O_z\to U_T
\]
of transformations in $\hilbert j\otimes\hilbert v\otimes \hilbert W$.
Here, $U_i$ are input-independent unitary transformations.
The transformation $O_z$ depends on the input string $z\in [q]^\vars$, and is given by $\ket j|j>\ket v|a>\mapsto \ket j|j>\ket v|a+z_j>$, where the addition is performed modulo $q$.
\mycommand{xm}{\widetilde{z}}
Let us define a Boolean string $\xm = (\xm_{j,a})_{j\in[\vars],a\in[q]}$ from $z$ by
\begin{equation}
\label{eqn:xja}
\xm_{j,a} = \begin{cases}
1,& \text{if $z_j=a$;}\\
0,& \text{otherwise.}
\end{cases}
\end{equation}
\mycommand{alg}{\cA}
We call Boolean strings thus obtained {\em valid}.  The main observation binding quantum query algorithm and multilinear polynomials is as follows.
\begin{lem}
\label{lem:amplitudes}
Let $\alg$ be a quantum query algorithm with $O_z$ as the input oracle, and let $\xm_{j,a}$ be defined as in~\refeqn{xja}.  Then, the amplitude of any basis state of $\alg$ after $t$ queries is given by a complex polynomial  of degree at most $t$ in the variables $\xm_{j,a}$.
\end{lem}

\pfstart
The proof proceeds by induction on $t$.  If $t=0$, the amplitude of the state does not depend on the input, hence, it is a degree-0 polynomial.  Assume the theorem holds for a value of $t$, and prove it for $t+1$ as follows.  By the inductive assumption, the state of the algorithm before the $(t+1)$st application of $O_z$ is of the form 
\[
\ket|\psi> = \sum_{j,a,i} P_{j,a,i}\ket j|j>\ket v|a>\ket W|i>
\]
with $\deg P_{j,a,i}\le t$.  Here $j,a$ and $i$ range over the computation basis elements of $\reg j$, $\reg v$ and $\reg W$, respectively.
Since the oracle maps $\ket j|j>\ket v|a>$ into $\ket j|j>\ket v|a+z_j>$, we have
\[
O_z\ket|\psi> = \sum_{j,a,i} \skB[\;\sum_{b\in[q]} \xm_{j,b} P_{j,a-b,i}]\ket j|j>\ket v|a>\ket W|i>.
\]
Thus, we see that the amplitudes after the application of $O_z$ (the expressions in the square brackets) can be expressed as complex polynomials of degree at most $t+1$.  After the application of the unitary $U_{t+1}$, the amplitudes are linear combinations of these polynomials, hence, are polynomials of degree at most $t+1$ themselves.
\pfend

Now we turn our attention to the function that the algorithm computes.
\begin{defn}[Approximating polynomial]
Let $f\colon [q]^\vars\supseteq \cD\to\{0,1\}$ be a function, and $\eps>0$ be a real number.  We say that a real polynomial $P$ in variables $(\xm_{j,a})_{j\in[\vars], a\in[q]}$ {\em $\eps$-approximates} $f$, if
\[
\mbox{$\abs|f(z) - P(\xm)|\le\eps$\quad for all $z\in\cD$},\qquad\mbox{ and }\qquad
\mbox{$0\le P(\xm)\le 1$\quad for all $z\in[q]^\vars$},
\]
where $\xm=(\xm_{j,a})$ are as defined in~\refeqn{xja}.
\end{defn}

\begin{thm}
\label{thm:approx}
Suppose there exists a quantum query algorithm $\alg$ evaluating a function $f\colon [q]^\vars\supseteq \cD\to\{0,1\}$ in $t$ queries with error probability $\eps$. Then, there exists a polynomial $P$ of degree at most $2t$ that $\eps$-approximates $f$.
\end{thm}

\pfstart
By \reflem{amplitudes}, the state of $\alg$ before the final measurement is of the form
\[
\ket|\psi> = \sum_{j,a,i} P_{j,a,i}\ket j|j>\ket v|a>\ket W|i>
\]
where $P_{j,a,i}$ are complex polynomials of degree at most $t$.  The acceptance probability is expressible as
\begin{equation}
\label{eqn:polAcceptProb}
\sum P_{j,a,i}P_{j,a,i}^*
\end{equation}
where the sum is over all $j$ and $a$, and those $i$ in which the output register $\reg O$ contains value 1.  (Recall, we assume that $\reg O$ is a part of the register $\reg W$.)
Each product $P_{j,a,i}P_{j,a,i}^*$ is a polynomial of degree at most $2t$.  Also, it has real coefficients, because it equals its complex conjugate.  Hence, the polynomial in~\refeqn{polAcceptProb} has real coefficients and degree at most $2t$.
\pfend

For Boolean functions this result can be stated in a bit nicer way.
\begin{cor}
Suppose there exists a quantum query algorithm $\alg$ evaluating a Boolean function $f\colon \{0,1\}^\vars\supseteq \cD\to\{0,1\}$ in $t$ queries with error probability $\eps$.  Then, there exists a polynomial $P$ in variables $(z_j)_{j\in[\vars]}$ of degree at most $2t$ such that
\[
\mbox{$\abs|f(z) - P(z)|\le\eps$\quad for all $z\in\cD$},\qquad\mbox{ and }\qquad
\mbox{$0\le P(z)\le 1$\quad for all $z\in\{0,1\}^\vars$},
\]
\end{cor}

\pfstart
Apply \refthm{approx}, and note that $\xm_{j,1} = z_j$ and $\xm_{j,0} = 1-z_j$.
\pfend

In this case, we also say that the polynomial $P$ $\eps$-approximates $f$.  If $\eps = 0$, we say that $P$ represents $f$ {\em exactly}.

\subsection{Lower Bounds by Polynomials}
\label{sec:polLowerBounds}
The idea behind proving lower bounds for quantum query algorithms using the polynomial method is as follows.  Assume a quantum query algorithm $\alg$ calculates a function $f$ in $t$ queries with error probability $\eps$.  By~\refthm{approx}, there exists a polynomial $P$ that $\eps$-approximates $f$ and has degree at most $2t$.  Transform $P$ into a univariate polynomial (i.e., a polynomial in one variable), and argue using lower bounds for univariate approximating polynomials.

\mycommand{sym}{\mathrm{sym}}
\mycommand{symgr}{\mathbb{S}}
We start with transforming a multivariate polynomial to a univariate one.  One of the basic ways is as follows.
Let $P$ be a polynomial in variables $z_1,\dots,z_n$.  If $\pi$ is a permutation on $n$ elements and $z\in \{0,1\}^\vars$, then denote by $\pi(z)$ the string $(z_{\pi(1)},\dots z_{\pi(n)})$ obtained from $z$ by permuting its elements according to $\pi$.  Let $\symgr_n$ denote the symmetric group of order $n$ consisting of all $n!$ permutations on $n$ elements.  Define the {\em symmetrisation} $P^\sym$ as the following polynomial in $n$ variables:
\[
P^\sym(z) = \frac{1}{n!}\sum_{\pi\in \symgr_n} P(\pi(z)).
\]

\begin{lem}
\label{lem:symmetrize}
If $P$ is a multilinear polynomial, then there exists a univariate polynomial $Q$ of degree at most $\deg P$ such that $P^\sym(z) = Q(|z|)$ for all $z\in\{0,1\}^n$.  (Here, $|z|$ denotes the Hamming weight of $z$.)
\end{lem}

\pfstart
By linearity, it is enough to consider the case of $P$ being a monomial $\prod_{j\in S} z_j$ with $S\subseteq[n]$.  Denote $k=|S|$.  For any $z$ of Hamming weight $\ell$, the value of $P$ on $\pi(z)$ is 1 if and only if
$\pi^{-1}(S)$ is contained in the value-1 variables of $z$.
Hence,
\(
P^\sym =  \ell^{\underline{k}}(n-k)!/n!
\),
where 
\[
\ell^{\underline{k}} = \ell(\ell-1)\cdots(\ell-k+1)
\]
is the {\em falling power}.  Thus, $P^\sym$ is degree-$k$ polynomial in $\ell$.
\pfend

We bound the degree of the approximating univariate polynomial using the following result:

\begin{lem}[Paturi~\cite{paturi:degree}]
\label{lem:paturi}
Let $Q$ be a real univariate polynomial of degree at most $d$, and $a<b$ be integers.  Assume the following holds
\itemstart
\item $|Q(i)|\le 1$ for all integers $i\in[a,b]$;
\item there exists a real number $\xi$ between $a$ and $b$ such that $\absA|Q(\lfloor\xi\rfloor) - Q(\xi)|=\Omega(1)$.
\itemend
Then,
\[
d = \Omega\s[\sqrt{(\xi-a+1)(b-\xi+1)}] .
\]
\end{lem}

Recall the $k$-threshold function from \rf(defn:threshold): The function evaluates to 1 if the Hamming weight of the input is at least $k$.

\begin{prp}
\label{prp:polThresholdLower}
The quantum query complexity of the $k$-threshold function on $\vars$ variables is $\Omega(\sqrt{k(\vars-k+1)})$.  Moreover, the same lower bound holds in the promise that the Hamming weight of the input is either $k-1$ or $k$.
\end{prp}

\pfstart
Let $\alg$ be a quantum algorithm calculating the $k$-threshold function in $t$ queries with error at most $1/3$.  By \refthm{approx}, there exists a polynomial $P$ in variables $z_1,\dots,z_n$ of degree at most $2t$ that approximates the $k$-threshold function.

Thus, we have $0\le P^\sym(y) \le 1/3$ if $|y|=k-1$, and $2/3\le P^\sym(x)\le 1$ if $|x|=k$.  Let $Q$ be the univariate polynomial that corresponds to $P$ per \reflem{symmetrize}.  Then, $0\le Q(k-1) \le 1/3$, $2/3\le Q(k) \le 1$, and $0\le Q(i)\le 1$ for all integers $i$ from 0 to $\vars$.

Consider $c = Q(k-1/2)$.  If $c\ge 1/2$, define $\xi = k-1/2$, $a=0$, and $b=\vars$.  With this choice, $|Q(\lfloor \xi\rfloor) - Q(\xi)|\ge 1/6$, hence, the conditions of \rf(lem:paturi) hold, and the degree of $Q$ is $\Omega\s[\sqrt{(k+1/2)(\vars-k+3/2)}] = \Omega(\sqrt{k(\vars-k+1)})$.

If $c < 1/2$, we obtain a similar result for the polynomial $\tilde Q(z) = Q(\vars-z)$ with $\xi = \vars-k+1/2$ and the same values of $a$ and $b$.  As the degree of $Q$ is at most the degree of $P$, we get that $t = \Omega(\sqrt{k(\vars-k+1)})$.
\pfend

This means that the algorithm in \rf(cor:kThresholdUpper) is tight.  In particular, the quantum query complexity of the OR function on $\vars$ variables is $\Omega(\sqrt{\vars})$, implying the optimality of Grover's search.  Also, by an argument similar to \rf(thm:R2bs), we get that the quantum query complexity of a function $f$ is $\Omega(\sqrt{\bs(f)})$.  By combining with \rf(thm:D2bs), we get the following result (recall that $Q(f)$ stands for the quantum query complexity of $f$):
\begin{thm}
\label{thm:DtoQ}
For any total Boolean function $f$, $D(f) =O(Q(f)^6)$.
\end{thm}
This is the best known lower bound on $Q(f)$ in terms of deterministic complexity.  The best known separation is quadratic, given by Grover's search.

\subsection{Collision Problem}
\label{sec:collisionLower}
Recall from \rf(defn:collision) that in the collision problem on $\vars=2n$ variables one has to detect whether the input is 1-to-1 or 2-to-1.
One of the main results obtained via the polynomial method is the lower bound for the collision problem.

\begin{thm}
\label{thm:collisionLower}
The quantum query complexity of the collision function on $\vars$ variables is $\Omega(\vars^{1/3})$. 
\end{thm}

Before we proceed with the proof, let us make some observations.
Firstly, the lower bound is tight because of \rf(prp:collision).
Also, it has the following important consequence:

\begin{cor}
\label{cor:distLower}
The quantum query complexity of the element distinctness problem on $n$ variables is $\Omega(n^{2/3})$.
\end{cor}

\pfstart
Assume we have a quantum algorithm $\alg$ solving the element distinctness problem for inputs of size $n$.  We will show how to construct a quantum algorithm for the collision function on inputs of size $\vars = n^2$ using $\alg$ as a subroutine.

Let $z \in [q]^{\vars}$ be the input to the collision problem.
 Select an $n$-subset $S$ of $[\vars]$ uniformly at random.  If all the elements of $z$ are distinct, then such are the elements inside $S$.  If $z$ is divided into pairs of equal elements then, by the birthday paradox, we obtain that $S$ contains two equal elements with probability $\Omega(1)$.
This means that we can apply $\alg$ in order to distinguish these two cases.  Thus, the complexity of $\alg$ is at least the complexity of the collision problem on an input of size $\vars$.  By \refthm{collisionLower}, it is $\Omega(\vars^{1/3}) = \Omega(n^{2/3})$.
\pfend

Recall form \rf(defn:setEquality), that the set equality problem is a special case of the collision problem, and the hidden shift problem is a special case of the set equality problem.  Consequently, the $O(\vars^{1/3})$ upper bound for the collision problem translates to the set equality problem.  The best known lower bound, however, is only $\Omega((\vars/\log \vars)^{1/5})$ as shown by Midrij\=anis~\cite{midrijanis:setEquality}.

The situation with the hidden shift problem is more interesting.  This problem reduces to the so-called hidden subgroup problem in the dihedral group~\cite{kuperberg:dihedral}, and the latter has logarithmic query complexity~\cite{ettinger:hspQuery}.

\pfstart[Proof of \refthm{collisionLower}.]
The first step is the same as in the proof of \rf(prp:polThresholdLower).  Let $f\colon [q]^\vars\to\{0,1\}$ be the collision function where $q\ge \vars$.  Assume there exists a quantum algorithm calculating $f$ in $t$ queries with error probability $1/3$.  Then, there exists a degree-$2t$ real polynomial $P$ that $1/3$-approximates $f$.

Next, it is necessary to obtain an univariate polynomial out of $P$.  This is done in two steps.  At first, a polynomial $Q$ in three variables is obtained using symmetrisation.  After that, an univariate polynomial is obtained from $Q$ using restrictions.  The last step requires some case analysis.

\mycommand{si}{\pi_{\mathrm i}}
\mycommand{sv}{\pi_{\mathrm v}}
We start with obtaining $Q$.  Assume $z\in [q]^\vars$ is an input.  If $\pi = (\si,\sv)\in \symgr_\vars\times \symgr_q$, let $\pi(z)$ denote the string $y\in[q]^\vars$ defined by $y\elem[j] = \sv\sA[{z\elem[\si(j)]}]$.  For a real polynomial $P$ in variables $\xm_{j,a}$, let its symmetrisation be the function on $[q]^\vars$ defined by
\[
P^\sym(x) = \frac{1}{\vars!q!} \sum_{\pi\in \symgr_\vars\times \symgr_q} P(\widetilde{\pi(z)}).
\]
where $\widetilde{\pi(z)}$ is defined similarly to~\rf(eqn:xja), that is
\[
\widetilde{\pi(z)}_{j,a} =
\begin{cases}
1, & \sv(z\elem[\si(j)]) = a;\\
0, & \mbox{otherwise.}
\end{cases}
\]
We call a triple $(m,b,b')$ of non-negative integers {\em good}, if $m\le \vars$, $b$ divides $m$ and $b'$ divides $\vars-m$.  Let $z^{(m,b,b')}$ be any input that is $b$-to-1 on $m$ input elements, and $b'$-to-1 on the remaining input elements.  In particular, one can take
\[
z^{(m,b,b')}_j = 
\begin{cases}
\lceil j/b\rceil,& 1\le j\le m;\\
m/b + \lceil (j-m)/b'\rceil, & m+1\le j\le \vars.
\end{cases}
\]

\begin{lem}
\label{lem:QfromPsym}
For each polynomial $P$ in variables $\xm_{j,a}$, there exists a polynomial $Q$ in variables $m$, $b$ and $b'$ such that $Q = P^\sym (z^{(m,b,b')})$ for all good triples $(m,b,b')$.  Moreover, $\deg Q\le \deg P$.
\end{lem}

\pfstart
Again, by linearity, we may assume $P$ equals a monomial $\xm_{j_1,a_1}\xm_{j_2,a_2}\cdots\xm_{j_d,a_d}$.  If there are equal elements among $j_1,\dots,j_d$, the value of the monomial is 0 for all valid strings $\tilde z$, because one of $\xm_{j,a}$ with equal values of $j$ will be equal to 0.  So, assume all $j$ are distinct.

In order to compute $P^\sym (z^{(m,b,b')})$, it suffices to count the number of permutations $\pi = (\si,\sv)\in \symgr_\vars\times \symgr_q$ such that $\xm\elem[\si(j_i),\sv(a_i)]=1$ for all $i\in[d]$.
Let $\ell$ be the number of distinct elements among $a_1,a_2,\dots,a_d$.  
Denote the values of the distinct elements by $c_1,\dots,c_\ell$, and let $k_i$ be the number of appearances of $c_i$ in the sequence $a_1,a_2,\dots,a_d$.  
Let us fix $L\subseteq [\ell]$ and consider those $\sv$ only that map $c_i$, with $i\in L$, to the $b$-to-1 part of the input (i.e., $1\le \sv(c_i)\le m/b$ for all $i\in L$ and only them).
For the ease of notation, let us assume $L = [\ell']$.
Then, the number of permutations $\pi$ mapping all $\xm\elem[j_i,a_i]$ to ones is
\begin{multline*}
m(b-1)^{\underline{k_1-1}}\; (m-b)(b-1)^{\underline{k_2-1}}\;\cdots\;(m-b(\ell'-1))(b-1)^{\underline{k_{\ell'}-1}}\\
\times (\vars-m)(b'-1)^{\underline{k_{\ell'+1}-1}}\cdots (\vars-m-b'(\ell-\ell'-1))(b'-1)^{\underline{k_{\ell}-1}}.
\end{multline*}
The first multiplier, $m$, is the number of ways to fix the value of $\si(j_i)$ for the first element with $a_i = c_1$.  The next multiplier, $(b-1)^{\underline{k_1-1}}$, is the number of ways to fix the indices of $\si(j_i)$ for the remaining $i$ with $a_i = c_1$, and so on.

This is a polynomial of degree at most $k_1+\cdots+k_\ell=d$ in variables $m$, $b$ and $b'$.  The value of $P^\sym$ is the sum over all choices of $L$ divided by $\vars!q!$.  Hence, it is a polynomial in $m$, $b$ and $b'$ of degree at most $d$ as well.
\pfend

Let $Q$ be as in \rf(lem:QfromPsym) for the degree-$2t$ polynomial $P$ that $1/3$-approximates $f$.  Then, $Q$ satisfies the following constraints:
\itemstart
\item $0\le Q(m,b,b')\le 1$ for all good triples $(m,b,b')$;
\item $0\le Q(m,1,1)\le 1/3$ for all $m$;
\item $2/3\le Q(m,2,2)\le 1$ for all even $m$.
\itemend
Let $M=2\lfloor \vars/4\rfloor$ be the closest even number to $\vars/2$.
Consider two cases:

\paragraph{Case I} Assume $Q(M, 1, 2)\ge 1/2$.  Define a univariate real polynomial $g(t) = Q(M,1,2t+1)/2$, and let $k$ be the least positive integer such that $|g(k)|>1$.  Thus, we have $|g(i)|\le 1$ for all integers $i$ in the range $[0,k-1]$.  Also,
\[
g(1/2)-g(0) = \frac{Q(M, 1, 2) - Q(M, 1, 1)}{2} \ge \frac1{12}.
\]
By \reflem{paturi}, we have
\begin{equation}
\label{eqn:polCol1}
\deg Q \ge \deg g  = \Omega\sA[\sqrt{k}].
\end{equation}

Now consider the polynomial $h(t) = Q(\vars-(2k+1)t, 1, 2k+1)$.  For any integer $0\le i\le \vars/(2k+1)$, the triple $(\vars-(2k+1)i,1,2k+1)$ is good, hence, $|h(i)|\le 1$.  But,
\[
\abs| h\sB[\frac{N-M}{2k+1}] | = \abs|Q(M, 1, 2k+1)| = 2|g(k)| \ge 2.
\]
Hence, by the same lemma,
\begin{equation}
\label{eqn:polCol2}
\deg Q\ge \deg h = \Omega(\vars/k).
\end{equation}

\paragraph{Case II} Now assume $Q(M,1,2)\le 1/2$.  In this case, the argument is similar.  Let $g(t) = Q(M,2t+1,2)/2$, and let $k$ be the least positive integer such that $|g(k)|>1$.  Again,
\[
g(1/2)-g(0) = \frac{Q(M, 2, 2) - Q(M, 1, 2)}{2} \ge \frac1{12},
\]
and by \reflem{paturi}, we obtain~\refeqn{polCol1}.  Consider the polynomial $h(t) = Q((4k+2)t, 2k+1, 2)$.  For any integer $0\le i\le \vars/(4k+2)$, the triple $((4k+2)i,2k+1,2)$ is valid, hence $|h(i)|\le 1$.  But,
\[
\abs| h\sB[\frac{M}{4k+2}] | = \abs|Q(M, 2k+1, 1)| = 2|g(k)| \ge 2,
\]
and by \reflem{paturi}, we again obtain~\refeqn{polCol2}

\paragraph{Finishing the Proof}  In both cases, we get that $\deg Q = \Omega\sA[\sqrt{k}]$ and $\deg Q = \Omega(\vars/k)$.  Hence, $\deg Q = \Omega(\vars^{1/3})$. \pfend

\section{Adversary Method}
\label{sec:adv}
In this section, we define the adversary lower bound on quantum query complexity.  We do this in a number of steps.  In \rf(sec:advBasic), we define the basic adversary method, and give a number of its applications.  In \rf(sec:advGeneral), we define its generalisation, the adversary bound.  
At first glance, they may seem different, but in \rf(sec:advPositive), we prove the generalisation relation between the two.  In \rf(sec:advProof), we prove that the adversary bound is indeed a lower bound on quantum query complexity.

\subsection{Basic Adversary Bound}
\label{sec:advBasic}
In this section, we describe a simple version of the adversary bound, as it was defined in the pioneering work by Ambainis~\cite{ambainis:adv}.  This version of the bound has been used extensively, because of its highly intuitive nature.  In order to prove an adversary lower bound for a function $f$ with Boolean output, one has to come up with a set $X$ of inputs from $f^{-1}(0)$, and a set $Y$ of inputs from $f^{-1}(1)$, that are hard to {\em distinguish by one query}.  A formal statement is as follows:

\begin{thm}
\label{thm:advBasic}
Let $f: [q]^\vars\supseteq \cD\to\{0,1\}$ be a function.  Suppose $X\subseteq f^{-1}(1)$, $Y\subseteq f^{-1}(0)$, and a relation $\sim$ between $X$ and $Y$ are such that
\itemstart
\item for each $x\in X$, there are at least $m$ different $y\in Y$ such that $x\sim y$;
\item for each $y\in Y$, there are at least $m'$ different $x\in X$ such that $x\sim y$.
\itemend
For $x\in X$ and $j\in[\vars]$, let $\ell_{x,j}$ (respectively, $\ell'_{y,j}$ for $y\in Y$) be the number of $y\in Y$ (respectively, $x\in X$) such that $x\sim y$ and $x_j\ne y_j$.  Let $\ell_{\max}$ denote the maximum of $\ell^{\phantom{*}}_{x,j}\ell'_{y,j}$ over all $x\sim y$ and $j\in[\vars]$ such that $x_j\ne y_j$.  In this case, any quantum algorithm evaluating $f$ uses $\Omega\s[ \sqrt{\frac{mm'}{\ell_{\max}}} ]$ queries.

In particular, $Q(f) = \Omega\s[\sqrt{\frac{mm'}{\ell\ell'}}]$ where $\ell = \max \ell_{x,j}$ and $\ell' = \max \ell'_{y,j}$.
\end{thm}

We will obtain this theorem as a special case of a more general result, \refthm{adv}.  But for now, let us give some examples of how this lower bound may be applied.  We start by reproving \rf(prp:polThresholdLower).

\begin{prp}
\label{prp:advThresholdLower}
The quantum query complexity of the $k$-threshold function on $\vars$ variables is $\Omega(\sqrt{k(\vars-k+1)})$.
\end{prp}

\pfstart
Let $X$ consist of all inputs of Hamming weight $k$, and let $Y$ consist of all inputs of Hamming weight $k-1$.  We say that $x\sim y$ if $x$ and $y$ differ in exactly one position.  In the notations of \rf(thm:advBasic), we get that $m = k$ and $m' = \vars-k+1$.  Also, $\ell = \ell' = 1$, hence, the quantum query complexity of the $k$-threshold function is $\Omega(\sqrt{k(\vars-k+1)})$.
\pfend

\begin{prp}[OR of ANDs, Ambainis~\cite{ambainis:adv}]
Consider the following function of $\vars = n^2$ Boolean variables $(z_{i,j})_{i,j\in [n]}$
\begin{equation}
\label{eqn:ORAND}
(z_{1,1}\wedge \cdots\wedge z_{1,n})\vee (z_{2,1}\wedge\cdots\wedge z_{2,n}) \vee \cdots \vee (z_{n,1}\wedge\cdots\wedge z_{n,n})
\end{equation}
where $\wedge$ stands for the logical AND, and $\vee$ stands for the logical OR.
The quantum query complexity of this function is $\Omega(n) = \Omega(\sqrt{\vars})$.
\end{prp}

\pfstart
Let $X$ consist of all inputs such that one block in~\refeqn{ORAND} evaluates to 1, and in all other blocks there is exactly one variable equal to 0.  Let $Y$ consist of all inputs such that in all blocks of~\refeqn{ORAND} there is exactly one variable equal to 0.  Clearly, $f(X)=\{1\}$, and $f(Y) = \{0\}$.  We say that $x\in X$ and $y\in Y$ are in the relation, $x\sim y$, if $x$ and $y$ differ in exactly one position. 

In the notations of \rf(thm:advBasic), $m=n$, because an input in $X$ may be transformed to an input in $Y$ by changing any variable in the block evaluating to 1.  Similarly, $m'=n$, because flipping any 0 to 1 changes an input in $Y$ to an input in $X$.  Again $\ell = \ell' = 1$.  By \refthm{advBasic}, the quantum query complexity of the function is $\Omega(n)$.
\pfend

The previous proof is very concise.  The proof of the same result using the polynomial method is much more complicated.  It has been an open problem for a long time, and only very recently it was proven~\cite{bun:andortree, sherstov:andortree} that the degree of the approximating polynomial is $\Omega(n)$.  
In \rf(exm:AmbainisFunction), we will see an example of a provable separation between the polynomial lower bound and the true quantum query complexity.
%There are also known examples when the adversary bound provides asymptotically better result than the polynomial method.  For an example, refer to the construction in~\cite{ambainis:polVsQCC}.

\begin{prp}[Graph Connectivity, D\"urr \etal~\cite{durr:quantumGraph}]
\label{prp:graphConnectivity}
Assume we are given a simple graph $G$ on $n$ vertices by its adjacency matrix $(z_{ij})_{1\le i<j\le n}$.  The task is to detect whether the graph $G$ is connected.  The quantum query complexity of this problem is $\Omega(n^{3/2}) = \Omega(N^{3/4})$.
\end{prp}

\begin{wrapfigure}{l}{0pt}
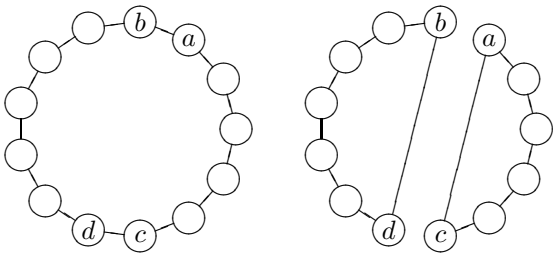

%\begin{figure}[htb]
\xygraph{!~*{\cir<6pt>{}} !{0;<1.7pc,0pc>:} \\
!P13"A"{ ~><{@{-}} ~:{(2,0):} \cir<6pt>{} }
%[o]=<10pt>{\xypolynode}*\frm{o}}
"A3"*{a} "A4"*{b} "A10"*{d} "A11"*{c}
}
\qquad
\xygraph{!~*{\cir<6pt>{}} !{0;<1.7pc,0pc>:} \\
!P13"A"{ ~><{@{}} ~:{(2,0):} \cir<6pt>{} }
%[o]=<10pt>{\xypolynode}*\frm{o}}
"A1" - "A2" - "A3" - "A11" - "A12" - "A13" - "A1"
"A4" - "A5" - "A6" - "A7" - "A8" - "A9" - "A10" - "A4"
"A3"*{a} "A4"*{b} "A10"*{d} "A11"*{c}
}
\caption{A graph from the set $X$ on the left, and a graph from $Y$ on the right.  They are in relation with each other.\vspace{-.5cm}}
\label{fig:connectivity}
\end{wrapfigure}

%\pfstart
\noindent \em Proof.\em\;\; 
Let $X$ be the set of all graphs on the vertex set $[n]$ consisting of one cycle going through all the vertices, and let $Y$ be the set of graphs formed by two cycles, each of length at least $n/3$, such that each vertex belongs to exactly one of the cycles.  Clearly, all graphs in $X$ are connected, while all graphs in $Y$ are not.  We say that a graph $x\in X$ is in the relation with a graph $y\in Y$, $x\sim y$, if there exist 4 distinct vertices $a,b,c,d$ such that $y$ can be obtained from $x$ by removing edges $ab$ and $cd$ and adding edges $ac$ and $bd$ (cf. \rf(fig:connectivity)).

For any $x\in X$ and any edge $ab$ of $x$, we have $\ell_{x,ab} = n/3$, because any edge of $x$ at distance at least $n/3$ from $ab$ may serve as $cd$.  If $ab$ is not an edge of $x$, then $\ell_{x,ab}$ is at most 2.  Similarly, for $y\in Y$, we have $\ell'_{y,ac}=\Theta(n)$ if $ac$ is an edge of $y$, and $\ell'_{y,ac}\le 4$, otherwise.  Since there are $n$ edges in a graph from $x$, we have $m = \frac12 \sum_{ab: x_{ab}=1} \ell_{x,ab} = \Omega(n^2)$.  Similarly, $m' = \Omega(n^2)$.  Also, we have that $\ell_{\max}=O(n)$, because, given that $x_{ab}\ne y_{ab}$, one of $\ell_{x,ab}$, $\ell'_{y,ab}$ is $O(1)$, and the second one is $O(n)$.  Thus, the quantum query complexity of the problem is $\Omega\s[ \sqrt{\frac{mm'}{\ell_{\max}}} ] = \Omega(n^{3/2})$. \hspace{\stretch{1}} \qedsymbol
%\pfend
%\end{exm}

\subsection{General Adversary Bound}
\label{sec:advGeneral}
The lower bound in \refsec{advBasic} proceeds by providing hard-to-distinguish input pairs that evaluate to different values of the function.  However, sometimes, in order to obtain a good lower bound, it is vital to take not-so-hard-to-distinguish pairs into consideration as well.  But since the distinguishability of the pairs is different, it is tempting to reflect this in the bound.  It is achieved by assigning different real {\em weights} to different input pairs.  This leads to the matrix formulation of the bound.  For the sake of generality, we consider functions with possibly non-Boolean output.

\begin{defn}
\label{defn:advMatrix}
Let $f\colon [q]^\vars\supseteq \cD\to [\ell]$ be a function.  An {\em adversary matrix} for the function $f$ is a non-zero $\cD\times\cD$ real symmetric matrix $\Gamma$ such that $\Gamma\elem[x,y]=0$ whenever $f(x)=f(y)$. And, for $j\in[\vars]$, let $\Delta_j$ denote the $\cD\times \cD$ matrix defined by
\[ \Delta_j\elem[x,y] = \begin{cases} 0,& x_j=y_j; \\ 1,&\text{otherwise}. \end{cases} \]
\end{defn}

\begin{defn}
\label{defn:adv}
Let $f$ be as in \refdefn{advMatrix}. The adversary bound is defined by
\begin{equation}
\label{eqn:adv}
\Adv(f) =  \max_{\Gamma} \frac{\|\Gamma\|}{\max_{j\in [\vars]} \|\Gamma\circ \Delta_j \|},
\end{equation}
where the outer maximisation is over all adversary matrices $\Gamma$ for $f$.
\end{defn}

As $\Gamma$ is real and symmetric, there exist a {\em real} unit vector $\delta$ such that $|\delta^* \Gamma \delta| = \norm|\Gamma|$.  Also, by substituting $\Gamma$ by $-\Gamma$, if necessary, we may assume that $\delta^* \Gamma\delta = \norm|\Gamma|$.  We will call a vector satisfying the last condition the {\em principal} eigenvector of $\Gamma$.  Everywhere in this chapter, we assume that entries of matrices and vectors are real.

%The use of the maximum operator in~\refeqn{adv} instead of supremum is justified by the following observations.  Firstly, the optimized function doesn't change under multiplying $\Gamma$ by a non-zero scalar.  Hence, we may assume $\Gamma$ runs through all adversary matrices of norm 1, and this is a compact set.  Next, at least one of $\Gamma\circ\Delta_i$ is non-zero if $\Gamma$ is non-zero.  Thus, the optimized function is continuous and attains its maximum value on the mentioned compact set.

%Now we are ready to state the general adversary lower bound, we prove in \refsec{advProof}.
The following theorem is the technical cornerstone of the thesis.

\begin{thm}
\label{thm:adv}
The quantum query complexity of a function $f\colon [q]^\vars\supseteq \cD\to[\ell]$ is $\Theta(\Adv(f))$.
%For any error probability $\eps<1/9$, and any function $f$,
%\[ Q_\eps(f) \ge \frac{1-2\sqrt{\eps(1-\eps)}-2\eps}{2}\Adv(f). \]
\end{thm}

We will prove the first half of the theorem (the lower bound) in \rf(sec:advProof), and the second half (the upper bound) in \rf(sec:advAlgorithms).  
The upper bound will be only proven in the case of functions with Boolean output.
That is the only case for which we apply the theorem.
Examples of applications will be given in \rf(sec:advExamples) after we introduce all the related notions.
We end this section by a number of small technical results useful in applications of \rf(thm:adv).

\begin{lem}[\cite{lee:stateConversion}]
\label{lem:gamma}
Let $\Delta_j$ be as in \refdefn{advMatrix}.  Then, for any matrix $A$ of the same size, 
\[
\norm|A\circ \Delta_j| \le 2 \norm|A|.
\]
\end{lem}
\mycommand{dar}{\stackrel{\Delta_j}{\longmapsto}}

We will use it to replace $\Gamma\circ \Delta_j$ in the denominator of~\refeqn{adv} with a matrix $\Gamma'$ such that $\Gamma\circ\Delta_j = \Gamma'\circ\Delta_j$.  By \reflem{gamma}, this gives the same result up to a factor of 2.  We will denote this relation between matrices by $\Gamma\dar\Gamma'$.

\begin{rem}
\label{rem:booleanOutput}
Assume the function $f$ has Boolean output (i.e., $\ell=2$).  This is the most common case in the thesis. Then, any adversary matrix can be represented in the following form
\[ \Gamma = 
\begin{pmatrix}
0 & \Gamma' \\
(\Gamma')^* & 0
\end{pmatrix},
\]
where the elements of $\cD$ are ordered so that the positive inputs precede the negative ones.
Moreover, the non-zero eigenvalues of $\Gamma$ are exactly the plus-minus singular values of $\Gamma'$. Hence, in particular, $\|\Gamma\| = \|\Gamma'\|$. The same is true for $\Gamma\circ\Delta_j$ as well. Thus, in the case of Boolean output, we usually abuse the notation and call the $f^{-1}(1)\times f^{-1}(0)$ matrix $\Gamma'$ an adversary matrix, and denote it $\Gamma$.
\end{rem}

\newcommand{\tcD}{\widetilde{\cD}}
\newcommand{\tDelta}{\widetilde{\Delta}}

In some cases, it is convenient to use the same label for different rows (and columns) in $\Gamma$.  More precisely, let the rows of a real symmetric matrix $\Gamma$ be labelled by elements of the form $(z,a)$ where $z\in\cD$, and $a$ is some additional parameter used to distinguish rows with the same value of $z$.  Still, it is required that $\Gamma\elem[(x,a),(y,b)] = 0$ if $f(x)=f(y)$.  Let $\widetilde{\cD}$ be the set of the labels of the rows of $\Gamma$.  Define the $\tcD\times\tcD$ matrix $\tDelta_j$ by
\[ \tDelta_j\elem[(x,a),(y,b)] = \begin{cases} 0,& x_j=y_j; \\ 1,&\text{otherwise}. \end{cases} \]

\begin{prp}
\label{prp:advMultipleRows}
If $\Gamma$ and $\tDelta_j$ are as above, then
\[
\Adv(f) \ge  \frac{\|\Gamma\|}{\max_{j\in [\vars]} \|\Gamma\circ \tDelta_j \|}.
\]
\end{prp}

\pfstart
Let $\delta=(\delta_{z,a})$ be the principal eigenvector of $\Gamma$.
Thus, $\delta^* \Gamma \delta = \|\Gamma\|$. We are going to construct an adversary matrix $\Gamma'$ in the sense of \refdefn{advMatrix} from $\Gamma$ and $\delta$. For all $x,y\in \cD$, let:
\[\delta'_x = \sqrt{\sum_{a:(x,a)\in\widetilde{\cD}} \delta_{x,a}^2}\qquad\mbox{and}\qquad
 \Gamma'\elem[x,y] = \frac1{\delta'_x\delta'_y}\sum_{\substack{a:(x,a)\in\widetilde{\cD}\\ b:(y,b)\in \widetilde{\cD}}} \delta_{x,a}\delta_{y,b} \Gamma\elem[(x,a),(y,b)]. \]
Then it is easy to see that $\delta' = (\delta'_{x})$ satisfies $\|\delta'\|=1$ and $(\delta')^* \Gamma'\delta' = \delta^* \Gamma \delta$, hence, $\|\Gamma'\|\ge \|\Gamma\|$.

And vice versa, if $\eps' = (\eps'_z)$ is such that $\|\eps'\|=1$ and $(\eps')^* (\Gamma'\circ \Delta_j)\eps' = \|\Gamma'\circ \Delta_j\|$, let $\eps_{z,a} = \delta_{z,a} \eps'_z/\delta'_z$. Again, $\|\eps\|=1$ and $\eps^*(\Gamma\circ\tDelta_j)\eps = (\eps')^*(\Gamma'\circ\Delta_i)\eps'$, hence, $\|\Gamma'\circ \Delta_j\| \le \|\Gamma\circ \tDelta_j\|$.
This means that $\Gamma'$ provides at least as good an adversary lower bound as $\Gamma$ does.
\pfend

\subsection{Positive-Weighted Adversary}
\label{sec:advPositive}
Although the formulation of the general adversary bound looks completely different from the basic adversary bound, there is a general way of relating the two.  It is based on the following
\begin{lem}
\label{lem:mathias}
Assume $A,B$ and $C$ are real matrices such that $A=B\circ C$.  Then,
\begin{equation}
\label{eqn:mathias}
\|A\| \le \max_{i,j\colon A\elem[i,j]\ne 0} r_i(B)c_j(C),
\end{equation}
where $r_i(B)$ is the $\ell_2$-norm of the $i$th row of $B$, and $c_j(C)$ is the $\ell_2$-norm of the $j$th column of $C$.
\end{lem}

\pfstart
We only prove that $\norm|A| \le \max_{i,j} r_i(B) c_j(C)$ that is a result from~\cite{mathias:nonnegative}.  
It is already enough to obtain the $\Omega(\sqrt{mm'/\ell\ell'})$ lower bound from \rf(thm:advBasic) using the construction of \rf(prp:advBasic2Spectral) further in the text.
For a proof of the general case, refer to~\cite{spalek:advEquivalent}.  

Let $\delta$ and $\delta'$ be real unit vectors such that $\norm|A| = |\delta^* A \delta'|$, and denote $D = \diag \delta$ and $D' = \diag \delta'$.  Then, by the Cauchy-Schwarz inequality,
\[
\norm |A| = |\ip<D B, CD'>| \le \normFrob |D B| \normFrob |CD'| = \sqrt{\sB[\sum_i \delta_i^2 r_i(B)^2]\sB[\sum_j (\delta'_j)^2 c_j(C)^2]} \le 
\max_{i,j} r_i(B) c_j(C). \qedhere
\]
\pfend

\begin{prp}
\label{prp:advBasic2Spectral}
\refthm{adv} implies \refthm{advBasic}.
\end{prp}

\pfstart
Assume $f$, $X$, $Y$ and $\sim$ are as in \refthm{advBasic}.  
Define the $f^{-1}(1)\times f^{-1}(0)$-matrix $\Gamma$ by $\Gamma\elem[x,y]=1$ if $x\sim y$, and $\Gamma\elem[x,y]=0$, otherwise.  (We use \rf(rem:booleanOutput) here.)

Let $\{e_z\}$ be the standard basis of $\C^\cD$.  Define $\delta_X = \frac1{\sqrt{|X|}}\sum_{x\in X} e_x$, and $\delta_Y = \frac1{\sqrt{|Y|}}\sum_{y\in Y} e_y$.  Then,
\[
\|\Gamma\| \ge \delta_X^*\Gamma\delta_Y = \frac{1}{\sqrt{|X||Y|}} |\{(x,y)\mid x\sim y\}| \ge \frac{\sqrt{(|X|m)(|Y|m')}}{\sqrt{|X||Y|}} = \sqrt{mm'}.
\]
Now, let $B_j = C_j = \Gamma\circ\Delta_j$.  In the notations of \refthm{advBasic}, $r_x(B_j) = \sqrt{\ell_{x,j}}$ for all $x\in X$, and $c_y(C_j) = \sqrt{\ell'_{y,j}}$ for all $y\in Y$.
By \reflem{mathias}, $\|\Gamma\circ\Delta_j\| \le \sqrt{\ell_{\max{}}}$.  Thus, by \refthm{adv}, $\Gamma$ provides an adversary lower bound of $\Omega(\sqrt{mm'/\ell_{\max}})$.
\pfend

In the preceding proof, we seek a decomposition $\Gamma\circ \Delta_j = B_j\circ C_j$ such that the right hand side of~\refeqn{mathias} is small.  Note that \reflem{mathias} only operates with the $\ell_2$-norms of the rows and the columns of $B$ or $C$, that does not depend on the signs of the entries.  Also, $\| |\Gamma| \| \ge \|\Gamma\|$, where $|\Gamma|$ is the entry-wise absolute value of $\Gamma$.
Thus, if \reflem{mathias} is used to estimate $\|\Gamma\circ\Delta_j\|$, it is without loss of generality to assume $\Gamma$, $B_j$ and $C_j$ all have non-negative entries.
(In fact, Ref.~\cite{mathias:nonnegative} proves that for any matrix $A$ with non-negative entries, the equality can be attained in~\rf(eqn:mathias).)
This justifies the following
\begin{defn}
\label{defn:advPositive}
The {\em positive-weighted adversary} $\pAdv(f)$ for the function $f$ is defined as in~\refeqn{adv} with the maximisation over all adversary matrices $\Gamma$ with non-negative entries.
\end{defn}
To distinguish the adversary bound from this special case, we sometimes call it {\em negative-weight}, or {\em general} adversary bound.
The applications of positive-weighted adversary are not limited to just \refthm{advBasic}, as we will see in \rf(exm:AmbainisFunction).  

The absence of a handy tool like \reflem{mathias} for tight estimation of $\|\Gamma\circ\Delta_j\|$ when $\Gamma$ has negative entries makes application of the latter rather complicated.  But it is important to keep in mind that the positive-weighted adversary, as intuitive it may be, is subject to some severe limitations we describe in \refsec{advLimitations}.

\subsection{First Half of the Proof of \TeXBug{Theorem \ref*{thm:adv}}}
\label{sec:advProof}
In this section, we prove the first half of \rf(thm:adv), namely, that $Q(f) = \Omega(\Adv(f))$. 
Assume we have a quantum query algorithm that evaluates $f$.  We define a quantity called the {\em progress function}.  It measures the distinguishability of the states of the quantum algorithm corresponding to pairs of inputs with different values of the function.

In the beginning of the algorithm, its state does not depend on the input, hence, the states are completely indistinguishable.  Just before the final measurement, the states of the algorithm, corresponding to inputs with different values of the function, are projected to orthogonal subspaces, hence, the distinguishability is high.  If one query to the input oracle does not change the progress function by much, this yields a lower bound.  
%We use the inner product to measure the distinguishablity of one pair of the states, and we use the elements of $\Gamma$ to average over all possible pairs of inputs.

More formally, let $T$ denote the total number of queries performed by the algorithm.  If $t$ is an integer between 0 and $T$, and $z\in\cD$, define the state of the algorithm corresponding to $z$ after $t$ queries by
\begin{equation}
\label{eqn:advpsit}
\psi^{(t)}_z = U_tO_z U_{t-1} O_z\cdots U_1 O_z U_0 \ket |0>.
\end{equation}
Recall from \rf(sec:query) that 
\(
O_z = \bigoplus_{j=0}^{\vars} O_{z_j}
\)
 where $(O_a)_{a\in[q]}$ is a family of unitaries.  Due to \rf(rem:realEntries), we may assume that the vectors $\psi^{(t)}_z$ have real entries.

Let $\delta = (\delta_z)$ be a principal eigenvector of $\Gamma$.  The progress functions is defined by
\begin{equation}
\label{eqn:progress}
W^{(t)} = \sum_{x,y\in\cD} \Gamma\elem[x,y]\delta_x\delta_y \ip<\psi^{(t)}_x, \psi^{(t)}_y>.
\end{equation}
We split the proof into three parts: proving that $W^{(0)}$ is large, proving that $W^{(T)}$ is small, and proving that $W^{(t)}-W^{(t+1)}$ is small.

\begin{clm}
\label{clm:adv1}
We have $W^{(0)}=\norm |\Gamma|$.
\end{clm}

\pfstart
This part is trivial.  We have $\psi^{(0)}_z = U_0\ket |0>$ no matter what $z$ is.  Hence, $\ip<\psi^{(0)}_x, \psi^{(0)}_y>=1$ for all $x, y\in\cD$.  Plugging this into~\refeqn{progress} gives
\[
W^{(0)} = \sum_{x,y\in\cD} \Gamma\elem[x,y]\delta_x\delta_y = \delta^* \Gamma \delta = \norm|\Gamma|.\qedhere
\]
\pfend

Before we proceed, we need a simple result from linear algebra.
\begin{lem}
\label{lem:simple}
Let $A$ be $n\times n$ matrix, and $U$ and $V$ be $m\times n$ matrices with columns $\{u_i\}_{i\in[n]}$ and $\{v_i\}_{i\in[n]}$, respectively.  Then,
\[
\absB|\sum_{i,j\in[n]} A\elem[i,j]\ip<u_i, v_j>| \le \norm|A| \normFrob|U| \normFrob|V|.
\]
\end{lem}

\pfstart
Indeed, using the Cauchy-Schwarz inequality, and the definition of the spectral norm:
\[
\absB|\sum_{i,j\in[n]} A\elem[i,j]\ip<u_i,v_j>| =
|\ip<A,U^*V>| = |\ip<U A,V>| \le \normFrob|U A|\normFrob|V| \le \norm|A| \normFrob|U| \normFrob|V|.\qedhere
\]
\pfend

\begin{clm}
\label{clm:adv3}
We have $W^{(T)}\le 2\sqrt{\eps} \norm|\Gamma|$.
\end{clm}

\mycommand{Pip}{\Pi^\perp}
\pfstart
Denote for brevity $\psi_z = \psi^{(T)}_z$.  By the assumption on the correctness of the algorithm, there exist orthogonal projectors $\{\Pi_a\}_{a\in[\ell]}$ such that $\norm|\Pi_{f(z)}\psi_z|^2\ge 1-\eps$ for all $z\in\cD$.  Denote $\Pip_a = I - \Pi_a$, so that $\|\Pip_{f(x)}\psi_z\|^2\le \eps$ for all $z\in\cD$.  We have
\begin{equation}
\label{eqn:Wt}
\begin{aligned}
W^T &= 
\sum_{x,y\in\cD} \Gamma\elem[x,y]\delta_x\delta_y \ip<\Pi_{f(x)}\psi_x, \Pi_{f(y)}\psi_y>
\\
&+ \sum_{x,y\in\cD} \Gamma\elem[x,y]\delta_x\delta_y \ip<\psi_x, \Pip_{f(y)}\psi_y> + 
\sum_{x,y\in\cD} \Gamma\elem[x,y]\delta_x\delta_y \ip<\Pip_{f(x)}\psi_x, \Pi_{f(y)}\psi_y>
\end{aligned}
\end{equation}
Note that the first term of~\refeqn{Wt} equals 0.  Indeed, if $f(x)=f(y)$, then $\Gamma\elem[x,y]=0$; otherwise $\Pi_{f(x)}$ and $\Pi_{f(y)}$ project to orthogonal subspaces.
For the second term, let $U$ and $V$ be the matrices having $u_x = \delta_x \psi_x$ and $v_y = \delta_y \Pip_{f(y)}\psi_y$ as their columns, respectively.  Then, by \reflem{simple}, the second term of~\refeqn{Wt} is at most $\|\Gamma\| \normFrob|U|\normFrob|V|$.  We have
\[
\normFrob|U|^2 = \sum_{x\in\cD}\delta_x^2\|\psi_x\|^2\le 1,\qquad\text{and}\qquad
\normFrob|V|^2 = \sum_{y\in\cD}\delta_y^2\|\Pip_{f(y)}\psi_y\|^2\le\eps.
\]
Thus, the first term  is at most $\sqrt{\eps}\|\Gamma\|$.  For the third term, the same estimate can be obtained.
\pfend

\begin{clm}
\label{clm:adv2}
We have $|W^{(t)}-W^{(t+1)}|\le 2\max_{j\in[\vars]} \norm|\Gamma\circ\Delta_j|$.
\end{clm}

\pfstart
Denote $\psi_z = \psi^{(t)}_z$ and $\psi'_z = \psi^{(t+1)}_z$.  
The vector $\psi_z$ can be decomposed as $\bigoplus_{j= 0}^{\vars} \psi_{z,j}$ where the decomposition is the same as for $O_z$ after \rf(eqn:advpsit).
%Let $\Gamma_j = \Gamma\circ\Delta_j$.  
The idea behind the proof is that if $x_j=y_j$, then the oracle does not change the inner product between $\psi_{x,j}$ and $\psi_{y,j}$, hence, the corresponding entry of $\Gamma$ can be ignored.  More formally, for any $x,y\in\cD$, we have
\[
\ip<\psi_x, \psi_y>-\ip<\psi'_x, \psi'_y>  = 
\ip<\psi_x, \psi_y>-\ip<O_x\psi_x, O_y\psi_y> = 
\sum_{j=0}^\vars \chi_{x,y,j},
\]
where $\chi_{x,y,j} = \ip<\psi_{x,j} , \psi_{y,j}> - \ip<O_{x_j}\psi_{x,j}, O_{y_j}\psi_{y,j} >$.  Note that $\chi_{x,y,j} = 0$ if $x_j=y_j$.  In particular, $\chi_{x,y,j}=0$, if $j=0$.  Thus,
\begin{align}
|W^{(t)}&-W^{(t+1)}| = \absC|\sum_{x,y\in\cD} \Gamma\elem[x,y]\delta_x\delta_y \s[\ip<\psi_x, \psi_y> - \ip<\psi'_x, \psi'_y>]| \notag\\
&\qquad=\absC| \sum_{x,y\in\cD}\sum_{j=0}^\vars\Gamma\elem[x,y]\delta_x\delta_y \chi_{x,y,j} |
=\absC|\sum_{j=1}^\vars \sum_{x,y\in\cD}(\Gamma\circ\Delta_j)\elem[x,y]\delta_x\delta_y \chi_{x,y,j} |\notag\\
\le& \sum_{j=1}^\vars \absC| \sum_{x,y\in\cD} (\Gamma\circ\Delta_j)\elem[x,y]\delta_x\delta_y \ip<\psi_{x,j} , \psi_{y,j}>  | +
\sum_{j=1}^\vars \absC| \sum_{x,y\in\cD} (\Gamma\circ\Delta_j)\elem[x,y]\delta_x\delta_y \ip<O_{x_j}\psi_{x,j} , O_{y_j}\psi_{y,j}> |. \label{eqn:Wtchange}
\end{align}
Let us estimate the second term, the first one being similar.  For $j\in[\vars]$, let $U_j$ be the matrix with columns $u_{j,z} = \delta_z O_{z_j}\psi_{z,j}$.  Thus, by \reflem{simple}, the second term of~\refeqn{Wtchange} is at most
\[
\sum_{j=1}^\vars \norm|\Gamma\circ\Delta_j| \normFrob|U_j|^2 \le \max_{i\in [\vars]} \norm|\Gamma\circ\Delta_i|\sum_{j=1}^\vars  \normFrob|U_j|^2,
\]
and, finally,
\[
\sum_{j=1}^\vars  \normFrob|U_j|^2 = \sum_{j=1}^\vars \sum_{z\in\cD} \delta_z^2 \norm|O_{z_j}\psi_{z,j}|^2 =
\sum_{z\in\cD} \delta_z^2 \sum_{j=1}^\vars  \norm|\psi_{z,j}|^2 \le
\sum_{z\in\cD} \delta_z^2 \norm|\psi_{z}|^2 = \norm|\delta|^2 = 1.
\]
By plugging this in \rf(eqn:Wtchange), and using the same estimate for the first term, we obtain the required inequality.
\pfend

Let $\eps$ be the error of the quantum query algorithm.  Usually it is $1/3$, but by \rf(lem:circuitAmplify), we may reduce it below any positive constant by introducing a constant multiplicative overhead.
If we take $\eps=1/5$, then $2\sqrt{\eps}<1$, hence, Claims~\ref{clm:adv1}, \ref{clm:adv2} and~\ref{clm:adv3} are enough to deduce that $T = \Omega(\Adv(f))$.

\section{Duality}
\label{sec:advDuality}
The aim of this section is to give an alternative formulation of the adversary bound as an optimisation problem.  This is achieved in \rf(sec:advDual).  This formulation will become more important in \rf(sec:advAlgorithms), when we show how to convert it into a quantum query algorithm.  As a by-product, in \rf(sec:advLimitations), we obtain some limitations on the positive-weighted adversary.  In \rf(sec:advExamples), we give a number of examples of applications of the adversary bound and its dual.  In \rf(sec:spanPrograms), we define span programs, a notion closely related to the dual adversary bound.

\subsection{Dual Adversary Bound}
\label{sec:advDual}
\begin{thm}
\label{thm:advDual}
Let $f\colon [q]^\vars \supseteq \cD\to\{0,1\}$ be a function.  Then, the adversary bound $\Adv(f)$ is equal to the optimal value of the both following optimisation problems:
\begin{subequations}
\label{eqn:advPrimal}
\begin{alignat}{3}
 &{\mbox{\rm maximise }} &\quad& \|\Gamma\| \label{eqn:advPrimalObjective} \\ 
 &{\mbox{\rm subject to }} && \|\Gamma\circ\Delta_j\| \le 1&\quad&\text{\rm for all $j\in[\vars]$;} \label{eqn:advPrimalCondition}\\
 &&& \mbox{\rm $\Gamma$ is an adversary matrix.} 
\end{alignat}
\end{subequations}
and
\begin{subequations}
\label{eqn:advDual}
\begin{alignat}{3}
&\mbox{\rm minimise} &\quad& \max_{z\in \cD}\sum\nolimits_{j \in [\vars]} X_j\elem[z,z]  \label{eqn:advDualObjective}\\
& \mbox{\rm subject to}&& \sum\nolimits_{j\colon x_j \ne y_j} X_j\elem[x,y] = 1 &\quad& \text{\rm for all $x\in f^{-1}(1)$ and $y\in f^{-1}(0)$;} \label{eqn:advDualCondition}\\
&&& X_j\succeq 0 && \mbox{\rm for all $j\in [\vars]$,} \label{eqn:advDualSemidefinite}
\end{alignat}
\end{subequations}
where $(X_j)_{j\in\vars}$ are positive semi-definite matrices with rows and columns labelled by the elements of $\cD$.
\end{thm}

\pfstart
Equation~\refeqn{advPrimal} is a mere restatement of the definition of the adversary bound.  Thus, it remains to prove that the optimisation problems~\refeqn{advPrimal} and~\refeqn{advDual} have equal optimal values.  This is done using semi-definite duality.  For a brief explanation of semi-definite optimisation, the reader may refer to \rf(sec:convexOptimisation).

First of all, we transform \refeqn{advPrimal} into a form more suitable for taking the dual. As in the proof of \refprp{advMultipleRows}, we may assume $\delta=(\delta_z)$ is a normalised $\|\Gamma\|$-eigenvector of $\Gamma$ with real entries. The objective value \refeqn{advPrimalObjective} equals
\newcommand{\xyho}{\substack{x\in f^{-1}(1)\\y\in f^{-1}(0)}}
\[
\norm|\Gamma| = \delta^* \Gamma\delta = 2\sum_{\xyho} \delta_x\delta_y \Gamma\elem[x,y].
\]
Also, from \refrem{booleanOutput} we know that since $f$ has Boolean output, the spectrum of $\Gamma\circ\Delta_j$ is symmetric with respect to 0. Thus, Eq. \refeqn{advPrimalCondition} holds if and only if 
\begin{equation}
\label{eqn:duality1}
\Gamma\circ\Delta_j \preceq I\qquad\text{for all $j\in [\vars]$.}
\end{equation}
Let us define the matrices 
\begin{equation}
\label{eqn:NandM}
\Lambda = (\lambda_{x,y})\qquad\mbox{and}\qquad M = \diag(\mu_z)
\end{equation}
by $\Lambda = \Gamma\circ(\delta\delta^*)$ and $M = I\circ (\delta\delta^*)=\diag (\delta_z^2)$, respectively. By taking the Hadamard product of both parts of \refeqn{duality1} with $\delta\delta^*\succeq 0$, we get $\Lambda\circ \Delta_j \preceq M$. Thus, we see that the optimisation problem \refeqn{advPrimal} is equivalent to the following one
\begin{subequations}
\label{eqn:advPrimalVar}
\begin{alignat}{3}
&{\mathrm{maximise}} &\quad& 2 \sum\nolimits_{\xyho} \lambda_{x,y}\\
& \mathrm{subject\  to} && \sum\nolimits_{z\in \cD} \mu_z = 1 \\ 
&&& M\succeq \Lambda\circ\Delta_j &\quad&\mbox{for all $j\in[\vars]$;}\label{eqn:advPrimalConditionM}\\
&&& \mu_z\ge 0 &&\mbox{for all $z\in\cD$;}\label{eqn:advPrimalConditionMu}\\
&&& \text{$\Lambda$ is an adversary matrix.}
\end{alignat}
\end{subequations}
%where $\Lambda$ and $M$ are as in \refeqn{NandM}.
Indeed, we have just shown how a feasible solution for \refeqn{advPrimal} can be transformed into a feasible solution for \refeqn{advPrimalVar}. For the reverse direction, we at first prove that $\mu_z=0$ implies that the $z$th row of $\Lambda$ contains only zeros. Assume that $\mu_z=0$, take any $z'\ne z$, choose any $j$ such that $z_j\ne z'_j$, and consider the $2\times 2$ submatrix of \refeqn{advPrimalConditionM} given by the rows and the columns labelled by $z$ and $z'$. The non-negativeness condition implies that $\lambda_{z,z'}=0$. Thus, the transformation from \refeqn{advPrimal} to \refeqn{advPrimalVar} can be reversed by assigning $\delta_z = \sqrt{\mu_z}$ and $\Gamma\elem[z,z'] = \lambda_{z,z'}/(\delta_z\delta_{z'})$ with the agreement that $0/0=0$.

Now we construct the dual of~\refeqn{advPrimalVar}. For that, we write out the Lagrangian:
\begin{equation}
\label{eqn:lagranzh1}
L(\mu, \lambda, t, X) = 
 2 \sum_{\xyho} \lambda_{x,y} + \s[1 - \sum_{z\in \cD} \mu_z] t + \sum_{j\in [\vars]} \tr ((M-\Lambda\circ \Delta_j) X_j) ,
\end{equation}
where $t\in \R$, and $\{X_j\}_{j\in[\vars]}$ are positive semi-definite $\cD\times\cD$ matrices. 
%Let $\mu^*_x$ and $\lambda_{x,y}^*$ be an optimal solution to~\refeqn{advPrimalVar}. Then, using the feasibility of $\mu^*_x$ and $\lambda_{x,y}^*$, we have for any choice of $t$ and $X_j$:
%\begin{equation}
%\label{eqn:lagranzhSup}
%\sup_{\mu_z\ge 0, \lambda_{x,y}} L(\mu_z, \lambda_{x,y}, t, X_j) \ge L(\mu_z^*, \lambda_{x,y}^*, t, X_j) \ge 2 \sum_{\xyho} \lambda_{x,y}^* ,
%\end{equation}
%that is the optimal solution to~\refeqn{advPrimalVar}. 
Let us transform the Lagrangian to aid with taking the supremum. Let $\{e_z\}$ be the standard basis of $\C^{\cD}$. Then the Lagrangian equals
%\begin{subequations}
\begin{align}
& \sum_{x,y\colon f(x)\ne f(y)} \lambda_{x,y} + \s[1 - \sum_{z\in \cD} \mu_z]t + \sum_{j\in [\vars]} \tr \sD[{\skC[\sum_{z\in\cD}\mu_z e_ze_z^* - \sum_{\substack{x,y\colon f(x)\ne f(y),\\ x_j\ne y_j}} \lambda_{x,y} e_x e_y^*]}X_j] \notag\\
& \qquad=\; t + \sum_{z\in\cD} \mu_z \sC[\sum_{j \in [\vars]} \tr(e_z e_z^* X_j) - t] + \sum_{x,y\colon f(x)\ne f(y)} \lambda_{x,y} \sC[1 - \sum_{j\colon x_j \ne y_j} \tr(e_x e_y^* X_j) ] \notag\\
&\qquad=\; t + \sum_{z\in\cD} \mu_z \sC[\sum_{j \in [\vars]} {X_j\elem[z,z]} - t] + \sum_{x,y\colon f(x)\ne f(y)} \lambda_{x,y} \sC[1 - \sum_{j\colon x_j \ne y_j} {X_j\elem[x,y]} ].\label{eqn:lagranzh2}
\end{align}
%\end{subequations}
By optimising \refeqn{lagranzh2}, we get
\[
\sup_{\mu\colon \mu_z\ge 0,\; \lambda} L(\mu, \lambda, t, X) = 
\begin{cases}
t,& \parbox{10cm}{if $t\ge \sum_{j \in [\vars]} X_j\elem[z,z]$ for all $z\in \cD$,\\ and $\sum_{j\colon x_j \ne y_j} X_j\elem[x,y]=1$ for all $x\in f^{-1}(1)$ and $y\in f^{-1}(0)$;}\\
+\infty,& \text{otherwise.}
\end{cases}
\]
%By~\refeqn{lagranzhSup}, this is an upper bound for~\refeqn{advPrimalVar} and, subsequently,~\refeqn{advPrimal}. 
%Looking for the best upper bound, we arrive at the following optimization problem
That gives us the optimisation problem:
\begin{subequations}
\label{eqn:advDualVar}
\begin{alignat}{3}
&\text{\rm minimise} &\quad& t\\
&\text{\rm subject to} && t\ge \sum\nolimits_{j \in [\vars]} X_j\elem[z,z] &\quad& \mbox{for all $z\in \cD$;} \\
&&& \sum\nolimits_{j: x_j \ne y_j} X_j\elem[x,y] =1 && \mbox{for all $x\in f^{-1}(1)$ and $y\in f^{-1}(0)$;}\\
&&& X_j\succeq 0 && \text{for all $j\in[\vars]$.}
\end{alignat}
\end{subequations}
Thus, the optimal value of~\refeqn{advDualVar} is at least the optimal value of~\refeqn{advPrimalVar}.  Also, it is easy to see that~\refeqn{advDualVar} is equivalent to~\refeqn{advDual}.
To prove the equality, we use Slater's condition.  By this condition, it is enough to prove that~\refeqn{advPrimalVar} is convex and strictly feasible.  The first condition is trivial, because the objective function and all the constraints are linear in $\lambda$ and $\mu$.  Strict feasibility also holds, because one may take $\lambda_{x,y}=0$ for all $x,y$, and $\mu_z = 1/|\cD|$ for all $z\in \cD$.
\pfend

\subsection{Limitations of Positive-Weighted Adversary}
\label{sec:advLimitations}
\begin{thm}
\label{thm:advPositiveDual}
Let $f\colon [q]^\vars\supseteq \cD\to\{0,1\}$ be a function.  The positive-weighted adversary $\pAdv(f)$ is equal to the optimal value of the following optimisation problem:
\begin{subequations}
\label{eqn:advPositiveDual}
\begin{alignat}{3}
&\mbox{\rm minimise} &\quad& \max_{z\in \cD}\sum\nolimits_{z \in [\vars]} X_j\elem[z,z] \label{eqn:advPositiveDual:objective}\\
& \mbox{\rm subject to}&& \sum\nolimits_{j\colon x_j \ne y_j} X_j\elem[x,y] \ge 1 &\quad& \text{\rm for all $x\in f^{-1}(1)$ and $y\in f^{-1}(0)$;} \label{eqn:advPositiveDual:condition}\\
&&& X_j\succeq 0 && \mbox{\rm for all $j\in [\vars]$,}
\end{alignat}
\end{subequations}
where $X_j$, for $j\in [\vars]$, are positive semi-definite matrices with rows and columns labelled by the elements of $\cD$.  Without any loss in the estimate, one can assume that matrices $X_j$ are rank-1, i.e., given by $X_j = \psi_j\psi_j^*$ with $\psi_j\in\R^{\cD}$.
\end{thm}

\pfstart
The proof goes along the lines of the proof of \refthm{advDual}.  By definition, the positive-weighted adversary equals the optimal value of~\refeqn{advPrimal} when $\Gamma$ ranges over matrices with non-negative entries.  Hence, it equals the optimal value of~\refeqn{advPrimalVar} with the additional condition that $\lambda_{x,y}\ge 0$ for all $x,y$.  The Lagrangian still equals~\refeqn{lagranzh2}.  By taking the supremum, we arrive at
\[
\sup_{\mu,\lambda \colon \mu_z\ge 0, \lambda_{x,y}\ge 0} L(\mu, \lambda, t, X) = 
\begin{cases}
t,& \parbox{10cm}{if $t\ge \sum_{j \in [\vars]} X_j\elem[z,z]$ for all $z\in \cD$,\\ and $\sum_{j\colon x_j \ne y_j} X_j\elem[x,y]\ge 1$ for all $x\in f^{-1}(1)$ and $y\in f^{-1}(0)$;}\\
+\infty,& \text{otherwise.}
\end{cases}
\]
This trivially yields~\refeqn{advPositiveDual}.  As $X_j\succeq 0$, there exist vectors $\phi_{j,z}\in \R^d$ where $z\in\cD$ and $d$ is some integer, such that $X_j\elem[z,z'] = \ip<\phi_{j, z}, \phi_{j,z'}>$ for all $z,z'\in\cD$.  Then, we may define $X_j = \psi_j\psi_j^*$ with $\psi_j\elem[z] = \| \phi_{j,z} \|$.  This does not change $X_j\elem[z,z]$, but may only increase $X_j\elem[z,z']$.
\pfend

Using this upper bound on the value of positive-weighted adversary, it is easy to prove some no-go results for this lower bound technique.  (The certificate complexities $C_0$ and $C_1$ are defined in \rf(sec:queryRelated).)

\begin{prp}
\label{prp:advLimitations}
Let $f\colon [q]^\vars\supseteq \cD\to\{0,1\}$ be a function.  Then,
\itemstart
\item[(a)] $\pAdv(f)\le \sqrt{\vars\min\{C_0(f), C_1(f)\}}$;
\item[(b)] if $f$ is total, i.e., $\cD = [q]^\vars$, then $\pAdv(f)\le \sqrt{C_0(f)C_1(f)}$;
\item[(c)] if $f$ is such that the Hamming distance between $f^{-1}(0)$ and $f^{-1}(1)$ is $\eps \vars$ then $\pAdv(f)\le 1/\eps$.
\itemend
\end{prp}

Points (a) and (b) are known as the {\em certificate complexity barrier}, and Point (c) is known as the {\em property testing barrier}.  In particular, Point (c) implies that the positive-weighted adversary cannot prove non-trivial lower bound for the collision or set equality problems, while Point (b) rules out an $\omega(\sqrt{\vars})$ positive-weighted adversary bound for the element distinctness problem.

\pfstart[Proof of \refprp{advLimitations}]
We start with (c).  Define $X_j$, for all $j$, as the matrix with all entries equal to $1/(\eps \vars)$.  It is a feasible solution for~\refeqn{advPositiveDual}, and its objective value is $1/\eps$.

For (a) and (b), we may assume, without loss of generality, that $C_1(f)\le C_0(f)$.  For each $z\in\cD$, choose a minimal certificate $M(z)$.  For (a), define $\psi_j$ by
\[
\psi_j\elem[z] = \begin{cases}
\sqrt{\vars/C_1(f)},& \text{$f(z)=1$, and $j\in M(z)$;}\\
\sqrt{C_1(f)/\vars},& \text{$f(z)=0$;}\\
0,&\text{otherwise.}
\end{cases}
\]
For each $x\in f^{-1}(1)$ and $y\in f^{-1}(0)$, there is at least one position $j\in M(x)$ such that $x_j\ne y_j$, hence this is a feasible solution to~\refeqn{advPositiveDual}.  Also, it is easy to see that the objective value is $\sqrt{\vars C_1(f)}$.

If $f$ is total, we use the observation that, for all $x\in f^{-1}(1)$ and $y\in f^{-1}(0)$, $M(x)\cap M(y)\ne\emptyset$.  We can define $\psi_j$ by
\[
\psi_j\elem[z] = \begin{cases}
\sqrt{C_0(f)/C_1(f)},& \text{$f(z)=1$, and $j\in M(z)$;}\\
\sqrt{C_1(f)/C_0(f)},& \text{$f(z)=0$, and $j\in M(z)$;}\\
0,&\text{otherwise,}
\end{cases}
\]
it is a feasible solution for~\refeqn{advPositiveDual}, and the objective value is $\sqrt{C_0(f)C_1(f)}$.
\pfend

\subsection{Examples}
\label{sec:advExamples}
\mycommand{thres}{\mathrm{Threshold}}
In this section, we give a number of examples of applications of the previously introduced notions.  Let $\thres_{k,\vars}$ denote the $k$-threshold function on $\vars$ variables defined in \rf(defn:threshold).

\begin{prp}
\label{prp:advThresholdExact}
$\Adv(\thres_{k,\vars}) = \sqrt{k(\vars-k+1)}$.
\end{prp}

This is the only place in the thesis where we calculate the {\em exact} value of the adversary bound.  In all other cases, we estimate the adversary bound up to a constant factor, that is sufficient in the light of \rf(thm:adv).  On the other hand, one may be also interested in the exact value, as we will see it below, in \rf(thm:advCompose).

\pfstart[Proof of \rf(prp:advThresholdExact).]
Let $f = \thres_{k,\vars}$.
The lower bound follows from the proof of \rf(prp:advThresholdLower) via the construction of \rf(prp:advBasic2Spectral).  For instance, for the $\thres_{3,4}$ function, we get the following matrix (the action of $\Delta_1$ on the matrix is also shown):
\[
\bordermatrix{
 & 0011 & 0101 & 0110 & 1001 & 1010 & 1100 \cr
0111 & 1 & 1 & 1 & 0 & 0 & 0 \cr
1011 & 1 & 0 & 0 & 1 & 1 & 0 \cr
1101 & 0 & 1 & 0 & 1 & 0 & 1 \cr
1110 & 0 & 0 & 1 & 0 & 1 & 1 \cr
}
\stackrel{\Delta_1}{\longmapsto}
\bordermatrix{
 & 0011 & 0101 & 0110 & 1001 & 1010 & 1100 \cr
& 0 & 0 & 0 & 0 & 0 & 0 \cr
& 1 & 0 & 0 & 0 & 0 & 0 \cr
& 0 & 1 & 0 & 0 & 0 & 0 \cr
& 0 & 0 & 1 & 0 & 0 & 0 \cr
}.
\]

Now let us prove the corresponding upper bound.  
In the following, we will assume that $x$ is a positive input, and $y$ is a negative one.
For an $x$, let $x'$ be the same string with all ones, beyond the first $k$, replaced by zeroes.  Similarly, for an $y$, let $y'$ be $y$ with all zeroes, beyond the first $\vars-k+1$, replaced by ones.  Thus, $x'$ and $y'$ still are positive and negative inputs, with Hamming weights of $k$ and $k-1$, respectively.  

We construct the matrices $X_j$s from \rf(eqn:advDual) so that they satisfy $X_j\elem[x,y] = 0$ unless $x'_j = 1$ and $y'_j=0$.  And if $x'_j = 1$ and $y'_j=0$, then
\begin{equation}
\label{eqn:thExact1}
X_j\elem[x,y] = \frac{1}{\absA|\{i\in[\vars]\mid x'_i = 1,\; y'_i = 0\}|}, \quad X_j\elem[x,x] = \sqrt{\frac{\vars-k+1}{k}}, \quad\mbox{and}\quad X_j\elem[y,y] = \sqrt{\frac{k}{\vars-k+1}}.
\end{equation}
In particular, an input $z$ is not used in the matrix $X_j$ if $j$ corresponds to a position where $z$ differs from $z'$.  
%Thus, we will further assume that $|x| = k$ and $|y|=k-1$.

For example, for the $\thres_{2,3}$ function, we can take the following matrices:
\begin{equation}
\label{eqn:maj3X1andX2}
X_1 = \bordermatrix {
& 111 & 110 & 101 & 010 & 001 & 000\cr
111 & 1& 1 & 1/2 & 1 & 1/2 & 1/2\cr
110 & 1& 1 & 1/2 & 1 & 1/2 & 1/2\cr
101 & 1/2 & 1/2 & 1 & 1/2 & 1 & 1\cr 
010 & 1 & 1 & 1/2 & 1 & 1/2 & 1/2 \cr
001 & 1/2 & 1/2 & 1 & 1/2 & 1 & 1\cr
000 & 1/2 & 1/2 & 1 & 1/2 & 1 & 1\cr
}\qquad
X_2 = \bordermatrix {
& 111 & 110 & 011 & 100 & 001 & 000\cr
111 & 1& 1 & 1/2 & 1 & 1/2 & 1/2\cr
110 & 1& 1 & 1/2 & 1 & 1/2 & 1/2\cr
011 & 1/2 & 1/2 & 1 & 1/2 & 1 & 1\cr 
100 & 1 & 1 & 1/2 & 1 & 1/2 & 1/2 \cr
001 & 1/2 & 1/2 & 1 & 1/2 & 1 & 1\cr
000 & 1/2 & 1/2 & 1 & 1/2 & 1 & 1\cr
}
\end{equation}
and
\begin{equation}
\label{eqn:maj3X3}
X_3 = \bordermatrix {
& 101 & 011 & 100 & 010 \cr
101 & 1 & 1/2 & 1 & 1/2\cr
011 & 1/2 & 1 & 1/2 & 1\cr 
100 & 1 & 1/2 & 1 & 1/2\cr
010 & 1/2 & 1 & 1/2 & 1\cr
}
\end{equation}
extended with zeroes for all other input pairs.

If the $X_j$s satisfy~\rf(eqn:thExact1) and $X_j\succeq 0$, then all the conditions in~\rf(eqn:advDual) are satisfied, and the objective value is $\sqrt{k(\vars-k+1)}$.  It remains to show that there exist such matrices.  First of all, we may write $X_j = B_j\circ D$, where $D$ is a positive semi-definite rank-1 block matrix
\[
\bordermatrix{
 & f^{-1}(1) & f^{-1}(0) \cr
f^{-1}(1) & \sqrt{(\vars-k+1)/k} & 1 \cr
f^{-1}(0) & 1 & \sqrt{k/(\vars-k+1)}  \cr
}
\]
and $B_j\succeq 0$ satisfies
\begin{equation}
\label{eqn:thExact2}
B_j\elem[x,y] = \frac{1}{\absA|\{i\in[\vars]\mid x'_i = 1,\; y'_i = 0\}|} \qquad\mbox{and}\qquad B_j\elem[x,x] = B_j\elem[y,y] = 1.
\end{equation}
In order to get $B_j$, we apply the following
\begin{lem}
\label{lem:1minus1overz}
Assume that $k>0$ is a real number, and a matrix $A\succeq 0$ is such that $\absA|A\elem[i,j]|<k$ for all $i$ and $j$.
Then the matrix $B$, defined by $B\elem[i,j] = (k - A\elem[i,j])^{-1}$, is also positive semi-definite.
\end{lem}

\pfstart
This follows from the series valid for all real $|x|<k$:
\[
\frac{1}{k-x} = \frac{1/k}{1-(x/k)} = \frac1k + \frac x{k^2} + \frac{x^2}{k^3} + \cdots,
\]
and the fact that the set of semi-definite matrices is a topologically closed convex cone, closed under Hadamard product (cf. \rf(sec:linearAlgebra)).
\pfend

Clearly, $\absA|\{i\in[\vars]\mid x'_i = 1,\; y'_i = 0\}| = k - \absA|\{i\in[\vars]\mid x'_i = y'_i = 1\}|$.  Thus, if we take a matrix $A_j\succeq 0$ satisfying
\[
A_j\elem[x,y] = \absA|\{i\in[\vars]\mid x'_i = y'_i = 1\}| \qquad\mbox{and}\qquad 
A_j\elem[x,x] = A_j\elem[y,y] = k-1,
\]
and apply \rf(lem:1minus1overz), we obtain a matrix $B_j$ satisfying~\rf(eqn:thExact2).  Finally, we can define $A_j = \sum_{i\in[\vars]\setminus\{j\}} C_i$, where $C_i\succeq 0$ is a rank-1 matrix given by 
\[
C_i\elem[x,y] =
\begin{cases}
1, & x'_i=y'_i = 1 \\
0, & \mbox{otherwise}
\end{cases}
\]
for all $x,y\in\{0,1\}^\vars$.
\pfend

One way of obtaining new functions from the existing ones is by composing them.  Assume $f\colon \{0,1\}^n\to\{0,1\}$ and $g_i\colon \{0,1\}^{m_i}\to \{0,1\}$ are Boolean functions, where $i$ ranges over $[n]$.  Then $f(g_1,\dots,g_n)$ is the composed Boolean function on $N=m_1+\cdots+m_n$ variables defined by
\[
(z_{1},\dots,z_{N}) \longmapsto f\sB[{g_1(z_{1},\dots,z_{m_1}),\;
g_2(z_{m_1+1},\dots,z_{m_1+m_2}),\;\dots\dots,\;g_n(z_{N-m_n+1},\dots,z_{N})}].
\]

The adversary bound behaves nicely under the composition operation.
\begin{thm}
\label{thm:advCompose}
Suppose $f,g_1,\dots,g_n$ are as above, and $\Adv(g_1)=\cdots=\Adv(g_n)$.  Then, $\Adv(f(g_1,\dots,g_n)) = \Adv(f)\Adv(g_1)$.
\end{thm}

The proof of this result is not complicated but rather bulky.  We leave it out, the interested reader may refer to the chapter notes for the corresponding references.
The condition on $\Adv(g_i)$ being equal may be dropped, but then the definition of the adversary bound must be modified to include {\em weights} of the variables.

In particular, it is possible to define the $d$th iteration $f^{d}\colon \{0,1\}^{n^d}\to\{0,1\}$ of the function $f$ by $f^{1} = f$ and $f^{d+1} = f(f^{d},\dots,f^{d})$.  By \rf(thm:advCompose), $\Adv(f^{d}) = \Adv(f)^d$.  
Combining this with \rf(thm:adv), we get that
\begin{equation}
\label{eqn:advAsIterated}
\Adv(f) = \lim_{d\to\infty} \sqrt[d]{Q(f^d)}
\end{equation}
where, recall, $Q$ stands for the quantum query complexity.  This gives an alternative definition of the adversary bound purely in terms of the quantum query complexity.
Also, from \rf(thm:advCompose) and \rf(prp:advThresholdExact), we get that $Q(\thres_{k,\vars}^d) = \Theta\s[(k(\vars-k+1))^{d/2}]$.

Another interesting consequence of~\rf(thm:advCompose) is that for every total Boolean function $f$, the functions $f^d$ and $(\mathop{\mathrm{NOT}} f)^d$ have asymptotically equal quantum query complexities.  This is not true for the randomised case, the AND function on 2 arguments providing a counterexample.

\begin{exm}
\label{exm:AmbainisFunction}
Consider the function $f\colon \{0,1\}^4\to\{0,1\}$ defined as follows.  The value of $f(z_1,z_2,z_3,z_4)$ is 1 if and only if the sequence $z_1z_2z_3z_4$ is monotone.  That it, $f$ evaluates to 1 on the following 8 inputs
\[
0000, 0001, 0011, 0111, 1111, 1110, 1100, 1000,
\]
and to 0 everywhere else.

The deterministic query complexity of $f$ is 3, as can be seen from the following algorithm.  Query the first and the third bits.  If they are equal, query the second one, otherwise, query the fourth one.  This is enough to determine the value of the function, and it is not possible to do so with fewer queries (the block sensitivity of the function is 3).

The degree-2 polynomial 
\[
P = 1 - z_2 - z_3 + z_2z_3 + z_1z_2 + z_3z_4 - z_1z_4
\]
of degree 2 represents $f$ exactly.  By composing it with itself, we get a polynomial of degree $2^k$ that represents $f^{k}$ exactly.  Thus, it is not possible to obtain a lower bound better than $\Omega(2^k)$ on $Q(f^{k})$ using \rf(thm:approx).

The adversary matrix $\Gamma$ for the function $f$ looks like
\[
\bordermatrix {
 & 0010 & 0101 & 1011 & 0110 & 1101 & 1010 & 0100 & 1001 \cr
0000& a& c& d& b& d& c& a& b \cr
0001& b& a& c& d& b& d& c& a \cr
0011& a& b& a& c& d& b& d& c \cr
0111& c& a& b& a& c& d& b& d \cr
1111& d& c& a& b& a& c& d& b \cr
1110& b& d& c& a& b& a& c& d \cr
1100& d& b& d& c& a& b& a& c \cr
1000& c& d& b& d& c& a& b& a 
} . 
\]
This layout takes into account all the symmetries of the problem.  For each $j$, the matrix $\Gamma\circ\Delta_j$ consists of two blocks, each equal to the matrix
\begin{equation}
\label{eqn:AmbFunc1}
\begin{pmatrix}
d&  d& c&  b \\
a&  d& b&  c \\
c&  b& d&  a \\
b&  c& d&  d 
\end{pmatrix}
\end{equation}
up to a permutation of rows and columns.  If we allow only non-negative entries, the best choice is $a=3/4$, $b=1/2$ and $c=d=0$.  This gives $\|\Gamma\| = 2(a+b+c+d) = 5/2$, as each row contains exactly two appearances of each $a$, $b$, $c$ and $d$.  The norm of $\Gamma\circ\Delta_j$ can be estimated using the decomposition
\[
\begin{pmatrix}
0&  0& 0&  1/2 \\
3/4&  0& 1/2&  0 \\
0&  1/2& 0&  3/4 \\
1/2&  0& 0&  0 
\end{pmatrix}
=
\begin{pmatrix}
0&  0& 0&  1 \\
\sqrt{3}/2&  0& 1/2&  0 \\
0&  1/2& 0&  \sqrt{3}/2 \\
1&  0& 0&  0 
\end{pmatrix}\circ
\begin{pmatrix}
0&  0& 0&  1/2 \\
\sqrt{3}/2&  0& 1&  0 \\
0&  1& 0&  \sqrt{3}/2 \\
1/2&  0& 0&  0 
\end{pmatrix}
\]
of the matrix in~\rf(eqn:AmbFunc1).  \rf(lem:mathias) then implies that $\|\Gamma\circ\Delta_j\|\le 1$.  This gives $\pAdv(f)=5/2$.  This is the best possible value for the non-negative adversary.  We will show this using \rf(thm:advPositiveDual).  Recall that $j$ is called a {\em sensitive variable} for input $z$, if flipping the value of the $j$th value changes the value of the function, i.e., $f(z)\ne f(z')$ where $z'_i = z_i$ for all $i\ne j$, and $z'_j = 1-z_j$.  Define $\psi_j$ from \rf(thm:advPositiveDual) by
\[
\psi_j\elem[z] = \begin{cases}
1,& \mbox{$j$ is a sensitive variable for $z$;} \\
1/2,& \mbox{otherwise.}
\end{cases}
\]
It is easy to check that each $z\in\cD$ has exactly two sensitive variables, that implies the value $5/2$ for the objective value~\rf(eqn:advPositiveDual:objective).  A simple case analysis shows that~\rf(eqn:advPositiveDual:condition) is satisfied.

With negative entries allowed, we can take $a=0.5788$, $b=0.7065$, $c = 0.1834$ and $d = -0.2120$.  (The values obtained numerically.)  This gives $\Adv(f) = 2.5135$.  Thus, the family of functions $f^{k}$ simultaneously gives asymptotical separations between the polynomial degree, the non-negative adversary bound and the quantum query complexity.
\end{exm}

\subsection{Span Programs}
\label{sec:spanPrograms}
%\mycommand{M}{M}
\mycommand{target}{\tau}
\mycommand{wsize}{\mathop{\mathrm{wsize}}\nolimits}
\mycommand{prostr}{\R^d}
In this section, we define a linear-algebraic model of computation having strong relation to the dual adversary SDP from \refsec{advDual}.  
A span program $\cP$ is specified by
\itemstart
\item a finite-dimensional inner product space $\prostr$;
\item a non-zero {\em target vector} $\target \in \prostr$;
\item a sequence of {\em input vectors} $(v_i)_{i\in \cI}$, where $\cI$ is a finite set of indices, and each $v_i$ is an element of $\prostr$. The set $\cI$ is split into the disjoint union: 
\(
\cI=\bigsqcup_{j\in[\vars],b\in\{0,1\}} \cI_{j,b}.
\)
\itemend

We say the input vectors in $\cI_{j,b}$ are {\em labelled} by (the tuple of) the $j$th input variable $x_j$ and its value $b$.  We often combine the input vectors into an $d\times \cI$-matrix that we denote by $V$.

For each input string $z\in\{0,1\}^\vars$, define the {\em available input vectors} as vectors $v_i$ with indices in $\cI(z)=\bigcup_{j\in[\vars]} \cI_{j,z_j}$.  Vectors $v_i$ with indices in $\cI\setminus \cI(z)$ are called {\em false input vectors}.  The program $\cP$ evaluates to 1 on input $z$ if $\target$ lies in the span of the available input vectors, and $\cP$ evaluates to 0 otherwise.  Thus, span programs define total Boolean functions.  It is also possible to define a span program for a partial Boolean function by ignoring the output on the inputs outside the domain.

Now we define a complexity measure of span programs.  Assume $\cP$ evaluates a function $f\colon \{0,1\}^\vars\supseteq \cD\to\{0,1\}$.  We consider the span program $(\cP, w)$ extended with {\em witnesses} proving that $\cP$ evaluates to the required value for every input in the domain:
\itemstart
\item
If $\cP$ evaluates to 1 on $x\in\cD$, a {\em positive witness} is a vector $w_x\in \R^\cI$ such that $\target = Vw_x$ and $w_x\elem[i]=0$ for all $i\notin \cI(x)$.  The {\em size} of the witness $w_x$ is defined as $\norm|w_x|^2$.
\item
If $\cP$ evaluates to 0 on $y\in\cD$, a {\em negative witness} is a vector $w'_y\in \prostr$ such that $\ip<\target, w'_y>=1$ and $w'_y\perp v_i$ for all $i\in \cI(y)$.  The {\em size} of the witness is defined as $\|V^* w'_y\|^2$.  This equals the sum of the squares of inner products of $w'_y$ with all the false input vectors.
\itemend
The {\em witness size} $\wsize(\cP, w)$ of the (extended) span program $(\cP, w)$ is defined as the maximal size of all its witnesses in $w$.  Usually, one is interested in the {\em minimal} size of a witness for $\cP$.  We will, however, have a {\em particular} witness in mind when designing a span program, thus, will include it in the definition of the span program.  This is analogous to the distinction between {\em optimal} and {\em feasible} solutions to an SDP like~\refeqn{advDual}.

\begin{rem}
\label{rem:SpanGeometricalMean}
In order to avoid unnecessary normalisation, we often calculate $\wsize(\cP,w)$ as $\sqrt{W_0W_1}$, where $W_0$ and $W_1$ are {\em negative} and {\em positive} witness sizes.  The witness size $W_b$, with $b\in\{0,1\}$, is defined as the maximum among all witness sizes for inputs in $f^{-1}(b)$.  This is justified because, if we replace the target vector by $\alpha\tau$, the positive witness size goes up by $\alpha^2$ while the negative witness size goes down by $\alpha^2$.  Thus, by choosing an appropriate value of $\alpha$, we can make them both equal to $\sqrt{W_0W_1}$.
\end{rem}

\begin{exm}
For the OR function on $\vars$ variables, we construct the following span program.  Take $d=1$, and $\tau = (1)$.  For each input variable $j\in[\vars]$, take the input vector $v_j = (1)$ that is available if the value of $z_j$ is 1.  If $x$ is a positive input, define $w_x$ by $w_x\elem[j]=1$ for some $j$ such that $x_j=1$, and $w_x\elem[i] = 0$ for all $i\ne j$.  In this case, we say that we {\em take} the available input vector $v_j$ with coefficient 1.  For the negative input $y$, let $w_y'=(1)$.  The positive witness size is $W_1 = 1$, and the negative witness size is $W_0 = \vars$.  By \rf(rem:SpanGeometricalMean), the witness size of the span program is $\sqrt{\vars}$ that agrees with the Grover search, \rf(prp:Grover).
\end{exm}

\mycommand{Hfree}{H_{\mathrm{free}}}
\paragraph{Free input vectors}  The following modification of span programs is often convenient.  Let $\Hfree$ be a subspace of $\prostr$.  It will be usually given as the span of a number of {\em free input vectors}.  We say that the span program evaluates to 1 iff $\target$ is in the span of the available input vectors and $\Hfree$.  A negative witness is defined as previously with the additional condition on $w'_y\perp \Hfree$.  A positive witness is a pair $(w_x, v_x)$ such that $\target = Vw_x + v_x$ and $v_x\in\Hfree$.  The size of the positive witness is still $\norm|w_x|^2$.

A span program $\cP$ with free input vectors can be converted into an ordinary span program $\cP'$ as follows.  Let $\Pi$ denote the projection onto the orthogonal complement of $\Hfree$.  Then, the target vector of $\cP'$ is $\Pi\target$, and the input vectors are given by $\Pi v_i$.  It is easy to see that a positive witness $w_x$ for $\cP$ is also a positive witness for $\cP'$ of the same size.  Also, if $w_y'$ is a negative witness, then $w_y'\perp \Hfree$ implies $\ip<w_y', v> = \ip<w_y', \Pi v>$ for any $v\in \prostr$.  Thus, $w_y'$ is also a negative witness for $\cP'$ of the same size.

\paragraph{Canonical Span Programs}
\mycommand{preimy}{f^{-1}(0)}
\mycommand{preimx}{f^{-1}(1)}
A span program $(\cP,w)$ is called {\em canonical} if it satisfies the following properties:
\itemstart
\item The vector space of $\cP$ has the set $\preimy$ as its orthonormal basis.  In the following, let $e_y$ denote the element of the standard basis corresponding to $y\in \preimy$;
\item The target vector is given by $\target = \sum_{y\in\preimy} e_y$;
\item For any negative input $y\in\preimy$, its witness is $e_y$.  In particular, $v_i\elem[y]=0$ for all $i\in \cI_{j,y_j}$.
\itemend

As the name suggests, any span program can be transformed into a canonical one.

\begin{prp}
\label{prp:toCanonical}
Let $(\cP, w)$ be a span program for a function $f\colon \{0,1\}^\vars\supseteq \cD\to\{0,1\}$.  Then, there exists a canonical span program $(\cP', w')$ that evaluates the same function and such that $\wsize(\cP',w') = \wsize(\cP, w)$.
\end{prp}

\pfstart
Consider the linear transformation $A\colon \prostr\to \R^{\preimy}$ given by $A = \sum_{y\in\preimy} e_y (w_y')^*$.  
Let $\cP'$ be the image of $\cP$ under this transformation.
Then, the target vector of $\cP'$ is given by
\[
A\tau = \sum_{y\in\preimy} e_y (w_y')^* \tau = \sum_{y\in\preimy} e_y
\] 
by the condition $\ip<w_y',\tau> = 1$.  The input vectors of $\cP'$ are given by $Av_i$.

It is easy to see that a positive witness $w_x$ for $\cP$ is also a positive witness for $\cP'$ of the same size.  Also, $\ip<e_y, Av> = \ip<w_y', v>$ for every $v\in \prostr$, hence, $e_y$ is a witness for a negative input $y$ in $\cP'$, and its size is the same as of $w'_y$ in $\cP$.
\pfend

It turns out that a canonical span program is essentially the same thing as a dual adversary SDP.

\begin{thm}
\label{thm:spanCanonical}
Let $f\colon \{0,1\}^\vars \supseteq \cD\to\{0,1\}$ be a function.  Then any feasible solution $(X_j)$ to the dual adversary SDP~\refeqn{advDual} can be transformed into a canonical span program $(\cP,w)$ evaluating $f$ such that the objective value~\refeqn{advDualObjective} of the program equals the witness size of $\cP$.  And vice versa, any canonical span program $(\cP,w)$ for $f$ can be transformed into a feasible solution to~\refeqn{advDual} with the objective value equal to the witness size of $(\cP,w)$.
\end{thm}

\pfstart
Assume we are given a feasible solution $(X_j)$ for~\refeqn{advDual}.  
As $X_j$ are positive semi-definite, we can find an integer $d$ and vectors $\psi_{j,z}\in \R^d$, with $j\in[\vars]$ and $z\in\cD$, such that $X_j\elem[x,y] = \ip<\psi_{j,x},\psi_{j,y}>$ for all $x,y\in\cD$.

Let us construct the span program.  Its vector space, target vector and negative witnesses are already determined by the canonicity requirement.  It remains to define the input vectors and the positive witnesses.  Let us start with the input vectors.  Define $\cI = [\vars]\otimes \{0,1\}\otimes [d]$, i.e., for each $j\in[\vars]$ and $b\in\{0,1\}$, we define $d$ input vectors $(v_{j,b,i})$ by
\[
v_{j,b,i}\elem[y] =
\begin{cases}
0,& y_j = b;\\
\psi_{j,y}\elem[i],& y_j \ne b.
\end{cases}
\]
For a positive input $x=(x_j)\in\preimx$, we define its witness by $w_x = \bigoplus_{j,b} \delta_{b,x_j} \psi_{j,x}$, where $\delta$ stands for the Kronecker delta.  

\begin{figure}[tbh] 
\[
\myincludegraphics{10cm}{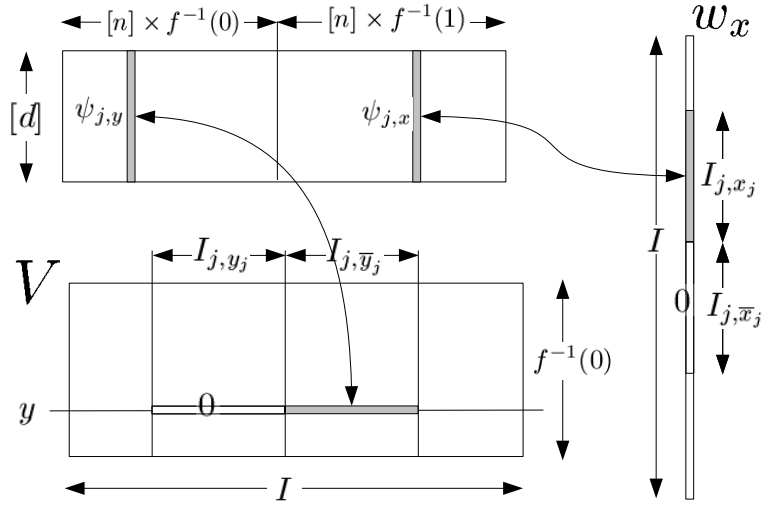}
\]
\caption{A correspondence between a canonical span program and a dual adversary SDP}
\label{fig:SpanPrgToSDP}
\end{figure}

In other words, we rearrange the entries of $\psi_{j,z}$ by placing the entries corresponding to the negative inputs ``horizontally'' into the matrix $V$, and the entries corresponding to the positive inputs ``vertically'' into positive witnesses, see \reffig{SpanPrgToSDP}.  Clearly, this operation is invertible.  Moreover, we have
\[
(Vw_x)\elem[y] = \sum_{j: x_j\ne y_j} \ip<\psi_{j,x}, \psi_{j,y}> = \sum_{j: x_j\ne y_j} X_j\elem[x,y].
\]
Thus, $Vw_x = \target$ for all $x\in\preimx$ if and only if~\refeqn{advDualCondition} holds.  Also, for all $x\in\preimx$ and $y\in\preimy$, their witness sizes are, respectively,
\[
\norm|w_x|^2 = \sum_{j\in[\vars]} \norm|\psi_{j,x}|^2,
\qquad\text{and}\qquad
\norm|V^*w'_y|^2 = \norm|V^*e_y|^2 = \sum_{j\in[\vars]} \norm|\psi_{j,y}|^2.
\]
Thus, the witness size of $(\cP, w)$ equals the objective value~\refeqn{advDualObjective} of the SDP.
\pfend

\begin{exm}
Let us construct a span program for the $\thres_{2,3}$ function.  In~\rf(eqn:maj3X1andX2) and~\rf(eqn:maj3X3), we constructed matrices satisfying the conditions of~\rf(eqn:advDual).  We use the construction from \rf(thm:spanCanonical), ignoring the inputs 000 and 111 as they only copy the inputs 001 and 110.  We get the target vector $\tau$ and 9 input vectors:
\[
\tau = 
\begin{pmatrix}
1 \\ 1 \\ 1
\end{pmatrix},\quad
\sfig{
\begin{pmatrix}
0 \\ \frac 1{\sqrt{2}} \\ \frac1{\sqrt{2}}
\end{pmatrix},
\begin{pmatrix}
0 \\ \frac 1{\sqrt{2}} \\ 0
\end{pmatrix},
\begin{pmatrix}
0 \\ 0 \\ \frac1{\sqrt{2}}
\end{pmatrix}},\quad
\sfig{
\begin{pmatrix}
\frac 1{\sqrt{2}} \\ 0 \\  \frac1{\sqrt{2}}
\end{pmatrix},
\begin{pmatrix}
\frac 1{\sqrt{2}} \\ 0 \\ 0
\end{pmatrix},
\begin{pmatrix}
0 \\ 0 \\  \frac1{\sqrt{2}}
\end{pmatrix}},\quad
\sfig{
\begin{pmatrix}
\frac 1{\sqrt{2}} \\ \frac1{\sqrt{2}} \\ 0
\end{pmatrix},
\begin{pmatrix}
\frac 1{\sqrt{2}} \\ 0 \\ 0
\end{pmatrix},
\begin{pmatrix}
0 \\ \frac1{\sqrt{2}} \\ 0
\end{pmatrix}
}.
\]
The first free input vectors are available for the value 1 of the variable $z_1$, the second three for the value 1 of $z_2$, and the last three for the value 1 of $z_3$.  The witness for the positive input $110$ is $w = \frac{1}{\sqrt{2}}(1,1,0,1,1,0,0,0,0)^*$.  The witness for the negative input $001$ is $w' = (0,0,1)^*$.
\end{exm}

\section{Algorithms}
\label{sec:advAlgorithms}
\mycommand{adjoint}{*}
We finish this chapter on lower bounds with a section on upper bounds.  The aim of this section is to show that the adversary bound is tight.  We do so by showing that the dual of the adversary bound from \rf(sec:advDual) can be transformed into a quantum query algorithm with at most a constant loss in the complexity.  Additionally, we show that a similar transformation can be performed to any span program from \rf(sec:spanPrograms).

We start by reviewing properties of a composition of two reflections from \rf(sec:szegedy).  Recall the notations from that section: $A$ and $B$ are matrices with the same number of rows and each having orthonormal columns, $\Pi_A = A A^\adjoint$, $\Pi_B = B B^*$, $R_A = 2\Pi_A-I$, $R_B = 2\Pi_B-I$, and $U = R_B R_A$.

\begin{lem}[Effective Spectral Gap Lemma] \label{lem:effective}
In the above notations, let $P_\delta$, where $\delta\ge0$, be the orthogonal projection on the span of the eigenvectors of $U$ with eigenvalues $\ee^{\ii\theta}$ satisfying $|\theta|\le \delta$.  Let $u$ be a vector in the kernel of $\Pi_A$. Then,
\[ \|P_\delta \Pi_B u \|\le \frac{\delta}{2}\|u\|. \]
\end{lem}

\begin{wrapfigure}{L}{0pt}
\xygraph{!{0;<2.4pc,-1.0pc>:<1.0pc,2.4pc>::}
(-[u(4)]{}[dr(.3)]*{b_j},
-[u(3)r(2)]{}[dr(.2)]*{a_j},
:[u(1.6)l(2.4)]{}
	([d(.3)]*{u_j},
	-@{.}[r(2.4)]{}="a"),
:"a"{}[u(.2)r(.5)]*{\Pi_B u_j}
)[u(.5)r(.2)]*-{\theta_j}
}
\caption{Effective spectral gap lemma in two dimensions.}
\label{fig:effective}
\end{wrapfigure}
\noindent \em Proof.\em\;\; This is an easy consequence of the Spectral \reflem{szegedy}.  
We use the notations from the lemma in the proof.
Let us assume $\delta<2$, otherwise the statement is trivial.  Consider the decomposition of the space into the eigenspaces of $U$.  On the $1$-eigenspace of $U$, the vectors from $\ker (\Pi_A) = \im (A)^\perp$ are vanished by $\Pi_B$.  On the intersection of the $(-1)$-eigenspace of $U$ and $\ker(\Pi_A)$, $\Pi_B$ is the identity and $P_\delta$ is zero.  

It remains to consider the behaviour of $P_\delta\Pi_B$ on the direct sum of $S_j$s.  We do it for each $S_j$ independently.  Let $\Pi_j$ be the orthogonal projector on $S_j$, and $u_j = \Pi_j u$.  Also, denote $a_j = \im(A)\cap S_j$ and $b_j = \im(B)\cap S_j$.  Restricted to $S_j$, the operator $\Pi_B$ projects onto $b_j$.

If $2\theta_j>\delta$, then $P_\delta$ is zero on $S_j$, and $P_\delta\Pi_B u_j = 0$.  Otherwise, $P_\delta$ is the identity on $S_j$.  Consider \reffig{effective}.  We have 
\[
\|\Pi_Bu_j\| = \|u_j\|\sin\theta_j\le \|u_j\|\theta_j.
\]
Hence,
$$
\|P_\delta\Pi_Bu\|^2 \le \sum_{j:\theta_j\le\delta/2} \|\Pi_Bu_j\|^2 \le \sum_{j:\theta_j\le\delta/2} \theta_j^2\|u_j\|^2 \le (\delta/2)^2\|u\|^2.\eqno\qedsymbol
$$

Due to technical reasons, we exclude the case of very small values of the adversary bound.

\begin{clm}
\label{clm:advNonConstant}
For any non-constant function $f$, $\Adv(f)\ge1/2$.
\end{clm}

\pfstart
This follows immediately from~\rf(eqn:adv) and \rf(lem:gamma).
\pfend

\begin{thm}
\label{thm:spanAlgorithm}
For any function $f\colon \{0,1\}^\vars\supseteq \cD\to\{0,1\}$ and any span program $(\cP, w)$ evaluating $f$, there exists a quantum algorithm that calculates $f$ in $O(\wsize(\cP, w))$ queries.
\end{thm}

\pfstart
If the function is constant, the statement is trivial.  Otherwise, by \rf(prp:toCanonical), \rf(thm:spanCanonical) and~\rf(clm:advNonConstant), $\wsize(\cP, w)\ge 1/2$.

Assume, as usually, that the span program $\cP$ has the vector space $\R^d$, the input vectors $\{v_i\}_{i\in \cI}$, and the target vector $\target$.  Let the positive and the negative witness sizes be $W_1$ and $W_0$, respectively, so that $W=\wsize(\cP, w)$ satisfies $W=\sqrt{W_0W_1}$.  We add an additional ``input vector'' $v_0 = \target/\alpha$, where $\alpha=C_1\sqrt{W_1}$ for some constant $C_1>0$ to be specified later.  In the implementation, the vector $v_0$ is treated as always available.  Let $\cI_0 = \cI\cup\{0\}$ where $0$ accounts for $v_0$.

\begin{algorithm}
\caption{Span Program Evaluation}
\label{alg:spanProgram}
\algbegin
\state {{\bf function} SpanProgram({\bf  quprocedure} InputOracle) {\bf :}}
\tab
  \state {{\bf attach} $\cI_0$-qudit $\qreg X$}
	\state {{\bf return not} PhaseDetection(StepOfWalk, $1/(C_2 W)$, $\qreg{X}$) }
\state {\ }
\state {{\bf quprocedure} StepOfWalk {\bf:}}
\tab
	\state {ReflectionAbout $\Lambda$ \label{line:spanProgram:Lambda} }
	\state {$\phase \qgets$ IsFalseInputVector($\qreg X$) \label{line:spanProgram:Pi} }
\untab
\state {\ }
\state {{\bf qufuncion} IsFalseInputVector({\bf registers} $\qreg X$, $\qreg O$) {\bf:}}
\tab
  \state {{\bf attach} $\vars$-qudit $\qreg{j}$, qubits $\qreg{b}$, $\qreg{v}$}
  \state {{\bf conditioned on} $\qreg{X}\ne 0$ \bf:}
  \tab
	  \state {$(\qreg j, \qreg b)\qgets$ the variable index and the value of input vector with index $\qreg X$}
	  \state {$\qreg v\qgets$ InputOracle($\qreg j$) }
	  \state {{\bf conditioned on} $\qreg{v}\ne \qreg{b}$ {\bf:} $\qreg{O}\qgets 1$ }
	\untab
\untab
\untab
\algend
\end{algorithm}

\mycommand{tV}{\widetilde V}
The description of the algorithm is given in \rf(alg:spanProgram). 
Let $\tV$ be the $d\times \cI_0$ matrix containing the input vectors of $\cP$ and $v_0$ as columns.  The algorithm runs in the vector space $\R^{\cI_0}$ with the standard basis $\{e_i\}$.  Let $\Lambda$ be the projection onto the kernel of $\tV$. For any input $z\in\cD$, let $\Pi_z=\sum e_ie_i^*$ where $i$ ranges over the indices of the available input vectors for the input $z$ (this includes $e_0$).  Denote by $R_\Pi=2\Pi_z-I$ and $R_\Lambda=2\Lambda-I$ the reflections about the images of $\Pi_z$ and $\Lambda$, respectively.   Lines~\ref{line:spanProgram:Lambda} and~\ref{line:spanProgram:Pi} of the algorithm perform reflections $R_\Lambda$ and $R_\Pi$, respectively.  The first reflection uses no query, thus, we do not describe it.

The step of the quantum walk is $U=R_\Pi R_\Lambda$.  The algorithm executes the phase detection subroutine (\rf(thm:detection)) on $U$ with the initial state $e_0$ and $\delta=1/(C_2 W)$ and accepts iff the phase 0 is detected.  Here $C_2$ is another constant to be specified.  Phase detection requires $O(W)$ (controlled) applications of $U$.  In each of them, $R_\Lambda$ requires no access to the input oracle, whereas $R_\Pi$ can be implemented in two oracle queries.  This proves the complexity estimate of the algorithm.

Let us prove the correctness.  
First assume the input $x$ is positive: $f(x)=1$.  In this case, we show that there is an 1-eigenvector $u$ of $U$ having a large overlap (inner product) with $e_0$.  Recall that $w_x\in\R^{\cI}$ is the witness for $x$, and denote 
\(
u = \alpha e_0 - w_x.
\)
Firstly,
\[
\tV u = \alpha v_0 - Vw_x = \target-\target=0,
\]
where $V$ is as in \rf(sec:spanPrograms).
Hence, $R_\Lambda u = u$.  
Next, $R_\Pi u = u$, because $w_x$ only uses the available input vectors. 
Thus, $u$ is a 1-eigenvector of $U$.  By definition,
\(
\norm|w_x|^2 \le W_1 = \alpha^2/C_1,
\)
hence $\ip<u/\|u\|, e_0>$ is at least a constant that can be tuned by adjusting the value of $C_1$.

Now assume the input $y$ is negative: $f(y)=0$. Let $P_\delta$ be the projection onto the span of the eigenvectors of $U$ with eigenvalues $\ee^{\ii\theta}$ satisfying $|\theta|\le\delta$. We have to prove that $\|P_\delta e_0\|$ is small.  The idea is to apply \reflem{effective}.  In this case, $w'_y$ is the witness for $y$. Define 
\[
u=\alpha \tV^*w'_y = e_0 + \alpha V^* w'_y.
\]
Since $u\in \im(\tV^*)$, we have $\Lambda u=0$. Also, $w'_y$ is orthogonal to all the available input vectors, hence, $\Pi_yu=e_0$. By \reflem{effective},
\[ 
\|P_\delta e_0 \|\le \frac{\delta}{2}\|u\| \le \frac{\sqrt{1+\alpha^2 W_0}}{2C_2 W} \le \frac{1+C_1 \sqrt{W_1W_0}}{2C_2W}\le \frac{2+C_1}{2C_2}.
\]
(In the last step, we used that $W\ge 1/2$.)  By adjusting the value of $C_2$, one may tune the acceptance probability of the algorithm.
\pfend

One may convert a feasible solution to the dual adversary SDP for a Boolean function into a quantum query algorithm by first translating it into a span program using \rf(thm:spanCanonical).  This does not work for non-Boolean functions, but for them, there exists a similar procedure.  We start with a small lemma.

\begin{lem}
\label{lem:muAndNu}
Let $q>0$ be an integer.  Then, there exist vectors $\mu_i, \nu_i\in \R^{q}$, where $i\in[q]$, such that $\ip<\mu_i, \nu_j> = 1-\delta_{ij}$ and $\norm|\nu_i|, \norm|\mu_j|\le \sqrt{2}$ for all $i,j\in[q]$. (Here $\delta_{ij}$ stands for the Kronecker delta.)
\end{lem}

\pfstart
Define $\alpha = \sqrt{\frac12 - \frac{\sqrt{q-1}}q}$,
\[
\mu_i = -\alpha \sqrt{\frac{2(q-1)}{q}}\; e_i + \sqrt{\frac{2(1-\alpha^2)}{q}} \sum_{j\ne i} e_j,
\quad\mbox{and}\quad
\nu_i = -\sqrt{\frac{2(q-1)(1-\alpha^2)}{q}}\; e_i + \alpha \sqrt{\frac{2}{q}} \sum_{j\ne i} e_j,
\]
where $e_i$ is the standard basis of $\R^q$.  A straightforward calculation reveals that these vectors satisfy the requirements.
\pfend

Now, we are finally able to prove the second half of \rf(thm:adv) for functions with Boolean output.

\begin{thm}
\label{thm:advAlgorithm}
Let $f\colon [q]^\vars\supseteq \cD\to\{0,1\}$ be a function, and let $(X_j)$ be a feasible solution to the dual adversary SDP~\rf(eqn:advDual) with the value $W$ of the objective function~\rf(eqn:advDualObjective).  Then, there exists a quantum algorithm that calculates $f$ in $O(W)$ queries.
\end{thm}

\pfstart
Again, we may assume that $W\ge 1/2$ due to \rf(clm:advNonConstant).  
Let, as in the proof of \rf(thm:spanCanonical), $\psi_{j,z}\in \R^d$ be such that $X_j\elem[x,y] = \ip<\psi_{j,x},\psi_{j,y}>$ for all $j\in [\vars]$ and $x,y\in\cD$.  The algorithm runs in the vector space $\C\oplus (\C^\vars\otimes \C^d\otimes\C^q)$.  Let $e_0$ be the normalised vector in the first $\C$, and $\{e_j\}_{j\in[\vars]}$ be the orthonormal basis of $\C^\vars$.

The algorithm is the same as in the proof of \rf(thm:spanAlgorithm) but with $\Lambda$ and $\Pi_z$ redefined.
Let $\alpha = C_1\sqrt{W}$, and define the projector $\Lambda$ as onto the orthogonal complement of the span of the vectors
\[
v_y = e_0 + \alpha \sum_{j\in[\vars]} e_j\otimes \psi_{j,z}\otimes \mu_{y_j}
\]
when $y$ ranges over $\preimy$.  The input-dependent projector $\Pi_z$ is onto the orthogonal complement of the image of 
\[
\sum_{j\in[\vars]} (e_je_j^*)\otimes I_d\otimes (\mu_{z_j}\mu_{z_j}^*),
\]
where $\mu$s are as in \rf(lem:muAndNu), and $I_d$ is the identity in $\C^d$.  The reflection about the image of $\Pi_z$ can be implemented in two oracle queries similarly to \rf(alg:spanProgram).  The complexity estimate of the algorithm is like in \rf(thm:spanAlgorithm).

Let us prove the correctness of the algorithm.  If $f(x)=1$, define a 1-eigenvector $u$ of $U$ by
\[
u = \alpha e_0 - \sum_{j\in[\vars]} e_j\otimes \psi_{j,x} \otimes \nu_{x_j}.
\]
Indeed, for each $y\in\preimy$, we have
\[
\ip<u, v_y> = \alpha - \alpha \sum_{j\in[\vars]} \ip<\psi_{j,x}, \psi_{j,y}>\ip<\mu_{x_j}, \nu_{y_j}> = \alpha - \alpha \sum_{j: x_j\ne y_j} X_j\elem[x,y] = 0.
\]
Hence, $R_\Lambda$ does not change $u$.  Also, $R_\Pi$ also does not change $u$ because $\nu_{x_j}$ is orthogonal to $\mu_{x_j}$ for all $j$.  Finally, $u/\|u\|$ has large overlap with $e_0$ that can be tuned using $C_1$.

In the negative case of $f(y)=0$, apply \rf(lem:effective) with $u = v_y$.  Clearly, $u\in\ker(\Lambda)$, and $\Pi_y u = e_0$.  Thus, we get
\[ 
\|P_\delta e_0 \|\le \frac{\Theta}{2}\|u\| \le \frac{\sqrt{1+4\alpha^2 W}}{2C_2 W} \le \frac{2+2C_1}{C_2}. \qedhere
\]
\pfend

\section{Chapter Notes}
The development of lower bounds on quantum algorithms started almost simultaneously with the development of quantum algorithms.  The first lower bound is for the OR function and it is due to Bennett, Bernstein, Brassard and Vazirani~\cite{bennett:strengths}.  Interestingly, this lower bound preceded Grover's discovery of his search algorithm.

The polynomial method was first developed for classical lower bounds.  Minsky and Papert used polynomial representation to gave a very precise characterisation of perceptrons, a special case of neural networks~\cite{minsky:perceptrons}.  The connection between polynomial degree and query complexity is first observed by Nisan and Szegedy~\cite{nisan:pol}.  For quantum query algorithm, this techniques were applied by Beals \etal~\cite{beals:pol}.  In particular, they prove \rf(lem:amplitudes) for Boolean functions (see also~\cite{fortnow:pol}).  The non-Boolean variant is due to Aaronson~\cite{aaronson:collisionLowerOriginal}.  \rf(lem:symmetrize) is due to Minsky and Papert~\cite{minsky:perceptrons}.  

Aaronson was the first to prove a non-trivial lower bound $\Omega(\vars^{1/5})$ on the quantum complexity of the collision problem~\cite{aaronson:collisionLowerOriginal} where $\vars$ is the number of elements.  This was soon improved to the optimal $\Omega(\vars^{1/3})$ by Shi~\cite{shi:collisionLowerOriginal}.  Ref.~\cite{shi:collisionLower} features a merged variant of both papers.  
%Again, this preceeded the matching upper bound by Ambainis~\cite{ambainis:distinctness}.
The results by Aaronson and Shi had a small catch: In order for them to apply, the size $q$ of the alphabet should have been large enough.  For the collision problem, it was $q\ge 3n/2$.  The proof we give is due to Kutin~\cite{kutin:collisionLower} and it does not require this assumption.  Due to the nature of reduction from element distinctness to the collision problem (\rf(cor:distLower)), for the lower bound to apply, the size of the alphabet for element distinctness should be $\Omega(\vars^2)$.  Ambainis~\cite{ambainis:collisionLower} has a general argument showing that, because of the symmetry of the problem, even the case of $q=\vars$ has complexity $\Omega(\vars^{2/3})$.

As already noted, the adversary bound was first defined by Ambainis~\cite{ambainis:adv} in 2000.  Afterwards, a variety of variants of the bound were defined.  In 2004, \v Spalek and Szegedy~\cite{spalek:advEquivalent} mention 7 variants of the adversary bound, including the positive-weighted adversary and even one based on Kolmogorov complexity~\cite{laplante:advKolmogorov}, and prove they all are equivalent.  The spectral formulation of the bound that we use, \refdefn{advPositive}, is due to Barnum \etal~\cite{barnum:advSpectral}.  The original formulation by Ambainis~\cite{ambainis:polVsQCC} and Zhang~\cite{zhang:advPower} uses weight schemes as outlined in \refsec{advPositive}.  Our proof of \refprp{advBasic2Spectral} is based on~\cite{hoyer:lowerSurvey}.

However, due to the limitations we describe in \refsec{advLimitations}, it was known none of them is tight.  In 2006, H{\o}yer \etal~\cite{hoyer:advNegative} proved that the spectral formulation of the bound still yields a lower bound if one allows negative entries: \refthm{adv}.  \rf(prp:advMultipleRows) is from~\cite{spalek:kSumLower}.  Our proof of the first half of \rf(thm:adv) in \rf(sec:advProof) is a fusion of the proofs in~\cite{hoyer:lowerSurvey} and~\cite{hoyer:advNegative} and is significantly shorter than the latter.

The limitations of the positive-weighted adversary from \rf(sec:advLimitations) were known before that.  Our formulation in \rf(thm:advPositiveDual) is similar to the minimax definition in \cite{laplante:advKolmogorov}.  Two different variants of the certificate barrier, \rf(prp:advLimitations) (a) and (b), are due to Szegedy~\cite{szegedy:triangle} and Zhang~\cite{zhang:advPower}, respectively.  The property testing barrier is a folklore result, see, e.g., \cite{hoyer:advNegative}, but a proof has not appeared in print up to our knowledge.

The research reflected in the second part of the chapter started by the algorithms for the NAND tree evaluation.  The first algorithm by Farhi, Goldstone and Gutmann~\cite{farhi:nandTree} was based on continuous-time quantum walks, the discrete-time quantum walk algorithm is due to Ambainis \etal~\cite{ambainis:formulaeEvaluation}.  
While trying to generalise the latter algorithm to formulae with arbitrary gates, Reichardt and \v Spalek developed a quantum algorithm for span program evaluation~\cite{reichardt:formulae}.

Span programs are a linear-algebraic model of computation first introduced by Karchmer and Wigderson in~\cite{karchmer:spanPrograms}.  Initially, they were used over finite fields and applications included, in particular, a log-space analogue of the complexity class inclusion $\mathsf{NP}\subseteq \oplus \mathsf{P}$, and secret sharing schemes.  Recently, span program have been applied in the area of non-interactive zero-knowledge proofs~\cite{gentry:nizkFromSpan, lipmaa:nizkFromSpan}.

With a replaced definition of complexity, span programs can be evaluated on a quantum computer.  Later, Reichardt~\cite{reichardt:spanPrograms} noticed that the complexity of a span program can be expressed as an SDP that is dual to the adversary SDP, thus proving Theorems~\ref{thm:advDual} and~\ref{thm:spanCanonical}.

The precise evaluation of the adversary bound of the threshold function is due to Reichardt~\cite{reichardt:spanPrograms}.  His proof uses span programs and is completely different from ours in \rf(prp:advThresholdExact).  The composition \rf(thm:advCompose) has a long history as well.  First version of the bound (in a slightly weaker form) was obtained for the positive-weighted adversary.  Ambainis proved one direction~\cite{ambainis:polVsQCC}, and the opposite direction was shown by Laplante, Lee and Szegedy~\cite{laplante:classicalFormulaSize}.  When introducing the general adversary bound, H\o yer, Lee and \v Spalek generalised the result of Ambainis, and the picture was finished by Reichardt~\cite{reichardt:spanPrograms, reichardt:advTight}.

The function in \rf(exm:AmbainisFunction) was first constructed by Ambainis (in a slightly different form) in order to prove a separation between quantum query complexity of a function and its degree~\cite{ambainis:polVsQCC}.  The form that we use is due Laplante, Lee and Szegedy~\cite{laplante:classicalFormulaSize}.  The lower bound on the general adversary bound for this function is due to H\o yer, Lee and \v Spalek~\cite{hoyer:advNegative}.  The upper bound did not appear in print up to our knowledge.

There are many more developments concerning the adversary bound we do not cover.  The adversary bound has been generalised to the problems of quantum state preparation~\cite{ambainis:symmetryAssisted} and quantum state conversion~\cite{lee:stateConversion}, and it is tight in both cases.  One may see that the adversary bound does not work well if one is interested in very small success probability.  Based on a paper by Ambainis~\cite{ambainis:newLowerBoundMethod}, \v Spalek introduced the multiplicative adversary method that works well in this regime~\cite{spalek:multiplicative}.  Combining this with the ideas for function composition, Lee and Roland proved strong direct product theorem for quantum query complexity~\cite{lee:strongDirect}.  Magnin and Roland demonstrated a relation between the multiplicative adversary and the polynomial method~\cite{magnin:relationLowerBounds}.

As mentioned earlier, most of the results in \rf(sec:advAlgorithms) go back to Reichardt.  In particular, \rf(thm:advAlgorithm) is from~\cite{reichardt:advTight}.  However, we give a simpler proof based on the Effective Spectral Gap Lemma from~\cite{lee:stateConversion}.  This reference gives an alternative proof of the lemma, not based on \rf(lem:szegedy).  The proof of \rf(thm:spanAlgorithm) is from~\cite{belovs:learningClaws} and the proof of \rf(thm:advAlgorithm) is based on~\cite{lee:stateConversion}.

Let us make some concluding remarks.
Firstly, in our opinion, the name ``adversary bound'' is inappropriate for the lower bound technique we defined in \rf(sec:adv).  In classical settings, by an adversary, one usually understands an {\em active} entity that communicates with the computational device by simulating the input data.  Responses of the adversary depend on the behaviour of the computational device, and its goal is to give the computational device the worst possible data string.  This does not apply for the quantum adversary bound.  One of its main advantages is that it is {\em static}, i.e., it does not depend on the actions of the computational device.  This greatly simplifies the reasoning about this bound.

Another shocking question arises in this concern:  How is it possible that such a simple lower bound is actually tight?
After all, nothing like this is known for randomised query complexity.  There is no concise optimisation problem that gives even a polynomial approximation for the randomised complexity.  No result like \rf(thm:advCompose) is known.  
The only non-trivial iterated function whose randomised query complexity has been evaluated is the NAND function on 2 arguments.  It is known~\cite{snir:nand, saks:nand} that $R(\mathrm{NAND}^d) = \Theta\sB[{\sA[\frac{1+\sqrt{33}}4]^d}] \approx \Theta(1.686^d)$.  (As $D(\mathrm{NAND}^d) = 2^d$, this function is conjectured to provide the largest possible separation between the randomised and deterministic query complexities for total Boolean functions.)
But even the value of the randomised query complexity of $\thres_{2,3}^d$ is still under consideration.  It is only known that it lies between $\Omega(2.55^d)$~\cite{leonardos:maj3} and $O(2.649^d)$~\cite{magniez:maj3}.

The reason that this is possible, in our opinion, lies in reversibility of quantum computation.  Because of this, every quantum query algorithm can be rewritten in the form~\rf(eqn:sequence) with the measurement only at the end of the algorithm.  Even more, it is not hard to transform any algorithm of the form~\rf(eqn:sequence) into an algorithm satisfying $U_1=\cdots = U_T$.  In this case, it becomes more clear that a tight lower bound can be obtained by estimating how much progress the algorithm can obtain by a single query that is performed {\em without any information on the input string}.  A randomised query algorithm, on the other hand, obtains little progress at first, but, as it learns the values of some variables, it can make more deliberate queries and obtain faster progress.
Thus, estimating its progress per query fails to provide a good lower bound.

%% file: _certificates.tex
\mycommand{cert}{{\cal C}}
\remycommand{marked}{M}

\mycommand{certM}{{A_\marked}}
\mycommand{vertices}{n}
\mycommand{prob}{p}

Determining the amount of computational resources required to solve a computational problem is one of the main problems in theoretical computer science.  At the current stage of knowledge, however, this task seems far out of reach for many problems.  In this case, it is possible to analyse the complexity of the problem under some simplifying assumptions.
We have already seen one example: the simplification made by the query model.
But, even the query complexity is too hard to evaluate for some functions.

%In this paper, we investigate a possibility of constructing an even simpler optimization problem under further simplifying assumptions.  Our assumptions are motivated by the class of algorithms based on quantum walks.  A popular framework for the development of such algorithms~\cite{magniez:walkSearch} includes a black-box {\em checking} subroutine that, given the information gathered during the walk, signals if this information is enough to accept the input string.  In many cases, the precise content of the gathered information is not relevant for the implementation of the quantum walk, what matters are the possible locations of these pieces of information.  We formalise this by the following definition. 

In this chapter, we make a further simplifying assumption and consider the framework of {\em certificate structures}.  This notion is partly motivated by the quantum walk algorithms from \rf(chp:walk).  Recall that the amplitude amplification and the MNRS quantum walk frameworks include a black-box checking subroutine (\rf(defn:quantumCheck)) that, given the information gathered during the walk, signals if this information is enough to accept the input string.  In many cases, the precise content of the gathered information is not relevant, what matters are the possible locations of these pieces of information.  
We formalise this by the notion of a certificate structure that can be considered as a more detailed version of the certificate complexity from \rf(sec:query).

Based on this notion, we develop the computational model of a {\em learning graph}.  It only relies on the certificate structure of the function being evaluated.  
We characterise the complexity of the learning graph by an optimisation problem that is significantly simpler than the general adversary SDP but still captures many aspects of the function.  We are able to get a tight solution for the certificate structures corresponding to the $k$-sum and triangle detection problems.

We also show that learning graphs are tight:  A learning graph can be transformed into a quantum query algorithm, and, for any certificate structure, there exists a function that requires that many queries.

This chapter is based on the following papers:

\begin{itemize}
\item[\cite{belovs:learning}]
A.~Belovs.
\newblock Span programs for functions with constant-sized 1-certificates.
\newblock In {\em Proc. of 44th ACM STOC}, pages 77--84, 2012, 1105.4024.

\item[\cite{spalek:kSumLower}]
A.~Belovs and R.~{\v Spalek}.
\newblock Adversary lower bound for the {$k$-sum} problem.
\newblock In {\em Proc. of 4th ACM ITCS}, pages 323--328, 2012, 1206.6528.

\item[\cite{belovs:onThePower}]
A.~Belovs and A.~Rosmanis.
\newblock On the power of non-adaptive learning graphs.
\newblock In {\em Proc. of 28th IEEE Complexity}, pages 44--55, 2013, 1210.3279.
\end{itemize}

In \rf(sec:certDefinition), we define the notion of certificate structure, and give some examples based on the functions we saw in the previous chapters.  
In \rf(sec:mainResults), we formulate the main result of this chapter and give some consequences of it.
In \rf(sec:learning), we define the notion of a learning graph, and prove that it can be converted into a quantum query algorithm for any function having the specified certificate structure.  In \rf(sec:procedureDriven), we give a method for constructing learning graphs for functions with a lot of symmetry.  In \rf(sec:learningApplications), we describe applications for the triangle and the associativity testing problems.  
In \rf(sec:learningDuality), we derive the dual formulation of the learning graph complexity, and give examples for the certificate structures of the $k$-sum and the triangle problems.  
 Finally, in \rf(sec:lower), we show that, for any certificate structure, the dual formulation can be transformed into a lower bound on the quantum query complexity of some function having this certificate structure.

\section{Definition and Examples}
\label{sec:certDefinition}
Recall the definition of a 1-certificate from \rf(sec:queryRelated):  If $f\colon [q]^\vars\supseteq \cD\to \{0,1\}$ is a function, and $x\in\preimx$ is a positive input, then a subset $S\subseteq[\vars]$ is called a 1-certificate iff any $z\in\cD$, that agrees with $x$ on $S$, satisfies $f(z)=1$.
We define the following subset
\begin{equation}
\label{eqn:certOfFunction}
M(f,x) = \{ S\subseteq [\vars]\mid \text{$S$ is a 1-certificate for input $x$ of function $f$} \}.
\end{equation}
As we will see, in many cases, the subsets $M(f,x)$ turn out more important than the function $f$ itself.  Thus, we abstract away from the function by the following definition.

\begin{defn}[Certificate Structure]
A {\em certificate structure} $\cert$ on $\vars$ variables is a collection of non-empty subsets of $2^{[\vars]}$ with each subset closed under taking supersets.  We say that a function $f\colon [q]^\vars\supseteq \cD\to\{0,1\}$ {\em has} certificate structure $\cert$ if, for every $x\in f^{-1}(1)$, one can find $\marked\in\cert$ such that
\begin{equation}
\label{eqn:certStructureCondition}
\forall S\in\marked\; \forall z\in\cD\colon z_S = x_S \Longrightarrow f(z)=1.
\end{equation}
For a fixed $M\in\cert$, the elements of $M$ are usually called {\em marked}.
\end{defn}

\begin{exm}
\label{exm:certTrivial}
The {\em trivial certificate structure} on $\vars$ variables is defined as $\{\{[\vars]\}\}$, i.e., it consists of one subset of $2^{[\vars]}$ made out solely of the set $[\vars]$ itself.  We call it trivial because any function on $\vars$ variables has this certificate structure.
\end{exm}

We usually assume that the elements of a certificate structure $\cert$ form an antichain under the set-theoretical inclusion relation, i.e., there exist no $M, M'\in\cert$ such that $M\subset M'$ and $M\ne M'$.  
This is motivated by the following observation.  Assume $M\subset M'$ are two elements of $\cert$.  If $M'$ satisfies~\rf(eqn:certStructureCondition) for some $x$ and $f$, then $M$ also satisfies it.  Hence, $M'$ can be removed from $\cert$ without affecting the set of functions having $\cert$ as their certificate structure.  In a similar spirit, we say that a certificate structure $\cert$ is {\em more precise} than a certificate structure $\cert'$ if, for all $M'\in \cert'$, there exists $M\in\cert$ such that $M\subseteq M'$.  For example, the trivial certificate structure is the least precise one.

\begin{defn}[Certificate Structure of a Function]
\label{defn:certOfFunction}
Assume $f\colon [q]^\vars\supseteq \cD\to\{0,1\}$ is a function.  {\em The certificate structure of the function $f$} is defined as the set of inclusion-wise minimal elements of $\{M(f,x)\mid x\in\preimx\}$.
\end{defn}

It is not hard to see that the certificate structure of \rf(defn:certOfFunction) is the most precise certificate structure of $f$.

In this chapter, we are interested in quantum algorithms performing equally well for any function with a fixed certificate structure.  
%Some examples of such algorithms are given in \refsec{examples}.  
More formally, consider the following definition:
\begin{defn}[Quantum Complexity]
\label{defn:certQuantumComplexity}
The {\em quantum query complexity} of a certificate structure $\cC$ is defined as the maximum quantum query complexity over all functions having $\cC$ as their certificate structure.
\end{defn}

Many existing quantum algorithms, implicitly or explicitly, work in these settings.  
The most celebrated examples are demonstrated by the Grover search algorithm (\rf(prp:Grover)), and the quantum walk on the Johnson graph.  For instance, \rf(thm:walkKDistinctness) can be reformulated as a quantum query algorithm evaluating any function with the following certificate structure:

\begin{exm}
\label{exm:certKSubset}
The {\em $k$-subset certificate structure} $\cert$ on $\vars$ elements with $k=O(1)$ is defined as follows.  It has ${\vars\choose k}$ elements, and, for each subset $S\subseteq[\vars]$ of size $k$, there exists unique $M\in\cert$ such that $T\in M$ if and only if $S\subseteq T\subseteq [\vars]$.

In particular, the 1-subset certificate structure corresponds to the OR function, and the 2-subset certificate structure---to the element distinctness problem.
We call the 1-subset certificate structure the {\em OR certificate structure}.
\end{exm}

The construction from the previous definition can be generalised in the following way:
\mycommand{cS}{{\cal S}}
\begin{defn}
Assume that $\cS$ is a family of subsets of $[\vars]$.  The certificate structure {\em generated} by $\cS$ consists of the elements $\{T\mid S\subseteq T\subseteq[\vars]\}$ where $S$ runs through $\cS$.
\end{defn}
Thus, the trivial certificate structure is generated by $\{[\vars]\}$, and the $k$-subset certificate structure is generated by the set of $k$-subsets of $[\vars]$.  Not all certificate structures can be constructed in this way, the following being an example:

\begin{exm}
\label{exm:certORMultiple}
The OR certificate structure from \rf(exm:certKSubset) can be generalised to the case when it is promised that each positive input contains at least $k$ ones.  Let $k$ be an integer between 1 and $\vars$.  The certificate structure $\cert$ has ${\vars\choose k}$ elements, and, for each subset $S\subseteq[\vars]$ of size $k$, there exists $M\in\cert$ such that $T\in M$ if and only if $S\cap T$ is non-empty.
\end{exm}

Inspired by \rf(thm:walkTriangle), we define the following certificate structure.

\begin{exm}
\label{exm:certTriangle}
The {\em triangle certificate structure} $\cert$ on $\vertices$ vertices is a certificate structure on $\vars={\vertices\choose 2}$ variables defined as follows.  Assume that the variables are labelled as $z_{ij}$ where $1\le i<j\le \vertices$.  Then, the certificate structure is generated by the set of triples
$\sfigA{ \{ab,bc,ac \} \mid 1\le a<b<c\le \vertices}$.
\end{exm}

The functions in \rf(defn:setEquality) give rise to the following certificate structures:

\begin{exm}
\label{exm:collision}
Each of the following certificate structures is defined on $\vars=2n$ input variables.  In the {\em collision certificate structure}, there is unique $M$ for each decomposition $[\vars]=\{a_1,b_1\}\sqcup\{a_2,b_2\}\sqcup\cdots\sqcup\{a_n,b_n\}$, and $S\in M$ if and only if $S\supseteq\{a_i,b_i\}$ for some $i\in[n]$.  The {\em set equality certificate structure} contains only those $M$ from the collision certificate structure that correspond to decompositions with $1\le a_i \le n$ and $n+1\le b_i\le \vars$ for all $i$.  

The {\em hidden shift certificate structure} $\cert$ has $n$ elements.  For each $d\in[n]$, there exists $M\in\cert$ such that $S\in M$ if and only if $S$ contains elements $i$ and $n+1+((i+d)\bmod n)$ for some $i\in[n]$.
\end{exm}

\rf(fig:certExamples) shows examples of certificate structures from Examples~\ref{exm:certKSubset}, \ref{exm:collision}, and~\ref{exm:certORMultiple}.

\def \hasse#1 {
\xygraph{!~*{\cir<6pt>{}} !{0;<1.8pc,0pc>:} !~:{@[red]@{-}}%
*{\emptyset}{}="0" [l(1.5)u]([]{}="1"(-"0") [r]{}="2"(-"0") [r]{}="3"(-"0") [r]{}="4"(-"0"))%
[ul]([r(.5)]{}="12" [r(.8)] {}="13" [r(.8)] {}="23" [r(.8)] {}="14" [r(.8)] {}="24" [r(.8)] {}="34")%
[ur]([]{}="123" [r] {}="124" [r] {}="134" [r] {}="234")%
[r(1.5)u] {}="1234"%
"12"(-"1"{},-"2"{})%
"13"(-"1"{},-"3"{})%
"14"(-"1"{},-"4"{})%
"23"(-"2"{},-"3"{})%
"24"(-"2"{},-"4"{})%
"34"(-"3"{},-"4"{})%
"123"(-"12"{},-"13"{},-"23"{},-"1234"{})%
"124"(-"12"{},-"14"{},-"24"{},-"1234"{})%
"134"(-"13"{},-"14"{},-"34"{},-"1234"{})%
"234"(-"23"{},-"24"{},-"34"{},-"1234"{})%
#1%
"1"*{1} "2"*{2} "3"*{3} "4"*{4}
}}

\begin{figure}[tbp]%
\release{
\[
\begin{array}{ccc}
\hasse{
"12" *++[o][F*:blue]{}
"123" *++[o][F*:blue]{}
"124" *++[o][F*:blue]{}
"1234" *++[o][F*:blue]{}
} &
\hasse{
"13" *++[o][F*:blue]{}
"123" *++[o][F*:blue]{}
"134" *++[o][F*:blue]{}
"1234" *++[o][F*:blue]{}
} &
\hasse{
"23" *++[o][F*:blue]{}
"123" *++[o][F*:blue]{}
"234" *++[o][F*:blue]{}
"1234" *++[o][F*:blue]{}
} \\
\strut \\
\hasse{
"14" *++[o][F*:blue]{}
"124" *++[o][F*:blue]{}
"134" *++[o][F*:blue]{}
"1234" *++[o][F*:blue]{}
} &
\hasse{
"24" *++[o][F*:blue]{}
"124" *++[o][F*:blue]{}
"234" *++[o][F*:blue]{}
"1234" *++[o][F*:blue]{}
} &
\hasse{
"34" *++[o][F*:blue]{}
"134" *++[o][F*:blue]{}
"234" *++[o][F*:blue]{}
"1234" *++[o][F*:blue]{}
} \\
\end{array}
\]
\[
\text{(a) 2-subset certificate structure}
\]
\vspace{.5cm}
\[
\begin{array}{ccc}
\hasse{
"12" *++[o][F*:blue]{}
"34" *++[o][F*:blue]{}
"123" *++[o][F*:blue]{}
"124" *++[o][F*:blue]{}
"134" *++[o][F*:blue]{}
"234" *++[o][F*:blue]{}
"1234" *++[o][F*:blue]{}
} &
\hasse{
"13" *++[o][F*:blue]{}
"24" *++[o][F*:blue]{}
"123" *++[o][F*:blue]{}
"124" *++[o][F*:blue]{}
"134" *++[o][F*:blue]{}
"234" *++[o][F*:blue]{}
"1234" *++[o][F*:blue]{}
} &
\hasse{
"14" *++[o][F*:blue]{}
"23" *++[o][F*:blue]{}
"123" *++[o][F*:blue]{}
"124" *++[o][F*:blue]{}
"134" *++[o][F*:blue]{}
"234" *++[o][F*:blue]{}
"1234" *++[o][F*:blue]{}
} \\
\end{array}
\]
\[
\text{(b) collision certificate structure}
\]
\vspace{.5cm}
\[
\begin{array}{ccc}
\hasse{
"1" *++[o][F*:blue]{}
"2" *++[o][F*:blue]{}
"12" *++[o][F*:blue]{}
"13" *++[o][F*:blue]{}
"14" *++[o][F*:blue]{}
"23" *++[o][F*:blue]{}
"24" *++[o][F*:blue]{}
"123" *++[o][F*:blue]{}
"124" *++[o][F*:blue]{}
"134" *++[o][F*:blue]{}
"234" *++[o][F*:blue]{}
"1234" *++[o][F*:blue]{}
} &
\hasse{
"1" *++[o][F*:blue]{}
"3" *++[o][F*:blue]{}
"12" *++[o][F*:blue]{}
"13" *++[o][F*:blue]{}
"14" *++[o][F*:blue]{}
"23" *++[o][F*:blue]{}
"34" *++[o][F*:blue]{}
"123" *++[o][F*:blue]{}
"124" *++[o][F*:blue]{}
"134" *++[o][F*:blue]{}
"234" *++[o][F*:blue]{}
"1234" *++[o][F*:blue]{}
} &
\hasse{
"1" *++[o][F*:blue]{}
"4" *++[o][F*:blue]{}
"12" *++[o][F*:blue]{}
"13" *++[o][F*:blue]{}
"14" *++[o][F*:blue]{}
"24" *++[o][F*:blue]{}
"34" *++[o][F*:blue]{}
"123" *++[o][F*:blue]{}
"124" *++[o][F*:blue]{}
"134" *++[o][F*:blue]{}
"234" *++[o][F*:blue]{}
"1234" *++[o][F*:blue]{}
} \\
\strut \\
\hasse{
"2" *++[o][F*:blue]{}
"3" *++[o][F*:blue]{}
"12" *++[o][F*:blue]{}
"13" *++[o][F*:blue]{}
"23" *++[o][F*:blue]{}
"24" *++[o][F*:blue]{}
"34" *++[o][F*:blue]{}
"123" *++[o][F*:blue]{}
"124" *++[o][F*:blue]{}
"134" *++[o][F*:blue]{}
"234" *++[o][F*:blue]{}
"1234" *++[o][F*:blue]{}
} &
\hasse{
"2" *++[o][F*:blue]{}
"4" *++[o][F*:blue]{}
"12" *++[o][F*:blue]{}
"14" *++[o][F*:blue]{}
"23" *++[o][F*:blue]{}
"24" *++[o][F*:blue]{}
"34" *++[o][F*:blue]{}
"123" *++[o][F*:blue]{}
"124" *++[o][F*:blue]{}
"134" *++[o][F*:blue]{}
"234" *++[o][F*:blue]{}
"1234" *++[o][F*:blue]{}
} &
\hasse{
"3" *++[o][F*:blue]{}
"4" *++[o][F*:blue]{}
"13" *++[o][F*:blue]{}
"14" *++[o][F*:blue]{}
"23" *++[o][F*:blue]{}
"24" *++[o][F*:blue]{}
"34" *++[o][F*:blue]{}
"123" *++[o][F*:blue]{}
"124" *++[o][F*:blue]{}
"134" *++[o][F*:blue]{}
"234" *++[o][F*:blue]{}
"1234" *++[o][F*:blue]{}
} \\
\end{array}
\]
\[
\text{(c) certificate structure from \rf(exm:certORMultiple) for $k=2$}
\]
}
\caption{Examples of certificate structures for $\vars=4$.  The elements of certificate structures are depicted on the Hasse diagram of $2^{[4]}$.}
\label{fig:certExamples}
\end{figure}

\section{Main Results}
\label{sec:mainResults}
Let $\cC$ be a certificate structure on $\vars$ elements.  In \rf(defn:certQuantumComplexity), we defined the quantum query complexity of $\cC$.  Later, in \rf(sec:learning), we define {\em learning graphs} that is a computational model depending on certificate structures by definition.  In particular, in \rf(defn:certLearningComplexity), we define the {\em learning graph complexity} of $\cC$ as the smallest possible complexity of a learning graph for $\cC$.

The main result of this chapter is as follows:

\begin{thm}
\label{thm:certificates}
For any certificate structure, its quantum query and learning graph complexities differ by at most a constant multiplicative factor.
\end{thm}

Thus, on one hand, given a function $f$ possessing a certificate structure $\cC$, one can obtain a quantum query algorithm for $f$ by constructing a learning graph for $\cC$.  In many cases, this gives a decent algorithm, but, of course, it need not be optimal.  But, on the other hand, for any certificate structure $\cC$, one can construct a function $f$ that has $\cC$ as its certificate structure and requires this number of queries.  Thus, if one wants to perform better, he must use other properties of the function besides the possible dispositions of its certificates.

The statement of \rf(thm:certificates) breaks into two halves.  The first half is proven in \rf(sec:learningProof), and the second half is proven in \rf(sec:lower).

Although \refthm{certificates} is a very general result, it is unsatisfactory in the sense that the function having the required quantum query complexity is rather artificial, and the size of the alphabet is astronomical.  However, for a special case of certificates structures we are about to define, it is possible to construct a relatively natural problem with a modestly-sized alphabet having high quantum query complexity.

\begin{defn}[Boundedly-Generated Certificate Structure]
\label{defn:boundedly}
We say that a certificate structure $\cert$ on $\vars$ variables is {\em boundedly-generated} if it is generated by a subset $\cS\subseteq 2^{[\vars]}$ satisfying $|S|=O(1)$ for all $S\in\cS$.
\end{defn}

For example, the $k$-subset and the triangle certificate structures from Examples~\ref{exm:certKSubset} and~\ref{exm:certTriangle} are boundedly-generated, while the certificate structures from Examples~\ref{exm:certORMultiple} and~\ref{exm:collision} are not.  The trivial certificate structure from \rf(exm:certTrivial) is {\em not} boundedly-generated as well, because $\vars\ne O(1)$.

%In order to define the function with a high complexity, we first have to introduce the following special case of a well-studied combinatorial object. 

\begin{defn}[Orthogonal Array]\label{defn:orthogonalArray}
Assume $T$ is a subset of $[q]^k$.  We say that $T$ is an \emph{orthogonal array} over alphabet $[q]$ iff, for every index $i \in [k]$ and for every sequence $x_1,\dots, x_{i-1}, x_{i+1},\dots, x_k$ of elements in $[q]$, there exist exactly $|T|/q^{k-1}$ choices of $x_i \in [q]$ such that $(x_1, \dots, x_k) \in T$.  We call $|T|$ the {\em size} of the array, and $k$---its {\em length}.
(Compared to a standard definition of orthogonal arrays (cf.~\cite{hedayat:orthogonal}), we always require that the so-called strength of the array equals $k-1$.)
\end{defn}

\begin{thm}
\label{thm:boundedly}
Assume a certificate structure $\cert$ is boundedly-generated, and let $\certM$ be like in \refdefn{boundedly}.  Assume the alphabet is $[q]$ for some $q\ge 2|\cert|$, and each $\certM$ is equipped with an orthogonal array $T_\marked$ over alphabet $[q]$ of length $|\certM|$ and size $q^{|\certM|-1}$.  Consider a function $f\colon [q]^\vars\to\{0,1\}$ defined by $f(x)=1$ iff there exists $\marked\in\cert$ such that $x_{\certM} \in T_\marked$.  Then, the quantum query complexity of $f$ is at least a constant times the learning graph complexity of $\cert$.
\end{thm}

For example, for a boundedly-generated certificate structure $\cert$, one can define the corresponding {\em sum} problem: Given $z\in[q]^\vars$, detect whether there exists $M\in\cert$ such that $\sum_{j\in\certM} z_j\equiv 0\pmod{q}$.  If $q\ge 2|\cert|$, \refthm{boundedly} implies that the quantum query complexity of this problem is at least a constant times the learning graph complexity of $\cert$.

\rf(thm:boundedly) is proven in \rf(sec:lower).

We apply Theorems~\ref{thm:certificates} and~\ref{thm:boundedly} to a number of functions.  For example, in
Propositions~\ref{prp:learningKSubsetUpper} and~\ref{prp:learningKSubsetLower}, we prove that the the learning graph complexity of the $k$-subset certificate  structure on $\vars$ variables is $O(\vertices^{k/(k+1)})$.  By combining this with \rf(thm:boundedly), we obtain the following results:

\begin{cor}
\label{cor:certificateDistLower}
The quantum query complexity of the element distinctness problem from \rf(defn:kdist), provided that the size of the alphabet $q>\vars^2$, is $\Theta(\vars^{2/3})$.
\end{cor}

\begin{cor}
\label{cor:ksumLower}
The quantum query complexity of the $k$-sum problem from \rf(defn:kdist) is $\Theta(\vars^{k/(k+1)})$ provided that the size of the alphabet $q>\vars^k$.
\end{cor}

The first of these corollaries reproves \rf(cor:distLower) using the adversary method.  Note that we prove this result directly, and not via the collision problem as it was done in \rf(sec:collisionLower).  It is still an open problem to reprove \rf(thm:collisionLower) using the adversary method.  The result in \rf(cor:ksumLower) is new, and it resolves the conjecture posed by Childs and Eisenberg~\cite{childs:subsetFinding} that the algorithm in \rf(thm:walkKDistinctness) is tight for the $k$-sum problem.  Also, by combining Propositions~\ref{prp:learningTriangleUpper} and~\ref{prp:learningKSubsetLower}, we get the following result:

\begin{cor}
Provided that the size of the alphabet $q>\vertices^3$, the quantum query complexity of the triangle-sum problem on $\vertices$ vertices is $\tilde \Theta(\vertices^{9/7})$.  
\end{cor}

Here, the triangle sum problem is defined in the obvious way.
Finally, note that all these results supersede the certificate complexity barrier, \rf(prp:advLimitations)(b), hence, the constructed adversary matrices use negative weights.

\section{Learning Graphs}
\mycommand{CN}{{\cal N}}
\mycommand{CP}{{\cal P}}
\mycommand{CT}{{\cal T}}
\label{sec:learning}
In this section, we define the computational model of a learning graph.  It is based on the notion of certificate structure from the previous section, and can be converted into a quantum query algorithm for any function having this certificate structure.

\begin{defn}[Learning Graph]
\label{defn:learning}
A {\em learning graph} ${\cal G}$ on $\vars$ input variables is a directed acyclic connected graph with vertices labelled by subsets of $[\vars]$, the input indices. It only has arcs connecting vertices labelled by $S$ and $S\cup\{j\}$ where 
$S\subset [\vars]$ and $j\in[\vars]\setminus S$. The root of ${\cal G}$ is the vertex labelled by the empty set $\emptyset$. Each arc $e$ is assigned a positive real {\em weight} $w_e$.
\end{defn}

Note that it is allowed to have several (or none) vertices labelled by the same subset $S\subseteq[\vars]$. If there is a unique vertex of $\cG$ labelled by $S$, we usually use $S$ to denote it. Otherwise, we denote the vertex by $(S,a)$ where $a$ is some additional parameter used to distinguish vertices labelled by the same subset $S$.

A learning graph can be thought of as a way of modelling the development of one's knowledge about the input during a query algorithm. Initially, nothing is known, and this is represented by the root labelled by $\emptyset$.  At a vertex labelled by $S\subseteq [\vars]$, the values of the variables in $S$ have been learned. Following an arc $e$ connecting vertices labelled by $S$ and $S\cup\{j\}$ can be interpreted as querying the value of the input variable $z_j$. We say the arc {\em loads} element $j$. When talking about a vertex labelled by $S$, we call $S$ the set of {\em loaded elements}.

The graph ${\cal G}$ itself has a very loose connection to the function being calculated. The following notion is the essence of the construction.

\begin{defn}[Flow]
\label{defn:flow}
Let ${\cal G}$ be a learning graph and $\cert$ be a certificate structure, both on $\vars$ input variables.  For each $M\in\cert$, we define a {\em flow} on ${\cal G}$ as a real-valued function $p_e = p_e(M)$ where $e$ is an arc of ${\cal G}$.  It has to satisfy the following properties:
\itemstart
\item the vertex $\emptyset$ is the only source of the flow, and the flow has value 1. In other words, the sum of $p_e$ over all $e$ leaving $\emptyset$ is 1;
\item a vertex labelled by $S$ is a sink only if $S\in M$.  Thus, if $S\ne\emptyset$ and $S$ is not marked, then, for a vertex labelled by $S$, the sum of $p_e$ over all in-coming arcs equals the sum of $p_e$ over all out-going arcs.
\itemend
\end{defn}

We always assume a learning graph ${\cal G}$ is equipped with a certificate structure $\cert$ and a flow $p$ that satisfy the constraints of \refdefn{flow}.  In this case, we say the learning graph $\cG$ is {\em for} the certificate structure $\cert$.  Define the {\em negative complexity} of ${\cal G}$ and the {\em positive complexity of $\cG$ for $M\in\cert$} as
\begin{equation}
\label{eqn:learningComp1}
\CN({\cal G}) = \sum_{e\in E} w_e\qquad\mbox{and}\qquad \CP({\cal G}, M) = \sum_{e\in E} \frac{p_e(M)^2}{w_e},
\end{equation}
respectively, where $E$ is the set of arcs of ${\cal G}$.  The {\em positive complexity} and the {\em (total) complexity} of ${\cal G}$ are defined as 
\begin{equation}
\label{eqn:learningComp2}
\CP({\cal G}) = \max_{x\in f^{-1}(1)} \CP(\cG, x)\qquad\mbox{and}\qquad \CT({\cal G}) = \max\{\CN({\cal G}), \CP({\cal G})\},
\end{equation}
respectively.

\begin{rem}
\label{rem:learningEquivalentComplexity}
Similarly to \rf(rem:SpanGeometricalMean), one can define the total complexity of the learning graph as $\sqrt{\CN(\cG)\CP(\cG)}$, or as $\sqrt{\CN(\cG)}$ subject to $\CP(\cG)\le 1$.  This can be achieved by multiplying the weights of all arcs by the same factor $\alpha$.  This operation simultaneously increases $\CN(\cG)$ $\alpha$ times, and decreases $\CP(\cG)$ by the same factor.  By choosing $\alpha$ appropriately, it is possible to convert one definition into another.
\end{rem}

If a certificate structure $\cert'$ is less precise than $\cert$, then any learning graph for $\cert'$ can be also considered as a learning graph for $\cert$.  We also say that a learning graph is {\em for} a function $f$ if it is for its certificate structure as in \rf(defn:certOfFunction).  Then, we write $p_e(x)$ instead of $p_e(M(f,x))$ where $M(f,x)$ is as in \rf(eqn:certOfFunction).

\begin{defn}[Learning Graph Complexity]
\label{defn:certLearningComplexity}
The {\em learning graph complexity} of a certificate structure $\cert$ is defined as the smallest possible complexity of a learning graph for $\cert$.
\end{defn}

\subsection{Examples}

Now we give a number of examples of learning graphs that 
replicate the quantum algorithms we saw in \rf(chp:walk).

\begin{exm}
\label{exm:learningTrivial}
We start with the trivial certificate structure $\cert$ from \rf(exm:certTrivial).  The corresponding learning graph is shown in \rf(fig:trivial)(a).  In order to save space, we only show the variables loaded by the arcs.  Thus, the vertices of the learning graph from left to right are $\emptyset, \{1\}, \{1,2\}, \{1,2,3\},\dots, [\vars]$.

There are $\vars$ arcs, each of weight 1, hence, the negative complexity $\CN(\cG)=\vars$.  There is only one choice for $M$ in $\cert$, where only the subset $[\vars]$ is marked.  We define the flow from $\emptyset$ to $[\vars]$ by setting $p_e=1$ for all arcs $e$ in the learning graph.  Thus, the positive complexity $\CP(\cG)$ also equals $\vars$, and the total complexity is $\vars$.  This corresponds to the fact that any function can be evaluated in $\vars$ queries.  In figures, we often replace paths as in (a) by ``super arcs'' as in (b) that we call {\em transitions}.
\end{exm}

\begin{figure}[htb]
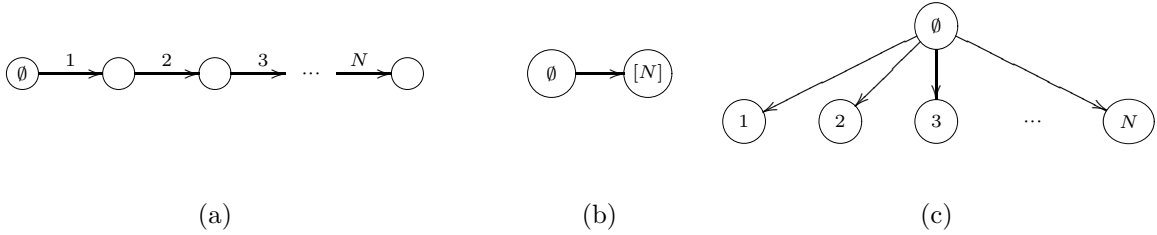
%
\[
\def\objectstyle{\scriptstyle}
\xygraph{ !{0;<3pc,0pc>:} !~-{@{->}}
*{\emptyset} *++[o][F]{} -^{1}[r] *++[o][F]{} -^{2}[r] *++[o][F]{} ([d(1.5)] *{\mbox{(a)}}) -^{3}[r] *++{\cdots} -^\vars[r] *++[o][F]{}
[r(1.5)]
*{\emptyset} *+++[o][F]{} -[r] *+++[o][F]{} ([]*{[\vars]})
( [d(1.5)l(.5)] *{\mbox{(b)}})
[u(.5)r(3)]
*++[o][F]{\emptyset}="0"
[l(2)d] *++[o][F]{1} & *++[o][F]{2} & *++[o][F]{3} ([d]*{\mbox{(c)}}) & \cdots  & *++[o][F]{\vars}
"0" (-"1",-"2",-"3",-"\vars")
}
\]
\caption{(a) A learning graph for the trivial certificate structure on $\vars$ variables.  (b) A shorthand for the path in (a).  (c) A learning graph for the 1-subset certificate structure on $\vars$ variables.  The weights of arcs in (a) and (c) are all equal to 1.}
\label{fig:trivial}
\end{figure}

\begin{exm}
As another example, consider the OR certificate structure from \rf(exm:certKSubset).  The learning graph $\cG$ can be found in \rf(fig:trivial)(c) where the weight of each arc is 1.  It corresponds to the Grover algorithm (\rf(prp:Grover)).  The negative complexity $\CN(\cG) = \vars$.  By definition, for each $M$ in the certificate structure, there exists $j\in[\vars]$ such that the singleton $\{j\}$ is marked.  Assign the flow 1 on the arc connecting $\emptyset$ and $\{j\}$, and set the zero flow on the remaining arcs.  Thus, the positive complexity $\CN(\cG)=1$.  By \rf(rem:learningEquivalentComplexity), the total complexity of the learning graph is $O(\sqrt{\vars})$.

This can be extended to the case of the certificate structure from \rf(exm:certORMultiple).  The learning graph remains the same.  For each $M$ in the certificate structure, there exists a $k$-subset $A$ of $[\vars]$ such that the subset $\{j\}$ is marked for all $j\in A$.  Define the flow as follows
\[
p_e(M) =
\begin{cases}
1/k,& \text{$e$ connects $\emptyset$ and $\{j\}\in M$;} \\
0,& \text{otherwise.}
\end{cases}
\]
The negative complexity is still $\vars$.  For the positive complexity, we have $\CP(\cG) = k(1/k)^2 = 1/k$.  Hence, the total complexity is $O(\sqrt{\vars/k})$ in accord with \rf(prp:Grover).
\end{exm}

The previous examples suggest that, in the learning graph, sequential loading of variables corresponds to a path, and amplitude amplification corresponds to branching.  In the following examples, we explore this intuition.  The first one corresponds to \rf(prp:collision).

\begin{prp}
\label{prp:learningCollisionUpper}
The learning graph complexity of the collision certificate structure is $O(\vars^{1/3})$.
\end{prp}

\begin{wrapfigure}{L}{0pt}
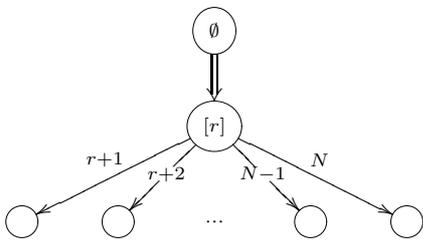

\def\objectstyle{\scriptstyle}
\xygraph{ !{0;<3pc,0pc>:} !~-{@{<-}} !~:{@{=>}}
*++[o][F]{\emptyset}
:[d]*++[o][F]{[r]}="0"
[l(2)d] *++[o][F]{}(-^{r+1}"0") & *++[o][F]{}(-|{r+2}"0") &  \cdots  
 & *++[o][F]{}(-|{\vars-1}"0")  & *++[o][F]{}(-_{\vars}"0")
}
\caption{Learning graph for the collision certificate structure.}
\label{fig:learningCollision}
\end{wrapfigure}
\noindent \em Proof.\em\;\;
The learning graph $\cG$ is shown in \rf(fig:learningCollision).  At first, there is transition from $\emptyset$ to the subset $[r]$ where $r=o(\vars)$ is some integer specified later.  Each arc in the transition has some weight $w$ we will also specify later.  Next, there are $\vars-r$ arcs that connect $[r]$ to all its $(r+1)$-supersets: $[r+1], \{1,2,\dots,r,r+2\},\dots,\{1,2,\dots,r,\vars\}$.  The weight of each of these arcs is 1.
The negative complexity of the learning graph is 
\(
\CN(\cG) = rw + \vars-r = O(\vars)
\) 
if we set $w=\vars/r$.

Now, let $M$ be an element of the collision certificate structure, and define the flow for $M$ as follows.  First, set flow 1 on the transition.  If $[r]\in M$, we are done.  Otherwise, there exist distinct elements $a_1,\dots,a_r$ such that $[r]\cup\{a_i\}\in M$ for all $i\in[r]$.  Define the flow $1/r$ on each arc loading an $a_i$, and zero elsewhere.  Thus, the positive complexity is 
\[
\CP(\cG) \le r/w + 1/r = r^2/\vars + 1/r = O(\vars^{-1/3})
\]
if we set $r = \vars^{1/3}$.  Thus, the total complexity of the learning graph is $\sqrt{\CN(\cG)\CP(\cG)} = O(\vars^{1/3})$. 
\hspace{\stretch{1}} \qedsymbol

\begin{exm}
Now, we consider \rf(prp:distinctnessOld).
Let $\cert$ be the 2-subset certificate structure and the learning graph $\cG$ be as in \rf(fig:learningOldDistinctness).  It consists of two layers.  In the first one, the root $\emptyset$ is connected to ${\vars\choose r}$ vertices: one for each $r$-subset of $[\vars]$.  Here we use transitions like in \rf(fig:trivial)(b).  In the second one, each $r$-subset is connected to all its $(r+1)$-supersets.  The arcs in the first layer have some weight $w$, and the arcs in the second one have weight 1.

\begin{figure}[htb]
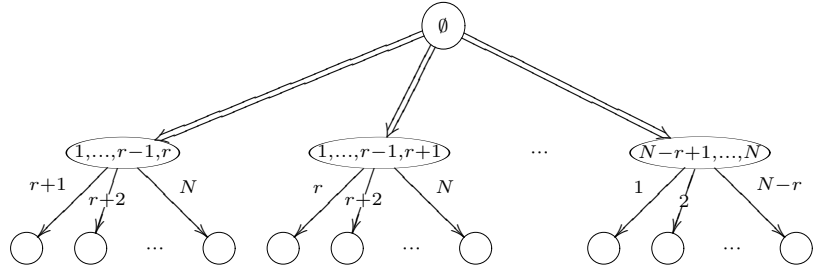
%
\[
\def\objectstyle{\scriptstyle}
\xygraph{ !{0;<2pc,0pc>:} !~-{@{->}} !~:{@{=>}}
*++[o][F]{\emptyset}="0" [ddlllll]
*+[o][F]{1,\dots,r-1,r}="1" [rrrr] *+[o][F]{1,\dots,r-1,r+1}="2" [r(2.5)] *{\cdots} [r(2.5)] *+[o][F]{\vars-r+1,\dots,\vars}="3"
"0"(:"1",:"2",:"3")
"1" 
( [ld(1.5)] *++[o][F]{}="11" [r] *++[o][F]{}="12" [r] *{\cdots} [r] *++[o][F]{}="13" ,
-"11"_{r+1}, -"12"|{r+2},  -"13"^{\vars})
"2" 
( [ld(1.5)] *++[o][F]{}="11" [r] *++[o][F]{}="12" [r] *{\cdots} [r] *++[o][F]{}="13" ,
-"11"_{r}, -"12"|{r+2},  -"13"^{\vars})
"3" 
( [ld(1.5)] *++[o][F]{}="11" [r] *++[o][F]{}="12" [r] *{\cdots} [r] *++[o][F]{}="13" ,
-"11"_{1}, -"12"|{2},  -"13"^{\vars-r})
 }
\]
\caption{A suboptimal learning graph for the 2-subset certificate structure.}
\label{fig:learningOldDistinctness}
\end{figure}

The negative complexity is $\CN(\cG) = {\vars\choose r}[rw+(\vars-r)] = O\s[\vars{\vars\choose r}]$ if we choose $w = \vars/r$.

Let us define the flow.  Assume $M\in\cert$ is defined by a 2-subset $\{a,b\}$, i.e., $S\in M$ iff $\{a,b\}\subseteq S$.  Let $V$ consist of all $r$-subsets of $[\vars]$ that include $a$, but not $b$.  Define the flow equal to ${\vars-2\choose r-1}^{-1}$ for all transitions that connect $\emptyset$ to an element $S\in V$ and for all arcs that connect an element $S\in V$ to $S\cup\{b\}$.  Everywhere else, the flow is zero.
The positive complexity is
\[
\CP(\cG) = {\vars-2\choose r-1} \cdot {\vars-2\choose r-1}^{-2} (r/w + 1) = (1+r^2/\vars){\vars-2\choose r-1}^{-1}.
\]
And the total complexity is
\[
\sqrt{\vars\s[1+\frac{r^2}\vars] {\vars\choose r} {\vars-2\choose r-1}^{-1} } =  
O\s[\sqrt{ \vars\s[1+\frac{r^2}\vars] \frac \vars r }] =
O(\vars^{3/4})
\]
if we set $r = \sqrt{\vars}$.
\end{exm}

The learning graphs in the previous examples directly follow the corresponding algorithms.  The analysis is similar with the exception of the weights.  In \rf(sec:procedureDriven), we will see a way how the weights can be calculated automatically.

\subsection{Proof of the First Half of \TeXBug{Theorem \ref*{thm:certificates}}}
\label{sec:learningProof}
In this section, we prove that for any function $f$ having certificate structure $\cC$ and for any learning graph $\cG$ for $\cG$, there exists a quantum algorithm evaluating $f$ that has query complexity $O(\CT(\cG))$.
We give two proofs of this result.  The first one is based on span programs and works only for Boolean functions.  The second proof is based on the dual of the adversary bound and is applicable for arbitrary functions.  Ideas from both of these proofs will be used later in the thesis: in Chapters~\ref{chp:claws} and~\ref{chp:kdist}, respectively.
But before that, we prove a result on the form of a learning graph.

\begin{prp}
\label{prp:learningNoDuplicates}
Assume $\cG$ is a learning graph for a certificate structure $\cert$.  Then, there exists a learning graph $\cG'$ for $\cert$ such that $\CT(\cG')\le \CT(\cG)$ and $\cG'$ has at most one vertex for any subset $S\subseteq[\vars]$.
\end{prp}

\pfstart
The learning graph $\cG'$ has a vertex $S\subseteq[\vars]$ iff $\cG$ has at least one vertex corresponding to this subset.  
For each pair $S,S'$ of vertices of $\cG'$ such that $S' = S\cup\{j\}$ for some $j\notin S$, perform the following transformation:
Let $e_1,\dots,e_k$ be all the arcs in $\cG$ connecting any vertex with label $S$ to a vertex with label $S'$.  Let $w_i$ be the weight of $e_i$.  Connect $S$ and $S'$ in $\cG'$ by an arc $e$ of weight $w_1+\cdots+w_k$.  Clearly, $\CN(\cG') = \CN(\cG)$.  Let $p$ be a flow in $\cG$ for some $M\in\cert$.  Set the flow $p_1+\cdots+p_k$ on $e$, where $p_i$ is the flow through $e_i$.  It is a valid flow for $M$, and $\CP(\cG',M) \le \CP(\cG, M)$, because
\[\frac{(p_1+p_2+\cdots+p_k)^2}{w_1+\cdots+w_k} \le \frac{p_1^2}{w_1}+\cdots+\frac{p_k^2}{w_k}.\]
(The last inequality follows from Jensen's inequality for the square function
\((\alpha_1 x_1+\cdots + \alpha_k x_k)^2 \le \alpha_1 x_1^2 + \cdots+ \alpha_k x_k^2,\)
with $\alpha_i = w_i/(w_1+\cdots+w_k)$ and $x_i = p_i/\alpha_i$.)
\pfend

Despite \rf(prp:learningNoDuplicates), we often use several vertices corresponding to the same subset because it makes the analysis simpler.  In the following proofs we assume the learning graph is transformed as in \rf(prp:learningNoDuplicates).

\pfstart[First Proof (\cite{belovs:learning})]
Assume a learning graph $\cG$ is for a function $f\colon \{0,1\}^\vars\supseteq \cD\to\{0,1\}$.  
The idea is to convert $\cG$ into a span program and apply \rf(thm:spanAlgorithm).  
%\paragraph{Description of the Span Program}
Each vertex $S$ of the learning graph is represented by $2^{|S|}$ vectors $\{t_\alpha\}$ where $\alpha$ is an assignment in $\{0,1\}^S$.  We assume all these vectors are orthonormal.  
The vector space of the span program is spanned by all $t_\alpha$s.
The vector $t_\emptyset$ that corresponds to the vertex $\emptyset$ of $\cG$ is the target vector of the span program.

If $\alpha\colon S\to \{0,1\}$ is a 1-certificate for $f$, $t_\alpha$ is a free input vector.
Consider an arc $e$ of $\cG$ from $S$ to $S\cup\{j\}$ with weight $w_e$. For each vector $t_\alpha$ with $\alpha$ having domain $S$, we add two input vectors
\begin{equation}
\label{eqn:vektora}
\sqrt{w_e} \s[t_\alpha - t_{\alpha \cup \{j\mapsto b\}} ], \qquad b=0,1.
\end{equation}
Here $\alpha \cup\{j\mapsto b\}$ is the assignment with domain $S\cup\{j\}$ that maps $i$ to $\alpha(i)$ for $i\in S$ and maps $j$ to $b$. Each of these two vectors is labelled by the corresponding value $b$ of the variable $j$.

%\paragraph{Negative Witness Size}
Let us describe a negative witness $w'$ of the span program on input $y\in \preimy$.  For each $t_\alpha$, we let $\langle w', t_\alpha\rangle = 1$ if $\alpha$ agrees with $y$, and $\langle w', t_\alpha\rangle = 0$ otherwise.
Consider a free input vector of the form $t_\alpha$. Since $f(y)=0$, and $\alpha$ is a 1-certificate, $\alpha$ does not agree with the input.  By construction, $t_\alpha$ is orthogonal to the witness $w'$.
Now consider an available input vector of the form \refeqn{vektora}. There are two cases:
\itemstart
	\item The inner product $\langle w', t_\alpha\rangle$ equals 0. In this case, $\alpha$ does not agree with the input, and, hence, none of $\alpha\cup\{j\mapsto 0\}$ and $\alpha\cup\{j\mapsto 1\}$ does.  Hence, both vectors of \refeqn{vektora} are orthogonal to the witness.
	\item The inner product $\langle w', t_\alpha\rangle$ equals 1.  In this case, only the vector corresponding to the value $y_j$ is available in~\rf(eqn:vektora), and the assignment $\alpha\cup\{j\mapsto y_j\}$ agrees with the input.  Hence, the available input vector is orthogonal to the witness.
\itemend

This proves that $w'$ is indeed a negative witness. Let us calculate the size of $w'$. Let $e$ be an arc of $\cG$ from $S$ to ${S\cup\{j\}}$.  We claim there is exactly one input vector that arises from $e$ and is not orthogonal $w'$.  Let $\alpha$ be an assignment with domain $S$.  By the first point above, if $\alpha$ does not agree with the input, both input vectors in~\refeqn{vektora} are orthogonal to $w'$. If $\alpha$ agrees with $y$, the inner product of the false input vector from \refeqn{vektora} and $w'$ is $\sqrt{w_e}$.  Summing up over all arcs, the size of $w'$ equals $\sum_e w_e = \CN(\cG)$.

%\paragraph{Positive Witness Size}
Now, let us construct a positive witness for an input $x\in f^{-1}(1)$.  Let $p_e=p_e(x)$ be the corresponding flow.  We describe a linear combination of the available input vectors that equals $t_\emptyset$.
Let $e$ be an arc from $S$ to ${S\cup\{j\}}$ with weight $w_e$. Let $\alpha=x_S$ and take the available input vector from \refeqn{vektora} with the coefficient $p_e/\sqrt{w_e}$.  Multiplied by the coefficient, the vector equals $p_e(t_{x_S} - t_{x_{S\cup\{j\}}})$.
Suppose a vertex $S$ is a sink.  Then, $t_{x_S}$ is a free input vector.  Take it with the coefficient equal to the difference of the in-flow to $S$ and the out-flow of $S$. 
By the definition of the flow, the sum of all these vectors equals the target $t_\emptyset$. The witness size is $\sum_e p_e^2/w_e = \CP(\cG,x)$.
\pfend

\pfstart[Second Proof (Lee~\cite{lee:learningKdistPrior}).] 
This time, we reduce to \rf(thm:advAlgorithm).  For each arc $e$ from $S$ to $S\cup\{j\}$, we define a block-diagonal matrix $X^e_j = \sum_\alpha Y_\alpha$, where the sum is over all assignments $\alpha$ on $S$. Each $Y_\alpha$ is defined as $\psi\psi^*$ where, for each $z\in\cD$:
\[
\psi\elem[z] =
\begin{cases}
p_e(z)/\sqrt{w_e},& \mbox{$f(z)=1$, and $z$ satisfies $\alpha$;}\\
\sqrt{w_e},& \text{$f(z)=0$, and $z$ satisfies $\alpha$;}\\
0,&\mbox{otherwise.}
\end{cases}
\]
Finally, we define $X_j$ in~\refeqn{advDual} as $\sum_e X_j^e$ where the sum is over all arcs $e$ loading $j$.

The condition~\rf(eqn:advDualSemidefinite) is trivial.  
Let us check the condition~\rf(eqn:advDualCondition).
Fix any $x\in f^{-1}(1)$ and $y\in f^{-1}(0)$.  By construction, $X_j^e\elem[x,y]=p_e(x)$ if $x_S = y_S$ where $S$ is the origin of $e$, otherwise, $X_j^e\elem[x,y] = 0$.  Thus, only the arcs $e$ from $S$ to $S\cup\{j\}$ such that $x_S=y_S$ and $x_j\ne y_j$ contribute to the sum in~\rf(eqn:advDualCondition).  These arcs define a cut between the source $\emptyset$ and all the sinks of the flow $p_e=p_e(x)$.  
Hence, the sum of the values of the flow on these arcs equals the total value of the flow, 1.

Let us calculate the objective value~\rf(eqn:advDualObjective).  In $X_j^e$, the diagonal entry corresponding to an element $z\in \cD$ equals $p_e(z)^2/w_e$ or $w_e$, if $f(z)$ equals 1 or 0, respectively.  The sum $\sum_{j\in[\vars]} X_j$ equals $\sum_e X_j^e$ where the summation is over all arcs of the learning graph.  Hence, the objective value equals the maximum of $\CN(\cG)$ and $\CP(\cG)$.
\pfend

\subsection{Procedure-Driven Description}
\label{sec:procedureDriven}
In this section, we give an interpretation of a learning graph as a randomised procedure for loading variables.  This interpretation is useful for symmetric problems.  
 Let $\cert$ be a certificate structure.  For each $M\in\cert$, its own procedure is built.  The goal is to end up in an element of $M$, and this must be achieved with certainty.  The value of the complexity of the learning graph arises from the interplay between the procedures for different inputs. 

%Informally, \refthm{symmetric} further in the text says that if the procedures are uniform enough, the complexity of loading each variable depends on how much information about the marked elements is revealed by the set of elements loaded so far, and the element being loaded. 

We illustrate this concept by an example of a learning graph for the $k$-subset certificate structure.  It corresponds to the algorithm in \rf(thm:walkKDistinctness).  Let $M\in\cert$ be given by $\{a_1,a_2\dots,a_k\}$, i.e., $S\in M$ iff $S$ contains all of $a_i$s.  
Our randomised procedure consists of $k+1$ stages and is given in \reftbl{old}.  Here $r=o(\vars)$ is some parameter to be specified later.  In this case, only stage I is probabilistic, and all other stages are deterministic.  Thus, the internal randomness of the procedure is concealed in the choice of the $r$ elements.  Each choice has probability $\prob={\vars-k\choose r}^{-1}$.

\begin{table}[htb]
\centering
\begin{tabular}{rp{11cm}}
\hline
I.& Load $r$ elements different from $a_1,\dots,a_k$ uniformly at random.\\
II.1.& Load $a_1$.\\
II.2.& Load $a_2$.\\
& \quad\vdots\\
II.$k$.& Load $a_k$.\\
\hline
\end{tabular}
\caption{Learning graph for the $k$-subset certificate structure.}
\label{tbl:old}
\end{table}

Let us describe how the graph $\cG$ and the flow $p$ is constructed from the procedure in \reftbl{old}. At first, we define the {\em key vertices} of $\cG$. If $d$ is the number of stages, the key vertices are in $V_0\cup\cdots\cup V_d$, where $V_0=\{\emptyset\}$ and $V_i$ consists of all possible sets of variables loaded after $i$ stages (over the choice of $M$ and the choice of the internal randomness).

For a fixed $M\in\cert$ and fixed internal randomness, the sets $S_{i-1}\in V_{i-1}$ and $S_i\in V_i$ of the variables loaded before and after stage $i$, respectively, are uniquely defined.  In this case, we connect $S_{i-1}$ and $S_i$ by a {\em transition} $e$, and say that the transition is {\em taken} for this choice of $M$ and the randomness.
The transition $e$ is the path
\[
S_{i-1}, (S_{i-1}\cup\{t_1\},e), (S_{i-1}\cup\{t_1,t_2\},e),\dots, (S_i\setminus\{t_{\ell}\},e), S_i
\]
in $\cG$, where $t_1,\dots,t_\ell$ are the elements of $S_i\setminus S_{i-1}$ in some arbitrary order (see also \rf(exm:learningTrivial)).  Additional labels $e$ in the internal vertices ensure that the paths corresponding to different transitions do not intersect, except at the ends.  We say that the transition $e$ and all arcs therein {\em belong} to stage $i$.  The number $\ell$ is the {\em length} of the transition.

We say a transition is {\em used} for $M\in\cert$, if it is taken for some choice of the internal randomness.
The set of transitions of $\cG$ is the union of all transitions used for all inputs in $M\in\cert$. 
For example, \rf(tbl:distinctnessParam) features the description of all transitions for each stage of the learning graph in \rf(tbl:old), and the condition when each of them is used for a particular $M\in\cert$.  \reffig{distinct} shows an example of the learning graph for one particular choice of the parameters.

\begin{figure}[tbh] 
\centering 
%\epsfig{file=distinct1.eps, width=8cm}
%\release{\includegraphics[width=8cm]{distinct1.eps} }
\[
\def\objectstyle{\scriptstyle}
\xygraph{ !{0;<2pc,0pc>:} !~-{@{->}} !~:{@{<-}}
*++[o][F]{\emptyset}="0" [d(1.5)l(4.5)]
(
[] *+[o][F]{1,2} [r] *+[o][F]{1,3} [r] *+[o][F]{1,4} [r] *+[o][F]{1,5} [r] *+[o][F]{2,3} [r] 
*+[o][F]{2,4} [r] *+[o][F]{2,5} [r] *+[o][F]{3,4} [r] *+[o][F]{3,5} [r] *+[o][F]{4,5}
)
[d(1.7)l(1.3)]
([]*+[o][F]{1,2,3} [r(1.3)] *+[o][F]{1,2,4} [r(1.3)]
*+[o][F]{1,2,5} [r(1.3)] *+[o][F]{1,3,4} [r(1.3)]
*+[o][F]{1,3,5} [r(1.3)] *+[o][F]{1,4,5} [r(1.3)]
*+[o][F]{2,3,4} [r(1.3)] *+[o][F]{2,3,5} [r(1.3)]
*+[o][F]{2,4,5} [r(1.3)] *+[o][F]{3,4,5})
[d(1.3)r(1.7)]
([]*+[o][F]{1,2,3,4} [rr] *+[o][F]{1,2,3,5} [rr] *+[o][F]{1,2,4,5} [rr] *+[o][F]{1,3,4,5} [rr] *+[o][F]{2,3,4,5} [rr] )
"0" (-"1,2", -"1,3", -"1,4", -"1,5", -"2,3", -"2,4", -"2,5", -"3,4", -"3,5", -"4,5")
"1,2" (-"1,2,3", -"1,2,4", -"1,2,5")
"1,3" (-"1,2,3", -"1,3,4", -"1,3,5")
"1,4" (-"1,2,4", -"1,3,4", -"1,4,5")
"1,5" (-"1,2,5", -"1,3,5", -"1,4,5")
"2,3" (-"1,2,3", -"2,3,4", -"2,3,5")
"2,4" (-"1,2,4", -"2,3,4", -"2,4,5")
"2,5" (-"1,2,5", -"2,3,5", -"2,4,5")
"3,4" (-"1,3,4", -"2,3,4", -"3,4,5")
"3,5" (-"1,3,5", -"2,3,5", -"3,4,5")
"4,5" (-"1,4,5", -"2,4,5", -"3,4,5")
"1,2,3,4" (:"1,2,3", :"1,2,4", :"1,3,4", :"2,3,4")
"1,2,3,5" (:"1,2,3", :"1,2,5", :"1,3,5", :"2,3,5")
"1,2,4,5" (:"1,2,4", :"1,2,5", :"1,4,5", :"2,4,5")
"1,3,4,5" (:"1,3,4", :"1,3,5", :"1,4,5", :"3,4,5")
"2,3,4,5" (:"2,3,4", :"2,3,5", :"2,4,5", :"3,4,5")
}\]
\caption{The learning graph for $k$-distinctness from \TeXBug{\reftbl{old}} in the case $k=2$, $\vars=5$ and $r=2$.}
\label{fig:distinct}
\end{figure}

The flow $p_e(M)$ is defined as the probability, over the internal randomness, that transition $e$ is taken for $M$.  All arcs forming the transition are assigned the same flow.  Thus, the transition $e$ is used by $M$ if and only if $p_e(M)>0$.  In the learning graph from \reftbl{old}, $p_e(M)$ attains two values only: 0 and $\prob$.  Also, for each transition $e$, we define its weight $w_e$, and all arcs in the transition also have this flow.

We define the {\em (total) complexity of stage $i$}, $\CT_i(\cG)$, similarly as $\CT(\cG)$ is defined in~\refeqn{learningComp1} and~\refeqn{learningComp2} with the summation over $E_i$, the set of all arcs on stage $i$, instead of $E$.  It is easy to see that $\CT(\cG)$ is at most $\sum_i \CT_i(\cG)$.

The description in \rf(tbl:old) said nothing about the weights of the transitions.  We define them using \refthm{symmetric} below.  But for that we need some additional notions.

\begin{defn}[Symmetric flow]
\label{defn:symmetric}
We say that the flow on stage $i$ is {\em symmetric} if all transitions on the stage can be divided into {\em classes} so that the following holds.  Firstly, all transitions in the same class has the same length.  Secondly, for each class, the flow $p_e(M)$ through a transition in the class takes two values only: 0 and some $\prob>0$.  The value of $\prob$ neither depends on the choice of $M\in\cert$, nor on the choice of $e$ in the class (but it may depend on the class).  And finally, the number of transitions in the class used by the flow does not depend on the choice of $M\in\cert$.
\end{defn}

The flow in the learning graph from \reftbl{old} is symmetric.  
The conditions of \rf(defn:symmetric) are satisfied if one puts all the transitions on each stage into one class.

We define the length of the class as the length of any transition in it.  The ratio of the total number of transitions in the class to the number of them used by the flow is called the {\em speciality} of the class.  The speciality $T_i$ of stage $i$ is the maximal speciality of all classes on stage $i$.  The length $L_i$ of stage $i$ is the average length of a class on stage $i$: $L_i = \sum_e p_e(M)\ell(e)$ where $\ell(e)$ is the length of the transition $e$ and the sum is over all transitions on stage $i$.  For a symmetric flow, both these quantities do not depend on the choice of $M$.

\begin{thm}%[\cite{belovs:learning}]
\label{thm:symmetric}
If the flow on stage $i$ is symmetric, the arcs on stage $i$ can be weighted so that the complexity of the stage becomes $L_i\sqrt{T_i}$.
\end{thm}

\pfstart
For each class $C$, let $|C|$, $\ell_C$, $q_C$, $T_C$ be, respectively, the number of transitions in $C$, the length of $C$, the non-zero value of the flow in $C$, and the speciality of $C$.  Let us assign the same weight $w_C$ to all the arcs in $C$.
Then, the negative complexity of the stage is $\sum_C |C|\ell_C w_C$, and the positive is $\sum_C \frac{|C|}{T_C} \ell_C \frac{q_C^2}{w_C}$.  If we define $w_C = q_C/\sqrt{T_i}$, both of these quantities become equal to
\[
\sum_{C} \frac{|C|}{\sqrt{T_C}} \ell_C q_C = \sum_e \sqrt{T_C} \ell(e) p_e(M)
\]
where the second sum is over all transitions on stage $i$.  Clearly, the last expression does not exceed $L_i\sqrt{T_i}$.
\pfend

\begin{table}[htb]
\[\begin{tabular}{r|p{5cm}c|cc}
Stage & Transitions & Used & Length & Speciality\\
\hline
I& From $\emptyset$ to $S$ of $r$ elements & $a_1\dots,a_k \notin S$ & $r$ & $O(1)$ \\
II.1& From $S$ to $S\cup\{j\}$ for $|S|=r$ and $j\notin S$ & $a_1\dots,a_k \notin S$, $j=a_1$ & 1 & $O(\vars)$ \\
II.2& From $S$ to $S\cup\{j\}$ for $|S|=r+1$ and $j\notin S$ & $a_1\in S, a_2,\dots,a_k \notin S$, $j=a_2$ & 1 & $O(\vars^2/r)$ \\
& $\vdots$ &\\
II.$k$& {From $S$ to $S\cup\{j\}$ \par for $|S|=r+k-1$ and $j\notin S$} & $a_1\dots,a_{k-1} \in S$, $j=a_k$ & 1 & $O(\vars^k/r^{k-1})$ \\
\hline
\end{tabular}\]
%\centering 
%\begin{tabular}{|r|ccccc|}
%\hline
%Stage & I & II.1 & II.2 &\dots & II.$k$ \\
%\hline
%Length & $r$ & 1 & 1 &\dots & 1\\
%Speciality & 1 & $n$ & $n^2/r$ & \dots & $n^k/r^{k-1}$\\
%\hline
%\end{tabular}
\caption{Description of the transitions for each stage of the learning graph in \TeXBug{\rf(tbl:old)}.  Additionally, it is described when the transition is used, and the length and speciality of each stage. }
\label{tbl:distinctnessParam}
\end{table}

Now we are able to calculate the complexity of the learning graph in \reftbl{old}.  The parameters of each stage are given in \rf(tbl:distinctnessParam).  
It is not hard to verify the table.  For example, a transition from $S$ to $S\cup\{j\}$ on stage II.$k$ is used by $M$ iff $a_1,\dots,a_{k-1}\in S$ and $j=a_k$. For a random choice of $S$ and $j\notin S$, the probability of $j=a_k$ is $1/\vars$, and the probability of $a_1,\dots,a_{k-1}\in S$, given $j=a_{k}$, is $\Omega(r^{k-1}/\vars^{k-1})$. Thus, the total probability is $\Omega(r^{k-1}/\vars^k)$ and the speciality is the inverse of that.
By \rf(thm:symmetric), the complexity of the learning graph is $O(r+\sqrt{\vars^k/r^{k-1}})$. It is optimised when $r = \vars^{k/(k+1)}$, and we have
\begin{prp}
\label{prp:learningKSubsetUpper}
The learning graph complexity of the $k$-subset certificate structure on $\vars$ variables is $O(\vars^{k/(k+1)})$.
\end{prp}

The main idea behind the learning graph in \reftbl{old} is to reduce the speciality of loading the certificate given by the element $a_1,\dots, a_k$.  In the learning graph from \reftbl{old}, it is achieved by loading $r$ non-marked elements before loading the certificate.   A transition of stage II.$k$ from a subset $S$ of size $r+k-1$ to its superset $S\cup\{j\}$ gets used for all $M\in\cert$ such that $a_1,\dots,a_{k-1}$ are in $S$.  This makes ${r+k-1\choose k-1}$ choices of $M$ for which the transition is used.  This reduces the speciality from $O(\vars^k)$ (if there were no stage I) to $O(\vars^k/r^{k-1})$.  In this case, we say that $a_1,\dots,a_{k-1}$ are {\em hidden} among the $r$ previously loaded elements.  This gives an intuitive way of calculating the specialities of stages of a learning graph.

We see that the larger the set we hide the elements $a_1,\dots,a_{k-1}$ into, the better.
Unfortunately, we can't make $r$ as large as we like, because loading the non-marked elements also counts towards the complexity.  At the equilibrium point $r=\vars^{k/(k+1)}$, we attain the optimal complexity of the learning graph.

\subsection{Applications}
\label{sec:learningApplications}
The paper~\cite{belovs:learning} contains a quantum $O(\vertices^{35/27})$-query algorithm for the triangle problem as a consequence of the above theory.  This was an improvement compared to the previously best known $O(\vertices^{13/10})$-query algorithm from \rf(thm:walkTriangle).  This result was generalised to arbitrary subgraphs in~\cite{zhu:learning, lee:learningSubgraphs}.  After that, all these results were improved by Lee, Magniez and Santha in~\cite{lee:learningTriangle}.  In particular, they gave an $O(\vertices^{9/7})$-query algorithm for the triangle detection, and $O(\vertices^{10/7})$-query algorithm for the associativity testing problem.  In this section, we describe both of these algorithms.  Also,~\cite{lee:learningTriangle} contains a general framework for subgraphs detection.  The complexity of the algorithms is expressed as the optimal value of a linear program.  Refer to the paper for more detail.  We start with the triangle certificate structure.

\begin{prp}
\label{prp:learningTriangleUpper}
The quantum query complexity of the triangle certificate structure on $\vertices$ vertices is $O(\vertices^{9/7}) = O(\vars^{9/14})$.
\end{prp}

\pfstart
Let $a,b$ and $c$ be the vertices of the graph that define an element $M$ of the certificate structure.  Consider the learning graph in \rf(tbl:learningTriangle) where $r_1, r_2$ and $\ell$ are some parameters that satisfy $r_1,r_2 = o(\vertices)$, $r_1,r_2 = \omega(1)$ and $\ell = o(r_2)$.
See also \rf(fig:triangleLearning).
The flow is defined similarly to \rf(sec:procedureDriven).  Every non-zero flow through a transition on stages I, II and III is $\prob = {\vertices-3\choose r_1, r_2}^{-1}$, where we use notation
\[
{n\choose r_1,\dots,r_k} = {n\choose r_1} {n-r_1\choose r_2}\cdots {n-r_1-\cdots-r_{k-1}\choose r_k}.
\]
On stages IV, V and VI, every non-zero flow equals $\prob{r_2\choose \ell}^{-1}$.

\begin{table}[htb]
\[\begin{tabular}{rp{11cm}}
\hline
I& Take disjoint subsets $A,B\subseteq [\vertices]\setminus\{a,b,c\}$ of sizes $r_1$ and $r_2$, respectively, uniformly at random, and load all the edges between $A$ and $B$\\
II& Add $a$ to $A$ and load all the edges between $a$ and $B$\\
III& Add $b$ to $B$ and load all the edges between $b$ and $A$ (including $ab$)\\
IV& Choose, uniformly at random, $\ell$ vertices in $B\setminus\{b\}$ and load all 
the edges connecting $c$ to these vertices\\
V& Load the edge $bc$\\
VI& Load the edge $ac$\\
\hline
\end{tabular}\]
\caption{Learning graph for the triangle certificate structure. }
\label{tbl:learningTriangle}
\end{table}

\begin{table}[htb]
\[\begin{tabular}{r|cccccc}
{Stage}& I & II & III & IV & V & VI\\
\hline
{Length} & $r_1 r_2$ & $r_2$ & $r_1$ & $\ell$ & 1 & 1\\
{Speciality} & 1 & $\vertices$ & $\vertices^2/r_1$ & $\vertices^3/(r_1r_2)$ & $\vertices^3/r_1$ & $\vertices^3/\ell$\\
\hline
\end{tabular}\]
\caption{Parameters of the stages of the learning graph in \TeXBug{\rf(tbl:learningTriangle)}.  All expressions are given up to constant multiplicative factors. }
\label{tbl:TriangleParameters}
\end{table}

The flow is symmetric if we put all transitions on one stage into one class.  The parameters of the learning graph are summarised in \rf(tbl:TriangleParameters).  The lengths of the stages are obvious.  Let us give some comments on the values of the specialities using the hiding argument.  It is also not hard to give a direct counting argument.
\itemstart
\item On stage III, $b$ is uniquely determined by the transition, $a$ is hidden among the $r_1+1$ element of $A$, and $c$ can be almost any vertex of the graph.
\item On stage IV, $c$ is uniquely determined, $a$ is hidden among the $r_1+1$ element of $A$, and $b$ is hidden among the $\Omega(r_2)$ elements of $B$ not being connected to $c$.
\item On stage V, $b$ and $c$ are uniquely determined, and $a$ is hidden among the $r_1+1$ element of $A$.
\item Finally, on stage VI, $a$ and $c$ are uniquely determined, and $b$ is hidden among the $\ell+1$ element of $B$ connected to $c$.
\itemend
By \rf(thm:symmetric), we get the complexity of the learning graph is
\[
r_1r_2 + r_2\sqrt{\vertices} + \vertices\sqrt{r_1} + \frac{\vertices^{3/2}\ell}{\sqrt{r_1r_2}} + \frac{\vertices^{3/2}}{\sqrt{r_1}} + \frac{\vertices^{3/2}}{\sqrt{\ell}}.
\]
It can be checked that the optimal complexity of $O(\vertices^{9/7})$ is achieved for the values of the parameters $r_1 = \vertices^{4/7}$, $r_2 = \vertices^{5/7}$ and $\ell = \vertices^{3/7}$.
\pfend

\begin{figure}[bth]
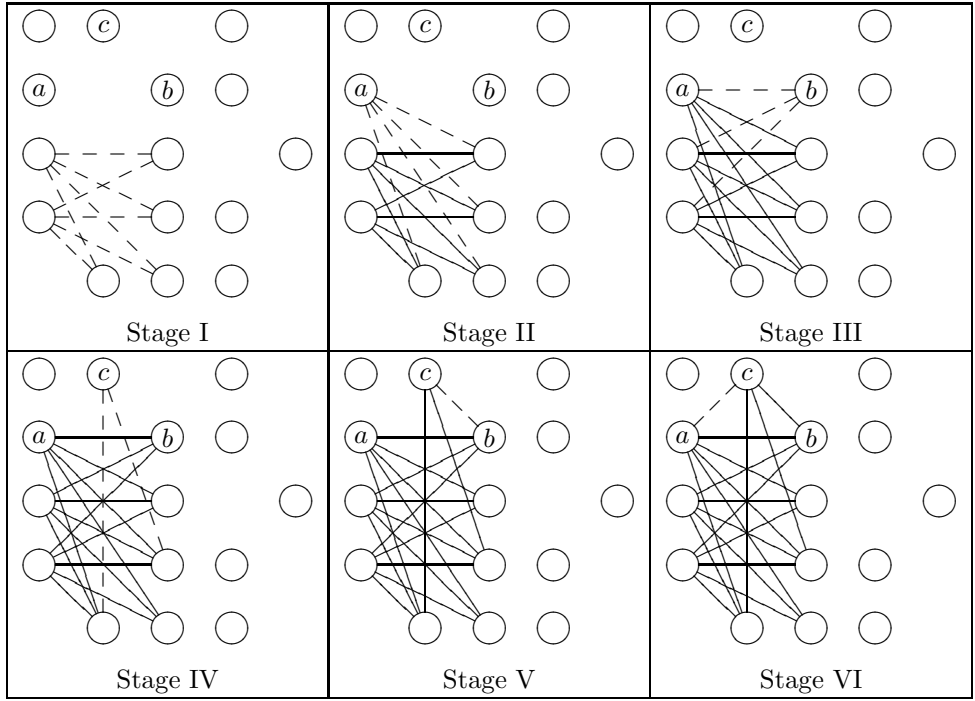

\[
\begin{tabular}{|c|c|c|}
\hline
{\xygraph{!~*{\cir<6pt>{}} !{0;<2pc,0pc>:} !~:{@{--}}
[l]{} & *{c}{}="c" & & {}\\
*{a}[]{}="a" & & *{b}{}="b" & {} \\
{}="a1" & & {}="b1" & & {} \\
{}="a2" & & {}="b2" & {}  \\
 & {}="b3" & {}="b4" &  *{\rule{0cm}{1cm}}{} 
"a1"(:"b1",:"b2",:"b3",:"b4")
"a2"(:"b1",:"b2",:"b3",:"b4")
}}
&
{\xygraph{!~*{\cir<6pt>{}} !{0;<2pc,0pc>:} !~:{@{--}}
[l]{} & *{c}{}="c" & & {}\\
*{a}[]{}="a" & & *{b}{}="b" & {} \\
{}="a1" & & {}="b1" & & {} \\
{}="a2" & & {}="b2" & {}  \\
 & {}="b3" & {}="b4" & {} 
"a1"(-"b1",-"b2",-"b3",-"b4")
"a2"(-"b1",-"b2",-"b3",-"b4")
"a"(:"b1",:"b2",:"b3",:"b4")
}}
&
{\xygraph{!~*{\cir<6pt>{}} !{0;<2pc,0pc>:} !~:{@{--}}
[l]{} & *{c}{}="c" & & {}\\
*{a}[]{}="a" & & *{b}{}="b" & {} \\
{}="a1" & & {}="b1" & & {} \\
{}="a2" & & {}="b2" & {}  \\
 & {}="b3" & {}="b4" & {} 
"a1"(-"b1",-"b2",-"b3",-"b4")
"a2"(-"b1",-"b2",-"b3",-"b4")
"a"(-"b1",-"b2",-"b3",-"b4")
"b"(:"a",:"a1",:"a2")
}}
\\
Stage I & Stage II & Stage III \\
\hline
{\xygraph{!~*{\cir<6pt>{}} !{0;<2pc,0pc>:} !~:{@{--}}
[l]{} & *{c}{}="c" & & {}\\
*{a}[]{}="a" & & *{b}{}="b" & {} \\
{}="a1" & & {}="b1" & & {} \\
{}="a2" & & {}="b2" & {}  \\
 & {}="b3" & {}="b4" & *{\rule{0cm}{1cm}}{} 
"a1"(-"b1",-"b2",-"b3",-"b4")
"a2"(-"b1",-"b2",-"b3",-"b4")
"a"(-"b1",-"b2",-"b3",-"b4")
"b"(-"a",-"a1",-"a2")
"c"(:"b2",:"b3")
}}
&
{\xygraph{!~*{\cir<6pt>{}} !{0;<2pc,0pc>:} !~:{@{--}}
[l]{} & *{c}{}="c" & & {}\\
*{a}[]{}="a" & & *{b}{}="b" & {} \\
{}="a1" & & {}="b1" & & {} \\
{}="a2" & & {}="b2" & {}  \\
 & {}="b3" & {}="b4" & {} 
"a1"(-"b1",-"b2",-"b3",-"b4")
"a2"(-"b1",-"b2",-"b3",-"b4")
"a"(-"b1",-"b2",-"b3",-"b4")
"b"(-"a",-"a1",-"a2")
"c"(:"b",-"b2",-"b3")
}}
&
{\xygraph{!~*{\cir<6pt>{}} !{0;<2pc,0pc>:} !~:{@{--}}
[l]{} & *{c}{}="c" & & {}\\
*{a}[]{}="a" & & *{b}{}="b" & {} \\
{}="a1" & & {}="b1" & & {} \\
{}="a2" & & {}="b2" & {}  \\
 & {}="b3" & {}="b4" & {} \\
"a1"(-"b1",-"b2",-"b3",-"b4")
"a2"(-"b1",-"b2",-"b3",-"b4")
"a"(-"b1",-"b2",-"b3",-"b4")
"b"(-"a",-"a1",-"a2")
"c"(:"a",-"b",-"b2",-"b3")
}}\\
Stage IV & Stage V & Stage VI \\
\hline
\end{tabular}
\]
\caption{Illustration of the learning graph in \TeXBug{\rf(tbl:learningTriangle)} for one particular choice of the number of vertices in the graph, vertices $a$, $b$ and $c$ defining the certificate $M$, and the internal randomness given by $A$ and $B$.  The edges loaded before the stage are shown as solid lines, and the edges being loaded during the stage are shown as hatched lines.}
\label{fig:triangleLearning}
\end{figure}

Now consider the associativity testing problem.  In this problem, we are given a binary algebraic operation $\circ$ defined on a set of $n$ elements.  The access to the operation is via an input oracle that, given elements $a,b\in[n]$, returns the value of $a\circ b$.  The problem is to check whether the operation is associative, i.e., whether $(b\circ c)\circ d = b\circ (c\circ d)$ for all $b,c,d\in[n]$.  This problem can be treated as a Boolean function on $\vars=n^2$ variables in $[n]$.
We treat the input to the associativity testing problem as a graph.  The set of vertices is $[n]$, and each edge $ab$ is labelled with the values of $a\circ b$ and $b\circ a$.  A query to the edge variables can be simulated by two queries to the input oracle.  

\begin{figure}[htb]
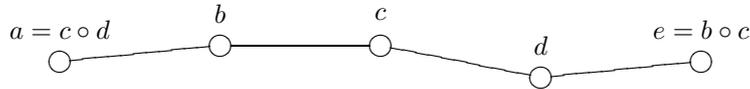

\vspace{-4pc}
\[
\xygraph{!~*{\cir<4pt>{}} !{0;<5pc,0pc>:} 
{} ([u(.2)]*{a = c\circ d})
-[u(.1)r] {} ([u(.2)]*{b})
-[r] {} ([u(.2)]*{c})
-[d(.2)r] {} ([u(.2)]*{d})
-[u(.1)r] {} ([u(.2)]*{e = b\circ c})
}
\]
\vspace{-4pc}
\caption{The associativity testing problem as having the 4-path certificate structure}
\label{fig:4path}
\end{figure}

In this representation, it suffices to construct a learning graph for the {\em 4-path certificate structure}.  In this certificate structure, the variables are given by pairs $ij$ with $1\le i<j\le n$, and each $M$ in the certificate structure is defined by five distinct elements $a,b,c,d,e\in[n]$ so that $S\in M$ if and only if $S\supseteq\{ab,bc,cd,de\}$.  In the last expression, we assume that $ji=ij$ for all $i<j$.
Indeed, we search for elements $b,c,d$ such that $(b\circ c)\circ d\ne b\circ (c\circ d)$.  Denote $a=c\circ d$ and $e = b\circ c$.  Then, the edges $ab,bc,cd$ and $de$ provide a certificate that the operation $\circ$ is not associative, see \rf(fig:4path).  There is a slight complication: It may happen that some of $a,b,c,d,e$ are equal.  The easiest way to overcome it is to copy each vertex of the graph 5 times.  In this case, even if some of the elements are equal, there is still a certificate formed by a 4-path.  From the next theorem, it follows that the quantum query complexity of the associativity testing problem is $O(n^{10/7}) = O(\vars^{5/7})$.  The best previously known quantum algorithm was based on simple Grover search and used $O(n^{3/2})$ queries.

\begin{prp}
\label{prp:associativity}
The learning graph complexity of the 4-path certificate structure is $O(n^{10/7})$.
\end{prp}

\pfstart
Let $a,b,c,d,e$ define an element $M$ of the certificate structure.  The learning graph is given in \rf(tbl:learning4path).  Again, the flow is treated as probability.  For instance, a non-zero flow through a transition on stage I equals
\[
{n-5\choose r_1,r_2,r_3,r_4}^{-1} s^{k_1+k_2} (1-s)^{r_1r_2 + r_3r_4 - k_1-k_2}
\]
where $k_1$ and $k_2$ are the number of edges the transition loads between $A$ and $B$, and $C$ and $D$, respectively.
The stages IV and VII are special.  No edges are loaded, instead of that, the flow before this stages is modified.  We scale the flow going to the vertices of the learning graph satisfying the condition so that it has value 1, and remove the flow from all other vertices.  By the Markov inequality, the probability the condition in stage IV is satisfied is at least $1/4$, so the positive complexity of the previous stages increases by at most a factor of 16.  Similarly for stage VII.
We calculate the complexity of the stages before the step with the old flow, and after the step---with the new one.  As the flow is scaled uniformly, it remains symmetric after the conditioning.

\begin{table}[htb]
\[\begin{tabular}{rp{14cm}}
\hline
I.& Take disjoint subsets $A,B,C,D\subseteq [\vertices]\setminus\{a,b,c,d,e\}$ of sizes $r_1 = n/10$, $r_2 = n^{4/7}$, $r_3 = n^{6/7}$, and $r_4 = n^{5/7}$, respectively, uniformly at random.  Load all edges between $B$ and $C$.  Load each edge between $A$ and $B$, and between $C$ and $D$, independently at random with probability $s = n^{-1/7}$. \\
II.& Add $a$ to $A$, and load each edge between $a$ and $B$ independently with probability $s$. \\
III.& Add $b$ to $B$, load all the edges between $b$ and $C$, and load each edge between $b$ and $A\setminus\{a\}$ independently with probability $s$.  \\
IV.& Condition on having at least $sr_1r_2/2$ and at most $r_1r_2/2$ edges between $A$ and $B$. \\
V.& Add $d$ to $D$, and load each edge between $d$ and $C$ independently with probability $s$.\\
VI.& Add $c$ to $C$, load all the edges between $c$ and $B$, and load each edge between $c$ and $D\setminus\{d\}$ independently with probability $s$.\\
VII.&  Condition on having at least $sr_3r_4/2$ and at most $r_3r_4/2$ edges between $A$ and $B$.\\
VIIII.& Load the edge $ab$.\\
IX.& Load the edge $cd$.\\
X.& Load the edge $de$.\\
\hline
\end{tabular}\]
\caption{Learning graph for the 4-path certificate structure. }
\label{tbl:learning4path}
\end{table}

The flow is symmetric if we define the class of the transition by the number of loaded edges between $A$ and $B$, and between $C$ and $D$ before and after the transition.  The lengths and specialities of the stages are summarised in \rf(tbl:4pathParameters).  It is straightforward to check that the complexity of each stage is $O(\vertices^{10/7})$.

\begin{table}[htb]
\[\begin{tabular}{r|cccccccc}
{Stage}& I & II & III &  V & VI & VIII & IX & X\\
\hline
{Length} & $s r_1r_2 + sr_3r_4 + r_2r_3$ & $sr_2$ & $sr_1+r_3$ & $sr_3$ & $sr_4 + r_2$ & 1 & 1 & 1\\
{Speciality} & 1 & $\vertices$ & $\vertices^2/r_1$ & $\frac{\vertices^3}{r_1r_2}$ & $\frac{\vertices^4}{r_1r_2r_4}$ & $\frac{\vertices^4}{r_3r_4}$ & $\frac{\vertices^4}{sr_1r_2}$ & $\frac{\vertices^5}{s^2 r_1r_2r_3}$\\
\hline
\end{tabular}\]
\caption{Parameters of the stages of the learning graph in \TeXBug{\rf(tbl:learning4path)}.  All expressions are given up to constant multiplicative factors. }
\label{tbl:4pathParameters}
\end{table}

The average lengths in \rf(tbl:4pathParameters) are easy to check.  Let us give some comments on specialities:
\itemstart
\item On stage III, $b$ is uniquely defined, and $a$ is hidden among the elements of $A$.  \item On stage V, $d$ is uniquely defined, and $ab$ is hidden among the non-loaded edges between $A$ and $B$.  This is the reason why we conditioned on having at most $r_1r_2/2$ edges loaded between $A$ and $B$.  Similarly on stages VI and VIII.
\item On stage IX, $c$ and $d$ are uniquely defined, and $ab$ is hidden among the loaded edges between $A$ and $B$.  This is why we conditioned on having at least $sr_1r_2/2$ edges loaded between $A$ and $B$.  Stage X is similar.
\itemend

For greater clarity, we also give a direct counting argument for stage X.  Let $k_1$ and $k_2$ be the number of edges between $A$ and $B$, and between $C$ and $D$ in a transition of a fixed class.  The total number of transitions in the class is 
\begin{equation}
\label{eqn:4path1}
{\vertices\choose r_1+1,r_2+1,r_3+1,r_4+1}{(r_1+1)(r_2+1)\choose k_1}
{(r_3+1)(r_4+1)\choose k_2}(r_4+1)(\vertices-r_1-r_2-r_3-r_4-4).
\end{equation}
For fixed $a,b,c,d$ and $e$, the number of transitions used by the flow is
\begin{equation}
\label{eqn:4path2}
{\vertices-5 \choose r_1,r_2,r_3,r_4}{(r_1+1)(r_2+1)-1\choose k_1-1}
{(r_3+1)(r_4+1)-1\choose k_2-1}.
\end{equation}
A simple calculation shows that the ratio of~\rf(eqn:4path1) to~\rf(eqn:4path2) is $O\s[\vertices^5/(s^2r_1r_2r_3)]$ if $k_1\ge sr_1r_2/2$ and $k_2\ge sr_3r_4/2$.  Similarly, one can check the specialities of all other stages.
\pfend

\section{Duality}
\label{sec:learningDuality}
\mycommand{arcs}{E}
\mycommand{source}{\mathrm{s}}
\remycommand{target}{\mathrm{t}}

This section is similar in spirit to \rf(sec:advDual).  We obtain a dual formulation of the learning graph complexity of a certificate structure that can serve as a lower bound.  After that, we use this dual formulation to show that the learning graphs we obtained in the previous sections are tight up to constant factors.  Since a learning graph seems a powerful tool for construction of quantum query algorithms, it is important to understand its limitations.  Even more, in the next section, we show how lower bounds on learning graph complexity can be transformed into lower bounds on quantum query complexity for some specific functions.

Recall that due to \rf(prp:learningNoDuplicates), we may assume a learning graph uses each subset of $[\vars]$ as a label of its vertex at most once.  Also, by introducing zero weights of arcs, we may, without loss of generality, assume the learning graph uses all subsets of $[\vars]$ exactly once, and uses all possible arcs between them.  Thus, let $\arcs$ by the set of pairs $(S,S')$ of subsets of $[\vars]$ such that $S'=S\cup\{j\}$ for some $j\notin S$.  For $e=(S,S\cup\{j\})\in\arcs$, let $\source(e)=S$ and $\target(e)=S\cup\{j\}$.

\begin{thm}
\label{thm:learningDuality}
The learning graph complexity of a certificate structure $\cert$ on $\vars$ variables is equal to the optimal value of the following two optimisation problems
\begin{subequations}
\label{eqn:learningPrimal}
\begin{alignat}{3}
 &{\mbox{\rm minimise }} &\quad& \sqrt{\sum\nolimits_{e\in\arcs} w_e} \label{eqn:flowObjective} \\ 
 &{\mbox{\rm subject to }} && \sum\nolimits_{e\in\arcs} \frac{p_e(\marked)^2}{w_e}\le 1 &\quad&\text{\rm for all $\marked\in\cert$;} \label{eqn:flowValue}\\
 &&& \sum_{e\in\arcs\colon \target(e)= S} p_e(\marked) = \sum_{e\in\arcs\colon \source(e)= S} p_e(\marked) &&\text{\rm for all $\marked\in\cert$ and $S\in 2^{[\vars]}\setminus(M\cup\{\emptyset\})$;} \label{eqn:flowCond1} \\
 &&& \sum\nolimits_{e\in\arcs\colon \source(e) = \emptyset} p_e(\marked) = 1 &&\text{\rm for all $\marked\in\cert$;}\label{eqn:flowCond2}\\
 &&& p_e(M)\in\R,\quad w_e\ge 0 && \text{\rm for all $e\in\arcs$ and $M\in\cert$;}
\end{alignat}
\end{subequations}
(here, $0/0$ in~\refeqn{flowValue} is defined to be 0), and
\begin{subequations}
\label{eqn:learningDual}
\begin{alignat}{3}
 &{\mbox{\rm maximise }} &\quad& \sqrt{\sum\nolimits_{\marked\in\cert} \alpha_\emptyset(\marked)^2} \label{eqn:alphaObjective} \\ 
 &{\mbox{\rm subject to }} && \sum\nolimits_{\marked\in\cert} \s[\alpha_{\source(e)}(\marked) - \alpha_{\target(e)}(\marked)\strut]^2\le 1 &\quad&\text{\rm for all $e\in\arcs$;} \label{eqn:alphaOne} \\
 &&& \alpha_S(\marked) = 0 && \text{\rm whenever $S\in \marked$;} \label{eqn:alphaZero}  \\
 &&& \alpha_S(\marked)\in\R && \text{\rm for all $S\subseteq[\vars]$ and $M\in\cert$.}
\end{alignat}
\end{subequations}
\end{thm}

\pfstart
Eq.~\refeqn{learningPrimal} is a restatement of the definition of a learning graph from \rf(sec:learning) with an application of \rf(rem:learningEquivalentComplexity).  The second expression~\refeqn{learningDual} is a new one, and requires a proof.  
The equivalence of the two expressions is obtained by duality.  We use basic convex duality~\cite[Chapter 5]{boyd:convex}.
First of all, we consider both programs with their objective values~\refeqn{flowObjective} and~\refeqn{alphaObjective} squared.  With this change, Eq.~\refeqn{learningPrimal} becomes a convex program  (for the convexity of~\refeqn{flowValue}, see~\cite[Section 3.1.5]{boyd:convex}).  The program is strictly feasible.  Indeed, it is easy to see that~\refeqn{flowCond1} and~\refeqn{flowCond2} are feasible.  To assure strong feasibility in~\refeqn{flowValue}, it is enough to take $w_e$ large enough.  Hence, by Slater's condition, the optimal values of~\refeqn{learningPrimal} and its dual are equal.  Let us calculate the dual.  The Lagrangian of~\refeqn{learningPrimal} is as follows
\begin{multline}
\label{eqn:lagrangian}
\sum_{e\in\arcs}w_e + \sum_{M\in\cert}\mu_M\s[\sum_{e\in\arcs} \frac{p_e(\marked)^2}{w_e}-1] \\
+\sum_{\substack{M\in\cert,\; S\subseteq[\vars]\\S\ne\emptyset,\; S\notin M}} \nu_{M,S}
\sD[\sum_{\substack{e\in\arcs\\\target(e)=S}} p_e(M) - \sum_{\substack{e\in\arcs\\\source(e)=S}} p_e(M)] +
\sum_{M\in\cert} \nu_{M,\emptyset}\sD[1-\sum_{\substack{e\in\arcs\\\source(e)=\emptyset}} p_e(M)].
\end{multline}
Here $\mu_M\ge 0$, and $\nu_{M,S}$ are arbitrary.  Let us first minimise over $p_e(M)$.  Each $p_e(M)$ appears three times in~\refeqn{lagrangian} with the following coefficients:
\[
p_e(M)^2\frac{\mu_M}{w_e} + p_e(M)\sA[\nu_{M,\target(e)}-\nu_{M,\source(e)}],
\]
where we assume $\nu_{M,S}=0$ for all $S\in M$.  The minimum of this expression clearly is 
\[
-\frac{w_e}{4\mu_M} \sA[\nu_{M,\target(e)}-\nu_{M,\source(e)}]^2.
\]
Plugging this into~\refeqn{lagrangian} yields
\begin{equation}
\label{eqn:lagrangian2}
\sum_{M\in\cert} (\nu_{M,\emptyset}-\mu_M) + 
\sum_{e\in\arcs} w_e\s[1 - \sum_{M\in\cert} \frac{\sA[\nu_{M,\target(e)}-\nu_{M,\source(e)}]^2}{4\mu_M}].
\end{equation}
Define $\alpha_S(M)$ as $\nu_{M,S}/(2\sqrt{\mu_M})$.  Minimising~\refeqn{lagrangian2} over $w_e$, the second term disappears if condition~\refeqn{alphaOne} is satisfied.  The first term is
\[
\sum_{M\in\cert} (2\sqrt{\mu_M}\alpha_\emptyset(M) - \mu_M).
\]
We can also maximise over $\mu_M$, that gives the square of~\refeqn{alphaObjective}.
\pfend

\mycommand{alphaz}{\alpha_\emptyset}
Now, we construct feasible solutions to the dual learning graph complexity~\refeqn{learningDual} for some of the certificate structures from \refsec{learning}.  
At first, we get the following matching lower bound to \rf(prp:learningCollisionUpper):
\begin{prp}
\label{prp:learningCollisionLower}
The learning graph complexity of the hidden shift (and, hence, the set equality and the collision) certificate structure is $\Omega(\vars^{1/3})$.
\end{prp}

\pfstart
We show that either we load too many elements, or the speciality of loading the second element of a pair in $M$ is too high.
Let $\cert$ be the hidden shift certificate structure.  Define 
\[
\alpha_S(M) = 
\begin{cases}
\vars^{-1/2}\max\sfig{\vars^{1/3}-|S|,\;	0}, & S\notin M; \\
0,&\mbox{otherwise.}
\end{cases}
\]
It is easy to check that the objective value~\refeqn{alphaObjective} is $\Omega(\vars^{1/3})$.  The condition~\rf(eqn:alphaZero) is trivial.  We now prove that the condition~\rf(eqn:alphaOne) holds up to a constant factor.  
Fix any $S\subset[\vars]$ and $j\notin S$.  If $|S|\ge \vars^{1/3}$, then $\alpha_S(M) = \alpha_{S\cup\{j\}}(M) =  0$ for all $M$, and we are done.  So, assume $|S|< \vars^{1/3}$.

There are $m$ choices of $M$ in $\cert$.  If $S\cup\{j\}\notin M$, then the value of $\alpha_S(M)$ decreases by $1/\sqrt{\vars}$ as $|S|$ increases by 1.  If $M$ is such that $S\notin M$ and $S\cup\{j\}\in M$, then $\alpha_S(M)$ changes by at most $\vars^{-1/6}$.  But there are at most $\vars^{1/3}$ such choices of $M$.  Thus,
\[
\sum_{M\in\cert} (\alpha_S(M) - \alpha_{S\cup\{j\}}(M))^2 \le 
m\cdot \frac1n + \vars^{1/3}\cdot \vars^{-1/3} = O(1).
\]
For the set equality and collision certificate structures, just assign $\alpha_S(M)=0$ for all $M$ that do not belong to the hidden shift certificate structure.
\pfend

Next, we obtain a matching lower bound to \rf(prp:learningKSubsetUpper).

\begin{prp}
\label{prp:learningKSubsetLower}
The learning graph complexity of the $k$-subset certificate structure on $\vars$ variables is $\Omega(\vars^{k/(k+1)})$.
\end{prp}

\pfstart
Let $\cert$ be the $k$-subset certificate structure.  Define $\alpha_S(M)$ as
\[
{\vars\choose k}^{-1/2}\max\sfig{\vars^{k/(k+1)} - |S|,\; 0}
\]
if $S\notin M$, and as 0 otherwise.  
Let us prove that~\refeqn{alphaOne} holds up to a constant factor.  Take any $S\subset[\vars]$ and let $j$ be any element not in $S$.  Again, we may assume $|S|\le \vars^{k/(k+1)}$.  There are ${\vars\choose k}$ choices of $M$.  If $S\cup\{j\}\notin M$, then the value of $\alpha_S(M)$ changes by ${\vars\choose k}^{-1/2}$ as the size of $|S|$ increases by 1.  Also, there are at most ${|S|\choose k-1}\le \vars^{k(k-1)/(k+1)}$ choices of $M\in\cert$ such that $S\notin M$ and $S\cup\{j\}\in M$.  For each of them, the value of $\alpha_S(M)$ changes by at most ${\vars\choose k}^{-1/2}\vars^{k/(k+1)}$.  Thus,
\[
\sum_{M\in\cert} (\alpha_S(M) - \alpha_{S\cup\{j\}}(M))^2 \le 
{\vars\choose k}^{-1} 
\sk[{\vars\choose k}\cdot1 + \vars^{k(k-1)/(k+1)}\vars^{2k/(k+1)} ] 
= O(1).
\]

On the other hand, for the objective value~\refeqn{alphaObjective}, we have
\[
\sqrt{\sum_{M\in\cert} \alphaz(M)^2} = \vars^{k/(k+1)}.\qedhere
\]
\pfend

\subsection{Triangle Certificate Structure}
The point of this section is to show that the learning graph from \rf(prp:learningTriangleUpper) for the triangle certificate structure is essentially tight.  
One can see that the proofs of the lower bounds in Propositions~\ref{prp:learningCollisionLower} and~\ref{prp:learningKSubsetLower} essentially proceed by showing, in a formal way, that all possible strategies of constructing the upper bound fail.  The complete proofs are short because of the simplicity of the corresponding certificate structures.  But even for the triangle certificate structure, the proof becomes rather bulky, and we lose a logarithmic factor compared to the upper bound, \rf(prp:learningTriangleUpper).

For greater clarity, we prove the lower bound in two steps, by showing an $\Omega(\vertices^{5/4})$ lower bound at first.

\begin{prp}
\label{prp:learningTriangle}
The learning graph complexity of the triangle certificate structure on $\vertices$ vertices is $\Omega(\vertices^{5/4})$.
\end{prp}

\pfstart
The idea of the proof is straightforward: we show that either we load too many edges incident to the vertices of the triangle, or the speciality of loading the last edge of the triangle is too high.

Let $\cert$ be the triangle certificate structure.  
Fix some $M\in\cert$ given by a triangle $abc$, and let $S$ be the set of loaded variables (edges of the input graph).
%Define $\alpha_S = \alpha_S(M)$ as follows.  
Let $T$ denote the subset of $\{a,b,c\}$ that consists of vertices incident to at least one not-yet-loaded edge of the triangle.  (Thus, $T$ is either empty, or consists of 2 or 3 vertices.)  Denote
\[
d = \max\left\{ \min_{x\in T} \deg_S x-\sqrt{\vertices},\;\; 0 \right\},
\]
where $\deg_S x$ is the number of loaded edges incident to $x$.
Define
\[
\alpha_S(M) = \begin{cases}
0,& ab,ac,bc \in S; \\
\vertices^{-3/2}\max\sfig{\vertices^{5/4} - |S| - \vertices^{3/4}d ,\; 0},&\text{otherwise.}
\end{cases}
\]
The objective value and~\rf(eqn:alphaZero) are clear.  So let us check~\rf(eqn:alphaOne) in various cases.  Assume $S$ and $j\notin S$ are fixed.  We may also assume $|S|<\vertices^{5/4}$.
\begin{itemize}
\itemsep0pt
\item adding the edge $j$ increases $|S|$ by 1.  The change in $\alpha_S(M)$ is $\vertices^{-3/2}$, and $O(\vertices^3)$ choices of $M$ satisfy this condition.  This gives the total contribution of $O(1)$ to the left hand side of~\rf(eqn:alphaOne).  
\item $d$ increases by 1, but $T$ does not change.  We may assume that $d<\sqrt{\vertices}$.  There are two cases.
\begin{itemize}
\itemsep0pt
\item $T=\{a,b,c\}$.  Thus, the degrees (in $S$) of $a$, $b$ and $c$ are at least $\sqrt{\vertices}$.  Then, one of the vertices of the triangle is incident to the new edge, and the two other vertices of the triangle are among at most  $2m^{5/4}/\vertices^{1/2} = O(\vertices^{3/4})$ vertices of high degree.  This gives $O(\vertices^{3/2})$ choices of $M$ that are affected by this change.
\item $T$ omits one of $a$, $b$, $c$.  In this case, one vertex of the triangle is incident to the new edge.  Also, as $d$ changes, our assumption on $d<\sqrt{\vertices}$ shows that the degree of this vertex is at most $2\sqrt{\vertices}$.
 Thus, the vertex omitted from $T$ is among at most $2\sqrt{\vertices}$ of its neighbours, and the second vertex of $T$ is among $O(\vertices^{3/4})$ vertices of high degree.
\end{itemize}
As $\alpha_S(M)$ changes by $\vertices^{-3/4}$, the total contribution is $O(1)$.
\item $j$ is the second edge of the triangle, and $d$ is affected as $T$ changes; or $j$ is the third edge of the triangle.  Again, we may assume that $d$ (before the change of $T$) is less than $\sqrt{\vertices}$.
In this case, two vertices of the triangle are determined by the new edge, and the third one is among at most $2\sqrt{\vertices}$ neighbours of one of the two.  There are $O(\sqrt{\vertices})$ choices of $M$, the change in $\alpha_S(M)$ is at most $\vertices^{-1/4}$, and the total contribution is $O(1)$. \qedhere
\end{itemize}

\pfend

Curiously, \rf(prp:learningTriangle) does not match the upper bound from \rf(prp:learningTriangleUpper), because it {\em is} possible to load many edges incident to the vertices of the triangle.  Consider the learning graph in \rf(tbl:fakeTriangle) where $a,b$ and $c$, as usual, denote the vertices of the triangle.  Its analysis is performed in \rf(tbl:fakeTriangleParameters).  If we set $r_2=\vertices^{3/4}$ and $r_1=r_3 = \vertices^{1/2}$, then the complexity of each stage, except stage VII, is $O(\vertices^{5/4})$.  Also, before stage V, each vertex of the triangle is incident to at least $\sqrt{\vertices}$ loaded edges.  

\begin{table}[htb]
\[\begin{tabular}{rp{11cm}}
\hline
I& Take disjoint subsets $B,C\subseteq [\vertices]\setminus\{a,b,c\}$ of sizes $r_2$ and $r_3$, respectively, uniformly at random, and load all the edges between $B$ and $C$\\
II& Add $a$ to $B$ and load all the edges between $a$ and $C$\\
III& Add $b$ to $B$ and load all the edges between $b$ and $C$\\
IV& Choose, uniformly at random, a subset $A\subseteq B\setminus\{a,b\}$ of size $r_1$ and load all the edges between $c$ and $A$\\
V& Load the edge $ac$, and add $a$ to $A$\\
VI& Load the edge $bc$\\
VII& Load the edge $ab$\\
\hline
\end{tabular}\]
\caption{An illustrative learning graph for the triangle certificate structure. }
\label{tbl:fakeTriangle}
\end{table}

\begin{table}[htb]
\[\begin{tabular}{r|ccccccc}
{Stage}& I & II & III & IV & V & VI & VII\\
\hline
{Length} & $r_2 r_3$ & $r_3$ & $r_3$ & $r_1$ & 1 & 1 & 1\\
{Speciality} & 1 & $\vertices$ & $\vertices^2/r_2$ & $\vertices^3/r_2^2$ & $\vertices^3/r_2$ & $\vertices^3/r_1$ & ? \\
\hline
\end{tabular}\]
\caption{Parameters of the stages of the learning graph in \TeXBug{\rf(tbl:fakeTriangle)}.  All expressions are given up to constant multiplicative factors. }
\label{tbl:fakeTriangleParameters}
\end{table}

To get a small complexity of the last stage, we would like to say that $c$ is hidden in $C$, but it is not true, because $c$ is the only vertex in $C$ that is not connected to the vertices in $B\setminus A$.  If we were able to erase the edges between $C$ and $B\setminus (A\cup\{b\})$, the complexity of stage VII would also become equal to $O(\vertices^{5/4})$.  But, we are not able to do this.  The analysis in \rf(prp:learningTriangle) does not take into account that we could have extra edges we would like to get rid of, and it gives a lower bound of $\Omega(\vertices^{5/4})$ because of the learning graph in \rf(tbl:fakeTriangle).  In order to give a tight lower bound, we have to take into account the degrees of the vertices outside the triangle.  Thus, we will be able to catch that the degree of $c$ in \rf(tbl:fakeTriangle) is different from the degrees of the vertices in $C$.  While doing so, we lose a logarithmic factor in the estimate.

\begin{thm}
\label{thm:learningTriangle}
The learning graph complexity of the triangle certificate structure on $\vertices$ vertices is $\Omega(\vertices^{9/7}/\sqrt{\log \vertices})$.
\end{thm}

\pfstart
 Let $E=\{uv\mid 1\le u < v\le \vertices\}$ be the set of input variables (potential edges of the graph).
Let $\cert$ be the triangle certificate structure.  We will construct a feasible solution to~\refeqn{learningDual} (with $[\vars]$ replaced by $E$) in the form
\begin{equation}
\label{eqn:form}
\alpha_S(M) = 
\begin{cases}
\max\{\vertices^{-3/14} - \vertices^{-3/2}|S| - \sum_{i=1}^k g_i(S,M),\; 0\},& S\notin M;\\
0,&\text{otherwise;}
\end{cases}
\end{equation}
where $g_i(S,M)$ is a non-negative function such that $g_i(\emptyset,M)=0$ and $g_i(S,M)\le \vertices^{-3/14}$.  The value of~\refeqn{alphaObjective} is ${\sqrt{{\vertices\choose 3}}}\; \vertices^{-3/14} = \Omega(\vertices^{9/7})$.  The hard part is to show that~\refeqn{alphaOne} holds up to logarithmic factors.  It is easy to see that $\alpha_S(M) = 0$ if $|S|\ge \vertices^{9/7}$, hence, we will further assume $|S|\le \vertices^{9/7}$.

% We first start with {\em partial} solutions of the form~\refeqn{form}.  This means the following.
 For $S\subset E$ and $j\in E\setminus S$, let $F(S,j)$ denote the subset of $M\in\cert$ such that $S\notin M$, but $S\cup\{j\}\in M$.
 We decompose $F(S,j) = F_1(S,j)\sqcup \cdots\sqcup F_k(S,j)$ as follows.
 Each $M\in\cert$ is defined by three vertices $a,b,c$ forming the triangle: $S\in M$ if and only if $ab,ac,bc\in S$.  An input index $j\in E$ satisfies $S\notin M$ and $S\cup\{j\}\in M$ only if $j\in\{ab,ac,bc\}$.  We specify to which of $F_i(S,j)$ an element $M\in F(S,j)$ belongs by the following properties:
\itemstart
\item to which of the three possible edges, $ab$, $ac$ or $bc$, the new edge $j$ is equal, and
\item the range to which the degree in $S$ of the third vertex of the triangle belongs: $[0,\vertices^{3/7}]$, $[\vertices^{3/7}, 2\vertices^{3/7}]$, $[2\vertices^{3/7}, 4\vertices^{3/7}]$, $[4\vertices^{3/7}, 8\vertices^{3/7}]\dots$
\itemend 
Hence, $k\approx 12/7 \log_2 \vertices$.  For notational convenience, let $j=bc$.  Then, the second property is determined by $\deg a = \deg_S a$, the degree of $a$ in the graph with edge set $S$.
 
% these three edges $j$ is equal.  For notational convenience, let it be $bc$.
% We define $F_i(S,j)$ using bounds on the degree $\deg a = \deg_S a$ of the third vertex of the triangle (in our case, $a$) in the graph having $S$ as the edge set: 
% $F_1(S,j)$ consists of all $M\in F(S,j)$  such that $\deg a\le n^{3/7}$ and, for $i\in[2,k]$, $F_i(S,j)$ consists of all $M\in F(S,j)$ such that $2^{i-2}n^{3/7}<\deg a\le 2^{i-1}n^{3/7}$. Hence we choose $k=\lceil 1+4/7\log_2n\rceil$.

For $i\in[k]$, we will define $g_i(S,M)$ so that, for all $S\subset E$ of size at most $\vertices^{9/7}$ and $j\in E\setminus S$:
\begin{equation}
\label{eqn:partial1}
\sum_{M\in\cert\setminus F(S,j)} \sA[g_i(S,M)-g_i(S\cup\{j\}, M)]^2 = O(1)
\end{equation}
and
\begin{equation}
\label{eqn:partial2}
\sum_{M\in F_i(S,j)} \sA[\vertices^{-3/14} - g_i(S,M)]^2 = O(1).
\end{equation}
Let $g_0(S,M)=\vertices^{-3/2}|S|$, for which~\refeqn{partial1} holds.
Even more, we will show that the set $K=K(S,j)$ of $i\in[0,k]$ such that~\refeqn{partial1} is non-zero has size $O(1)$.
 Thus, for the left hand side of~\refeqn{alphaOne}, we will have 
\begin{equation*}
\begin{split}
\sum_{M\in\cert} (\alpha_S(M) - \alpha_{S\cup\{j\}}(M))^2 \le &\;
|K| \sum_{i\in K}\sum_{M\in\cert\setminus F(S,j)} \sA[g_i(S,M) - g_i(S\cup\{j\}, M)]^2 \\
& + \sum_{i=1}^k \sum_{M\in F_i(S,j)} \sA[ \vertices^{-3/14} - g_i(S,M) ]^2,
\end{split}
\end{equation*}
where the former term on the right hand size is $O(1)$ and the latter one is $O(\log \vertices)$.
By scaling all $\alpha_S(M)$ down by a factor of $O(\sqrt{\log \vertices})$, we obtain a feasible solution to~\refeqn{learningDual} with the objective value $\Omega(\vertices^{9/7}/\sqrt{\log \vertices})$.

It remains to construct the functions $g_i(S,M)$.  
In the following, let $\mu(x)$ be the median of $0$, $x$, and $1$, i.e., $\mu(x) = \max\{0,\min\{x,1\}\}$.
The first interval of $\deg a$ will be considered separately from the rest.

% % %

%\paragraph{Case $i=1$}
\paragraph{First interval}
  Assume the condition $\deg a \le \vertices^{3/7}$.  Define
\begin{equation}
\label{eqn:g1}
g_i(S, M) = \begin{cases}
\vertices^{-3/14}\,\mu(2 - \vertices^{-3/7}\deg a),&\text{$ab,ac\in S$;}\\
0,&\text{otherwise.}
\end{cases}
\end{equation}
Clearly, $g_i(\emptyset,M)=0$ and $g_i(S,M)\ge 0$.  There are two cases how $g_i(S,M)$ may be influenced.  We show that the total contribution to~\refeqn{partial1} is $O(1)$.
\itemstart
\item It may happen if $|\{ab,ac\}\cap S| = 1$ and $j\in \{ab,ac\}$, i.e., the transition from the second case of~\refeqn{g1} to the first one happens.  Moreover, $g_1(S,M)$ changes only if $\deg a\le 2\vertices^{3/7}$.  Then $j$ identifies two vertices of the triangle, and the third one is among the neighbours of an endpoint of $j$ having degree at most $2\vertices^{3/7}$.  Thus, the total number of $M$ satisfying this scenario is at most $4\vertices^{3/7}$.  The contribution to~\refeqn{partial1} is at most $O(\vertices^{3/7})(\vertices^{-3/14})^2 =O(1)$.
\item Another possibility is that $ab,ac\in S$ and $\deg a$ changes.  In this case, $a$ is determined as an endpoint of $j$, and $b$ and $c$ are among its at most $2\vertices^{3/7}$ neighbours.  The number of $M$ influenced is $O(\vertices^{6/7})$, and the contribution is $O(\vertices^{6/7})(\vertices^{-9/14})^2 = o(1)$.
\itemend
Finally, we have to show that~\refeqn{partial2} holds.  If $M$ satisfies the condition, then $ab,ac\in S$ and $\deg a\le \vertices^{3/7}$.  In this case, the left hand side of~\refeqn{partial2} is 0.

%\paragraph{Case $i\in[2,k]$}
\paragraph{Other intervals}
  Now assume the condition $d < \deg a\le 2d$ with $d\ge \vertices^{3/7}$. Define a piece-wise linear function $\tau$ as follows
\[
\tau(x) =
\begin{cases}
0,& x<d/2;\\
(2x-d)/d,& d/2\le x<d;\\
1,& d\le x < 2d;\\
(5d-2x)/d,& 2d\le x\le 5d/2;\\
0,& x\ge 5d/2.
\end{cases}
\xygraph{!{0;<1pc,0pc>:}
[d(3)r(5)](
	[ld(.4)]*{0},
	:[u(6)][d(.2)r(1.1)]*{\tau(x)},
	:[r(15)][r(.2)d(.5)]*{x},
	[u(4)](
		[l(.4)]*{1},
		-@{.}[r(14)]
		),
	[d(.8)][r(2.5)]*{\smash{d/2}}[r(2.5)]*{\smash{d}}[r(5)]*{\smash{2d}}[r(2.5)]*{\smash{5d/2}},
	[r(2.5)]-[r(2.5)u(4)](-@{.}[d(4)])-[r(5)](-@{.}[d(4)])-[r(2.5)d(4)]
	)
}
\]
It can be interpreted as a continuous version of the indicator function that a vertex has a right degree.  Define
\[
\nu(S,M) = \sum_{v\in N(b)\cap N(c)} \tau(\deg v),
\]
where the sum is over the common neighbours of $b$ and $c$.  Let
\[
g_i(S,M) = \vertices^{-3/14}\,\mu\sB[\min\sfig{\frac{2\deg a}{d}, \frac{\nu(S,M)}{\vertices^{3/7}} }-1].
\]
Let us consider how $g_i(S, M)$ may change and how this contributes to~\refeqn{alphaOne}.
Now there are three cases how $g_i(S,M)$ may be influenced.  We again show that the total contribution to~\refeqn{partial1} is $O(1)$.

\itemstart
\item It may happen that $j$ is incident to a common neighbour of $b$ and $c$, and thus $\nu(S)$ may change.  This means $b$ and $c$ are among the neighbours of an endpoint of $j$ of degree at most $5d/2$.  Hence, this affects $O(\vertices d^2)$ different $M$.  The contribution is $O(nd^2)(\vertices^{-9/14}/d)^2 = o(1)$.

\item The set $N(b)\cap N(c)$ may increase.  This causes a change in $g_i(S, M)$ only under the following circumstances.  The new edge $j$ is incident to $b$ or $c$.  The second vertex in $\{b,c\}$ is among $\Theta(d)$ neighbours of the second end-point of $j$.  Finally, $\deg a\ge d/2$, that together with $|S|\le \vertices^{9/7}$ implies that there are $O(\vertices^{9/7}/d)$ choices for $a$.  Altogether, the number of $M$ affected by this is $O(\vertices^{9/7})$, and the change in $g_i(S, M)$ does not exceed $\vertices^{-9/14}$.  The contribution is $O(1)$.

\item The degree of $a$ may change.  Let us calculate the number $P$ of possible pairs $b$ and $c$ affected by this.  There is a change in $g_i(S, M)$ only if $b$ and $c$ are connected to at least $\vertices^{3/7}$ vertices of degrees between $d/2$ and $5d/2$.  Denote the set of these vertices by $A$.  Since $|S|\le \vertices^{9/7}$, we have $|A|=O(\vertices^{9/7}/d)$.  

Let us calculate the number of paths of length 2 in $S$ having the middle vertex in $A$.   On one hand,  this number is at least $P\vertices^{3/7}$.  On the other hand, it is at most $O(d^2 |A|) = O(d\vertices^{9/7})$.  Thus, $P = O(d\vertices^{6/7})$.  Since $a$ is determined as an end-point of $j$, the contribution is $O(d\vertices^{6/7})(\vertices^{-3/14}/d)^2 = O(1)$, as $d\ge \vertices^{3/7}$.
\itemend

Finally, $j$ may be the last edge of the triangle.  We know that $\deg a> d$, hence, either $\vertices^{-3/14} - g_i(S,M) = 0$, or $\nu(S,M)\le 2\vertices^{3/7}$, in which case, there are $O(\vertices^{3/7})$ choices of $a$ satisfying the condition.  Hence, the left hand side of~\refeqn{partial2} is $O(\vertices^{3/7}) (\vertices^{-3/14})^2 = O(1)$.

%A careful analysis shows that there is a subset $K$ of conditions with $|K|=O(1)$ and $g_i(S,M) - g_i(S\cup\{j\}, M)= 0$ for all $i\notin K$ and $M\in\cert\setminus F(S,j)$.  (In two words: in all three subcases, the value of $d$ may be determined with some ambiguity from $j$.) 
If $g_i(S,M) - g_i(S\cup\{j\}, M)\ne 0$, then, in the first three cases, the value of $d$, up to a small ambiguity, may be determined from the degree of one of the end-points of $j$.  Hence, the set $K=K(S,j)$, as stated previously in the proof, exists.
\pfend

%Automatically, we obtain that the quantum query complexity of the {\em triangle sum} problem is $\widetilde{\Omega}(n^{9/7})$.  Thus, any quantum query algorithm, willing to improve the $O(n^{9/7})$ bound for the triangle detection problem, will have to take differences between the triangle detection and triangle sum problems into consideration.

\section{Lower Bound}
\label{sec:lower}
This section is devoted to finishing the proof of \rf(thm:certificates) and proving \rf(thm:boundedly).  The results are strongly connected: In the second one we prove a stronger statement from stronger premisses.  As a consequence, the proofs also have many common elements.
In the proofs, we define a number of matrices and argue about their spectral properties.  For convenience, we describe the main parameters of the matrices, such as the labelling of their rows and columns, as well as their mutual relationships in one place, \refsec{outline}.  In \refsec{common}, we state the intermediate results important to both Theorems~\ref{thm:certificates} and~\ref{thm:boundedly}.  In~\refsec{boundedly}, we finish the proof of \refthm{boundedly}.  In \refsec{fourier}, we recall the definition and main properties of the Fourier basis, and define the important notion of the Fourier bias.  Finally, in \refsec{general}, we prove \refthm{certificates}.

\subsection{Outline}
\label{sec:outline}
Let us briefly outline how Theorems~\ref{thm:certificates} and~\ref{thm:boundedly} are proven.  Let $\cert$ denote the certificate structure.  Let $\alpha_S(M)$ satisfy~\refeqn{learningDual}, and be such that~\refeqn{alphaObjective} equals the learning graph complexity of $\cert$.  We define an explicit function $f\colon\cD\to \{0,1\}$ with $\cD\subseteq [q]^\vars$ having the objective value~\refeqn{alphaObjective} of program~\refeqn{learningDual} as a lower bound on its quantum query complexity.  The latter is proven using the adversary bound, \refthm{adv}.  For that, we define a number of matrices, as illustrated in \reffig{matrices}.
%And the adversary matrix $\Gamma$ is constructed from the values of $\alpha_S(M)$.

\mycommand{tGamma}{\widetilde{\Gamma}}
\mycommand{tG}{\widetilde{G}}
\mycommand{hG}{\widehat{G}'}
\mycommand{hGamma}{\widehat{\Gamma}'}

\begin{figure}[tbh]
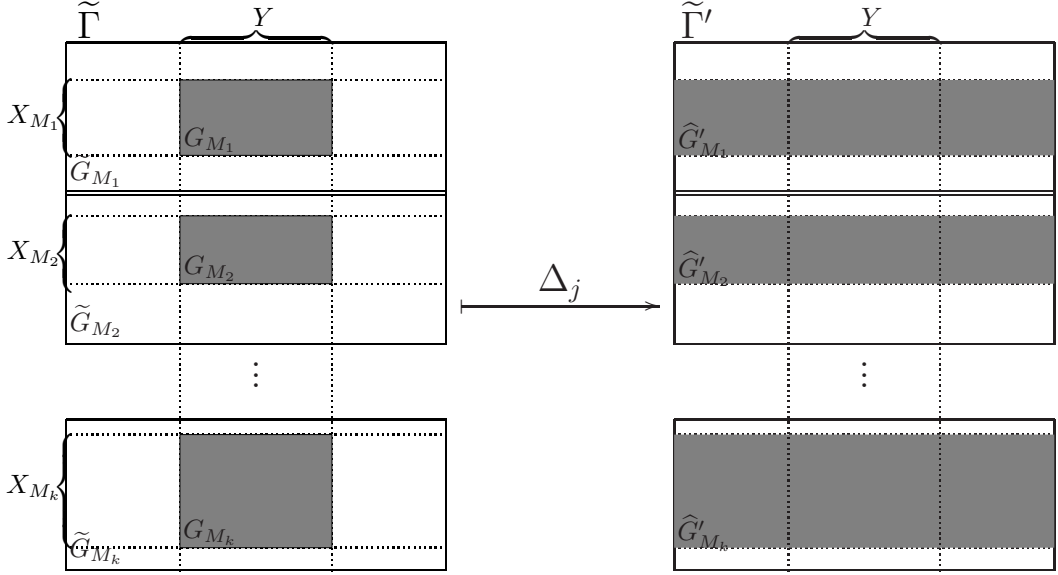

\release{
\[
\xy %<0.12cm,0cm>:
="st"
 @={+(0,0), -(0,19.7),-(0,0.6), -(0,19.7), -(0,10), -(0,20)},
 @@{; p +(50,0) **@{-}},
 "st" 
 @={+(0,0), +(50,0)},
 @@{; p -(0,40) **@{-}, -(0,10); p-(0,20) **@{-}},
 "st" 
 @={+(15,0)="d", +(20,0)="d2"},
 @@{; p -(0,70) **@{.}},
 "st"+(25,0.4)."d" ; "d2" **\frm{^\}}; +(1,3) *{Y},
 "st" 
 @={-(0,5)="a", -(0,10)="a2", "st"-(0,23)="b", -(0,9)="b2", "st"-(0,52)="c", -(0,15)="c2"},
 @@{; p +(50,0) **@{.}},
 "st"+(3,3) *\txt{\Large$\tGamma$},
 "a"-(0.2,5)."a" ; "a2" **\frm{\{} ; -(4,0) *{X_{M_1}},
 "b"-(0.2,5)."b" ; "b2" **\frm{\{} ; -(4,0) *{X_{M_2}},
 "c"-(0.2,7)."c" ; "c2" **\frm{\{} ; -(4,0) *{X_{M_k}},
 "a"+(15,0); "a2"+(35,0).p *[F*:gray]\frm{-},
 "b"+(15,0); "b2"+(35,0).p *[F*:gray]\frm{-},
 "c"+(15,0); "c2"+(35,0).p *[F*:gray]\frm{-},
 "st"+(4,-17) *{\tG_{M_1}}, -(0,20) *{\tG_{M_2}}, -(0,30) *{\tG_{M_k}},
 "a2"+(19,2) *{G_{M_1}}, "b2"+(19,2) *{G_{M_2}}, "c2"+(19,2) *{G_{M_k}},
 "st"+(25,-43) *\txt{\Large$\vdots$},
"st"+(80,0)="st",
 @={+(0,0), -(0,19.7),-(0,0.6), -(0,19.7), -(0,10), -(0,20)},
 @@{; p +(50,0) **@{-}},
 "st" 
 @={+(0,0), +(50,0)},
 @@{; p -(0,40) **@{-}, -(0,10); p-(0,20) **@{-}},
 "st" 
 @={-(0,5)="a", -(0,10)="a2", "st"-(0,23)="b", -(0,9)="b2", "st"-(0,52)="c", -(0,15)="c2"},
 @@{; p +(50,0) **@{.}},
 "st"+(3,3) *\txt{\Large$\tGamma'$},
% "a"-(0.2,5)."a" ; "a2" **\frm{\{} ; -(4,0) *{X_{M_1}},
% "b"-(0.2,5)."b" ; "b2" **\frm{\{} ; -(4,0) *{X_{M_2}},
% "c"-(0.2,7)."c" ; "c2" **\frm{\{} ; -(4,0) *{X_{M_\ell}},
 "a2"+(50,0)."a" *[F*:gray]\frm{},
 "b2"+(50,0)."b" *[F*:gray]\frm{},
 "c2"+(50,0)."c" *[F*:gray]\frm{},
 "st" 
 @={+(15,0)="d", +(20,0)="d2"},
 @@{; p -(0,70) **@{.}},
 "st"+(25,0.4)."d" ; "d2" **\frm{^\}}; +(1,3) *{Y},
 "a2"+(4,2) *{\hG_{M_1}}, "b2"+(4,2) *{\hG_{M_2}}, "c2"+(4,2) *{\hG_{M_k}},
 "st"+(25,-43) *\txt{\Large$\vdots$},
"st"-(28,35); p+(13,3)*\txt{\Large$\Delta_j$} ;
\ar @{|->} +(26,0)
\endxy
\]

\caption{The relationships between matrices used in \refsec{lower}.  The parts marked in grey form the matrix $\Gamma$ on the left, and $\hGamma$ on the right.  Note that they are {\em not} submatrices of $\tGamma$ and $\tGamma'$, respectively: 
They have additional multiplicative factor as specified in~\refeqn{GammaElements} and~\refeqn{hGM}.
}}
\label{fig:matrices}
\end{figure}

\paragraph{Matrix $\tGamma$} 
At first, we construct a matrix $\tGamma$ satisfying the following properties.  Firstly, it has rows labelled by the elements of $[q]^\vars\times\cert$, and columns labelled by the elements of $[q]^\vars$.  Thus, if we denote $\cert = \{\marked_1,\dots,\marked_k\}$, the matrix $\tGamma$ has the following form
\begin{equation}
\label{eqn:tGamma}
\tGamma = 
\begin{pmatrix}
\tG_{\marked_1}\\ \tG_{\marked_2}\\\vdots\\ \tG_{\marked_{k}}
\end{pmatrix},
\end{equation}
where each $\tG_\marked$ is an $[q]^\vars\times[q]^\vars$-matrix.  Next, $\|\tGamma\|$ is at least the objective value~\refeqn{alphaObjective}.  And finally, for each $j\in[\vars]$, there exists $\tGamma'$ such that $\tGamma\dar\tGamma'$ and $\|\tGamma'\|\le 1$.  
%Here, $\dar$ is as described after \reflem{gamma}.  
The matrix $\tGamma'$ has a decomposition into blocks $\tG'_M$ similar to~\refeqn{tGamma}.

Thus, $\tGamma$ has a good value of~\refeqn{adv}.  But, we cannot use it, because it is not an adversary matrix: It uses all possible inputs as labels of both rows and columns.  However, due to the specific way $\tGamma$ is constructed, we will be able to transform $\tGamma$ into a true adversary matrix $\Gamma$ such that the value of~\refeqn{adv} is still good.  Before we describe how we do it, let us outline the definition of the function $f$.

\mycommand{sz}{\ell}
\mycommand{fii}{f^{-1}}
\paragraph{Defining the function}
Let $M$ be an element of the certificate structure $\cert$.  Let $A_{M}^{(1)},\dots,A_{M}^{(\sz(M))}$ be all the inclusion-wise minimal elements of $M$.  (In a boundedly-generated certificate structure, $M$ has only one inclusion-wise minimal element $\certM$.)  For each $A_{M}^{(i)}$, we choose an orthogonal array $T_{M}^{(i)}$ of length $|A_{M}^{(i)}|$ over the alphabet $[q]$, and define
\begin{equation}
\label{eqn:Xm}
X_M = \sfig{x\in[q]^\vars\mid \mbox{$x_{A_{M}^{(i)}}\in T_{M}^{(i)}$ for all $i\in[\sz(M)]$} }.
\end{equation}
The orthogonal arrays are chosen so that $X_M$ is non-empty and satisfies the following {\em orthogonality property}:
\begin{equation}
\label{eqn:orthogonality}
\forall S\in 2^{[\vars]}\setminus M\;\; \forall z\in [q]^S\;:\; \abs|\strut\{x\in X_M \mid x_S = z \}| = |X_M|/q^{|S|}.
\end{equation}
For boundedly-generated certificate structures, this property is satisfied automatically.

The set of positive inputs is defined by $\fii(1) = \bigcup_{M\in\cert} X_M$.  The set of negative inputs is defined by
\begin{equation}
\label{eqn:fii}
\fii(0) = \sfig{y\in[q]^\vars\mid \mbox{$y_{A_{M}^{(i)}}\notin T_{M}^{(i)}$ for all $M\in\cert$ and $i\in[\sz(M)]$} }.
\end{equation}
It is easy to see that $f$ has $\cert$ as its certificate structure.  The parameters will be chosen so that $|\fii(0)|=\Omega(q^\vars)$.  %If $\cert$ is boundedly generated, the function $f$ is total.

\paragraph{Remaining matrices}
Let us define $X = \{(x,M)\in[q]^\vars\times\cert \mid x\in X_M\}$ and $Y = \fii(0)$.  The matrix $\Gamma$ is an $X\times Y$ matrix defined by
\begin{equation}
\label{eqn:GammaElements}
\Gamma\elem[(x,M), y] = \sqrt{\frac{q^\vars}{|X_M|}}\; \tGamma\elem[(x,M),y].
\end{equation}
Thus, $\Gamma$ consists of blocks $G_M$, like in~\refeqn{tGamma}, where $G_M = \sqrt{q^\vars/|X_M|}\; \tG_M\elem[X_M,Y]$.  (The latter notation stands for the submatrix formed by the specified rows and columns).  We also show that $\|\Gamma\|$ is not much smaller than $\|\tGamma\|$.

The matrix $\Gamma'$ is obtained similarly from $\tGamma'$.  It is clear that $\tGamma\dar\tGamma'$ implies $\Gamma\dar\Gamma'$.  We show that the norm of $\Gamma'$ is small by showing that $\|\hGamma\| = O(\|\tGamma'\|)$ where $\hGamma$ is an $X\times [q]^\vars$-matrix with 
\[
\hGamma\elem[(x,M),y] = \sqrt{\frac{q^\vars}{|X_M|}}\; \tGamma'\elem[(x,M),y].
\]
As $\Gamma'$ is a submatrix of $\hGamma$ and $\|\tGamma'\|\le 1$, we obtain that $\|\Gamma'\| = O(1)$ as required.  We denote the blocks of $\hGamma$ by $\hG_M$.  That is, 
\begin{equation}
\label{eqn:hGM}
\hG_M = \sqrt{\frac{q^\vars}{|X_M|}}\; \tG_M'\elem[X_M,{[q]^\vars}].
\end{equation}

\subsection{Common Parts of the Proofs}
\label{sec:common}

%We apply the adversary bound~\refthm{adv}.  The construction is similar to the one in~\cite{spalek:kSumLower}. 
Let $e_0,\dots,e_{q-1}$ be an orthonormal basis of $\C^q$ such that $e_0=1/\sqrt{q}(1,\dots,1)$.  Denote $E_0 = e_0e_0^*$ and $E_1 = \sum_{i>0} e_ie_i^*$.  These are $q\times q$ matrices. All entries of $E_0$ are equal to $1/q$, and the entries of $E_1$ are given by
\begin{equation}
\label{eqn:E1entries}
E_1\elem[x,y] = \begin{cases}
1-1/q,& x=y;\\
-1/q,& x\ne y.
\end{cases}
\end{equation}
For a subset $S\subseteq[\vars]$, let $E_S$ denote $\bigotimes_{j\in[\vars]} E_{s_j}$ where $s_j=1$ if $j\in S$, and $s_j=0$ otherwise.  These matrices are orthogonal projectors:
\begin{equation}
\label{eqn:ESOrthogonal}
E_SE_{S'} =
\begin{cases}
E_S,& S=S'\\
0,&\text{otherwise.}
\end{cases}
\end{equation}

We define the matrices $\tG_M$ from~\refeqn{tGamma} by
\begin{equation}
\label{eqn:tG}
\tG_\marked = \sum_{S\subseteq[\vars]} \alpha_S(\marked) E_S,
\end{equation}
where $\alpha_S(M)$ are as in~\refeqn{learningDual}. 

\begin{lem}
\label{lem:GammaNorm}
If $\tGamma$ and $\Gamma$ are defined as in \refsec{outline}, all $X_M$ satisfy the orthogonality property~\refeqn{orthogonality} and $|Y|=\Omega(q^\vars)$, then 
\begin{equation}
\label{eqn:GammaNorm}
\|\Gamma\| = \Omega\sD[\sqrt{\sum_{M\in\cert} \alphaz(M)^2}].
\end{equation}
\end{lem}

\mycommand{summa}{\mathrm{s}}
\mycommand{gm}{G_\marked}
\mycommand{tgm}{\tG_\marked}
\pfstart
Recall that $\gm = \sqrt{q^\vars/|X_M|}\tgm\elem[X_M,Y]$, hence, by~\refeqn{tG}:
\[%begin{equation}
%\label{eqn:tG2}
\gm = \sqrt{\frac{q^\vars}{|X_M|}}\; \alphaz(M) E_0^{\otimes \vars}\elem[X_M,Y] + \sqrt{\frac{q^\vars}{|X_M|}} 
\sum_{S\ne\emptyset} \alpha_S(\marked) E_S\elem[X_M,Y].
\]%end{equation}
Let us calculate the sum $\summa(G_M)$ of the entries of $G_M$.
In the first term, each entry of $E_0^{\otimes \vars}$ equals $q^{-\vars}$.  There are $|X_M|$ rows and $|Y|$ columns in the matrix, hence, the sum of the entries of the first term is $\sqrt{|X_M|/q^\vars}\;|Y|\alphaz(M)$.  

We claim that, in the second term, $\summa\s[\strut \alpha_S(M){E_S\elem[X_M,Y]}]=0$ for all $S\ne \emptyset$.  Indeed, if $S\in M$, then $\alpha_S(M)=0$ by~\refeqn{alphaZero}.  Otherwise,
\[
\summa(E_S\elem[X_M,Y]) = \sum_{y\in Y}\sum_{x\in X_M} E_S\elem[x, y] = 
q^{|S|-\vars} \sum_{y\in Y}\sum_{x\in X_M} E_1^{\otimes|S|} \elem[x_S, y_S] = 
\frac{|X_M|}{q^{\vars}} \sum_{y\in Y} \sum_{z\in [q]^S} E_1^{\otimes|S|} \elem[z, y_S] = 0.
\]
%\[
%\s(E_S\elem[X_M,Y]) = \sum_{y\in Y} \s(E_S\elem[X_M,\{y\}]) = \sum_{y\in Y} \frac{|X_M|}{q_S} \s(E_1^{\otimes|S|}\elem[{[q]}^S, \{{y\elem[S]}\}] ) = 0.
%\]
(On the third step, the orthogonality condition~\refeqn{orthogonality} is used.  On the last step, we use that the sum of the entries of every column of $E_1^{\otimes k}$ is zero if $k>0$.)  Summing up,
\[
\summa(\gm) = \sqrt{\frac{|X_M|}{q^\vars}}\;|Y|\alphaz(M).
\]

We are now ready to estimate $\norm|\Gamma|$.  Define two unit vectors $u\in\R^X$ and $v\in\R^Y$ by 
\[
u\elem[(x,M)] = \frac{\alphaz(M)}{\sqrt{|X_M|\sum_{M\in\cert} \alphaz(M)^2}}\qquad\text{and}\qquad
v\elem[y] = \frac1{\sqrt{|Y|}}
\]
for all $(x,M)\in X$ and $y\in Y$.  Then,
%\begin{equation}
%\label{eqn:GammaNorm}
\[
\norm|\Gamma| \ge u^*\Gamma v = \frac{\sum_{M\in\cert} \alphaz(M)\summa(G_M)}{\sqrt{|X_M|\;|Y|\sum_{M\in\cert} \alphaz(M)^2}} = \sqrt{\frac{|Y|}{q^\vars}\sum_{M\in\cert} \alphaz(M)^2} = \Omega\sD[\sqrt{\sum_{M\in\cert} \alphaz(M)^2}] .\qedhere
\]%end{equation}
%by the assumption on the size of $|Y|$.
\pfend

In the remaining part of this section, we define the transformation $\tGamma\dar\tGamma'$ and state some of the properties of $\tGamma'$ that will be used in the subsequent sections.
Using~\refeqn{E1entries}, we can define the action of $\Delta$ on $E_0$ and $E_1$ by 
\[%begin{equation}
%\label{eqn:delta}
E_0\stackrel{\Delta}{\longmapsto} E_0\qquad\text{and}\qquad E_1\stackrel{\Delta}{\longmapsto} -E_0.
\]%end{equation}
We define $\tGamma'$ by applying this transformation to $E_0$ and $E_1$ in the $j$th position in the tensor product of~\refeqn{tG}.  The result is again a matrix of the form~\refeqn{tGamma}, but with each $\tgm$ replaced by
\begin{equation}
\label{eqn:gmPrim}
\tgm' = \sum_{S\subseteq[\vars]} \beta_S(M)E_S,
\end{equation}
where $\beta_S(M) = \alpha_S(M) - \alpha_{S\cup\{j\}}(M)$.  In particular, $\beta_S(M)=0$ if $j\in S$ or $S\in M$.  Thus,
\begin{equation}
\label{eqn:GammaPrim}
(\tGamma')^*\tGamma' = \sum_{M\in\cert} (\tgm')^*\tgm' = \sum_{S\in 2^{[\vars]}} \sB[\sum_{M\in\cert} \beta_S(M)^2] E_S.
\end{equation}
In particular, we obtain from~\refeqn{alphaOne} that $\|\tGamma'\|\le 1$.

\subsection{Boundedly-Generated Certificate Structures}
\label{sec:boundedly}
In this section, we finish the proof of \refthm{boundedly}.  In the settings of the theorem, the orthogonal arrays $T_M^{(i)}$ in~\refeqn{Xm} are already specified.  Since each $M\in\cert$ has only one inclusion-wise minimal element $\certM$, we drop all upper indices $(i)$ in this section.  

From the statement of the theorem, we have $|X_M| = q^{\vars-1}$, in particular, they are non-empty.  Also, $X_M$ satisfy the orthogonality property~\refeqn{orthogonality}, and, by~\refeqn{fii}, we have
\begin{equation}
\label{eqn:Ysize}
|Y| = \absC|[q]^\vars\setminus\bigcup_{M\in\cert} X_M| \ge q^\vars - \sum_{M\in\cert} |X_M| = q^\vars - |\cert|q^{\vars-1} \ge \frac{q^\vars}2.
\end{equation}
Thus, the conditions of \reflem{GammaNorm} are satisfied, and~\refeqn{GammaNorm} holds.

Recall from \refsec{outline} that in order to estimate $\|\Gamma'\|$ we consider the matrix $\hGamma$.  The matrix $\Gamma'$ is a submatrix of $\hGamma$, hence, it suffices to estimate $\|\hGamma\|$.  Let $k=\max_{M\in\cert} |\certM|$.  By \refdefn{boundedly}, $k=O(1)$.  

Fix an arbitrary order of the elements in each $\certM = \{a_{M,1},\dots,a_{M,|\certM|}\}$, and let $L_{M,i}$, where $M\in\cert$ and $i\in[k]$, be subsets of $2^{[\vars]}$ satisfying the following properties:
\itemstart
\item for each $M$, the set $2^{[\vars]}\setminus M$ is the disjoint union $L_{M,1}\sqcup\cdots\sqcup L_{M,k}$;
\item for each $M$ and each $i\le |\certM|$, all elements of $L_{M,i}$ omit $a_{M,i}$;
\item for each $M$ and each $i$ such that $|\certM|<i\le k$, the set $L_{M,i}$ is empty.
\itemend
Recall that, if $S\subseteq[\vars]$ and $(s_j)$ is the corresponding characteristic vector, $E_S = \bigotimes_{j\in[\vars]} E_{s_j}$.  The main idea behind defining $L_{M,i}$s is as follows.

\begin{clm}
\label{clm:LMi}
If $S,S'\in L_{M,i}$, then
\[
(E_S\elem[X_M,{[q]^\vars}])^* (E_{S'}\elem[X_M,{[q]^\vars}]) = 
\begin{cases}
E_S/q,& S=S';\\
0,&\text{\rm otherwise.}
\end{cases}
\]
\end{clm}

\pfstart
If we strike out the $a_{M,i}$th element in all elements of $X_M$, we obtain $[q]^{\vars-1}$ by the definition of an orthogonal array.  All elements of $L_{M,i}$ omit $a_{M,i}$, hence, $E_S$ has $E_0$ in the $a_{M,i}$th position for all $S\in L_{M,i}$.  Thus, the $a_{M,i}$th entries of $x$ and $y$ has no impact on the value of $E_S\elem[x,y]$.  

Let $(s_j)$ and $(s'_j)$ be the characteristic vectors of $S$ and $S'$.  Then,
\[
E_S\elem[X_M,{[q]^\vars}] = \sC[\bigotimes_{j\in[\vars]\setminus\{a_{M,i}\}} E_{s_j}]\otimes \frac{e_0^*}{\sqrt{q}}.
\]
(Here $e_0^*$ is on the $a_{M,i}$th element of $[q]^\vars$.)  Similarly for $S'$, and the claim follows from~\refeqn{ESOrthogonal}.
\pfend

For each $M$, decompose $\tgm'$ from~\refeqn{gmPrim} into $\sum_{i\in[k]} \tG'_{M,i}$, where
\[
\tG'_{M,i} = \sum_{S\in L_{M,i}} \beta_S(M) E_S.
\]
Define similarly to \refsec{outline}, 
\[
\hG_{M,i} = \sqrt{\frac{q^\vars}{|X_M|}}\; \tG'_{M,i}\elem[X_M,{[q]^\vars}] = \sqrt{q}
\sum_{S\in L_{M,i}} \beta_S(M) E_S\elem[X_M,{[q]^\vars}],
\]
and let $\hGamma_i$ be the matrix consisting of $\hG_{M,i}$, for all $M\in\cert$, stacked one on another like in~\refeqn{tGamma}.  Then, $\hGamma = \sum_{i\in [k]} \hGamma_i$.
We have
\[%begin{equation}
%\label{eqn:urav1}
(\hGamma_i)^*\hGamma_i = \sum_{M\in \cert} (\hG_{M,i})^* \hG_{M,i} = \sum_{M\in\cert} \sum_{S\in L_{M,i}} \beta_S(M)^2E_S,
\]%end{equation}
by \refclm{LMi}.  Similarly to~\refeqn{GammaPrim}, we get $\|\hGamma_i\|\le 1$.  By the triangle inequality, $\normS|\hGamma|\le k$, hence, $\norm|\Gamma'| \le k=O(1)$.  Combining this with~\refeqn{GammaNorm}, and using \refthm{adv}, we obtain the necessary lower bound.  This finishes the proof of \refthm{boundedly}.

\subsection{Fourier Basis and Bias}
\label{sec:fourier}
\mycommand{group}{Z}
\mycommand{fb}{\chi}
\def \fnorm|#1|{\|#1\|_{\mathrm{u}}}
In \refsec{common}, we defined $e_i$ as an arbitrary orthonormal basis satisfying the requirement that $e_0$ has all its entries equal to $1/\sqrt{q}$.  In the next section, we will specify a concrete choice for $e_i$.  Its construction is based on the Fourier basis we briefly review in this section.

Let $p$ be a positive integer, and $\Z_p$ be the cyclic group of order $p$, formed by the integers modulo $p$. Consider the complex vector space $\C^{\Z_p}$.  The vectors $(\fb_a)_{a\in \Z_p}$, defined by $\fb_a\elem[b] = \ee^{2\pi\ii ab/p}/\sqrt{p}$, form its orthonormal basis.  Note that the value of $\fb_a\elem[b]$ is well-defined because $\ee^{2\pi\ii}=1$.

If $U\subseteq\Z_p$, then the {\em Fourier bias}~\cite{tao:additive} of $U$ is defined by
\begin{equation}
\label{eqn:fnorm}
\fnorm|U|  % = \abs|\max_{a\in\Z_p\setminus\{0\}} \summa(\fb_a\elem[U]) |
= \frac{1}{p}\; \absC| \max_{a\in\Z_p\setminus\{0\}} \sum_{u\in U} \ee^{2\pi\ii au/p} |.
\end{equation}
It is a real number between 0 and $|U|/p$.  In the next section, we will need the following result stating the existence of sets with small Fourier bias and arbitrary density.

\begin{thm}
\label{thm:FourierBias}
For any real $0<\delta<1$, it is possible to construct $U\subseteq \Z_q$ such that $|U|\sim\delta q$, $\fnorm|U| = O(\polylog(q)/\sqrt{q})$ and $q$ is arbitrary large.  In particular, $\fnorm|U| = o(1)$.
\end{thm}

For instance, one may prove a random subset satisfies these properties with high probability~\cite[Lemma~4.16]{tao:additive}.  There also exist explicit constructions~\cite{gillespie:fourierBias}.

\subsection{General Certificate Structures}
\label{sec:general}
In this section, we finish the proof of \refthm{certificates}.  There are two main reasons why it is not possible to prove a general result like~\refthm{boundedly} for arbitrary certificate structures.

One counterexample is provided by \rf(prp:learningCollisionLower) and the discussion after \rf(defn:setEquality):  The learning graph complexity of the hidden shift certificate structure is $\Theta(\vars^{1/3})$, but the quantum query complexity of the hidden shift problem is $O(\log \vars)$.
The proof of \refsec{boundedly} cannot be applied here, because $k$ in the decomposition of $\tG'_M$ into $\sum_{i\in[k]} \tG'_{M,i}$ would not be bounded by a constant.  We solve this by considering much ``thicker'' orthogonal arrays $T_M^{(i)}$.

Next, the orthogonality property~\refeqn{orthogonality} is not satisfied automatically for general certificate structures.  For instance, assume $A_{M}^{(1)} = \{1,2\}$, $A_{M}^{(2)}=\{2,3\}$, and the orthogonal arrays are given by the conditions $x_1=x_2$ and $x_2=x_3$, respectively.  Then, for any input $x$ satisfying both conditions, we have $x_1 = x_3$, and the orthogonality condition fails for $S = \{1,3\}$.

The problem in the last example is that the orthogonal arrays are not independent because $A_{M}^{(1)}$ and $A^{(2)}_{M}$ intersect.  We cannot avoid that $A_{M}^{(i)}$s intersect, but we still can have $T_M^{(i)}$s independent by defining them on independent parts of the input alphabet.  

More formally, let $\sz = \max_{M\in\cert} \sz(M)$, where $\sz(M)$ is defined in \refsec{outline} as the number of inclusion-wise minimal elements of $M$.  We define the input alphabet as $\group = \Z_p^\sz$ for some $p$ to be defined later.  Hence, the size of the alphabet is $q = p^\sz$.  

Let $Q_M^{(i)}$ be an orthogonal array of length $|A_M^{(i)}|$ over the alphabet $\Z_p$.  We will specify a concrete choice in a moment.  From $Q_M^{(i)}$, we define $T_M^{(i)}$ in~\refeqn{Xm} by requiring that the $i$th components of the elements in the sequence satisfy $Q_M^{(i)}$.  The sets $X_M$ are defined as in~\refeqn{Xm}.  We additionally define
\[
X_{M}^{(i)} = \sfig{x\in \Z_p^\vars \mid x_{A_{M}^{(i)}}\in Q_{M}^{(i)} },
\]
for $i\le\sz(M)$, and $X_{M}^{(i)} = \Z_p^\vars$ otherwise.  Note that $X_M = \prod_{i=1}^\sz X_{M}^{(i)}$ in the sense that, for each sequence $x^{(i)}\in X_{M}^{(i)}$ with $i=1,\dots,\sz$, there is a corresponding element $x\in X_M$ with $x_j = (x_j^{(1)}, \dots, x_j^{(\sz)})$.

Now we make our choice for $Q_M^{(i)}$.  Let $U\subseteq \Z_p$ be a set with small Fourier bias and some $\delta = |U|/p$ that exists due to \refthm{FourierBias}.  We define $Q_M^{(i)}$ as consisting of all $x\in \Z_p^{A_{M}^{(i)}}$ such that the sum of the elements of $x$ belongs to $U$.  With this definition,
\begin{equation}
\label{eqn:density}
|X_M^{(i)}| = \delta p^\vars.
\end{equation}
Hence, there are exactly $\delta q^\vars$ elements $x\in\group^\vars$ such that $x_{A_M^{(i)}}\in T_M^{(i)}$.  If we let $\delta = 1/(2\sz|\cert|)$, a calculation similar to~\refeqn{Ysize} shows that $|Y|\ge q^\vars/2$.  Also, by considering each $i\in[\sz]$ independently, it is easy to see that all $X_M$ satisfy the orthogonality condition.  Thus, \reflem{GammaNorm} applies, and~\refeqn{GammaNorm} holds.

Now it remains to estimate $\|\Gamma'\|$, and it is done by considering matrix $\hGamma$ as described in \refsec{outline}, and performed once in \refsec{boundedly}.  If $\tGamma'=0$, then also $\Gamma'=0$, and we are done.  Thus, we further assume $\tGamma'\ne 0$.
\mycommand{tB}{\widetilde{B}}
\mycommand{hB}{\widehat{B}}
Recall that $(\fb_a)_{a\in\Z_p}$ denotes the Fourier basis of $\Z_p$.  The basis $e$ is defined as the Fourier basis of $\C^\group$.  It consists of the elements of the form $e_a = \bigotimes_{i=1}^\sz \fb_{a^{(i)}}$ where $a=(a^{(i)})\in\group$.  Note that $e_0$ has the required value, where $0$ is interpreted as the neutral element of $\group$.  

If $v = (v_j) = (v_j^{(i)})\in \group^\vars$, we define $e_v = \bigotimes_{j=1}^\vars e_{v_j}$, and 
$v^{(i)}\in \Z_p^\vars$ as $(v_1^{(i)},\dots,v_n^{(i)})$.  Also, for $w = (w_j)\in \Z_p^\vars$, we define
$\fb_w = \bigotimes_{j=1}^\vars \fb_{w_j}$.

Fix an arbitrary $M\in\cert$.  Let $\tB_M = (\tG'_M)^*\tG'_M$ and $\hB_M = (\hG_M)^*\hG_M$.  We aim to show that
\begin{equation}
\label{eqn:difference}
\|\tB_M - \hB_M\| \to 0\quad \mbox{as}\quad p\to\infty,
\end{equation}
because this implies
\[
 \|(\tGamma')^*\tGamma' - (\hGamma)^*\hGamma \| = \normB|\sum_{M\in\cert} (\tB_M - \hB_M)  | 
\le \sum_{M\in\cert} \| \tB_M - \hB_M \| \to 0
\]
as $p\to\infty$.  As $\|\tGamma'\|>0$, this implies that $\|\Gamma'\| \le 2 \|\tGamma'\|$ for $p$ large enough, and together with~\refeqn{GammaNorm} and \refthm{adv}, this implies \refthm{certificates}.

From~\refeqn{gmPrim}, we conclude that the eigenbasis of $\tB_M$ consists of the vectors $e_v$, with $v\in\group^\vars$, defined above.  In order to understand $\hB_M$ better, we have to understand how $e_v\elem[X_M]$ behave.  We have
\begin{equation}
\label{eqn:evInner}
(e_v\elem[X_M])^*(e_{v'}\elem[X_M]) = \prod_{i=1}^\sz (\fb_{v^{(i)}}\elem[X_{M}^{(i)}])^*(\fb_{v'^{(i)}}\elem[X_{M}^{(i)}]).
\end{equation}
Hence, it suffices to understand the behaviour of $\fb_w\elem[X_M^{(i)}]$.  For $w\in \Z_p^\vars$, $A\subseteq [\vars]$ and $c\in\Z_p$, we write $w+cA$ for the sequence $w'\in \Z_p^\vars$ defined by
\[
w'_j = \begin{cases}
w_j + c,& j\in A;\\
w_j,&\mbox{otherwise.}
\end{cases}
\]
In this case, we say that $w$ and $w'$ are obtained from each other by a {\em shift on $A$}.

\begin{clm}
\label{clm:zeta}
Assume $w,w'\in\Z_p^\vars$, and let $\xi = (\fb_w\elem[X_M^{(i)}])^*(\fb_{w'}\elem[X_M^{(i)}])$.  If $w=w'$, then $\xi = \delta$.  If $w\ne w'$, but $w$ can be obtained from $w'$ by a shift on $A_M^{(i)}$, then $|\xi|\le \fnorm|U|$.  Finally, if $w$ cannot be obtained from $w'$ by a shift on $A_M^{(i)}$, then $\xi = 0$.
\end{clm}

\pfstart
Arbitrary enumerate the elements of $U = \{u_1,\dots,u_m\}$ where $m=\delta p$.  Denote, for the sake of brevity, $A = A_M^{(i)}$.  Consider the decomposition $X_{M}^{(i)} = \bigsqcup_{k=1}^m X_k$, where
\[
X_k = \sfig{w\in \Z_p^\vars \mid \sum\nolimits_{j\in A} w_j = u_k} .
\]
\mycommand{ww}{\bar{w}}
Fix an arbitrary element $a\in A$ and denote $\ww = w-w_aA$ and $\ww' = w'-w'_aA$.  In both of them, $\ww_a = \ww'_a=0$, and by an argument similar to \refclm{LMi}, we get that
\begin{equation}
\label{eqn:chivInner}
(\fb_{\ww}\elem[X_k])^* (\fb_{\ww'}\elem[X_k])=
\begin{cases}
1/p,& \ww = \ww';\\
0,& \text{otherwise.}
\end{cases}
\end{equation}

If $x\in X_k$, then
\begin{multline*}
\fb_w\elem[x] = \prod_{j=1}^\vars \fb_{w_j}\elem[x_j] = 
\frac1{\sqrt{p^\vars}}\exp\skC[\frac{2\pi\ii}{p} \sum_{j=1}^\vars w_jx_j]
\\ = \frac1{\sqrt{p^\vars}}\exp\skC[\frac{2\pi\ii}{p}{\sB[ \sum_{j=1}^\vars \ww_j x_j + w_a\sum_{j\in A} x_j]}] = 
\exp\sB[\frac{2\pi\ii\; w_au_k}{p}]
\fb_{\ww}\elem[x] .
\end{multline*}
Hence,
\begin{equation}
\label{eqn:chivInner2}
(\fb_{w}\elem[X_M^{(i)}])^* (\fb_{w'}\elem[X_M^{(i)}])
= \sum_{k=1}^m (\fb_{w}\elem[X_k])^* (\fb_{w'}\elem[X_k])
= \sum_{k=1}^m 
\ee^{2\pi\ii (w'_a-w_a)u_k/p} 
%\exp\sB[\frac{2\pi\ii (w'_a-w_a)u_k}p ]
(\fb_{\ww}\elem[X_k])^* (\fb_{\ww'}\elem[X_k]).
\end{equation}
If $w'$ cannot be obtained from $w$ by a shift on $A$, then $\ww\ne\ww'$ and~\refeqn{chivInner2} equals zero by~\refeqn{chivInner}.  If $w=w'$, then~\refeqn{chivInner2} equals $m/p = \delta$.  Finally, if $w'$ can be obtained from $w$ by a shift on $A$ but $w\ne w'$, then $\ww=\ww'$ and $w_a\ne w'_a$.  By~\refeqn{chivInner} and~\refeqn{fnorm}, we get that~\refeqn{chivInner2} does not exceed $\fnorm|U|$ in absolute value.
\pfend

Let $v\in\group^\vars$, and $S=\{j\in[\vars]\mid v_j\ne 0\}$.  Let $v'\in\group^\vars$, and define $S'$ similarly.  By~\refeqn{hGM}, \refeqn{gmPrim}, \refeqn{density} and~\refeqn{evInner}, we have
\begin{equation}
\label{eqn:hBmEntry}
e_v^* \hB_M e_{v'} = \frac{q^\vars\beta_S(M)\beta_{S'}(M)}{|X_M|}(e_v\elem[X_M])^*(e_{v'}\elem[X_M]) = \frac{\beta_S(M)\beta_{S'}(M)}{\delta^{\sz}} \prod_{i=1}^\sz (\fb_{v^{(i)}}\elem[X_{M}^{(i)}])^*(\fb_{v'^{(i)}}\elem[X_{M}^{(i)}]).
\end{equation}
By this and \refclm{zeta}, we have that
\begin{equation}
\label{eqn:hBmDiagonal}
e_v^* \hB_M e_{v} = \beta_S(M)^2 = e_v^* \tB_M e_{v}.
\end{equation}
Call $v$ and $v'$ {\em equivalent}, if $\beta_S(M)$ and $\beta_{S'}(M)$ are both non-zero and, for each $i\in[\sz]$, $v^{(i)}$ can be obtained from $v'^{(i)}$ by a shift on $A_M^{(i)}$.  By~\refeqn{hBmEntry} and \refclm{zeta}, we have that $e_v^* \hB_M e_{v'}$ is non-zero only if $v$ and $v'$ are equivalent.

For each $i\in[\sz]$, there are at most $|A_M^{(i)}|\le \vars$ shifts of $v^{(i)}$ on $A_M^{(i)}$ that have an element with an index in $A_M^{(i)}$ equal to 0.  By~\refeqn{alphaZero}, the latter is a necessary condition for $\beta_S(M)$ being non-zero.  Hence, for each $v\in \group^\vars$, there are at most $\vars^\sz$ elements of $\group^\vars$ equivalent to it.

Thus, in the basis of $e_v$s, the matrix $\hB_M$ has the following properties.  By~\refeqn{hBmDiagonal}, its diagonal entries equal the diagonal entries of $\tB_M$, and the latter matrix is diagonal.  Next, $\hB_M$ is block-diagonal with the blocks of size at most $\vars^\sz$.  By~\refeqn{hBmEntry} and \refclm{zeta}, the off-diagonal elements satisfy 
\[
|e_v^*\hB_M e_{v'}| \le \frac{\fnorm|U|}{\delta} \beta_S(M)\beta_{S'}(M),
\]
because $\fnorm|U|\le\delta$. 
Since the values of $\beta_S(M)$ do not depend on $p$, and by \refthm{FourierBias}, the off-diagonal elements of $\hB_M$ tend to zero as $p$ tends to infinity.  Since the sizes of the blocks also do not depend on $p$, the norm of $\tB_M-\hB_M$ also tends to 0, as required in~\refeqn{difference}.  This finishes the proof of \refthm{certificates}.

\section{Summary}
In this chapter, we introduced the notion of the certificate structure of a function, and analysed the complexity of quantum query algorithms that are based only on the certificate structure of the problem.  We developed the computational model of learning graphs that tightly characterise this complexity:  For any function with a fixed certificate structure, the corresponding learning graph can be converted into a quantum query algorithm, and some of the functions require this number of queries.

For symmetric functions, we developed an intuitive approach for the construction of learning graphs.  The task of loading the certificate is divided into a number of stages, and for each stage we calculate two parameters: its length and speciality.  The total complexity of the learning graph is a simple function of these quantities.
The analysis of the constructed learning graph requires just simple combinatorial tools.
The main intuition is the following ``hiding technique'': the variables of the certificate are hidden among the previously loaded dummy variables that mimic the structure of the certificate.  The more dummy variables are loaded, the smaller is the complexity of loading the certificate, but the loading of dummy variables also requires resources.  At the equilibrium point, the optimum is attained.

With the help of this approach, we constructed quantum query algorithms for the triangle and the associativity testing problems.  Additionally, we proved tight lower bounds on the quantum query complexity of the $k$-sum and triangle-sum problems.  The analysis is also purely combinatorial, although, more involved than the analysis of the corresponding upper bounds.

Of course, there are more possible applications out there.
We will only mention the problem of characterising the learning graph complexity of the subgraph detection problem.  So far, we have succeeded with the case of the triangle, and Ref.~\cite{lee:learningTriangle} mentions some upper bounds.

The main limitation of the results in this chapter stems from the same source as their handiness: They are bounded to certificate structures.  It is a smaller problem for the upper bounds, and, indeed, we will consider some algorithms beyond the certificate structure framework in the next chapter.  Because of this, we postpone further discussion of quantum query algorithms based on the dual adversary SDP till \rf(sec:kdistSummary).

For the lower bounds, this is a more important issue.
The technique of switching to $[q]^\vars\times [q]^\vars$ matrices relies heavily on the fact that the set of negative inputs is close to $[q]^\vars$, and that the set of positive inputs can be obtained by small alternations.  Thus, we require the assumptions of Theorems~\ref{thm:certificates} and~\ref{thm:boundedly}.  In particular, this approach fails immediately for the collision problem.  

Another issue is the size of the alphabet.  Some lower bound on the size of the alphabet is required, as can be seen from the element distinctness problem with $q<\vars$, but we expect that our requirements can be lowered.  For instance, we require $q>\vars^2$ for the element distinctness problem, but we know~\cite{ambainis:collisionLower} that the true bound is $q\ge\vars$.

Solving these problems would require different and more complicated techniques.  We are especially interested in constructing adversary lower bounds for the collision, the set equality and the $k$-distinctness problems.

%% file: _kdist.tex
\mycommand{stage}{s}

In the previous chapter, we introduced the model of a learning graph.  A learning graph can be converted into a dual adversary SDP that can be further transformed into a quantum query algorithm.  The learning graph approach is nice because, on the one hand, it ignores the internal organisation of the algorithm: Once the matrices $X_j$ satisfy constraints \rf(eqn:advDual), \rf(thm:advAlgorithm) will do all the remaining work for us.  On the other hand, the graph structure of the learning graph still appeals to the intuition of solving a query problem.

So far, we only saw the applications within the framework of certificate structures.  But, for many functions, their quantum query complexity is smaller than the complexity of their certificate structures.  
In this chapter, we show how the ideas from the previous chapter can help in constructing quantum query algorithms for such functions.  
We still use the term ``learning graph'' to describe algorithms in this chapter, although they do not satisfy the definition in \rf(sec:learning).

This chapter is based on the following paper:
\begin{itemize}
\item[\cite{belovs:learningKDist}]
A.~Belovs.
\newblock Learning-graph-based quantum algorithm for $k$-distinctness.
\newblock In {\em Proc. of 53rd IEEE FOCS}, pages 207--216, 2012, 1205.1534.
\end{itemize}

The main result of this chapter is a new quantum query algorithm for the $k$-distinctness problem.  We describe it in Sections~\ref{sec:first} and~\ref{sec:final}.  Before that, we give two warm-up examples of algorithms for the promise threshold function in \rf(sec:threshold), and the graph collision problem in~\rf(sec:collision).

\section{Threshold Problem}
\label{sec:threshold}
In this section, we construct a learning-graph-based quantum query algorithm for the {\em promise threshold function}.  This is a partial function $f\colon \{0,1\}^\vars\supseteq \cD\to\{0,1\}$ defined by
\[
f(z) = \begin{cases}
0, & |z|\le k;\\
1, & |z|\ge k+d;\\
\end{cases}
\]
where $|z|$ stands for the Hamming weight of $z$ (the number of ones), and $k$ and $d$ are some positive integers less than $\vars$.  
One can also construction a dual adversary for this problem by generalising the construction in \rf(prp:advThresholdExact).
For simplicity, we assume that $d=O(k)$, although it is not crucial for the algorithm.

\begin{prp}
\label{prp:counting}
The quantum query complexity of the promise threshold function is $O(\sqrt{\vars k}/d)$ if $d=O(k)$.
\end{prp}

For the case of $d=1$, we obtain the same complexity as in \rf(cor:findingAllOnes).  The case of larger $d$ is also well-known: this estimate can be obtained using quantum counting~\cite{brassard:counting}.

The function has the $(k+1)$-subset certificate structure.  Thus, a learning graph, as defined in \rf(sec:learning), cannot get the complexity claimed in \rf(prp:counting).  Thus, we have to change the definition of the learning graph $\cG$.
We assume that the weight $w_e$ of an arc $e$ may depend on the input.  More precisely, if $e$ goes from $S$ to $S\cup\{j\}$, and $z$ is the input string, then the weight of the arc may depend on the values of the variables in $S$: $w_e(z) = w_e(z_S)$.  The flow is defined in the same way as in \rf(defn:flow).  The negative and the positive complexities of $\cG$ on a particular input are defined by
\[
\CN(\cG, y) = \sum_{e\in E} w_e(y)
\qquad\mbox{and}\qquad
\CP(\cG, x) = \sum_{e\in E} \frac{p_e(x)^2}{w_e(x)}
\]
where $y\in\preimy$ and $x\in\preimx$.  The negative, the positive and the total complexities are
\[
\CN(\cG) = \max_{y\in\preimy} \CN(\cG,y),\qquad
\CP(\cG) = \max_{x\in\preimx} \CP(\cG,x),\qquad\mbox{and}\qquad
\CT(\cG) = \max\{\CN(\cG), \CP(\cG)\}.
\]
It is not hard to check that both proofs in \rf(sec:learningProof) can be adapted to include this definition of the learning graph.  We leave out the details.  This model of learning graphs is called {\em adaptive learning graph} in~\cite{belovs:learning}.

\pfstart[Proof of \rf(prp:counting)]
The vertices of the learning graph $\cG$ are formed by the subsets of $[\vars]$ of sizes at most $k+1$, and we have all possible arcs between them.  Let $e$ be an arc from $S$ to $S\cup\{j\}$ where $|S|\le k$ and $j\notin S$.  Define the weight $w_e(z)$ as follows.  If $z_j=0$ for at least one $j\in S$, define $w_e(z)=0$.  Otherwise, $w_e(z) = w_{|S|}$ only depends on the size of $S$.  We say that the arc $e$ is on the $|S|$th {\em step}.
For a positive input $x\in\preimx$, if $S$ is such that $|S|=k+1$ and $x_j=1$ for all $j\in S$, then $S$ is marked.

The negative complexity of the learning graph is maximised when $|y|=k$, and 
\begin{equation}
\label{eqn:counting1}
\CN(\cG) = \sum_{i=0}^k {k\choose i}(\vars-i)w_i \le \vars \sum_{i=0}^k w_i {k\choose i}.
\end{equation}
The positive complexity is maximised when $|x|=k+d$.  For each vertex $S$ of the learning graph, we distribute the flow uniformly to all arcs loading an element equal to 1.  Thus, on the $i$th step, there are ${k+d\choose i}(k+d-i)$ arcs used by the flow, and the flow is equal among all these arcs.  Thus, the positive complexity is
\begin{equation}
\label{eqn:counting2}
\CP(\cG) = \sum_{i=0}^k \sk[{k+d\choose i}(k+d-i) w_i]^{-1}.
\end{equation}
The optimal choice of $w_i$ is
\[
w_i = \sk[{k\choose i} {k+d\choose i}(k+d-i)]^{-1/2},
\]
so that the $i$th term in the sums from~\rf(eqn:counting1) and~\rf(eqn:counting2) is
\begin{align*}
\sqrt{{k\choose i} \sk[ {k+d\choose i}(k+d-i) ]^{-1}} &=\sqrt{\frac{k!}{(k-i)!i!}\frac{(k+d-i)!i!}{(k+d)!(k+d-i)}}\\&= \sqrt{\frac{(k-i+1)\cdots(k-i+d-1)}{(k+1)\cdots(k+d)}}\le \frac1{\sqrt{k+1}} \s[\frac{k+d-i-1}{k+d}]^{\frac{d-1}{2}}.
\end{align*}
Thus, the complexity of the learning graph is at most
\[
\sqrt{\frac{\vars}{k+1}} \sum_{i=0}^{k} \s[ \frac{k+d-i-1}{k+d} ]^{\frac{d-1}{2}}\le \frac{\sqrt{\vars}(k+d)}{\sqrt{k+1}}\int_0^1x^{\frac{d-1}{2}}\, dx = \frac{2\sqrt{\vars}(k+d)}{\sqrt{k+1}(d+1)} = O\s[\frac{\sqrt{\vars k}}{d}].\qedhere
\]
\pfend

\section{Graph Collision}
\label{sec:collision}
In this section, we describe a learning-graph-based algorithm for the graph collision problem with an additional promise.  It is a learning graph version of the algorithm by Andris Ambainis (personal communication).

Recall the graph collision problem from \rf(defn:graphCollision).  The problem is parametrised by a simple graph $G$ on $\vars$ vertices.  The input string consists of $\vars$ Boolean variables: one for each vertex of the graph.  The function evaluates to 1 if there exists an edge of $G$ with both endpoints marked by value 1, and to 0 otherwise.

The $O(\vars^{2/3})$ query algorithm we saw in \rf(cor:walkKDist) is the best known quantum algorithm for a general graph $G$.  For specific classes of graphs, however, one can do better. For instance, if $G$ is the complete graph, graph collision is equivalent to the 2-threshold problem that can be solved in $O(\sqrt{\vars})$ queries by \rf(prp:counting).  The algorithm in this section may be interpreted as an interpolation between this trivial special case and the general case.

Recall that the independence number $\alpha(G)$ of a simple graph $G$ is the maximal cardinality of a subset of vertices of $G$ such that no two of them are connected by an edge.

\begin{thm}
\label{thm:collision}
Graph collision on an $\vars$-vertex graph $G$ can be solved in $O(\sqrt{\vars}\alpha^{1/6})$ quantum queries, where $\alpha=\alpha(G)$ is the independence number of $G$.
\end{thm}

Note that if $G$ is a complete graph, $\alpha(G)=1$, and we get the previously mentioned $O(\sqrt{\vars})$-algorithm for this trivial case. In the general case, $\alpha(G)=O(\vars)$, and the complexity of the algorithm is $O(\vars^{2/3})$ that coincides with the complexity of the algorithm in \rf(cor:walkKDist).

Jeffery {\em et al.}~\cite{jeffery:matrixMultiplication} have built a quantum algorithm solving graph collision on $G$ in $O(\sqrt{\vars}+\sqrt{m})$ queries if $G$ misses $m$ edges to be a complete graph. This algorithm is incomparable to the one in \refthm{collision}: For some graphs the algorithm from \refthm{collision} performs better, for some graphs, vice versa.

\pfstart[Proof of \refthm{collision}] 
Let $f$ be the graph collision function specified by the graph $G$.  At first, we distinguish the case when the number of ones in the input is at most $\alpha$, and when it is at least $2\alpha$.  By \rf(prp:counting), the complexity of this step is less than $O(\sqrt{\vars})$.  
Inputs of Hamming weight between $\alpha$ and $2\alpha$ may fall in both categories.

If we know that the number of ones is greater than $\alpha$, we may claim that a graph collision exists.  Otherwise, we may assume the number of ones is at most $2\alpha$.  In this case, we execute the following learning graph $\cG$.

The learning graph is essentially the one for the 2-subset certificate structure from \reftbl{old}.  The certificate for a positive input $x$ is a pair $M = \{a,b\}$ of vertices such that $ab$ is an edge of $G$ and $x_a=x_b=1$.  (For notational convenience, we switched notation from $a_1,a_2$ to $a,b$.)

We would like to use the fact that the Hamming weight of the input string is small.  As the certificate is given by the elements with value 1, one possibility is to use an adaptive learning graph as it is done in the proof of \rf(prp:counting).
This works if we know the exact number of ones in the input.  Unfortunately, the analysis of the proof of \rf(prp:counting) reveals that the number of arcs of the learning graph (and, hence, the negative complexity) depends heavily on the Hamming weight of the input.  This excludes the possibility of an universal adaptive learning graph that would work for all possible Hamming weights of the input.  In principle, it is possible to estimate the number of ones using the quantum counting~\cite{brassard:counting} or \rf(prp:counting) prior executing the learning graph.  But this feels like an artificial solution, and is not readily applicable for the $k$-distinctness problem that we aim for.

Instead of that, we use all possible subsets of $[\vars]$ as vertices of the learning graph, regardless the content, but use an analogue of \rf(cor:findingAllOnes) to reduce the complexity.  For that, we utilise an idea due to Robin Kothari (personal communication).  We make the weight of an arc dependent (although, in a restricted form) on the value of the variable being loaded.

We do not reduce to any result like Theorems~\ref{thm:symmetric} or~\ref{thm:certificates}, but construct the matrices $X_j$ from~\rf(eqn:advDual) directly.  However, our construction is similar to the second proof of \rf(thm:symmetric).
We also use the randomised procedure language from \rf(sec:procedureDriven).

%We reduce the complexity of the learning graph by utilizing the bias between the number of zeros and ones induced by the small independence number, as outlined in Points~\ref{intro} and~\ref{rank} of \refsec{outline}. 

%One could prove the correctness of the algorithm completely analogously to the correctness proof of the algorithm from \refsec{goal}. However, in the preparation for future discard of the notion of flow (Point~\ref{drop} from \refsec{outline}), we use language from \refsec{final}. The reader is encouraged to compare both ways of the proof.

Let $x$ be a positive input, and $M=\{a,b\}$ be a 1-certificate.
The key vertices of the learning graph are $V_1\cup V_2$, where $V_1$ and $V_2$ consist of all subsets of $[\vars]$ of sizes $r$ and $r+1$, respectively, where $r=o(\vars)$ is some parameter specified later.%

A vertex in $V_1$ completely specifies the internal randomness of the loading procedure. For each $R\in V_1$, we fix an arbitrary order of its elements: $R=\{t_1,\dots, t_r\}$. We say that the choice of randomness $R\in V_1$ is {\em consistent} with $x$ if $\{a,b\}\cap R=\emptyset$. For each $x\in f^{-1}(1)$, there are exactly ${\vars-2\choose r}$ choices of $R\in V_1$ consistent with it.  We take each of them with probability $\prob={\vars-2\choose r}^{-1}$. 

For a fixed input $x$ and fixed randomness $R=\{t_1,\dots,t_r\}\in V_1$ consistent with $x$, the elements are loaded 
%(we are going to define what this means later) 
in the following order: 
\begin{equation}
\label{eqn:collisionT}
t_1,\dots,t_r,t_{r+1}=a,t_{r+2} = b.
\end{equation}

%The non-key vertices of $\cG$ appear in transitions from $\emptyset$ to elements of $V_1$. 
The non-key vertices of $\cG$ are of the form $v=(\{t_1,\dots,t_\ell\}, R)$, where $0\le\ell<r$, $R\in V_1$, and $t_i$ are as in~\refeqn{collisionT}. Recall that the first element of the pair is the set of loaded elements, and the second one is an additional label used to distinguish vertices with the same set of loaded elements.

An {\em arc} of the learning graph is a process of loading one variable. We denote an arc by $A^v_j$. Here, $j$ is the variable being loaded, and $v$ is the vertex of $\cG$ the arc originates in.  The arcs are as follows.  The arcs of the stage I have $v=(\{t_1,\dots,t_\ell\}, R)$ and $j=t_{\ell+1}$ with $0\le \ell<r$.
 The arcs of stages II.1 and II.2 have $v=S$, with $S\in V_1$ and $S\in V_2$, respectively, and $j\notin S$.

For a fixed $x\in f^{-1}(1)$, and fixed internal randomness $R\in V_1$ consistent with $x$, the following arcs are {\em taken}:
\begin{equation}
\label{eqn:collisionTaken}
A^{(\{t_1,\dots,t_\ell\},R)}_{t_{\ell+1}}\; \mbox{for $0\le \ell<r$,}\qquad A^R_a\qquad\text{and}\qquad A^{R\cup\{a\}}_b.
\end{equation}
We say that $x$ {\em satisfies} an arc if the arc is taken for some $R\in V_1$ consistent with $x$.  Note also, that, for a fixed positive input, no arc is taken for two different choices of the randomness.

Like in the second proof in \rf(sec:learningProof), for each arc $A^v_j$, we assign a semi-definite matrix $X^v_j\succeq 0$.  Then, $X_j$ in \refeqn{advDual} are given by $X_j = \sum_v X^v_j$ with the sum over all vertices.
Fix $A^v_j$, and let $S$ be the set of loaded elements in $v$.  
%Recall that an assignment on $S$ as a function $\alpha\colon S\to \{0,1\}$.  An input $z\in \{0,1\}^n$ {\em satisfies} assignment $\alpha$ iff $z_t=\alpha(t)$ for each $t\in S$. We say inputs $x$ and $y$ {\em agree} on $S$, if they satisfy the same assignment $\alpha$. 
Define $X^v_j = \sum_\alpha Y_\alpha$, where the sum is over all assignments $\alpha$ on $S$. The matrix $Y_\alpha$ is defined as $\prob(\psi\psi^*+\phi\phi^*)$, where, for each $z\in\{0,1\}^\vars$,
\[
\psi\elem[z]=
\begin{cases}
1/\sqrt{w_1},& \parbox{4.3cm}{$f(z)=1$, $z_j=1$,\\$z$ satisfies $\alpha$ and the arc $A^v_j$;}\\[\bigskipamount]
\sqrt{w_1},& \parbox{4cm}{$f(z)=0$, $z_j=0$,\\and $z$ satisfies $\alpha$;}\\[\medskipamount]
0,&\mbox{otherwise;}
\end{cases}
\mbox{and}\qquad
\phi\elem[z]=
\begin{cases}
1/\sqrt{w_0},& \parbox{4.3cm}{$f(z)=1$, $z_j=0$,\\$z$ satisfies $\alpha$ and the arc $A^v_j$;}\\[\bigskipamount]
\sqrt{w_0},& \parbox{4cm}{$f(z)=0$, $z_j=1$,\\ and $z$ satisfies $\alpha$;}\\[\medskipamount]
0,&\mbox{otherwise.}
\end{cases}
\]
Here $w_0$ and $w_1$ are parameters (the weights of the arcs) specified later.  They depend only on the stage the arc belongs to.  In other words, $X^v_j$ consists of the blocks of the following form:
\begin{equation}
\label{eqn:blockCollision}
\begin{array}{r|cc|cc|}
& x_j=1 & x_j=0 & y_j=1 & y_j=0 \\
\hline
x_j=1 & \prob/w_1 & 0 & 0 & \prob \\
x_j=0 & 0 & \prob/w_0 & \prob & 0\\
\hline
y_j=1 & 0 & \prob & \prob w_0 & 0  \\
y_j=0 & \prob & 0 & 0 & \prob w_1\\
\hline
\end{array}
\end{equation}
Here each of the 16 elements corresponds to a block in $Y_\alpha$ with all entries equal to this element. The first and the second columns represent the elements from $f^{-1}(1)$ that satisfy $\alpha$ and $A^v_j$, and such that their $j$th element equals $1$ and $0$, respectively. Similarly, the third and the fourth columns represent elements from $f^{-1}(0)$ that satisfy $\alpha$ and such that their $j$th element equals $1$ and $0$, respectively.

\paragraph{Feasibility} 
Assume $x\in\preimx$ and $y\in\preimy$ are some fixed inputs.  Let $R\in V_1$ be a choice of the internal randomness consistent with $x$. Let $Z_j$ be the matrix corresponding to the arc loading $j$ that is taken for this choice of $R$. That is,
$Z_j$ is the matrix in~\refeqn{collisionTaken} with sub-index $j$ if $j\in R\cup\{a,b\}$, or $Z_j=0$, otherwise.  We are going to prove that
\begin{equation}
\label{eqn:collisionSum}
\sum_{j\colon x_j\ne y_j} Z_j\elem[x,y] = \prob.
\end{equation}
Since there are ${\vars-2\choose r} = 1/p$ choices of $R$ consistent with $x$, and no arc is taken for two different choices of the randomness, this proves the feasibility condition~\refeqn{advDualCondition}.

Consider the order~\refeqn{collisionT} in which the elements are loaded for this particular choice of $x$ and $R$.
Before any element is loaded, both inputs agree (they satisfy the same assignment $\alpha\colon\emptyset\to\{0,1\}$).  After all the elements are loaded, $x$ and $y$ disagree, because it is not possible that $y_a=x_a$ and $y_b=x_b$.  With each element loaded, the assignments become more and more specific.  This means that there exists an element $j=t_i$ such that $x$ and $y$ agree before loading $j$, but disagree afterwards.  In particular, $x_j\ne y_j$.  By construction, this $j$ contributes $\prob$ to the sum in~\refeqn{collisionSum}.  All other $j$ contribute 0 to the sum. Indeed, if $j'=t_{i'}$ with $i'<i$ then $x_{j'}=y_{j'}$, hence, $j'$ contributes 0.  For $j'=t_{i'}$ with $i'>i$, $x$ and $y$ disagree on $\{t_1,\dots,t_{i'-1}\}$, hence, $Z_{j'}\elem[x,y]=0$ by construction.

\paragraph{Complexity} 
%\label{sec:collisionComplexity}
Similarly to \rf(sec:procedureDriven), let us define the complexity of stage $i$ on input $z\in\{0,1\}^\vars$ as $\sum_{j\in [\vars]} X'_j\elem[z,z]$, where $X'_j = \sum_v X^v_j$ with the sum over $v$ such that $A^v_j$ belongs to stage $i$.  Also, define the complexity of stage $i$ as the maximum complexity over all inputs $z\in\{0,1\}^\vars$.  Clearly, the objective value~\rf(eqn:advDualObjective) of the whole program is at most the sum of the complexities of all the stages.

Let us start with stages II.1 and II.2.%
  \footnote{For stages II.1 and II.2, the hiding intuition from \rf(sec:procedureDriven) works. For stage II.1, the length is 1, and the speciality is $O(\vars)$. For stage II.2, the length is 1, and the speciality is $O(\vars^2/r)$. Hence, the complexities are $O(\sqrt{\vars})$ and $O(\vars/\sqrt{r})$, respectively.}
  For any $x\in f^{-1}(1)$, on each of stages II.1 and II.2, there are ${\vars-2\choose r}$ arcs satisfying $x$.  These are the arcs $A^R_a$ and $A^{R\cup\{a\}}_b$, respectively, for all choices of $R\in V_1$ consistent with $x$.  By~\refeqn{blockCollision}, each of them contributes $\prob/w_1$ to the complexity of $x$ on the corresponding stage.  Thus, the complexity of $x$ on each of the stages is ${\vars-2\choose r}\prob/w_1=1/w_1$.
Since we are guaranteed that $x_j=1$ in notations from~\refeqn{blockCollision}, we may set $w_0=0$. 

The total number of arcs on stages II.1 and II.2 are $(\vars-r){\vars\choose r}$ and $(\vars-r-1){\vars\choose r+1}$, respectively.  Each of them contributes at most $\prob w_1$ to the complexity of any $y\in\preimy$.
Hence, the complexities of stages II.1 and II.2 on $y$ are at most $(\vars-r){\vars\choose r}\prob w_1 = O(\vars w_1)$ and $(\vars-r-1){\vars\choose r+1}\prob w_1 = O(\vars^2w_1/r)$, respectively.  If we set $w_1=1/\sqrt{\vars}$ on stage II.1 and $w_1 = \sqrt{r}/\vars$ on stage II.2, the complexities of these stages become $O(\sqrt{\vars})$ and $O(\vars/\sqrt{r})$, respectively.

Consider stage I now. Let $k\le 2\alpha$ be the number of variables with value 1 in the input ($x$ or $y$). The total number of arcs on this stage is $r{\vars\choose r}$. Out of them, exactly $k{\vars-1\choose r-1}$ load a variable with value 1. Thus, for $y$, the complexity of stage I is at most
\[
qr{\vars\choose r}w_0 + qk{\vars-1\choose r-1}w_1 = O\s[rw_0 + \frac{\alpha r}{\vars}w_1]. 
\]
Similarly, for $x\in f^{-1}(1)$, the complexity of stage I is $O(r/w_1+\alpha r/(\vars w_0))$. If we set $w_0=\sqrt{\alpha/\vars}$ and $w_1 =\sqrt{\vars/\alpha}$, then the complexity of stage I becomes $O(r\sqrt{\alpha/\vars})$.  The total complexity of the learning graph is
\[O\s[r\sqrt{\frac{\alpha}{\vars}} + \frac{\vars}{\sqrt{r}}] = O\s[\sqrt{\vars}\alpha^{1/6}],\]
if $r=\vars\alpha^{-1/3}$. 
\pfend

\section{\texorpdfstring{$k$-Distinctness}{k-Distinctness}: First Attempt}
\label{sec:first}
In this section, we develop a quantum query algorithm for the $k$-distinctness problem from \rf(defn:kdist).  As usual, we assume $k=O(1)$, and consider the complexity as $\vars\to\infty$.  
In particular, the factors behind the Big-Oh notation are functions of $k$.
The best previously known algorithm, described in \rf(cor:walkKDist), uses $O(\vars^{k/(k+1)})$ quantum queries.
As element distinctness reduces to $k$-distinctness by repeating each element $k-1$ times, the lower bound of $\Omega(\vars^{2/3})$ from \rf(cor:distLower) carries over to $k$-distinctness (this argument is attributed 
to Aaronson in~\cite{ambainis:distinctness}). This simple lower bound is the best known so far.
In the remaining part of this chapter, we prove the following

\begin{thm}
\label{thm:kdist}
For arbitrary but fixed integer $k\ge2$, the $k$-distinctness problem can be solved by a quantum algorithm in $O\s[\vars^{1-2^{k-2}/(2^k-1)}]$ queries.
\end{thm}

Note that $O\s[\vars^{1-2^{k-2}/(2^k-1)}] = o(\vars^{3/4})$.  Thus, our algorithm solves the $k$-distinctness problem in asymptotically fewer queries than the best previously known algorithm spends on 3-distinctness.  Let throughout Sections~\ref{sec:first} and~\ref{sec:final}, $f\colon [q]^\vars\to\{0,1\}$ be the $k$-distinctness function.  

Similarly to the analysis in~\cite{ambainis:distinctness}, we may assume that there is a unique $k$-tuple of equal elements in any positive input.
One of the simplest reductions to this special case is to take a sequence $T_i$ of uniformly random subsets of $[\vars]$ of sizes $(2k/(2k+1))^i \vars$, and to run the algorithm, for each $i$, with the input variables outside $T_i$ removed. One can prove that if there are $k$ equal elements in the input, then there exists $i$ such that, with probability at least $1/2$, $T_i$ will contain unique $k$-tuple of equal elements. The complexities of the executions of the algorithm for various $i$ form a geometric series, and their sum is equal to the complexity of the algorithm for $i=0$ up to a constant factor.  See~\cite{ambainis:distinctness} for more detail and alternative reductions.  

Let $x$ be a positive input, and $M = \{a_1,\dots,a_k\}$ be the $k$-tuple of equal elements.
At a very high level, our learning graph is similar to the one in \rf(tbl:old).  We hide elements $a_1,\dots,a_{k-1}$ in $S$ during the loading of $a_k$.  And, as in \rf(sec:collision), we use a bias between the values of variables to load more of them for the same cost.

The bias comes from the following observation.
Let us divide the set $S$ into $k$ subsets: $S=S_1\sqcup \cdots\sqcup S_{k-1}$, where $\sqcup$ denotes disjoint union.  Set $S_i$ has size $r_i=o(\vars)$. We use $S_i$ to hide $a_i$ when loading $a_k$. 
Consider the situation before loading $a_k$. If an element $j\in S_2$ is such that $x_j\ne x_t$ for all $t\in S_1$, then this element cannot be a part of the certificate (i.e., it cannot be $a_2$), and its precise value is irrelevant. 
In this case, we say that $j$ has no {\em match} in $S_1$, and represent $j$ by a special symbol $\star$.  Otherwise, we {\em uncover} the element, i.e., load its precise value.  Similarly, when loading $S_i$ with $i>2$, we only uncover those elements that have a match among the uncovered elements of $S_{i-1}$.  All of this is summarised in \rf(tbl:new).

\begin{table}
\begin{tabular}{rp{13cm}}
\hline
I.1& Load a set $S_1$ of $r_1$ elements not from $M$.\\
I.2& Load a set $S_2$ of $r_2$ elements not from $M$, uncovering only those elements that have a match in $S_1$.\\
I.3& Load a set $S_3$ of $r_3$ elements not from $M$, uncovering only those elements that have a match among the uncovered elements of $S_2$.\\
& \quad\vdots\\
I.($k-1$)& Load a set $S_{k-1}$ of $r_{k-1}$ elements not from $M$, uncovering only those elements that have a match among uncovered elements of $S_{k-2}$.\\
II.1& Load $a_1$ and add it to $S_1$.\\
& \quad\vdots\\
II.$(k-1)$& Load $a_{k-1}$ and add it to $S_{k-1}$.\\
II.$k$& Load $a_k$.\\
\hline\end{tabular}
\caption{Learning graph for the $k$-distinctness problem.}
\label{tbl:new}
\end{table}

In this section, we start the description of the learning graph, but it has a flaw that we describe in \rf(sec:firstFeasible).  We fix the flaw in \rf(sec:final).

\subsection{Construction}
\label{sec:firstConstruction}
Again, we construct the matrices $X_j$ in~\rf(eqn:advDual) directly.  The construction deviates from the graph representation: a bit in \refsec{first}, and quite strongly in \refsec{final}.  Nevertheless, we keep the term ``vertex'' for an entity describing some knowledge of the values of the input variables, and the term ``arc'' for a process of loading a value of a variable (possibly, only partially).  Each arc originates in a vertex, but we do not specify where it goes.  Inspired by \rf(sec:procedureDriven), the vertices are divided into {\em key} vertices denoted by the set of loaded variables $S$ with additional structure.  The non-key vertices are denoted by $(S,R)$ where $S$ is the set of loaded variables, and $R$ is an additional label used to distinguish vertices with the same $S$.  Also, we use the ``internal randomness'' term.
At first, we describe the learning graph in the terms of vertices and arcs, and then explain how they are converted into the matrices $X_j$.

The key vertices of the learning graph are $V_1\cup\cdots \cup V_k$, where $V_\stage$, for $\stage\in[k]$, consists of $(k-1)$-tuples $S=(S_1,\dots,S_{k-1})$ of pairwise disjoint subsets of $[\vars]$.  For $V_\stage$, we require that $|S_i|=r_i+1$ for $i<\stage$, and $|S_i|=r_i$ for $i\ge \stage$.  Here, $r_1,\dots, r_{k-1}$ are some constants specified later.  Denote also $r = \sum_i r_i$.

A vertex $R=(R_1,\dots,R_{k-1})\in V_1$ completely specifies the internal randomness.  We assume that, for any $R\in V_1$, an arbitrary order $t_1,\dots, t_r$ of the elements in $\bigcup R = R_1\cup \cdots \cup R_{k-1}$ is fixed so that all the elements of $R_i$ precede all the elements of $R_{i+1}$ for all $i\le k-2$. We say that $R\in V_1$ is {\em consistent} with the input $x$ iff $\{a_1,\dots,a_k\}\cap (\bigcup R)=\emptyset$. 

For each $x\in f^{-1}(1)$, there are exactly ${\vars-k\choose r_1,\dots,r_{k-1}}$ choices of $R\in V_1$ consistent with $x$.  Recall the notation
\[
{n\choose b_1,\dots,b_i} = {n\choose b_1}{n-b_1\choose b_2}\cdots{n-b_1-\cdots-b_{i-1}\choose b_i}.
\]
We take each of them with probability $\prob={\vars-k\choose r_1,\dots,r_{k-1}}^{-1}$.  For a fixed input $x$ and fixed randomness $R\in V_1$ consistent with $x$, the elements are loaded in the following order: 
\begin{equation}
\label{eqn:t}
t_1,t_2,\dots,t_r,t_{r+1}=a_1,t_{r+2} = a_2,\dots,t_{r+k}=a_k.
\end{equation}

We use a convention to name the vertices and the arcs of the learning graph similar to the proof of \refthm{collision}. The non-key vertices of $\cG$ are of the form $v=(R\cap \{t_1,\dots,t_\ell\}, R)$, where $R\in V_1$, $0\le \ell<r$, and $\{t_i\}$ are as in~\refeqn{t}. Here we use notation $R\cap T = (R_1\cap T,\dots,R_{k-1}\cap T)$.

Let us describe the arcs $A^v_j$ of $\cG$, where $j$ is the variable being loaded and $v$ is the vertex it originates in.  Arcs of stages I.$\stage$ have $v=(R\cap \{t_1,\dots,t_\ell\}, R)$ and $j=t_{\ell+1}$ with $0\le \ell<r$.  The arc belongs to stage I.$\stage$ iff $t_{\ell+1}\in R_\stage$.  The arcs of stage II.$\stage$ have $v=S$, with $S\in V_\stage$, and $j\notin \bigcup S$.

For a fixed $x\in f^{-1}(1)$ and fixed internal randomness $R\in V_1$ consistent with $x$, the following arcs are {\em taken}: 
\begin{equation}
\label{eqn:taken}
A^{(R\cap \{t_1,\dots,t_\ell\}, R)}_{t_{\ell+1}}\; \mbox{for $0\le \ell<r$}\qquad\text{and}\qquad A^{R[a_1,\dots,a_\ell]}_{a_{\ell+1}}\; \text{with $0\le\ell < k$}.
\end{equation}
Here
\[R[a_1,a_2,\dots,a_\ell] = (R_1\cup\{a_1\},R_2\cup\{a_2\}\dots,R_\ell\cup\{a_\ell\},R_{\ell+1},\dots,R_{k-1}).\] 
We say that $x$ {\em satisfies} all these arcs. Note that, for a fixed $x$, no arc is taken for two different choices of $R$.

Again, for each arc $A^v_j$, we assign a matrix $X^v_j\succeq 0$, so that $X_j$ in \refeqn{advDual} are given by $X_j = \sum_v X^v_j$. Assume $A_j^v$ is fixed. Let $S=(S_1\dots,S_{k-1})$ be the corresponding set of loaded elements. Define an assignment on $S$ as a function $\alpha\colon \bigcup S \to [q]\cup\{\star\}$, where $\star$ represents the covered elements of stages I.$\stage$.  We have that $\star\notin \alpha(S_1)$ and $\alpha(S_{i+1}) \subseteq \alpha(S_i)\cup\{\star\}$ for $1\le i\le k-2$.  An input $z\in [q]^\vars$ {\em satisfies} the assignment $\alpha$ iff, for each $t\in \bigcup S$,
\[\alpha(t) = \begin{cases}
z_t,& \mbox{$t\in S_1$; or $t\in S_i$ for $i>1$ and $z_t\in \alpha(S_{i-1})$;}\\
\star,& \mbox{otherwise}.
\end{cases}\]
Each input $z$ satisfies a unique assignment on $S$.  For a fixed input $z$, we say that an element $j\in \bigcup S$ is {\em covered} in $S$ if $\alpha(j)=\star$, where $\alpha$ is the unique assignment $z$ satisfies on $S$.  (Sometimes we say that $z_j$ is covered to indicate the input.)  Otherwise, the element is {\em uncovered}.
We also say that inputs $x$ and $y$ {\em agree} on $S$, if they satisfy the same assignment on $S$.

We define $X^v_j$ as $\sum_\alpha Y_\alpha$ where the sum is over all assignments $\alpha$ on $S$. The definition of $Y_\alpha$ depends on whether $A^v_j$ is on stage I.$\stage$ with $\stage>1$, or not. If $A^v_j$ is not on one of these stages, then $Y_\alpha = \prob \psi \psi^*$ where, for each $z\in [q]^\vars$,
\[
\psi\elem[z] = \begin{cases}
1/{\sqrt{w}},& \text{$f(z)=1$, and $z$ satisfies $\alpha$ and the arc $A_j^v$;}\\
\sqrt{w},& \text{$f(z)=0$, and $z$ satisfies $\alpha$;}\\
0,&\text{otherwise.}
\end{cases}
\]
Here $w$ is a positive real number: the weight of the arc. It only depends on the stage of the arc, and will be specified later. Thus, $X_j^v$ consists of the blocks of the following form:
\begin{equation}
\label{eqn:tablica1}
\begin{array}{r|cc|}
& x & y \\
\hline
x & \prob/w & \prob\\
y & \prob & \prob w\\
\hline
\end{array}
\end{equation}
Here $x\in \preimx$ and $y\in\preimy$ represent inputs satisfying some assignment $\alpha$. The inputs represented by $x$ have to satisfy the arc $A_j^v$ as well. 

If $A_j^v$ is on stage I.$\stage$ with $\stage>1$, the elements having a match in $S_{\stage-1}$ and the ones that don't must be treated differently. In this case, $Y_\alpha=\prob(\psi\psi^*+\phi\phi^*)$, where
\[
\psi\elem[z]=
\begin{cases}
1/\sqrt{w_1},& \parbox{4.5cm}{$f(z)=1$, $z_j\in \alpha(S_{\stage-1})$,\\and $z$ satisfies $\alpha$ and $A_j^v$;}\\[\medskipamount]
\sqrt{w_1},& \mbox{$f(z)=0$, and $z$ satisfies $\alpha$};\\
0,&\mbox{otherwise};
\end{cases}
\qquad
\phi\elem[z]=\begin{cases}
1/\sqrt{w_0},& \parbox{5cm}{$f(z)=1$, $z_j\notin \alpha(S_{\stage-1})$,\\and $z$ satisfies $\alpha$ and $A_j^v$;}\\[\bigskipamount]
\sqrt{w_0},& \parbox{5cm}{$f(z)=0$, $z_j\in \alpha(S_{\stage-1})$,\\ and $z$ satisfies $\alpha$;}\\[\medskipamount]
0,&\mbox{otherwise}.
\end{cases}
\]
Here $w_0$ and $w_1$ are again parameters to be specified later. In other words, $X_j^v$ consists of the blocks of the following form:
\begin{equation}
\label{eqn:tablica}
\begin{array}{r|cc|cc|}
& x_j\in \alpha(S_{\stage-1}) & x_j\notin \alpha(S_{\stage-1}) & y_j\in \alpha(S_{\stage-1}) & y_j\notin\alpha(S_{\stage-1}) \\
\hline
x_j\in \alpha(S_{\stage-1}) & \prob/w_1 & 0 & \prob & \prob \\
x_j\notin \alpha(S_{\stage-1}) & 0 & \prob/w_0 & \prob & 0\\
\hline
y_j\in \alpha(S_{\stage-1}) & \prob & \prob & \prob(w_0+w_1) & \prob w_1  \\
y_j\notin \alpha(S_{\stage-1}) & \prob & 0 & \prob w_1 & \prob w_1\\
\hline
\end{array}
\end{equation}
Here $x$ and $y$ are like in~\refeqn{tablica1}. This is a generalisation of the construction from \refthm{collision}. Note that if $x_j$ and $y_j$ are both represented by $\star$ in the assignments on $(S_1,\dots,S_{s-1},S_s\cup\{j\},S_{s+1},\dots,S_{k-1})$, then they satisfy $X^v_j\elem[x,y]=0$.

\subsection{Complexity}
\label{sec:complexity}
Let us estimate the complexity of the learning graph.  We use the notion of the complexity of a stage from the proof of \rf(thm:collision).

Let us start with stage I.1. We set $w=1$ for all arcs on this stage. There are $r_1{\vars\choose r_1,\dots,r_{k-1}}$ arcs on this stage, and, by~\refeqn{tablica1}, each of them contributes at most $\prob$ to the complexity of each $z\in\{0,1\}^\vars$. Hence, the complexity of stage I.1 is $O\s[\prob r_1{\vars\choose r_1,\dots,r_{k-1}}] = O(r_1)$.

Now consider stage II.$\stage$ for $\stage\in[k]$.%
\footnote{\label{foot:symmetric} The complexities of stages I.1 and II.$\stage$ can be explained by a similar argument like in \rf(sec:procedureDriven). For stage I.1, the length is $r_1$, and the speciality is $O(1)$. For stage II.$\stage$, the length is 1, but the speciality is $O(\vars^\stage/(r_1\cdots r_{\stage-1}))$, because there are $\stage$ marked elements involved, but $a_i$, for $i<\stage$, is hidden in $S_i$ of size $r_i$. }
   The total number of arcs on the stage is $(\vars-r-\stage+1){\vars\choose r_1+1,\dots,r_{\stage-1}+1,r_\stage,\dots,r_{k-1}}$.  By~\refeqn{tablica1}, each of them contributes $\prob w$ to the complexity of each $y\in f^{-1}(0)$.  Out of these arcs, for any $x\in f^{-1}(1)$, exactly ${\vars-k\choose r_1,\dots,r_{k-1}}$ satisfy $x$.  And each of them contributes $\prob/w$ to the complexity of $x$.  Thus, the complexities of stage II.$\stage$ for any input in $f^{-1}(0)$ and $f^{-1}(1)$ are
\[(\vars-r-\stage+1){\vars\choose r_1+1,\dots,r_{\stage-1}+1,r_\stage,\dots,r_{k-1}}\prob w = O\s[\frac{\vars^\stage w}{r_1\cdots r_{\stage-1}}], \]
and
\[{\vars-k\choose r_1,\dots,r_{k-1}}\frac{\prob}{w} = \frac{1}{w}, \]
respectively.  By setting $w = \s[\vars^s/(r_1\cdots r_{\stage-1})]^{-1/2}$, we get complexity $O\s[\sqrt{\vars^s/(r_1\cdots r_{\stage-1})}]$ of stage II.$\stage$. The maximal complexity is attained for stage II.$k$.

Now let us calculate the complexity of stage I.$\stage$ for $\stage>1$.  The total number of arcs on this stage is $r_\stage{\vars\choose r_1,\dots,r_{k-1}}$.  Consider an input $z\in [q]^\vars$, and a choice of the internal randomness $R=(R_1,\dots,R_{k-1})\in V_1$.  An element $j$ is uncovered on stage I.$\stage$ for this choice of $R$ if and only if there is an $\stage$-tuple $(b_1,\dots,b_\stage)$ of elements such that $j=b_\stage$, $b_i\in R_i$ and $z_{b_i}=z_{b_j}$ for all $i, j\in[\stage]$.  By our assumption on the uniqueness of a $k$-tuple of equal elements in a positive input, the total number of such $\stage$-tuples is $O(\vars)$.  And, for each of them, there are ${\vars-\stage \choose r_1-1,\dots,r_\stage-1,r_{\stage+1},\dots,r_{k-1}}$ choices of $R\in V_1$ such that $b_i\in R_i$ for all $i\in[s]$.  By~\refeqn{tablica}, the complexities of this stage for an input in $f^{-1}(0)$ and in $f^{-1}(1)$ are, respectively, at most
\[\prob\left[O(\vars) {\vars-\stage\choose r_1-1,\dots,r_\stage-1,r_{\stage+1},\dots,r_{k-1}}w_0+r_\stage{\vars\choose r_1,\dots,r_{k-1}}w_1 \right]=
O\s[\frac{r_1\cdots r_\stage}{\vars^{\stage-1}}w_0+r_\stage w_1] \]
and
\[\prob\left[O(\vars) {\vars-\stage\choose r_1-1,\dots,r_\stage-1,r_{\stage+1},\dots,r_{k-1}}\frac{1}{w_1}+r_\stage{\vars\choose r_1,\dots,r_{k-1}}\frac1{w_0} \right]=
O\s[\frac{r_1\cdots r_\stage}{\vars^{\stage-1}w_1}+\frac{r_\stage}{w_0}]. \]
By assigning $w_0=\sqrt{\vars^{\stage-1}/(r_1\cdots r_{\stage-1})}$ and $w_1 = \sqrt{r_1\cdots r_{\stage-1}/\vars^{\stage-1}}$, both these quantities become $O\s[r_\stage\sqrt{r_1\cdots r_{\stage-1}/\vars^{\stage-1}}]$.  

With this choice of the weights, the value of the objective function in~\refeqn{advDualObjective} is
\begin{equation}
\label{eqn:complexity}
O\s[r_1+ r_2\sqrt{\frac{r_1}{\vars}}+\cdots+r_{k-1}\sqrt{\frac{r_1\cdots r_{k-2}}{\vars^{k-2}}} + \sqrt{\frac{\vars^k}{r_1\cdots r_{k-1}}}].
\end{equation}
Assuming all terms in~\refeqn{complexity}, except the last one, are equal, and denoting $\rho_i = \log_n r_i$, we get that
\[\rho_i +\frac12(\rho_1+\cdots+\rho_{i-1}) - \frac{i-1}{2} = \rho_{i+1} + \frac12(\rho_1+\cdots+\rho_{i}) - \frac{i}{2},\qquad\mbox{for $i=1,\dots,k-2$;} \]
or, equivalently,
\[\rho_{i+1} = \frac{1 + \rho_{i}}2,\qquad\mbox{for $i=1,\dots,k-2$.} \]
Assuming the first term, $r_1$, equals the last one, $\sqrt{\vars\frac{\vars}{r_1}\cdots\frac{\vars}{r_{k-1}}}$, we get
\[ \rho_1 = \frac{1+(1-\rho_1)+\cdots+(1-\rho_{k-1})}2 = \frac12 + \s[\frac12 + \cdots + \frac{1}{2^{k-1}}](1-\rho_1) =  \frac12 + \s[1-\frac{1}{2^{k-1}}](1-\rho_1).\]
From here, it is straightforward that $\rho_1 = 1 - 2^{k-2}/(2^k-1)$, hence, the complexity of the algorithm is $O\s[\vars^{1-2^{k-2}/(2^k-1)}]$.

\subsection{(In)feasibility}
\label{sec:firstFeasible}
 Assume $x$ and $y$ are inputs such that $f(x)=1$ and $f(y)=0$. Let $R=(R_1,\dots,R_{k-1})\in V_1$ be a choice of the internal randomness consistent with $x$. Similarly to the proof of \refthm{collision}, let $Z_j$ be the matrix corresponding to the arc loading $j$ that is taken for input $x$ and randomness $R$ (i.e., the one from~\refeqn{taken} with sub-index $j$, or the zero matrix, if there are none). 

Again, we would like to prove that~\refeqn{collisionSum} holds. Unfortunately, it does not always hold. Assume $x$, $y$ and $R\in V_1$ are such that $x$ and $y$ agree on $R$. Thus, the contribution to the sum in~\refeqn{collisionSum} is 0 from all arcs of stages I.$\stage$. Now assume that $x_{a_1}=y_{a_1}$ and there exists $b\in R_2$ such that $y_b = x_{a_1}$. This doesn't contradict that $x$ and $y$ agree on $R$, because $y_b$ is represented by $\star$ in the assignment it satisfies on $R$. 

But $x$ and $y$ disagree on $R[a_1]$, because $y_b$ gets uncovered there. Thus, the contribution to~\refeqn{collisionSum} is 0 from all arcs of stages II.$\stage$ as well. Thus, equation~\refeqn{collisionSum} does not hold. We deal with this problem in the next section.

\section{\texorpdfstring{$k$-Distinctness}{k-Distinctness}: Final Version}
\label{sec:final}
In \refsec{firstFeasible}, we saw that the learning graph from \reftbl{new} is incorrect.  This is due to {\em faults}. A fault is an element $b$ of $R_i$ with $i>1$ such that $y_b = x_{a_1}$.  This is the only element that can suddenly become uncovered after adding $a_{i-1}$ to $R_{i-1}$ on stage II.$(i-1)$.  Indeed, we assumed $x$ contains a unique $k$-tuple of equal elements, hence, if $R\in V_1$ is consistent with $x$, no $b$ in $\bigcup R$ satisfies $x_b=x_{a_1}$.
 
Since $y$ is a negative input, there are at most $k-1=O(1)$ faults for every choice of $x$.  Thus, all we need is to develop a fault-tolerant version of the learning graph from~\reftbl{new} that is capable of dealing with this number of faults.

As an introductory example, consider the case of $k=3$.  Here, a fault may only occur in $R_2$, and a fault may come in action only if $y_{a_1}=x_{a_1}$, hence, we may assume there is at most one fault.  Split $R_2$ into two subsets: $R_2 = R_2(1)\sqcup R_2(2)$.  We know that at least one of them is not faulty.  Hence, we could could try both cases: adding $a_2$ to $R_2(1)$, and adding it $R_2(2)$.  At least one of them will work.  But it is not enough: If they both work, the contribution is $2p$, and we want it to be exactly $p$ in all cases.

To solve this complication, we split $R_1$ into three subsets: $R_1 = R_1(1)\sqcup R_1(2)\sqcup R_1(1,2)$.  We uncover an element in $R_2(i)$ iff it has a match in $R_1(i)\cup R_1(1,2)$.
Consider three cases:
\itemstart
\item $a_1$ goes to $R_1(1)$, and $a_2$ goes to $R_2(1)$;
\item $a_1$ goes to $R_1(2)$, and $a_2$ goes to $R_2(2)$; and
\item $a_1$ goes to $R_1(1,2)$, and $a_2$ goes to $R_2(1)$.
\itemend
Also, we set the third case to give contribution $-p$, whereas the first two give contribution $p$ as before.  Again, at least one of the first two cases will work.  Moreover, the third case will work if and only if both of the first two cases work.  Thus, in dependence on the case, the contribution is $p+0+0$, $0+p+0$, or $p+p-p$, that equals $p$.
The construction in the general case is a direct generalisation of this idea.

\subsection{Construction}
\label{sec:construction}
The variables loaded in vertices of the learning graph are split into collections of pairwise disjoint subsets: $S = \sA[S_i(d_1,d_2,\dots, d_{i-1}, D) ]$, where $i\in [k-1]$, $d_j\in [k-j]$, and $\emptyset\subset D\subseteq [k-i]$.
If $S$ is as above, let $S_i = \bigcup_{d_1,\dots,d_{i-1},D} S_i(d_1,\dots,d_{i-1},D)$, and $\bigcup S = \bigcup_i S_i$.

For a non-empty subset $D\subset \N$, let $\mu(D)$ denote the minimal element of $D$ (or any other fixed element of $D$).  For each sequence $(D_1,\dots,D_{\stage-1})$, where $D_i$ is a non-empty subset of $[k-i]$, let $V_\stage(D_1,\dots,D_{\stage-1})$ consist of all collections $\sA[S_i(d_1,d_2,\dots, d_{i-1}, D) ]$ such that 
\[
|S_i(d_1,\dots,d_{i-1}, D)| = \begin{cases}
r_i+1,& \mbox{$i<\stage$, $d_1=\mu(D_1),\dots, d_{i-1}=\mu(D_{i-1})$, and $D=D_i$;}\\
r_i,&\mbox{otherwise.}
\end{cases}
\]

The key vertices of the learning graph are $V_1\cup\cdots \cup V_{k}$, where $V_\stage$ is the union of $V_\stage(D_1,\dots,D_{\stage-1})$ over all choices of $(D_1,\dots,D_{\stage-1})$.

Again, a vertex $R = \sA[R_i(d_1,d_2,\dots, d_{i-1}, D) ]\in V_1$ completely specifies the internal randomness.  For each of them, we fix an arbitrary order $t_1,\dots, t_r$ of the elements in $\bigcup R$.
The order is such that all the elements of $R_i$ precede all the elements of $R_{i+1}$ for any $i\le k-2$.  We say that $R$ is {\em consistent} with $x$ if $\{a_1,\dots,a_k\}$ is disjoint from $\bigcup R$.  Let $\prob$ be the inverse of the number of $R\in V_1$ consistent with $x$. (Clearly, this number is the same for all choices of $x$.)

The elements still are loaded in the order from~\refeqn{t}. We use a similar convention to name the arcs of the learning graph as in \refsec{first}. Arcs of stages I.$\stage$ are of the form $A^{(R\cap \{t_1,\dots, t_\ell \}, R)}_{t_{\ell+1}}$ for $R\in V_1$ and $0\le \ell < r$. Here, $R\cap T = \sA[S_i(d_1,d_2,\dots, d_{i-1}, D) ]$ is defined by $S_i(d_1,d_2,\dots, d_{i-1}, D) = R_i(d_1,d_2,\dots, d_{i-1}, D)\cap T$.  Arcs of stage II.$\stage$ are of the form $A^R_j$ with $R\in V_\stage$ and $j\notin\bigcup R$. 

Fix an arc $A^v_j$, and let $S=\sA[S_i(d_1,d_2,\dots, d_{i-1}, D) ]$ be the set of loaded elements. This time, we define an assignment on $S$ as a function $\alpha\colon \bigcup S\to [q]\cup\{\star\}$ such that $\star\notin \alpha(S_1)$, and, for all $i>1$ and all possible choices of $d_1,\dots,d_{i-1}$ and $D$:
\[\alpha(S_i(d_1,d_2,\dots,d_{i-1},D)) \subseteq \{\star\}\cup \bigcup_{K\ni d_{i-1}} \alpha(S_{i-1}(d_1,\dots,d_{i-2}, K)). \]

An input $z\in [q]^\vars$ satisfies the assignment $\alpha$ iff, for each $t\in \bigcup S$,
\begin{equation}
\label{eqn:satisfiesAlpha}
\alpha(t) = \begin{cases}
z_t,& \mbox{$t\in S_1(D)$ for some $D$;}\\
z_t,& \mbox{$t\in S_i(d_1,\dots,d_{i-1},D)$ and $z_t\in \bigcup_{K\ni d_{i-1}} \alpha(S_{i-1}(d_1,\dots,d_{i-2},K))$;}\\
\star,& \mbox{otherwise}.
\end{cases}
\end{equation}
The covered, uncovered elements and the agreement relation are defined as before.

For any $x\in f^{-1}(1)$ and $R\in V_1$ consistent with $x$, the following arcs are taken.  On stage I.$\stage$, for $\stage\in [k-1]$, these are arcs $A^{(R\cap \{t_1,\dots, t_\ell\}, R)}_{t_{\ell+1}}$, where $t_{\ell+1}$ belongs to one of $R_\stage$.  On stage II.$\stage$, for any fixed $R$ consistent with $x$ and $\stage\in[k]$, we have many arcs loading $a_\stage$.  For each choice of $(D_i)_{i\in[\stage-1]}$, where $D_i$ is a non-empty subset of $[k-i]$, the arc 
$A_{a_\stage}^{R[D_1\leftarrow a_1,\dots,D_{\stage-1}\leftarrow a_{\stage-1}]}$ is taken.  Here, $R[D_1\leftarrow a_1,\dots,D_{\stage-1}\leftarrow a_{\stage-1}] = \sA[S_i(d_1,d_2,\dots, d_{i-1}, D) ]$ is defined by
\[
S_i(d_1,\dots,d_{i-1}, D) = 
\begin{cases}
R_i(d_1,\dots,d_{i-1}, D)\cup\{a_i\},& \mbox{$i<\stage$, $d_1=\mu(D_1),\dots, d_{i-1}=\mu(D_{i-1})$, and $D=D_i$;}\\
R_i(d_1,\dots,d_{i-1}, D),&\mbox{otherwise.}
\end{cases}
\]

The main property of this vertex is as follows:
\begin{clm}
\label{clm:Rextended}
The vertex $S = R[D_1\leftarrow a_1,\dots,D_{\stage-1}\leftarrow a_{\stage-1}]$ belongs to $V_\stage(D_1,\dots, D_{\stage-1})$.  Moreover, all the elements $a_1,\dots,a_{\stage-1}$ are uncovered in this vertex.
\end{clm}

\pfstart
The first statement is obvious.  Let us prove the second one.  The element $a_1$ is uncovered because it belongs to $S_1$ (the first case of~\rf(eqn:satisfiesAlpha)).  We proceed further by induction.  Assume $a_{i}\in S_i(\mu(D_1),\dots,\mu(D_{i-1}), D_i)$ is uncovered.  As $D_i\ni \mu(D_i)$, we get that $a_{i+1}\in S_{i+1}(\mu(D_1),\dots,\mu(D_{i}), D_{i+1})$ is uncovered by the second case of~\rf(eqn:satisfiesAlpha).
\pfend

\begin{figure}[tb]
\release{\centering \includegraphics[width=9cm]{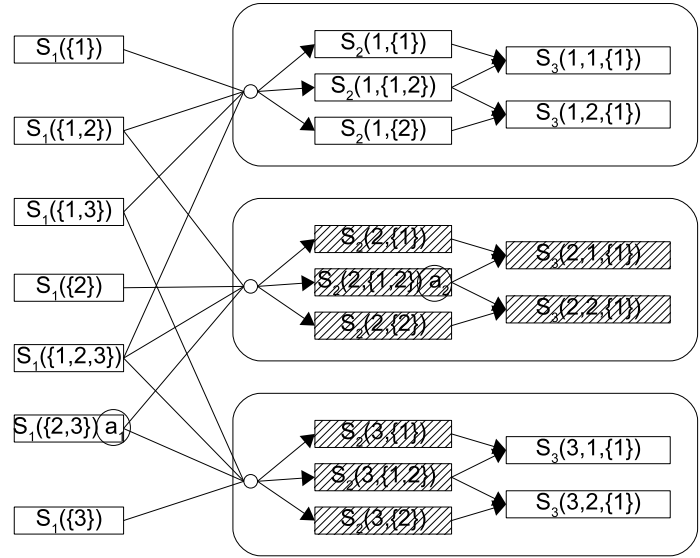} }
\caption{A structure of a vertex of a learning graph for 4-distinctness. The vertex belongs to $V_2(\{2,3\},\{1,2\})$. If there is an arrow between two subsets, a match in the first one is enough to uncover an element in the second one. After $a_1$ is added to $S_1(\{2,3\})$ and $a_2$ is added to $S_2(2,\{1,2\})$, $x$ and $y$ disagree if there is a fault in one of the hatched subsets.}
\label{fig:figa}
\end{figure}

Again, for each arc $A^v_j$, we define a positive semi-definite matrix $X^v_j$ so that $X_j$ in~\refeqn{advDual} are given by $\sum_v X^v_j$. 
The matrix $X^v_j$ is defined as $\sum_\alpha Y_\alpha$ where the sum is over all assignments $\alpha$ on $S$.  For the arcs on stage I.1, $Y_\alpha$ are defined as in~\refeqn{tablica1}, and the arcs on stage I.$\stage$, for $\stage>1$, are defined as in~\refeqn{tablica} with $\alpha(S_{\stage-1})$ replaced by $\bigcup_{K\ni d_{\stage-1}} \alpha(S_{\stage-1}(d_1,\dots,d_{\stage-2}, K))$. 

Now consider stage II.$\stage$. Let $A_j^S$ be an arc with $S\in V_\stage(D_1,\dots,D_{\stage-1})$. In this case, $Y_\alpha = \prob \psi\psi^*$ where
\[
\psi\elem[z] = \begin{cases}
1/{\sqrt{w}},& \text{$f(z)=1$, and $z$ satisfies $\alpha$ and the arc $A_j^S$;}\\
\sqrt{w},& \text{$f(z)=0$, $z$ satisfies $\alpha$, and $\stage+|D_1|+\cdots+|D_{\stage-1}|$ is odd;}\\
-\sqrt{w},& \text{$f(z)=0$, $z$ satisfies $\alpha$, and $\stage+|D_1|+\cdots+|D_{\stage-1}|$ is even;}\\
0,&\text{otherwise.}
\end{cases}
\]
Thus, depending on the parity of $\stage+|D_1|+\cdots+|D_{\stage-1}|$, $X_j^S$ consists of the blocks of one of the following two types:
\begin{equation}
\label{eqn:tablicaFinal}
\begin{array}{r|cc|}
& x & y \\
\hline
x & \prob/w & \prob\\
y & \prob & \prob w\\
\hline
\end{array}\qquad\mbox{or}\qquad
\begin{array}{r|cc|}
& x & y \\
\hline
x & \prob/w & -\prob\\
y & -\prob & \prob w\\
\hline
\end{array}
\end{equation}

\paragraph{Complexity} The complexity analysis follows the same lines as in \refsec{complexity}. The complexity of stages I.$s$ is proved similarly, by taking $R_i = \bigcup_{d_1,\dots,d_{i-1}, D} R_i(d_1,\dots,d_{i-1},D)$, and noting that $|R_i| = O(k!)r_i = O(r_i)$.  Now, having a match in $R_{i-1}$ is not sufficient for an element in $R_i$ to be uncovered, but this only reduces the complexity. The analysis of stage II.$\stage$ is also similar, but this time instead of one arc loading an element $a_\stage$ for a fixed choice of $x$ and $R\in V_1$, there are $2^{O(k^2)}=O(1)$ of them.

\subsection{Feasibility} 
Fix inputs $x\in f^{-1}(1)$ and $y\in f^{-1}(0)$, and let $R\in V_1$ be a choice of the internal randomness consistent with $x$.  Compared to the learning graph in \refsec{first}, many arcs of the form $A^v_j$ are taken for a fixed $j\in [\vars]$.  Let $\cZ$ be the set of arcs taken for this choice of $x$ and $R$. The complete list is in \refsec{construction}. We prove that
\begin{equation}
\label{eqn:sum} 
\sum_{A^v_j\in \cZ\colon x_j\ne y_j} X^v_j\elem[x,y] = \prob.
\end{equation}
Since, again, no arc is taken for two different choices of $R\in V_1$, this proves feasibility~\rf(eqn:advDualCondition). 

If $x$ and $y$ disagree on $R$ then~\refeqn{sum} holds. The reason is similar to the proof of \refthm{collision}:  It is not hard to check that there exists $i\in[r]$ such that $x$ and $y$ disagree on $R\cap \{t_1,\dots,t_{i'}\}$ if and only if $i'\ge i$. Let $j=t_i$, $T=\{t_1,\dots,t_{i-1}\}$, $S = R\cap T$ and $S' = R\cap (T\cup\{j\})$. We claim that $X^{(S,R)}_j\elem[x,y]=\prob$ and $x_j\ne y_j$. 

Indeed, let $\alpha$ be the assignment $x$ and $y$ both satisfy on $S$, and let $\alpha_x$ and $\alpha_y$ be the assignments $x$ and $y$, respectively, satisfy on $S'$.  As, for any $i$, the elements of $R_{i+1}$ are loaded only after all the elements of $R_i$ have been loaded, we get that $\alpha(t)=\alpha_x(t)=\alpha_y(t)$ for all $t\in T$.  Since $x$ and $y$ disagree on $S'$, it must hold that $\alpha_x(j)\ne \alpha_y(j)$.  Hence, $x_j\ne y_j$, and at least one of them is uncovered on $S'$. Thus, $X^{(S,R)}_j\elem[x,y]=\prob$ by~\refeqn{tablica1} or~\refeqn{tablica}, in dependence on whether $A^{(S,R)}_j$ belongs to stage I.1 or not.

We claim the contribution to the sum in~\refeqn{sum} from the arcs in $\cZ$ loading $t_{i'}$ for $i'\in [r+k]\setminus\{i\}$ is zero. For $i'>i$, this follows from that $x$ and $y$ disagree before loading $t_{i'}$.  Now consider $i'<i$.  Inputs $x$ and $y$ agree on $S = R\cap\{t_1,\dots,t_{i'}\}$.  Let $j'=t_{i'}$ and $\alpha$ be the assignment $x$ and $y$ both satisfy on $S$.  We have that either $x_{j'}=y_{j'}$, or they both are represented by $\star$ in $\alpha$.  In both cases, the contribution is zero (in the second case, by~\refeqn{tablica}).

Now assume $x$ and $y$ agree on $R$. The contribution to~\refeqn{sum} from the arcs of stages I.$\stage$ is 0 by the same argument as in the previous paragraph. Let $\stage$ be the first element such that $x_{a_\stage}\ne y_{a_\stage}$. We claim that if $\stage'\ne\stage$, the contribution to~\refeqn{sum} from the arcs $A^S_{a_{\stage'}}\in \cZ$ with $S\in V_{\stage'}$ is 0.

Indeed, if $s'<s$ then $x_{a_{\stage'}}=y_{a_{\stage'}}$. If $\stage'>\stage$, then for each choice of $(D_i)_{i\in[\stage'-1]}$, $x$ and $y$ disagree on $R[D_1\gets a_1,\dots, D_{\stage'-1}\gets a_{\stage'-1}]$, because, by~\rf(clm:Rextended), all $a_i$ with $i<\stage'$ are uncovered in the assignment of $x$.

The total contribution from the arcs $A^S_{a_\stage}\in \cZ$ with $S\in V_\stage$ is $\prob$. This is a special case of \reflem{1} below. Before stating the lemma we have to introduce additional notations.  For a vertex $S = R[D_1\leftarrow a_1,\dots,D_\ell\leftarrow a_\ell]$ with $\ell<\stage$, define the {\em block} on this vertex as the set of vertices
\[\cB(S) = \left\{R[D_1\gets a_1,\dots, D_{\stage-1}\leftarrow a_{\stage-1}] \mid \mbox{$\emptyset\subset D_i\subseteq [k-i]$ for $i=\ell+1,\dots,\stage-1$}\right\}. \]
Also, define the {\em contribution} of the block as
$\cC(S) = \sum_{S'\in \cB(S)} X^{S'}_{a_\stage}\elem[x,y]$. We prove the following lemma by induction on $\stage-\ell$:
\begin{lem}
\label{lem:1}
Let $R$ and $s$ be as above.  If $x$ and $y$ agree on $S = R[D_1\leftarrow a_1,\dots,D_\ell\leftarrow a_\ell]$ then the contribution from the block on $S$ is $(-1)^{\ell+|D_1|+\cdots+|D_{\ell}|} \prob$. Otherwise, it is 0.
\end{lem}

Note that if $\ell=0$, the lemma states that the contribution of the block on $R$ is $\prob$. But this contribution is exactly from all the arcs of the form $A^S_{a_\stage}$ from $\cZ$.  This proves~\refeqn{sum}.

\pfstart[Proof of \reflem{1}]
If $x$ and $y$ disagree on $S$, they disagree on any vertex from the block, hence, the contribution is 0. 

So, assume $x$ and $y$ agree on $S$. If $\ell=\stage-1$, there is only $S$ in the block. Hence, the contribution is $(-1)^{\ell+|D_1|+\cdots+|D_{\ell}|} \prob$ by~\refeqn{tablicaFinal}, because $x$ and $y$ agree on $S$ and $x_{a_\stage}\ne y_{a_\stage}$.
Now assume the lemma holds for $0<\ell<\stage$, and let us prove it for $\ell-1$.  
Fix $S = R[D_1\leftarrow a_1,\dots,D_{\ell-1}\gets a_{\ell-1}]$.  
The block $\cB(S)$ can be expressed as the following disjoint union:
\[\cB(S) = \bigsqcup_{\emptyset\subset D_{\ell}\subseteq [k-\ell]} \cB(R[D_1\leftarrow a_1,\dots,D_{\ell-1}\leftarrow a_{\ell-1},D_{\ell}\leftarrow a_{\ell}]).\]
Let $I$ be the set of $d_{\ell}\in [k-\ell]$ such that $\bigcup_{D} R_{\ell+1}(\mu(D_1),\dots,\mu(D_{\ell-1}),d_{\ell}, D)$ does not contain a fault.

\begin{clm}
\label{clm:xyagree}
The inputs $x$ and $y$ agree on $R[D_1\leftarrow a_1,\dots,D_{\ell}\leftarrow a_{\ell}]$ if and only if $D_{\ell}\subseteq I$.
\end{clm}

\pfstart
The inputs $x$ and $y$ disagree on this vertex if and only if there is a fault in one of $R_{\ell+1}(d_1,\dots,d_{\ell}, D)$ that became uncovered after the addition of $a_{\ell}$.  As $a_{\ell}$ is added to $R_{\ell} (\mu(D_1),\dots,\mu(D_{\ell-1}), D_{\ell})$, equation~\rf(eqn:satisfiesAlpha) indicates that we are only interested in the faults with $d_i = \mu(D_i)$ for all $i\in[\ell-1]$ and $D_{\ell}\ni d_{\ell}$.  Hence, $x$ and $y$ disagree if and only if $D_\ell \not\subseteq I$.
\pfend

Since $y_{a_1}=\cdots=y_{a_{\stage-1}}=x_{a_1}$ and there are at most $k-1$ elements in $y$ equal to $x_{a_1}$, there are at most $k-1-(\stage-1) < k-\ell$ faults. Hence, $I$ is non-empty. Using the inductive assumption and~\rf(clm:xyagree),
\begin{align*}
\cC(S) &= \sum_{\emptyset\subset D_{\ell}\subseteq [k-\ell]} \cC(R[D_1\leftarrow a_1, D_{\ell-1}\leftarrow a_{\ell-1}, \dots,D_\ell\leftarrow a_\ell]) \\
&= \sum_{\emptyset\subset D_{\ell}\subseteq I} (-1)^{\ell+|D_1|+\cdots+|D_{\ell}|}\prob = (-1)^{\ell-1+|D_1|+\cdots+|D_{\ell-1}|}\prob,
\end{align*}
by inclusion-exclusion.
\pfend

\section{Summary}
\label{sec:kdistSummary}
In this chapter, we considered a number of applications of learning graphs beyond the framework of certificate structures.  Two main applications are the special case of the graph collision problem and the $k$-distinctness problem.  We preserved the ``hiding'' intuition from \rf(sec:procedureDriven), but added a number of new ingredients: the weights of arcs that depend on the values of the loaded variables, generalised assignments in vertices of the learning graph, partial dependence on the value of the variable being loaded, inclusion-exclusion-based techniques.

There are more problems that can be approached with the dual adversary SDP.  We find the graph collision problem (in its general form) most interesting.  We saw in \rf(chp:cert) that the largest possible quantum query complexity of a function with 1-certificate complexity 2 is $\Theta(\vars^{2/3})$.  However, the size of the alphabet increases with $\vars$.  The graph collision problem is equivalent to a {\em Boolean} function with 1-certificate complexity 2.  Similarly, it is also possible to analyse Boolean function with 1-certificate complexity $k=O(1)$.

Another open problem is whether the $k$-distinctness algorithm from \rf(thm:kdist) is optimal.

In general, the learning graph approach seems more flexible than the approach based on the quantum walk on the Johnson graph.  We were able to analyse more complicated underlying graphs, and we did it without any spectral analysis.  Also, we did not have to bother about the internal organisation of the algorithm.  One aspect of this is that algorithms based on the dual adversary SDP have built-in amortisation.  In the analysis of stage I of the algorithms in Theorems~\ref{thm:collision} and~\ref{thm:kdist} (and, to smaller extent, \rf(prp:associativity)), it was sufficient to calculate the {\em average} complexity of the stage.  Ordinary quantum algorithms usually have to wait until all computations in the superposition are finished, hence, they spend the {\em maximal} complexity on the subroutine.  Converting it to the average is non-trivial, see, e.g.,~\cite{ambainis:searchVariableTimes, ambainis:amplificationVariableTimes}.

One drawback of the approach based on the dual adversary SDP is that it only gives query-efficient algorithms, and says nothing about their time complexity.  We solve this problem, to some extent, in Part III of the thesis.  In particular, we obtain a time-efficient implementation for the 3-distinctness problem.  However, the resulting algorithm is rather different from the one presented in this chapter.

%% file: _claws.tex
\newcommand{\tcH}{{\cal \tilde H}}
\newcommand{\comment}[1]{}

\def\adjoint{*}
\newcommand{\binomial}[2]{\ensuremath{\left(\begin{smallmatrix}#1 \\ #2 \end{smallmatrix}\right)}}

\renewcommand{\target}{\tau}

\renewcommand{\im}{\mathop{\mathrm{im}}}

\renewcommand{\bra}[1]{{\langle#1|}}
\newcommand{\braket}[2]{{\langle#1|#2\rangle}}
\newcommand{\ketbra}[2]{{\ket{#1}\!\bra{#2}}}
\newcommand{\lbra}[1]{{\bra{\overline{#1}}}}
\newcommand{\lket}[1]{{\ket{\overline{#1}}}}
\newcommand{\bigabs}[1]{{\big\lvert #1\big\rvert}}
\newcommand{\Bigabs}[1]{{\Big\lvert #1\Big\rvert}}

\newcommand{\bignorm}[1]{{\big\| #1 \big\|}}
\newcommand{\Bignorm}[1]{{\Big\| #1 \Big\|}}
\newcommand{\Biggnorm}[1]{{\Bigg\| #1 \Bigg\|}}
\newcommand{\trnorm}[1]{{\| #1 \|_{\mathrm{tr}}}}
\newcommand{\bigtrnorm}[1]{{\bigl\| #1 \bigr\|_{\mathrm{tr}}}}	% by not calling \bignorm, the subscript height is independent of the argument
\newcommand{\Bigtrnorm}[1]{{\Bigl\| #1 \Bigr\|_{\mathrm{tr}}}}
\newcommand{\Biggtrnorm}[1]{{\Biggl\| #1 \Biggr\|_{\mathrm{tr}}}}

In this chapter, we use the computational model of a span program from \rf(sec:spanPrograms) to come up with efficient quantum algorithms for some problems on graphs.  First, we present a new quantum algorithm for the $st$-connectivity problem, that uses exponentially less space and, in many cases, runs faster compared to the previous known algorithm.  Second, we give an optimal quantum algorithm for detecting presence of paths and claws of arbitrary fixed size in a graph given by its adjacency matrix.
Although \rf(thm:spanAlgorithm) only claims the existence of a query-efficient quantum algorithm corresponding to a span program, we are able to implement our span programs time-efficiently.

This chapter is based on the following paper:
\begin{itemize}
\item[\cite{belovs:learningClaws}]
A.~Belovs and B.~W. Reichardt.
\newblock Span programs and quantum algorithms for $st$-connectivity and claw
  detection.
\newblock In {\em Proc. of 20th ESA}, volume 7501 of {\em LNCS}, pages
  193--204, 2012, 1203.2603.
\end{itemize}

\section{Preliminaries} 
\label{sec:graphPreliminaries}

All graphs in this chapter are considered simple, i.e, undirected and without parallel edges.  
Let $K_n$ be the complete graph on~$n$ vertices, and 
let $K_{m,n}$ be the complete bipartite graph with the parts of sizes $m$ and~$n$.
A {\em star} is a complete bipartite graph of the form $K_{1,m}$, and the {\em claw} is the star~$K_{1,3}$.  %A $\{k_1, k_2, k_3\}$-claw is a claw with the edges subdivided into non-intersecting paths of lengths $k_1, k_2$ and~$k_3$; it has $k_1 + k_2 + k_3 + 1$~vertices.  
%A $\{k_1,\dots,k_m\}$-star is the star $K_{1,m}$ with the edges subdivided into non-intersecting paths of length $k_1,\dots,k_m$. It has $k_1 + k_2 + k_3 + 1$ vertices.  A claw is the star~$K_{1,3}$.  %For the subdivided claws we use notation $\{k_1,k_2,k_3\}$-claws.  
%We give efficient quantum algorithms for detecting paths, subdivided claws, and stars with two subdivided legs (i.e., $\{k_1,k_2,1,\dots,1\}$-stars.)  See \reffig{graphs}.  
%Except in the formulation of the Perron-Frobenius theorem, all graphs we consider are simple.  

A graph $T$ is said to be a {\em subgraph} of a graph $G$, if $T$ can be obtained from $G$ by repeatedly deleting edges and isolated vertices.  The subgraph of $G$, {\em induced} by a subset $V$ of the vertex set of $G$, has $V$ as its vertex set, all edges of $G$ having their both endpoints in $V$, and only them.

A graph $T$ is a {\em minor} of $G$, if it can be obtained from $G$ by deleting and contracting edges, and deleting isolated vertices.  To {\em contract} an edge $uv$ is to replace $u$ and $v$ by a new vertex that is adjacent to the union of the neighbours of $u$ and $v$.
There is an alternative way of describing the minor relation.  Let $H$ be a graph, and $\{V_x\}$, where~$x$ runs through all the vertices of a graph $T$, be a collection of pairwise disjoint subsets of vertices of~$H$ such that the subgraph of $H$, induced by $V_x$, is connected for each $x$.  We write $H = M T$ if the following holds: there is an edge $uv$ in $T$ if and only if there is an edge between a vertex of $V_x$ and a vertex of $V_y$.  If this holds, the sets $V_x$ are called the {\em branch sets} of $MT$.  A graph~$T$ is contained in $G$ as a minor if and only if some $MT$ is contained in $G$ as a subgraph.

A graph $H$ is called a {\em subdivision} of a graph $T$ if it can be obtained by repeatedly {\em subdividing edges} of $T$.  The subdivision of an edge $uv$ of a graph $G=(V,E)$ results in the graph $(V\cup\{w\}, E\cup\{uw, wv\}\setminus\{uv\})$.  Informally, it places a new vertex $w$ of degree 2 in the middle of the edge $uv$.

\paragraph{$st$-connectivity}
In the $st$-connectivity problem, we are given a simple $n$-vertex graph $G$ with two selected vertices $s$ and~$t$.  As usual in our thesis, the graph $G$ is given by its adjacency matrix, i.e., the $n\times n$ matrix $(z_{uv})$, with~$z_{uv}=1$ if the edge $uv$ is present and $z_{uv} = 0$ otherwise.  As $z_{uv} = z_{vu}$, the problem is given by $\vars = {n\choose 2}$ input variables.
The task is to determine whether there is a path from $s$ to $t$ in $G$.  This problem is also known as {USTCON} or {UPATH}.  Classically, it can be solved in quadratic time (in the number of vertices) by a variety of algorithms, and it is not hard to see that this is optimal.  With more (but still polynomial) time, it can be solved in logarithmic space~\cite{aleliunas:connectivity}, even by a deterministic algorithm~\cite{reingold:connectivity}.

D{\"u}rr {\em et al.}\ gave a quantum algorithm for this problem that makes $O(n^{3/2})$ queries~\cite{durr:quantumGraph}.  In fact, with an approach based on Bor\r{u}vka's algorithm~\cite{boruvka:graph}, they solve a more general problem of finding a minimum spanning forest in~$G$, i.e., a cycle-free edge set of maximal cardinality that has minimum total weight.  In particular, the algorithm outputs the list of the connected components of the graph.  The algorithm's time complexity is also $O(n^{3/2})$ up to logarithmic factors.  The algorithm works by executing a quantum subroutine that uses $O(\log n)$ qubits and requires a QRACM array of size $O(n)$.  The content of the array is updated classically between the runs of the subroutine.

Our algorithm has the same time complexity as that of D\"urr {\em et al.}\ in the worst case, it has logarithmic space complexity, and does not use any QRACM arrays.  Moreover, the time complexity reduces to $\tilde O(n\sqrt{d})$, if it is known that the shortest path between $s$ and $t$, if one exists, has length at most $d$.
This promise often appears in applications as we will see later in the chapter. 
Finally, we note that our algorithm only detects the presence of an $st$-path, and does not output any.

\paragraph{Graph properties}
Another class of functions we consider in this chapter is also related to graphs.  A {\em graph property} is a total Boolean function of the adjacency matrix $(z_{ij})$ that is invariant under permuting the vertices of~$G$, i.e., it is a function of the graph, and not of its representation.  We say that a graph {\em possesses} the property if the value of the function on the graph equals 1.  The property is {\em trivial} if either all graphs or none possess it.  We say the property is {\em monotone} if it is either increasing or decreasing with respect to subgraph relation.  In the first case, adding an edge to a graph possessing the property results in a graph also having the property.  In the second case, removing an edge cannot make the graph lose the property.  
For instance, $st$-connectivity is {\em not} a graph property, because it features two fixed vertices $s$ and $t$.  On the other hand, connectivity (\rf(prp:graphConnectivity)) {\em is} a monotone non-trivial graph property.  

Query complexity of graph properties is a broad topic both classically and quantumly.  
The Aanderaa-Karp-Rosenberg conjecture~\cite{rosenberg:conjecture} states that any non-trivial monotone graph property has deterministic query complexity exactly ${n \choose 2}$ where $n$ is the number of vertices in the graph.  That is, in the worst case, any deterministic algorithm computing the property must query all edges of the graph, the property also known as {\em evasiveness}.  
The conjecture remains unsolved, however, it is known that the complexity is $\Omega(n^2)$ as by result due to Rivest and Vuillemin~\cite{rivest:RosenbergConjecture}.  
The quantum {\em exact} complexity is also known~\cite{buhrman:boundsForSmallError} to be $\Omega(n^2)$.
The randomised complexity is also believed~\cite{saks:nand} to be $\Omega(n^2)$.  But the (bounded-error) quantum query complexity of many monotone graph properties is $o(n^2)$.  Actually, by using a threshold function (\rf(cor:kThresholdUpper)) on the number of edges, one can obtain a monotone graph property with any intermediate polynomial between $n$ and $n^2$ as its quantum query complexity.  The best known lower bound on a general monotone graph property is $\Omega(n^{2/3}\log^{1/6} n)$ as observed by Yao and mentioned in~\cite{magniez:triangle}.

In fact, we have already seen some monotone graph properties in the thesis.  The triangle property from \rf(thm:walkTriangle) is an example, and we also know that its complexity, as per \rf(prp:learningTriangleUpper), is $O(n^{9/7}) = o(n^2)$.  This can be generalised to {\em forbidden subgraph properties} (FSP): Given a finite list $H_1,\dots,H_m$ of graphs, a graph $G$ possesses the property iff it does not contain any of the $H_j$s as a subgraph.  Clearly, it is a monotone graph property.

Similarly, one can define a monotone graph property that evaluates to 1 iff the input graph $G$ does not contain any of the $H_j$s as a {\em minor}.  It is easy to see that this property is {\em minor-closed}: If a graph $G$ does not contain any $H_j$ as a minor, then any minor of $G$ also does not.  Minor-closed properties include properties like being a forest (the forbidden minor is $K_3$), or being embeddable into a fixed two-dimensional manifold.  
Robertson and Seymour have famously shown~\cite{robertson:wagnerConjecture} that any minor-closed property can be described by a finite list of forbidden minors.  They also have developed a cubic-time deterministic algorithm for evaluating any minor-closed graph property~\cite{robertson:minorAlgorithm}.

Childs and Kothari have studied quantum query complexity of general minor-closed graph properties~\cite{childs:graphProperties}.  Mader's theorem~\cite{mader:minors} implies that any minor-closed graph property is {\em sparse}, i.e., for any graph possessing the property, the number of edges is at most linear in the number of vertices.  By \rf(cor:findingAllOnes) this gives a trivial quantum $O(n^{3/2})$ query algorithm.  Childs and Kothari showed that this is optimal for any minor-closed property that is not simultaneously a forbidden subgraph property.  This generalises, for instance, the $\Omega(n^{3/2})$ lower bound on quantum query complexity of planarity by Ambainis \etal~\cite{ambainis:smallOnSets}.  The proof technique is similar to \rf(prp:graphConnectivity).

For minor-closed FSPs, on the other hand, it is always possible to do better, i.e., there exists a $o(n^{3/2})$-query quantum algorithm.  In particular, Childs and Kothari gave a quantum query algorithm for deciding if a graph contains a path of length~$k$ that uses $\tilde O(n)$ queries if $k \leq 4$, $\tilde O(n^{3/2 - 1/(\lceil k/2\rceil-1)})$ queries if $k \geq 9$, and certain intermediate polynomials for $5 \leq k \leq 8$.  For subdivided claws, the quantum query complexity was also given by a polynomial whose exponent approaches $3/2$ as the size of the forbidden subgraph increases.

We make further progress on characterising the quantum query complexity of minor-closed FSPs by giving an optimal quantum $O(n)$ query query algorithm for any minor-closed FSP that is characterised by a {\em single} forbidden subgraph.  The graph is then necessarily a collection of disjoint paths and subdivided claws.  This is optimal.  While the algorithm by Childs and Kothari uses complicated quantum walks that utilise the sparsity of the input graph, our algorithm is built on span programs and can be considered as a generalisation of the $st$-connectivity algorithm.  Moreover, we show that our algorithm can be implemented in $\tilde O(n)$ quantum time.

%\vspace{-.6cm}
\paragraph{Organisation}
In \refsec{st}, we present the algorithm for $st$-connectivity, and analyze its query complexity.  In \refsec{subgraph}, we define the subgraph/not-a-minor promise problem, and solve it for the cases when the subgraph is a subdivided star or the triangle.  In \refsec{K5}, we show that our technique does not work for arbitrary subgraphs.  In \refsec{efficient}, we present a framework for span program evaluation, and prove that the above algorithms can be implemented time efficiently.  %Finally, in \refsec{startwo}, we generalize the reduction given above for path detection and give an $O(n)$-query quantum algorithm for detecting as a subgraph a star with two subdivided~legs.  

\section{Span Program for \texorpdfstring{$st$-Connectivity}{st-Connectivity}} \label{sec:st}

\begin{thm}
\label{thm:st}
Consider the $st$-connectivity problem on a $n$-vertex graph $G$ given by its adjacency matrix.  Assume there is a promise that if $s$ and $t$ are connected by a path, then they are connected by a path of length at most~$d$.  Then, the problem can be solved in $O(n\sqrt{d}) = O(\sqrt{\vars d})$ quantum queries.  
\end{thm}

It is easy to see that the quantum query complexity of the problem is $1$ if $d = 1$, and $\Theta(\sqrt{n})$ if $d = 2$.  If~$d \geq 3$, and~$d = O(1)$, then the algorithm of \refthm{st} is optimal, that can be seen by reduction from the unordered search problem.  The algorithm is also optimal if $d = \Theta(n)$ by an argument similar to \rf(prp:graphConnectivity).

\pfstart[Proof of \rf(thm:st).]
Define a span program $\cP$ using the vector space~$\R^n$, with the vertex set of $G$ as an orthonormal basis.  As usually, for a vertex $u$ of $G$, $e_u$ denotes the element of the basis corresponding to $u$.  The target vector is $\target = e_t - e_s$.  For each pair of distinct vertices $\{u, v\}$, order the vertices arbitrarily and add the input vector $e_u - e_v$ labelled by the presence of the edge~$uv$ in the input graph, i.e., $e_u - e_v$ is available when the entry $(u,v)$ of the adjacency matrix is~$1$. 
The edge orientation is not important since $e_v - e_u = -(e_u - e_v)$.  

Assume that $s$ is connected to~$t$ in~$G$, and let $t = u_0, u_1, \dots, u_m = s$ be a path between them of length $m \le d$.  All vectors $e_{u_i} - e_{u_{i+1}}$ are available, and their sum is $e_t - e_s$.  Thus, the span program evaluates to~$1$.  The positive witness size is at most $d$.  

Now assume that $t$ and $s$ are in different connected components of $G$.  Define the negative witness $w'$ by $\ip<w',e_u> = 1$ if~$u$ is in the connected component of $t$, and 0 otherwise.  Then $\ip<w', \tau> = 1$ and $w'$ is orthogonal to all available input vectors.  Thus, the span program evaluates to~$0$.  Since there are $O(n^2)$ false input vectors, and the inner product of each of them with $w'$ is at most 1, the negative witness size is $O(n^2)$.  

Thus, by \rf(rem:SpanGeometricalMean), the witness size of $\cP$ is $O(n\sqrt{d})$.  By \rf(thm:spanAlgorithm), the quantum query complexity of the problem is~$O(n \sqrt{d})$.
\pfend

One can observe the similarity between this span program and the one from the second proof in \rf(sec:learningProof).  In \refsec{efficient}, we will prove that the $st$-connectivity algorithm can be implemented in $\tilde O(n\sqrt{d})$ time and $O(\log n)$ space.

As a warm-up before the algorithm in \rf(sec:subgraph), let us briefly describe an application of this algorithm for detecting $k$-paths in the input graph, where $k=O(1)$.  At first, we describe a classical algorithm for $k$-path detection from~\cite{alon:colorCoding}.  It is based on the {\em colour-coding} technique.  Let $G$ be the input graph.  Colour each vertex of $G$ uniformly at random with an integer in $\{0,1,\dots,k\}$.  Using the dynamic programming, it is not hard to detect whether $G$ contains a {\em correctly coloured} $k$-path, i.e., one coloured with consecutive integers $0,\dots,k$ from one end to the other.  Indeed, let $V_0$ be the set of all vertices of $G$ coloured with 0.  Let $V_1$ consist of all vertices of $G$ coloured in colour 1 and connected to a vertex in $V_0$.  And so on: Let $V_{i+1}$ consist of all vertices of $G$ coloured in colour $i+1$ and having a neighbour in $V_i$.  There is a correctly coloured path in $G$ if and only if $V_k$ is non-empty.  If the graph is given by its adjacency matrix, this procedure takes time $O(n^2)$.  If $G$ contains a $k$-path, the probability it is coloured correctly is $2k^{-k} = \Omega(1)$.  Thus, we can get sufficiently high probability of success by testing a constant number of different colourings.

We do not know how to perform general dynamic programming quantumly, and because of that we replace it by the $st$-connectivity algorithm.  We perform the same colouring, and construct an ancillary graph $H$ from $G$ as follows.  Add two vertices $s$ and $t$.  Connect $s$ to all vertices of colour 0, and $t$ to all vertices of colour $k$.  Remove all edges of $G$ that do not connect vertices of consecutive colours.  If $G$ contains a correctly coloured $k$-path, then there is a path from $s$ to $t$ in $H$ of length $k+2$.  If $G$ does not contain a $k$-path, $s$ and $t$ are disconnected for any possible colouring.  By \rf(thm:st), this algorithm works in $O(n)$ queries.

Note that, contrary to the classical case, an $st$-path in $H$ does not imply a correctly coloured path in $G$.  Indeed, the path from $s$ to $t$ may zigzag back and forth between consecutive layers of $H$.  This means our quantum ``dynamical programming'' can have false positives.  The classical algorithm can be easily generalised to detect arbitrary fixed trees in $G$.  Quantumly, we are able to generalise it to subdivisions of stars in \rf(thm:star), but the construction becomes more complicated.

\section{Subgraph/Not-a-Minor Promise Problem} 
\label{sec:subgraph}
In this section, we study minor-closed forbidden subgraph properties.  A natural strategy for testing a minor-closed FSP is to take the list of forbidden subgraphs and test the input graph $G$ for the presence of each subgraph one by one.  Let $T$ be a forbidden subgraph from the list.  To simplify the problem of detecting~$T$, we can add the promise that $G$ either contains $T$ as a subgraph or does not contain~$T$ \emph{as a minor}.  We call this problem the {\em subgraph/not-a-minor} promise problem for~$T$.

We develop an approach to the subgraph/not-a-minor problem using span programs.  We first show that the approach achieves the optimal $O(n)$ query complexity in the case when~$T$ is a subdivided star.  In \refsec{claws:triangle}, we extend the approach to the case when~$T$ is a triangle.  In \refsec{K5}, we show that the approach fails for the case $T = K_5$.

Before describing the algorithm, we state a lower bound that proves the optimality of all these algorithms: 

\begin{prp} \label{prp:optimal}
If the graph $T$ has at least one edge, then the quantum query complexity of the subgraph/not-a-minor problem for $T$ is $\Omega(n)$, and the randomised query complexity is $\Omega(n^2)$.
\end{prp}

\pfstart
This is a standard argument by reduction from the unordered search problem; see, e.g.,~\cite{buhrman:distinctness}.  Let $H$ be the smallest connected component of $T$ of size at least 2. Let $H'$ be $H$ with a vertex removed.  Let $G$ be constructed as $T \setminus H$ together with $n$ disjoint copies of $H'$ and $n$ isolated vertices.  The graph~$G$ has $O(n)$ vertices and does not contain a $T$-minor.  

Let $z_{i,j}$, for $i, j \in [n]$, be boolean variables.  Define $G(x)$ as $G$ with the $j$th isolated vertex connected to all vertices of the $i$th copy of $H'$ for all $i, j$ such that $z_{i,j} = 1$.  The graph $G(x)$ contains~$T$ as a subgraph if and only if at least one $z_{i,j}$ is~$1$.  This gives the reduction.  Unordered search on $n^2$ inputs requires $\Omega(n)$ quantum queries (\rf(prp:polThresholdLower)) and, $\Omega(n^2)$ randomised queries (\rf(thm:randomORlower)).  
\pfend

\subsection{Subdivision of a Star} \label{sec:star}
\newcommand{\eee}[2]{e^{(#2)}_{#1}}
\newcommand{\hhh}[2]{h^{(#2)}_{#1}}

In this section, we give an optimal quantum query algorithm for the subgraph/not-a-minor promise problem for a subdivided star.  As a special case, this implies an optimal quantum query algorithm for deciding minor-closed forbidden subgraph properties that are determined by a single forbidden subgraph.  

\begin{thm} \label{thm:star}
Let $T$ be a subdivision of a star.  Then, there exists a quantum algorithm that, given query access to the adjacency matrix of a simple graph $G$ with $n$ vertices, makes $O(n) = O(\sqrt{\vars})$ queries, and, with probability at least~$2/3$, accepts if $G$ contains~$T$ as a subgraph and rejects if $G$ does not contain $T$ as a minor.  
\end{thm}

In \refsec{efficient}, we prove that the algorithm from \refthm{star} can be implemented efficiently, in~$\tilde O(n)$ time and $O(\log n)$ space.  

\pfstart[Proof of \refthm{star}]
We use the colour-coding technique~\cite{alon:colorCoding}.  
Let $T$ be a star with~$d$ legs of lengths $\ell_1, \ldots, \ell_d > 0$.  
Denote the root vertex by~$r$ and the vertex at depth~$i$ along the $j$th leg by~$t_{j,i}$.  
The vertex set of~$T$ is $V_T = \{ r, t_{1,1}, \ldots, t_{1, \ell_1}, \ldots, t_{d, 1}, \ldots, t_{d, \ell_d} \}$.
Refer to \rf(fig:star)(a) for an example.
Colour every vertex~$u$ of $G$ with an element $c(u) \in V_T$ chosen independently and uniformly at random.  For $v \in V_T$, let $c^{-1}(v)$ be its preimage in the set of vertices of~$G$.  We design a span program that 
\itemstart
\item accepts if there is a correctly coloured $T$-subgraph in $G$, i.e., an injection $\iota$ from $V_T$ to the vertices of~$G$ such that the composition $c \circ \iota$ is the identity, and $uv$ being an edge of~$T$ implies that $\iota(u)\iota(v)$ is an edge of~$G$; 
\item rejects if $G$ does not contain $T$ as a minor, no matter the colouring~$c$.
\itemend
If $G$ contains a $T$-subgraph, then the probability it is coloured correctly is at least $|V_T|^{-|V_T|} = \Omega(1)$.
Evaluating the span program for a constant number of independent colourings suffices to detect the presence of~$T$ with probability at least~$2/3$.  

\begin{figure}[p]
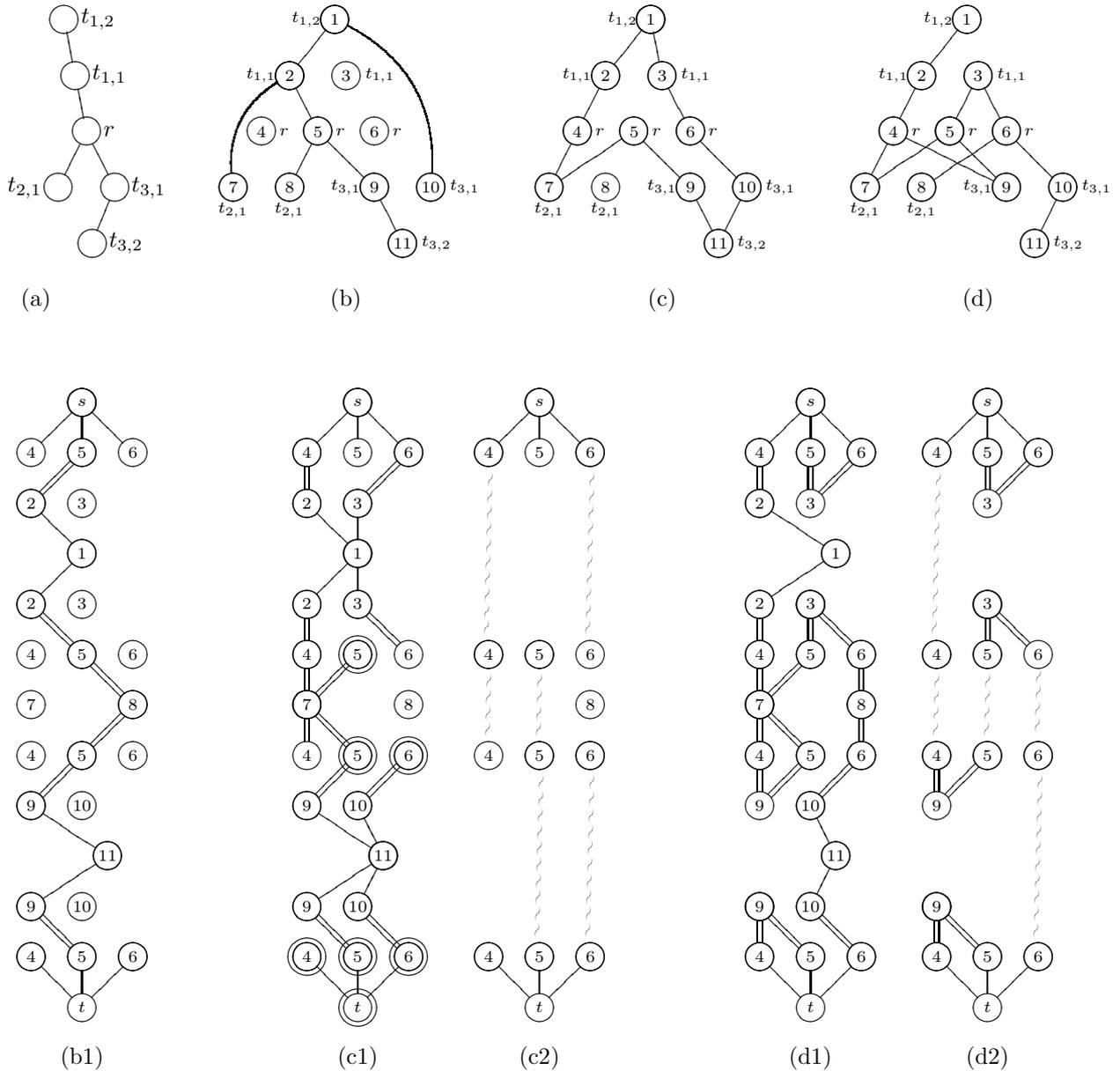

\vspace{-.3in}
\[
\def\objectstyle{\scriptstyle}
\xygraph{!~*{\cir<6pt>{}} !{0;<2pc,0pc>:} \\
{}([r(.6)]*{\textstyle t_{1,2}}) -[r(.2)d]{}([r(.6)]*{ \textstyle t_{1,1}}) -[r(.2)d]{}([r(.4)]*{\textstyle r},
-[l(.5)d]{}([l(.6)]*{\textstyle t_{2,1}}),
-[r(.5)d]{}([r(.6)]*{\textstyle t_{3,1}})-[l(.4)d]{}([r(.6)]*{\textstyle t_{3,2}}) [ld]*{\mbox{(a)}})\\
} \qquad
\xygraph{!~*{\cir<6pt>{}} !{0;<2pc,0pc>:} \\
[r(1.8)]{}([l(.5)]*{t_{1,2}}="")[]*{1}\\
&{}([l(.5)]*{t_{1,1}}="")[]*{2} & {}([r(.6)]*{t_{1,1}}="")[]*{3}\\
[r(.5)] {}([r(.4)]*{r}="")[]*{4} & {}([r(.4)]*{r}="")[]*{5} & {}([r(.4)]*{r}="")[]*{6} \\
{}([d(.4)]*{t_{2,1}}="")[]*{7} & {}([d(.4)]*{t_{2,1}}="")[]*{8} & [r(.5)]{}([l(.5)]*{t_{3,1}}="")[]*{9} & {}([r(.6)]*{t_{3,1}}="")[]*{10}\\
([drr]*{\mbox{(b)}})[rr] & {}([r(.6)]*{t_{3,2}}="")[]*{11}
"1"{}-"2"{}-"5"{}(-"8"{}, -"9"{}-"11"{})
"2"{}-@/_10pt/"7"{}
"1"{}-@/^15pt/"10"{}
}\qquad 
\xygraph{!~*{\cir<6pt>{}} !{0;<2pc,0pc>:} \\
[r(1.8)]{}([l(.5)]*{t_{1,2}}="")[]*{1}\\
&{}([l(.5)]*{t_{1,1}}="")[]*{2} & {}([r(.6)]*{t_{1,1}}="")[]*{3}\\
[r(.5)] {}([r(.4)]*{r}="")[]*{4} & {}([r(.4)]*{r}="")[]*{5} & {}([r(.4)]*{r}="")[]*{6} \\
{}([d(.4)]*{t_{2,1}}="")[]*{7} & {}([d(.4)]*{t_{2,1}}="")[]*{8} & [r(.5)]{}([l(.5)]*{t_{3,1}}="")[]*{9} & {}([r(.6)]*{t_{3,1}}="")[]*{10}\\
([drr]*{\mbox{(c)}})[rr] & {}([r(.6)]*{t_{3,2}}="")[]*{11}
"1"{}-"2"{}-"4"{}-"7"{}-"5"{}-"9"{}-"11"{}-"10"{}-"6"{}-"3"{}-"1"{}
}\qquad
\xygraph{!~*{\cir<6pt>{}} !{0;<2pc,0pc>:} \\
[r(1.8)]{}([l(.5)]*{t_{1,2}}="")[]*{1}\\
&{}([l(.5)]*{t_{1,1}}="")[]*{2} & {}([r(.6)]*{t_{1,1}}="")[]*{3}\\
[r(.5)] {}([r(.4)]*{r}="")[]*{4} & {}([r(.4)]*{r}="")[]*{5} & {}([r(.4)]*{r}="")[]*{6} \\
{}([d(.4)]*{t_{2,1}}="")[]*{7} & {}([d(.4)]*{t_{2,1}}="")[]*{8} & [r(.5)]{}([l(.5)]*{t_{3,1}}="")[]*{9} & {}([r(.6)]*{t_{3,1}}="")[]*{10}\\
([drr]*{\mbox{(d)}})[rr] & {}([r(.6)]*{t_{3,2}}="")[]*{11}
"1"{}-"2"{}-"4"{}(-"9"{}-"5"{})-"7"{}-"5"{}-"3"{}-"6"{}(-"8"{})-"10"{}-"11"{}
} \]
\strut
\[
\def\objectstyle{\scriptstyle}
\xygraph{ !~*{\cir<6pt>{}} !~:{@{}} !{0;<1.8pc,0pc>:} 
*{s}\\
{}(-"s"{})[]*{4}& {}(-"s"{})[]*{5} & {}(-"s"{})[]*{6} \\
{}(-@{=}"5"{})[]*{2} & {}[]*{3}\\
 & {}(-"2"{})[]*{1} \\
{}(-"1"{})[]*{2} & {}[]*{3}\\ 
{}[]*{4}& {}(-@{=}"2"{})[]*{5} & {}[]*{6} \\ 
{}[]*{7} && {}(-@{=}"5"{})[]*{8}\\
{}[]*{4}& {}(-@{=}"8"{})[]*{5} & {}[]*{6} \\ 
{}(-@{=}"5"{})[]*{9} & {}[]*{10}\\
[r(.5)]  & {}(-"9"{})[]*{11} \\
{}(-"11"{})[]*{9} & {}[]*{10}\\
{}()[]*{4}& {}(-@{=}"9"{})[]*{5} & {}[]*{6} \\
&{}(-"4"{},-"5"{},-"6"{})[]*{t}
[d]*{\mbox{(b1)}}
}\qquad\qquad\qquad
\xygraph{ !~*{\cir<6pt>{}} !~:{@{=}} !{0;<1.8pc,0pc>:} 
*{s}\\
{}(-"s"{})[]*{4}& {}(-"s"{})[]*{5} & {}(-"s"{})[]*{6} \\
{}(:"4"{})[]*{2} & {}(:"6"{})[]*{3}\\
& {}(-"2"{},-"3"{})[]*{1} \\
{}(-"1"{})[]*{2} & {}(-"1"{})[]*{3}\\ 
{}(:"2"{})[]*{4}& {}()[]*{5}[]*{\cir<8pt>{}} & {}(:"3"{})[]*{6} \\ 
{}(:"4"{},:"5"{})[]*{7} && {}()[]*{8}\\
{}(:"7"{})[]*{4}& {}(:"7"{})[]*{5}[]*{\cir<8pt>{}} & {}()[]*{6}[]*{\cir<8pt>{}} \\ 
{}(:"5"{})[]*{9} & {}(:"6"{})[]*{10}\\
[r(.5)]  & {}(-"9"{},-"10"{})[]*{11} \\
{}(-"11"{})[]*{9} & {}(,-"11"{})[]*{10}\\
{}()[]*{4}[]*{\cir<8pt>{}}& {}(:"9"{})[]*{5}[]*{\cir<8pt>{}} & {}(:"10"{})[]*{6}[]*{\cir<8pt>{}} \\
&{}(-"4"{},-"5"{},-"6"{})[]*{t}[]*{\cir<8pt>{}}
[d]*{\mbox{(c1)}}
}\qquad
\xygraph{ !~*{\cir<6pt>{}} !~:{@{=}} !{0;<1.8pc,0pc>:} 
*{s}\\
{}(-"s"{})[]*{4}& {}(-"s"{})[]*{5} & {}(-"s"{})[]*{6} \\
\\
\\
\\ 
{}(-@{~}"4"{})[]*{4}& {}()[]*{5} & {}(-@{~}"6"{})[]*{6} \\ 
 && {}()[]*{8}\\
{}(-@{~}"4"{})[]*{4}& {}(-@{~}"5"{})[]*{5} & {}()[]*{6} \\ 
\\
\\
\\
{}()[]*{4}& {}(-@{~}"5"{})[]*{5} & {}(-@{~}"6"{})[]*{6} \\
&{}(-"4"{},-"5"{},-"6"{})[]*{t}
[d]*{\mbox{(c2)}}
}\qquad\qquad\qquad
\xygraph{ !~*{\cir<6pt>{}} !~:{@{=}} !{0;<1.8pc,0pc>:} 
*{s}\\
{}(-"s"{})[]*{4}& {}(-"s"{})[]*{5} & {}(-"s"{})[]*{6} \\
{}(:"4"{})[]*{2} & {}(:"5"{},:"6"{})[]*{3}\\
[r(1.5)] {}(-"2"{})[]*{1} \\
{}(-"1"{})[]*{2} & {}[]*{3}\\ 
{}(:"2"{})[]*{4}& {}(:"3"{})[]*{5} & {}(:"3"{})[]*{6} \\ 
{}(:"4"{},:"5"{})[]*{7} && {}(:"6"{})[]*{8}\\
{}(:"7"{})[]*{4}& {}(:"7"{})[]*{5} & {}(:"8"{})[]*{6} \\ 
{}(:"4"{},:"5"{})[]*{9} & {}(:"6"{})[]*{10}\\
[r(1.5)] {}(-"10"{})[]*{11}  \\
{}[]*{9} & {}(-"11"{})[]*{10}\\
{}(:"9"{})[]*{4}& {}(:"9"{})[]*{5} & {}(:"10"{})[]*{6} \\
&{}(-"4"{},-"5"{},-"6"{})[]*{t}
[d]*{\mbox{(d1)}}
}\qquad
\xygraph{ !~*{\cir<6pt>{}} !~:{@{=}} !{0;<1.8pc,0pc>:} 
*{s}\\
{}(-"s"{})[]*{4}& {}(-"s"{})[]*{5} & {}(-"s"{})[]*{6} \\
 & {}(:"5"{},:"6"{})[]*{3}\\
 \\
 & {}[]*{3}\\ 
{}(-@{~}"4"{})[]*{4}& {}(:"3"{})[]*{5} & {}(:"3"{})[]*{6} \\ 
\\
{}(-@{~}"4"{})[]*{4}& {}(-@{~}"5"{})[]*{5} & {}(-@{~}"6")[]*{6} \\ 
{}(:"4"{},:"5"{})[]*{9} \\
 \\
{}[]*{9} \\
{}(:"9"{})[]*{4}& {}(:"9"{})[]*{5} & {}(-@{~}"6"{})[]*{6} \\
&{}(-"4"{},-"5"{},-"6"{})[]*{t}
[d]*{\mbox{(d2)}} [d]*{\strut}
}
\]
\vspace{-.2in}
\caption{An example of the span program construction from \TeXBug{\refthm{star}}.
(a) Coloured graph $T$.  
(b), (c) and (d) contain graphs on 11 vertices coloured with the vertices of $T$.  (b1), (c1) and (d1) contain the corresponding graphs $H$.  Here, paired edges are given by double lines.  
In (b), there is a correctly coloured $T$-subgraph, and the target vector $t-s$ can be obtained as a sum of available input vectors.
(c) does not contain $T$ as a minor.  There is a path from $s$ to $t$ in $H$, but in $H'$, that is given in (c2), $s$ and $t$ are disconnected.  The wavy lines show the added ancillary edges.  The double circles in (c1) show vertices of $H$ having inner product 1 with the negative witness $w'$.
(d) shows a graph $G$ that does not contain a correctly coloured $T$-subgraph.  (d1) shows the corresponding $H$ graph, and (d2) shows $H'$.  There is a path from $s$ to $t$ in $H'$, hence, there is a $T$-minor in $G$.  The branch sets are as follows.  $r:\{4,9,5,3,6\}$, $t_{1,1}: \{2\}$, $t_{1,2}:\{1\}$, $t_{2,1}: \{7\}$, $t_{3,1}:\{10\}$, $t_{3,2}:\{11\}$.
} \label{fig:star}
\end{figure}

\paragraph{Span program }
The vector space of the span program has the following orthonormal basis 
\begin{equation} \label{eqn:starspanprogrambasis}
\{ s,  t \} \cup 
\sfigB{ \hhh ub \mid c(u) = r,\; b\in \{0,\dots,d\} } \cup
\sfigB{ \eee ub \mid c(u) \ne r,\; b\in [\deg c(u)] }
\end{equation}
where $\deg$ stands for the degree of a vertex in the graph $T$.  That is, there is only vector $\eee u1$ if $c(u) = t_{j,\ell_j}$ is a dangling vertex in $T$; otherwise, there are two vectors $\eee u1$ and $\eee u2$.  However, for notational convenience, we use notation $e_u = \eee u2 = \eee u1$ even if $\deg c(u)=1$.
The target vector is $\tau = t - s$.  For each $u \in c^{-1}(r)$, there are free input vectors $\hhh u0 - s$ and $t - \hhh ud$.  
For each $j \in [d]$, there are the following input vectors: 
\itemstart
\item 
For $i \in [\ell_j - 1]$, $u \in c^{-1}(t_{j,i})$ and $v \in c^{-1}(t_{j,i+1})$, the input vectors $\eee v1 - \eee u1$ and $\eee u2 - \eee v2$ are available when the edge $uv$ is present in~$G$;
\item 
For $u \in c^{-1}(r)$ and $v \in c^{-1}(t_{j,1})$, the input vector $(\eee v1 - \hhh u{j-1}) + (\hhh uj - \eee v2)$ is available when the edge $uv$ is present in~$G$.
\itemend

For visualising and arguing about this span program, it is convenient to define a graph~$H$ whose vertices are the basis vectors in~\refeqn{starspanprogrambasis}.  
Edges of~$H$ correspond to the available input vectors of the span program.  For an input vector with two terms, $\beta - \alpha$, add an edge $\alpha\beta$, and for the four-term input vectors $(\eee v1 - \hhh u{j-1}) + (\hhh uj - \eee v2)$ add two ``paired" edges, $\hhh u{j-1} \eee v1 $ and $\eee v2 \hhh uj$.  For an example, refer to \rf(fig:star).

\paragraph{Positive case.}
Assume that there is a correctly coloured $T$-subgraph in $G$, given by a map $\iota$ from $V_T$ to the vertices of~$G$.  Then, the target vector $t - s$ is achieved as the sum of the input vectors spanned by $s$, $t$ and the basis vectors of the form $\eee ui$ and $\hhh ui$ with $u\in \iota(V_T)$ and $i$ arbitrary.  All these vectors are available.  
This sum has a term $\beta - \alpha$ for each pair of consecutive vertices $\alpha, \beta$ in the following path from $s$ to $t$ in $H$:
\begin{multline*}
s, \hhh{\iota(r)}{0}, \eee{\iota(t_{1,1})}{1}, \eee{\iota(t_{1,2})}{1}, \ldots, 
\eee{\iota(t_{1,\ell_1-1})}{1},
e_{\iota(t_{1,\ell_1})}^{\phantom{(1)}},
\eee{\iota(t_{1,\ell_1-1})}{2} , \ldots, \eee{\iota(t_{1,2})}{2}, \eee{\iota(t_{1,1})}{2}, \hhh{\iota(r)}{1}, \\ \eee{\iota(t_{2,1})}{1}, \ldots, \eee{\iota(t_{2,1})}{2}, \hhh{\iota(r)}{2}, \ldots \ldots, \hhh{\iota(r)}{d}, t.
\end{multline*}
 Pulled back to~$T$, the path goes from~$r$ out and back along each leg, in order.  The positive witness size is $O(1)$, since there are $O(1)$ input vectors along the path.  An example can be found in \rf(fig:star)(b).

%This argument shows much of the intuition for the span program.  A $T$-subgraph is detected as a path from~$s$ to~$t$, starting at a vertex in~$G$ with colour~$r$ and traversing each leg of~$T$ in both directions, out and back.  
It is not enough just to traverse~$T$ in this manner, though, because the path might each time use different vertices of colour~$r$ like in a graph in \rf(fig:star)(c).  The purpose of the four-term input vectors $(\eee v1 - \hhh u{j-1}) + (\hhh uj - \eee v2)$ is to enforce that if the path goes out along an edge $\hhh u{j-1}\eee v1$, then it must return using the paired edge $\eee v2\hhh uj$.

\paragraph{Negative case.}
Assume that $G$ does not contain $T$ as a minor.  It may still happen that $s$ is connected to~$t$ in~$H$.  We construct an ancillary graph $H'$ from $H$ by removing some vertices and adding some extra edges, so that $s$ is disconnected from~$t$ in~$H'$.  Then, we use this graph to construct the negative witness.

The graph $H'$ is defined starting with~$H$.  
Let $T_j = \{ t_{j,1}, \ldots, t_{j, \ell_j} \}$, $H_j = \sfigA{\eee ub \mid c(u) \in T_j,\; b \in \{0,1\}}$ and $R_j = \{ \hhh uj \mid c(u) = r \}$.
Perform the following transformations:
\itemstart
\item For $j \in [d]$ and $u \in c^{-1}(r)$, add the edge $\hhh{u}{j-1}\hhh uj$ to $H'$ if $\hhh {u}{j-1}$ is connected to $R_j$ in~$H$ via a path with all internal vertices (vertices besides the two endpoints) in $H_j$;
\item For $j\in [d]$, remove all vertices in $H_j$ that are connected to both $R_{j-1}$ and $R_j$ in~$H$ via paths with all internal vertices in~$H_j$.
\itemend

Note that in the second case, for each $u \in c^{-1}(T_j)$, either both $\eee u1$ and $\eee u2$ are removed, or neither~is.  Indeed, if there is a path from $\eee u1$ to $R_j$, then it necessarily must pass through a vertex $e_v$ with $\deg c(v) = 1$.  Then backtracking along the path before this vertex, except with the upper index switched $1 \leftrightarrow 2$, gives a path from $\eee u1$ to $\eee u2$.  Similarly, $\eee u1$ is connected to $\eee u2$ if there is a path from $\eee u2$ to $R_{j-1}$.

Define the negative witness $w'$ by $\ip<\alpha, w'> = 1$ if $t$ is connected to $\alpha$ in~$H'$, and $\ip<\alpha, w'> = 0$ otherwise (this includes the case when $\alpha$ is removed from $H'$).  The vector $w'$ is orthogonal to all available input vectors.  In particular, it is orthogonal to any available four-term input vector $(\eee v1 - \hhh u{j-1}) + (\hhh uj - \eee v2)$, corresponding to two paired edges in~$H$, because either the same edges are present in~$H'$, or $\eee v1$ and $\eee v2$ are removed and a new edge $\hhh{u}{j-1}\hhh uj$ is added.  For an example, refer to \rf(fig:star)(c).

In order to verify that $w'$ is a negative witness, it remains to prove that $s$ is disconnected from $t$ in~$H'$, for then $\ip<w',t-s>=1$.  Assume that $s$ is connected to~$t$ in~$H'$, via a simple path~$p$.  Based on the path~$p$, we will construct a $T$-minor in~$G$, giving a contradiction.  

By the structure of the graph~$H'$, $p$ must pass in order through some vertices $\hhh {u_0}{0}$, $\hhh{u_1}{1}, \ldots, \hhh{u_{d}}{d}$, where $c(u_j) = r$ for all $j$.
Consider the segment of the path from $\hhh{u_{j-1}}{j-1}$ to $\hhh{u_j}{j}$.  Due to the construction of $H'$, this segment must cross a new edge added to~$H'$, $\hhh{v_j}{j-1}\hhh{v_j}{j}$, for some $v_j\in c^{-1}(r)$.  Thus, the path~$p$ has the form 
\begin{equation*}
s, \ldots, \hhh{v_1}{0}, \hhh{v_1}{1}, \ldots, \hhh{v_2}{1}, \hhh{v_2}{2}, \ldots \ldots, \hhh{v_d}{d-1}, \hhh{v_d}{d}, \ldots, t \enspace .  
\end{equation*}
Based on this path, we can construct the $T$-minor as follows.  The branch set of the root~$r$ consists of all the vertices in~$G$ that correspond to the vertices along~$p$ (i.e., the corresponding subindices).   
Furthermore, for each edge $\hhh{v_j}{j-1}\hhh{v_j}{j}$, there is a path in~$H$ from $\hhh{v_j}{j-1}$ to $R_j$ with every internal vertex in $H_j$.  The first $\ell_j$ vertices along the path give a minor for the $j$th leg of~$T$.  It is vertex-disjoint from the minors for the other legs because the colours are different.  It is also vertex-disjoint from the branch set of~$r$ because no vertices along the path are present in~$H'$.  (Here, we again use that $\eee u1$ is present in $H'$ if and only if $\eee u2$ is present.)  For an example, refer to \rf(fig:star)(d).

Since each coefficient of $w'$ is zero or one, the inner product of $w'$ with any false input vector is at most two in magnitude.  Since there are $O(n^2)$ input vectors, the negative witness size is $O(n^2)$.  
Thus, the total witness size of the learning graph is $O(n)$.
\pfend

%\begin{figure}
%\vspace{-.4cm}
%\centering
%%\subfigure[$T$]{\raisebox{0cm}{$\!\!\!\!\!\!\!\!$\includegraphics{images/exampleacceptsminorT}}}
%%\subfigure[$G$]{\raisebox{0cm}{\includegraphics{images/exampleacceptsminorG}}}
%%\subfigure[$H$]{\includegraphics[scale=2]{images/exampleacceptsminorH}}
%\begin{tabular}{c@{$\;$}c@{$\quad\!\!$}c}
%\subfigure[$T$]{\raisebox{.95cm}{$\!\!\!\!\!\!\!\!$\includegraphics[scale=.9]{images/exampleacceptsminorT}}}
%&
%\subfigure[$G$]{\raisebox{.95cm}{\includegraphics[scale=.9]{images/exampleacceptsminorG}}}
%&
%\subfigure[$H$]{\includegraphics[scale=1.3]{images/exampleacceptsminorH}}
%\end{tabular}
%\caption{Let $T$ be the subdivided star with four legs of lengths $\ell_1 = \cdots = \ell_4 = 2$.  Then the span program from the proof of \refthm{star} accepts the colored graph~$G$ in~(b), even though $G$ contains~$T$ only as a minor and not as a subgraph.  The corresponding graph~$H$ is shown in~(c), together with a flow that indicates the combination of input vectors adding to $\ketAAA t - \ketAAA s$.  Notice that the flow is balanced at all vertices except~$s$ and~$t$, and also that the flows along paired edges are of equal strengths in opposite directions.  
%} \label{fig:exampleacceptsminor}
%\end{figure}

It can be checked that if $T$ is a path or a subdivision of a claw then a graph $G$ contains $T$ as a minor if and only if it contains $T$ as a subgraph.  Moreover, disjoint collections of paths and subdivided claws are the only graphs~$T$ with this property.  This implies the following corollary: 

\begin{cor} \label{cor:starforest}
Let $T$ be a collection of vertex-disjoint subdivided stars.  Then there exists a quantum algorithm that, given query access to the adjacency matrix of a simple graph $G$ with $n$ vertices, makes $O(n)$ queries, and, with probability at least~$2/3$, accepts if $G$ contains~$T$ as a subgraph and rejects if~$G$ does not contain $T$ as a minor.  
\end{cor}

\pfstart
It is not enough to apply \refthm{star} once for each component of~$T$, because some components might be subgraphs of other components.  Instead, proceed as in the proof of \refthm{star}, but for each fixed colouring of~$G$ by the vertices of~$T$ run the span program once for every component on the graph $G$ restricted to vertices coloured by that component.  This ensures that in the negative case, if the span programs for all components accept, then there are vertex-disjoint minors for every component, which together form a minor for~$T$.  
\pfend

%Paired edges are more complicated to work with than ordinary edges.  For implementing the quantum algorithm time efficiently, in \refthm{efficient} below, we will therefore work with a slightly different span program in which the vertices $\ketAAA{u, b}$, for $(u, b) \in c^{-1}(r) \times [d-1]$, are split in four and the vertices $\ketAAA{u, b}$, for $(u, b) \in c^{-1}(r) \times \{0, d\}$ are split in two.  The modified span program computes the same function on allowed input graphs~$G$, with nearly the same witness size, but has the advantage that any vertex is incident to at most one paired edge.  

\begin{figure}[bt]
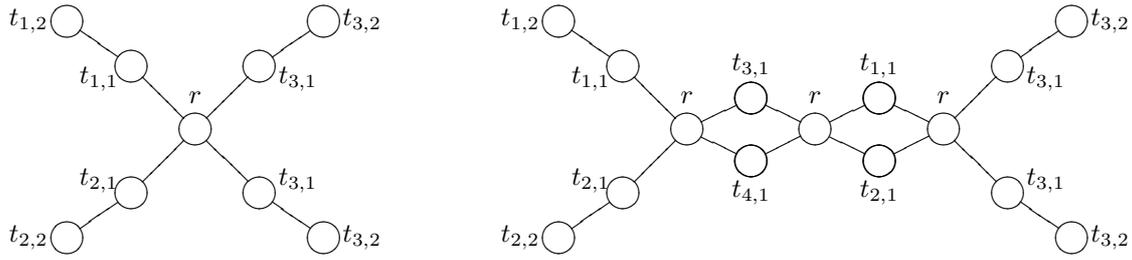

\vspace{-.3in}
\[
\xygraph{!~*{\cir<6pt>{}} !{0;<2pc,0pc>:} \\
{}([u(.5)]*{\textstyle r},
-[lu]{}([d(.2)l(.5)]*{\textstyle t_{1,1}}) - [u(.7)l]{}([l(.6)]*{\textstyle t_{1,2}}),
-[ld]{}([u(.2)l(.5)]*{\textstyle t_{2,1}}) - [d(.7)l]{}([l(.6)]*{\textstyle t_{2,2}}),
-[ru]{}([d(.2)r(.6)]*{\textstyle t_{3,1}}) - [u(.7)r]{}([r(.6)]*{\textstyle t_{3,2}}),
-[rd]{}([u(.2)r(.6)]*{\textstyle t_{3,1}}) - [d(.7)r]{}([r(.6)]*{\textstyle t_{3,2}})
)} \qquad\qquad
\xygraph{!~*{\cir<6pt>{}} !{0;<2pc,0pc>:} \\
{}([u(.5)]*{\textstyle r},
-[lu]{}([d(.2)l(.5)]*{\textstyle t_{1,1}}) - [u(.7)l]{}([l(.6)]*{\textstyle t_{1,2}}),
-[ld]{}([u(.2)l(.5)]*{\textstyle t_{2,1}}) - [d(.7)l]{}([l(.6)]*{\textstyle t_{2,2}}),
-[u(.5)r]{} ([u(.5)]*{\textstyle t_{3,1}}) []="1",
-[d(.5)r]{} ([d(.5)]*{\textstyle t_{4,1}}) []="2")
[rr]{} ([u(.5)]*{r}, -"1"{}, -"2"{},
-[u(.5)r]{} ([u(.5)]*{\textstyle t_{1,1}}) []="1",
-[d(.5)r]{} ([d(.5)]*{\textstyle t_{2,1}}) []="2")
[rr]{} ([u(.5)]*{r}, -"1"{}, -"2"{},
-[ru]{}([d(.2)r(.6)]*{\textstyle t_{3,1}}) - [u(.7)r]{}([r(.6)]*{\textstyle t_{3,2}}),
-[rd]{}([u(.2)r(.6)]*{\textstyle t_{3,1}}) - [d(.7)r]{}([r(.6)]*{\textstyle t_{3,2}})
)}
\]
\caption{A graph $T$ on the left, and a coloured graph $G$ on the right.  It is not hard to check that the learning graph from \TeXBug{\rf(thm:star)} accepts on $G$, although that does not contain a $T$-subraph.}
\label{fig:4star}
\end{figure}

Finally, we note that it is not possible to replace the subgraph/not-a-minor promise problem in the formulation of \rf(thm:star) by the ordinary subgraph containment problem.  The span program can accept even if $G$ does not contain $T$ as a subgraph, see \rf(fig:4star).

\subsection{Triangle}
\label{sec:claws:triangle}
The technique used in the proof of \refthm{star} extends to other problems.  As an example, we consider the case of~$T$ being a triangle.
As per \rf(prp:learningTriangleUpper), a $K_3$-subgraph in an $n$-vertex graph $G$ can be detected in $O(n^{9/7})$ quantum queries.  We show that if $G$ is promised not to contain a $K_3$-minor in the negative case (i.e., it is a forest), this problem can be solved in linear number of queries.  As forests are sparse, it is also apt to mention an $O(n^{7/6})$-query quantum algorithm for finding triangles in sparse graphs~\cite[Theorem~4.4]{childs:graphProperties}.

\begin{thm} \label{thm:claws:triangle}
There exists a quantum algorithm that, given query access to the adjacency matrix of a simple graph $G$ with $n$ vertices, makes $O(n) = O(\sqrt{\vars})$ queries and distinguishes the cases when $G$ contains a triangle and when $G$ is a forest, except with error probability at most~$1/3$.  
\end{thm}

\pfstart
The algorithm is similar to the one in \refthm{star}.  Let $c$ be a uniformly random map from the vertex set~$V_G$ of~$G$ to $\{ 0, 1, 2 \}$.  Define a span program on the vector space with orthonormal basis 
\begin{equation} \label{eqn:trianglespanprogrambasis}
\{ s, t \} \cup \{ \eee {u}{c(u)} : u \in V_G \} \cup \{ \eee u3 \mid u \in c^{-1}(0) \}.
\end{equation}
The target vector is again $\tau = t - s$.  The free input vectors are $t - s + \eee u0 - \eee u3$ for $u \in c^{-1}(0)$.  For $j \in \{0,1,2\}$ and $(u,v) \in c^{-1}(j) \times c^{-1}(j+1 \bmod 3)$, add an input vector $\eee{v}{j+1} - \eee{u}{j}$ that is available iff the edge $uv$ is present in~$G$.  

The intuition behind this construction is similar to \refthm{star}.  By using a four-term input vector $t - s + \eee u0 - \eee u3$ for $u \in c^{-1}(0)$, instead of two separate input vectors $\eee u0 - s$ and $t - \eee{u}{3}$, we prevent the span program from accepting paths $u_0, u_1, u_2, v_0$ with $c(v_0) = 0$ but $v_0 \neq u_0$.  

Let us make this intuition precise.  The positive case is straightforward:  If $G$ contains a triangle, then the triangle is coloured correctly with probability $2/9$.  (By a correct colouring, we mean a colouring that assigns distinct colours to different vertices of the triangle.)
Assume the triangle is $\{u_0, u_1, u_2\}$, with $c(u_j) = j$.  Since the sum of the input vectors $t - s + \eee{u_0}{0} - \eee{u_0}{3}$, $\eee{u_1}{1} - \eee{u_0}{0}$, $\eee{u_2}{2} - \eee{u_1}{1}$ and $\eee{u_0}{3} - \eee{u_2}{2}$ equals $t - s$, the span program accepts.  The witness size is~$3$.  

%Let~$H$ be the graph with a vertex for every basis vector in Eq.~\refeqn{trianglespanprogrambasis}, two paired edges $(\ketAAA s, \ketAAA{u,0})$ and $(\ketAAA{u,3}, \ketAAA t)$ for every $u \in c^{-1}(0)$, and an edge $(\ketAAA{u,j}, \ketAAA{u',j'})$ for every available input vector $\ketAAA{u',j'} - \ketAAA{u,j}$.  The intuition behind this span program is that the triangle corresponds to a flow in $H$ along the path $\ketAAA s, \ketAAA{u_0, 0}, \ketAAA{u_1, 1}, \ketAAA{u_2, 2}, \ketAAA{u_0, 3}, \ketAAA t$.  The paired edges $(\ketAAA s, \ketAAA{u_0,0})$ and $(\ketAAA{u_0,3}, \ketAAA t)$ force the path to pass through both $\ketAAA{u_0, 0}$ and $\ketAAA{u_0, 3}$.  Without these edges paired, the span program would accept if $G$ contained a path $u_0, u_1, u_2, u_0'$ with $c(u_0') = 0$ but $u_0' \neq u_0$.

For the negative case, assume that $G$ is acyclic.  We argue that the span program rejects by constructing a negative witness~$w'$.  Unlike \refthm{star}, the coefficients of $w'$ will not be only~$0$ or~$1$, and the worst-case negative witness size is $\Theta(n^4)$.  We will, however, prove that the expected (with respect to the colouring) negative witness size is $O(n^2)$.

Fix arbitrarily a root for every tree component of~$G$, and measure depths in every component of $G$ from these root vertices.  Let~$H$ be the same graph as~$G$, except with edges connecting vertices of the same colour removed.  For every tree component in~$H$, set the root to be the unique vertex in that component with the least depth in~$G$.  For a vertex~$u$, let $d(u)$ be its depth in~$H$. 
Observe that because $G$ is acyclic, every edge is removed independently with probability $1/3$ when going from $G$ to~$H$.  

Let $H'$ be the same as $H$ but with each vertex $u \in c^{-1}(0)$ split into two vertices: $\eee u0$ and $\eee u3$, so that $\eee u0$ is connected to $u$'s neighbours of colour~$1$, and $\eee u3$ is connected to $u$'s neighbours of colour~$2$.  Also, add an edge from $\eee u0$ to $\eee u3$.  Additionally, rename each vertex $u\in c^{-1}(\{1,2\})$ into $\eee u{c(u)}$.  Thus, we get the graph corresponding to our span program with vertices $s$ and $t$ removed.  Clearly, $H'$ is also acyclic.

%, and therefore for any vertex~$u$, the expectation of $d'(u)$ is bounded by $\frac13 \sum_{i=0}^\infty (i+1)^2(2/3)^i = O(1)$.

Using the graph $H'$, we can specify the negative witness $w'$.  Let $\ip<s, w'> = 0$ and $\ip<t, w'>=1$.  The vertices of $H'$ are in one-to-one correspondence with the remaining basis vectors of~\refeqn{trianglespanprogrambasis}, so it is enough to specify the coefficients for each vertex of~$H'$.  Note that for any $u \in c^{-1}(0)$, the condition that $w'$ is orthogonal to the free input vector $t - s + \eee u0 - \eee u3$ is equivalent to $\ip <\eee u0, w'> = \ip< \eee u3, w'> - 1$.  Up to an additive factor, this condition determines the coefficients of $w'$ for each connected component of~$H'$.  Let $r$ be the root of the component.  For a vertex $u$ in the component, define the level $\ell(u)$ as the number of $\eee u0\eee u3$ edges minus the number of $\eee u3\eee u0$ edges traversed along the simple path from $r$ to~$u$.  Let $\ip <u, w'> = \ell(u)$.  Note that $\ell(u) \leq d(u)+1$ because no two new edges are adjacent.  

Unfortunately, the coefficients of~$w'$ may grow as large as $\Omega(n)$, resulting in a negative witness size of order $n^4$.  However, the probability of this event is negligible.  Indeed, the negative witness size is bounded by
\begin{equation*}
\sum_{u, v \in H'} \ip<u-v, w'>^2 \le \sum_{u, v \in H'} 2(\ip<u,w'>^2 + \ip<v,w'>^2) \le 8 n \sum_{u \in H'} (d(u)+1)^2 \enspace ,
\end{equation*}
because $H'$ has at most $2n$ vertices.  For a fixed $u$, the expectation of $(d(u)+1)^2$ is bounded by the series 
\[
\frac13 \sum_{i=0}^{+\infty} (i+1)^2(2/3)^i = O(1).
\]
By the linearity of expectation, the expected size of the negative witness is~$O(n^2)$.  By Markov inequality, for any $\eps > 0$ one may choose $C$ so that the probability the negative witness size exceeds $C n^2$ is less than~$\eps$.  This adds at most~$\eps$ to the algorithm's error probability.  If the negative witness size is at most $C n^2$, then the total witness size is~$O(n)$.  
\pfend

\subsection{A Counterexample for \texorpdfstring{$K_5$}{K5}} \label{sec:K5}

The algorithms in Sections~\ref{sec:star} and~\ref{sec:claws:triangle} suggest a general approach for solving the subgraph/not-a-minor problem for a graph $T$: randomly colour $G$ by the vertices of~$T$, and construct a span program for a traversal of~$H$, using the paired-edge trick to assure that the same vertex of $G$ is chosen for all appearances of a vertex of $T$ in the traversal.  
%Natural candidate graphs to consider next include general trees and cycles.  
In this section, we show that this approach fails for some graphs~$T$.  

\begin{figure}
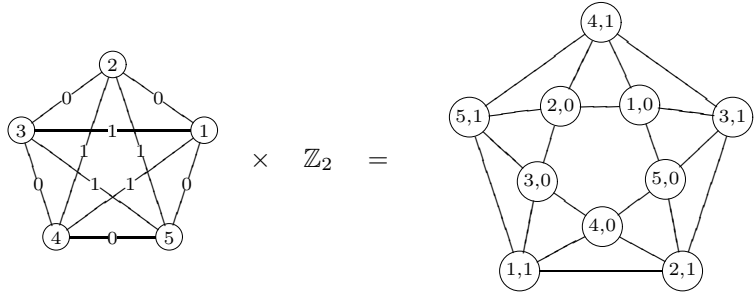
%[thb]
\vspace{-.4cm}
\[
\def\objectstyle{\scriptstyle}
\xygraph{!~-{@{-}|1}
!P5"A"{ ~><{@{-}|0}
[o]=<10pt>{\xypolynode}*\frm{o}}
"A1" - "A3" - "A5" - "A2" - "A4" - "A1"}
%\raisebox{-1.35cm}{\includegraphics{images/K5labeled}}
\quad \times \quad \Z_2 \quad = \qquad 
\xygraph{
!P5"A"{~={55}~:{(0.7,0):}
[o]=<15pt>{\xypolynode,0}*\frm{o}}
!P5"B"{~={234}~:{(1.44,0):}
[o]=<15pt>{\xypolynode,1}*\frm{o}}
"A1" - "B3" - "A5" - "B2" - "A4" - "B1"
- "A3" - "B5" - "A2" - "B4" - "A1"
}
%\raisebox{-1.75cm}{\includegraphics{images/K5skewlabeled}}
\]
\caption{A skew product of $K_5$ and $\Z/2\Z$ gives a planar graph that does not contain~$K_5$ as a minor.  This example is due to Jim Geelen.}
\label{fig:skew}
\end{figure}

Consider the following operation that is a special case of the skew product of a graph and a group~\cite{kumjian:Calgebras}.  Let $T$ be a graph with each edge $e$ marked by an element $s_e$ of the cyclic group $\Z_2$ of size 2.  The skew product of $T$ and $\Z_2$ is the graph $T_2$ with vertices $(v, i)$, where $v$ is a vertex of~$T$ and $i \in Z_2$.  The graph $T_2$ has two edges for each edge $(u, v)$ of~$T$: $(u,i)(v,i+s_{(u,v)})$ for $i \in Z_2$.  

The span program built along the lines of the algorithms from Theorems~\ref{thm:star} and~\ref{thm:claws:triangle} accepts on $T_2$ if it is coloured correctly, i.e., if both vertices $(v,0)$ and $(v,1)$ of $T_2$ are coloured by~$v$.  
Indeed, the positive witness for $G = T_2$ can use all available input vectors with uniform coefficients~$1/2$.  

In general, however, $T_2$ does not contain $T$ as a minor.  For instance, \rf(fig:skew) shows an example of a skew product of $K_5$ and $\Z_2$ that does not contain a $K_5$-minor.  It is easy to check, however, that if $T$ is a tree or a triangle, $T_2$ does contain $T$ as a minor---and even as a subgraph, in the case of a tree.  

This shows that our algorithm does not work for all subgraph/not-a-minor promise problems.  Similarly, one can define a (total) minor-closed forbidden subgraph property for which our algorithm fails.  The property of having as a minor neither $K_5$ nor the eleven-vertex path~$P_{10}$ is a forbidden subgraph property.

\section{Time-Efficient Implementations} 
\label{sec:efficient}
As described in \rf(thm:spanAlgorithm), a span program~$\cP$ can be evaluated by a quantum algorithm that makes $O(\wsize(\cP))$ queries to the input string.  
The algorithm alternates a fixed input-independent reflection $R_\Lambda$ with a simple input-dependent reflection $R_\Pi$.  
The reflection $R_\Pi$ can be implemented efficiently in most cases, but implementing $R_\Lambda$, in general, is difficult.  Since many functions have much larger time complexity than query complexity, this should be expected.  

In this section, we show how to implement $R_\Lambda$ time-efficiently for the span programs from Theorems~\ref{thm:st}, \ref{thm:star} and~\ref{thm:claws:triangle}.
The idea is similar to the MNRS quantum walk from \rf(sec:mnrs).  
The graph's constant spectral gap allows for implementing this reflection to within inverse polynomial precision using only logarithmically many steps of the walk.  The graph's uniform structure allows for implementing each step efficiently.  

\begin{thm} 
\label{thm:efficient}
The algorithm from \refthm{st} can be implemented in $\tilde O(n \sqrt d)$ quantum time, and the algorithms from Theorems~\ref{thm:star} and~\ref{thm:claws:triangle} can be implemented in $\tilde O(n)$ quantum time.  In these implementations, the algorithms from Theorems~\ref{thm:st} and~\ref{thm:star} use $O(\log n)$ qubits of space.
\end{thm}

\paragraph{Preliminaries}
Before we start with the proof of \refthm{efficient}, we state some ancillary facts.  We start with {\em $k$-wise independent hash functions}; see, e.g.~\cite{luby:kWiseHash}.  This is a collection of functions $h_m\colon [n]\to[\ell]$ such that, for any~$k$ distinct elements $a_1,\dots,a_k$, the probability over the choice of~$m$ that $(h_m(a_1),\dots,h_m(a_k))$ takes a particular value in $[\ell]^k$, is $\ell^{-k}$.  The simplest construction, that suffices for our purposes, is to assume that $\ell \le n$ are powers of two, and define~$h_m$ as the $\log_2\ell$ lowest bits of the value of a random polynomial over $GF(n)$ of degree $k-1$.  ($GF(n)$ stands for the finite field with $n$ elements.)  Then, $O(k \log n)$ bits suffice to specify~$h_m$, from which $h_m(a)$ can be calculated in $O(k\log^2 n)$ time.  

We will also need the following simple result from linear algebra.  In the following, $I_n$ is the $n \times n$ identity matrix, and $J_n$ is the $n \times n$ all-1 matrix.  

\begin{lem} \label{lem:blockmatrixspectrum}
Fix $\ell \times \ell$ symmetric matrices $A$ and~$B$.  For $n \in \N$, let $M_n = A \otimes I_n + \frac{1}{n} B \otimes J_n$.  Then the spectrum of~$M_n$, i.e., the set of eigenvalues without 
multiplicities, is independent of~$n$.
\end{lem}

\begin{proof}
Let $\{ u_i\}_{i\in [n]}$ be an orthonormal eigensystem for $J_n/n$, with the corresponding eigenvalues $\lambda_i \in \{0, 1\}$.  For $i \in [n]$, let $M(i) = A + \lambda_i B$.  If $v$ is a $\lambda$-eigenvector of $M(i)$, then $v \otimes u_i$ is a $\lambda$-eigenvector of $M_n$.  These derived eigenvectors span the whole $(\ell n)$-dimensional space, and hence the set of eigenvalues of $M_n$ does not depend on~$n$.  
\end{proof}

Essentially, the above argument works because $I_n$ and~$J_n/n$ commute and have spectra independent of~$n$.  

\paragraph{General approach}
Now, we describe a general approach to span program implementation.  After that, we apply it to the span programs in the proofs of Theorems~\ref{thm:st}, \ref{thm:star} and~\ref{thm:claws:triangle}.  
\mycommand{tV}{\tilde V}

In a time-efficient implementation, we do not allow free input vectors.  As described in \rf(sec:spanPrograms), the implementation of free input vectors accounts for a projection, and we do not want to change the way our vectors are represented in the vector space.  Instead of that, we allow {\em always available} input vectors.  They have the same functionality as free input vectors, but they are not free.  One possible interpretation is that all input strings are extended with a new variable $z_{0}$ with value 0.  Then, always available input vectors may be considered as labelled by the value 0 of $z_{0}$.  More interestingly, we also find use of {\em never available} input vectors.  They may be interpreted as labelled by the value 1 of $z_{0}$.

We require some notations from the proof of \rf(thm:spanAlgorithm).  
Recall that $\{v_i\}_{i\in\cI}$ are the input vectors of the span program, $\cI_0 = \cI\cup\{0\}$, and $v_0 = \tau/\alpha$ is an additional input vector, where $\tau$ is the target vector and $\alpha\in\R$.  The algorithm runs in the vector space $\R^{\cI_0}$ and uses two reflections $R_\Lambda$ and $R_\Pi$.  The first one, $R_\Lambda = 2\Lambda-I$, where $\Lambda$ is the projector onto the kernel of the $d \times \cI_0$ matrix $\tV$ having $\{v_i\}_{i\in\cI_0}$ as its columns.  The second reflection, $R_\Pi$, reflects about the span of the elements of the computational basis corresponding to the available input vectors.

At first, we are free to replace the matrix $\tV$ by a matrix $V'$ obtained by a rescaling of the rows (indeed, this does not affect the kernel).  Next, inspired by the MNRS quantum walk from \rf(sec:mnrs), we ``factor'' $V'$ into two sets of unit vectors as follows.  Let unit vectors $a_k \in \hilbert X$, one for each row $k \in [d]$, and $b_i \in \hilbert Y$, one for each column $i\in\cI_0$, be such that $a_k\elem[i] b_i\elem[k] = V'\elem[k,i]$ for all $k$ and $i$.

\rf(alg:spanProgram) uses an $\cI_0$-qudit $\reg X$.  We add a $d$-qudit $\reg Y$.
Denote by $\{e_i\}$ and $\{h_k\}$ the standard bases of $\R^{\cI_0}$ and $\R^d$, respectively, and let 
\[
A = \sum_{k\in[d]} (a_k \otimes h_k) h_k^*\qquad \mbox{and}\qquad B = \sum_{i\in \cI_0} (e_i \otimes b_i) e_i^*.
\]
Thus, the linear operators $A$ and $B$, map $\R^d$ and $\R^{\cI_0}$, respectively, into $H = \R^{d}\otimes \R^{\cI_0}$.
The quantum algorithm runs in $H$ and we identify $\R^{\cI_0}$ from the algorithm of \rf(thm:spanAlgorithm) with its image under the isometry $B$.

Then, $R_\Pi$ can be implemented in $\im(B)$ by exactly the same procedure as in \rf(alg:spanProgram) applied to the register $\reg X$.
As in \rf(sec:szegedy), we denote $R_A = 2AA^*-I_H$ and $R_B = 2BB^*-I_H$.
The reflection $R_\Lambda$ can be implemented on $\im(B)$ as the reflection about the $(-1)$-eigenspace of~$R_B R_A$.  
Indeed, the corresponding discriminant matrix $D = A^*B = V'$, and by \reflem{szegedy}, the $(-1)$-eigenspace equals $B (\ker V')$ plus a part that is orthogonal to $\im(B)$ and, therefore, irrelevant.

The reflection about the $(-1)$-eigenspace of $R_B R_A$ can be implemented using the phase detection subroutine, \rf(thm:detection), applied to $-R_BR_A$.  The efficiency depends on two factors: 
\begin{enumerate}
\item The implementation costs of $R_A$ and $R_B$.  They can be easier to implement than $R_\Lambda$ directly, because they decompose into local reflections.  The reflection $R_A$ about the span of the columns of~$A$ equals the reflection about $a_i$ controlled by the index of the column~$i$, and similarly for~$R_B$.  

\item The spectral gap around the $(-1)$-eigenvalue of $R_B R_A$ necessary to implement the reflection about the eigenspace.  By \reflem{szegedy}, this gap is determined by the spectral gap of $D = V'$ around the singular value zero.
\end{enumerate}

\paragraph{Proof of \refthm{efficient}}
So far the arguments have been general.  Let us now specialise to the span programs in Theorems~\ref{thm:st}, \ref{thm:star} and~\ref{thm:claws:triangle}.  These span programs are sufficiently uniform that neither of the above two factors is a problem.  Both reflections can be implemented efficiently, in poly-logarithmic time.  Similarly, we can show that $D = A^*B$ has an $\Omega(1)$ spectral gap around singular value zero.  Therefore, approximating to within an inverse polynomial the reflection about the $(-1)$-eigenspace of $R_B R_A$ takes only poly-logarithmic~time.  
We give the proof for the algorithms from Theorems~\ref{thm:star} and~\ref{thm:claws:triangle}.  The argument for $st$-connectivity, \refthm{st}, is similar and actually easier.  %\comment{Should the arXiv version include all proofs?}  

Both algorithms look similar.  In each case, the span program is based on the graph~$H$, whose vertices form an orthonormal basis for vector space of the span program.  The vertices of~$H$ can be divided into a sequence of layers that are monochromatic according to the colouring~$c$ induced from~$G$.
The edges only go between consecutive layers.  Precisely, place the vertices $s$ and~$t$ each on their own separate layer at the beginning and end, respectively, and set the layer of a vertex~$v$ to be the distance from~$s$ to $c(v)$ in the graph~$H$ for the case that $G = T$.  For example, in the span program for detecting a subdivided star with branches of lengths $\{ \ell_1, \ldots, \ell_d \}$, there are $\ell = 3 + 2 \sum_{j \in [d]} \ell_j$ layers, because the $s$-$t$ path is meant to traverse each branch of the star out and back.  There are $\ell = 6$ layers of vertices for the triangle-detection span program.  
%Some layers of edges are paired, meaning that corresponding edges are paired to each other; for example, in \reffig{exampleH} there are paired edges from vertices colored~$r$ to and from vertices colored~$v_{1,1}$, and therefore the second and fourth edge layers are paired to each other.  

In order to facilitate finding factorisations~$\{ a_k \}$ and~$\{ b_i \}$ such that~$R_A$ and~$R_B$ are easily implementable, we make two modifications to the span programs.  

First, the span programs, as presented, depend on the random colouring $c$ of~$G$.  This dependence makes it difficult to specify a general factorisation of~$V$.  To fix this, we add dummy vertices so that every layer becomes of size~$n$.
More specifically, each vertex $k$ of the graph $H$ is represented by a tuple $(j,\sigma)$, where $j\in[\ell]$ denotes the layer of the vertex, and $\sigma\in[n]$ denotes the index of the vertex of $G$ that corresponds to $k$.  Thus, a vertex $(j,\sigma)$ is not dummy if $c(\sigma)$ equals the colour of the layer $j$.
Special vertices $s$ and $t$ are represented by $(1,1)$ and $(\ell,1)$, respectively.  All other vertices in the layers 1 and $\ell$ are dummy.

We fill in the graph with never-available edges between adjacent layers, including between the layers of~$s$ and~$t$, so that every vertex has degree exactly~$2 n$.  If the edges in two layers are paired, then pair the corresponding newly added edges; each edge pair corresponds to one never-available, four-term input vector.
%
%This transformation is equivalent to making the colouring part of the input, in the following sense: if $e_u$ and~$e_v$ are vertices of $H$ in adjacent layers, corresponding to vertices~$u$ and~$v$ of~$G$, then the $e_ue_v$ edge input vector is available if $uv$ is an edge in~$G$ \emph{and} if $u$ and~$v$ are both coloured appropriately.

Second, we scale the input vectors corresponding to paired edges down by a factor of $\sqrt 2$.  This is performed to ensure that $\|b_i\|=1$, independently of whether $b_i$ corresponds to a two-term or to a four-term input vector.
We connect~$s$ and~$t$ by \emph{two} edges.  The first one corresponds to the scaled target vector $v_0 = \tilde \target = (h_t - h_s)/\alpha$.  The second one corresponds to a never-available input vector $\sqrt{1 - 1/\alpha^2} (h_t - h_s)$.  We may assume that $\alpha = C_1 \sqrt{W_1} \geq 1$.  (Thus, $s$ and $t$ have degree $2n+1$.)

It is easy to verify that the span program after this transformation still computes the same function, and the positive and the negative witness sizes remain $O(1)$ and $O(n^2)$, respectively.  After the modifications, the graph~$H$ has a simple uniform structure that allows for facile factorisation.  There is a complete bipartite graph between any two adjacent layers.  

We specify the vector $a_k$ for each vertex~$k$ of the graph $H$.  For $k \notin \{s, t\}$, let~$a_k$ be the vector with uniform $1/\sqrt{2n}$ coefficients for all incident edges.  For $k \in \{s, t\}$, let $a_k$ have coefficients $1/(\alpha\sqrt{2n})$ and $\sqrt{(1-1/\alpha^2)/(2n)}$ for the two edges between~$s$ and~$t$, and coefficients $1/\sqrt{2n}$ for the other $2n-1$ edges.  
For each input vector $i$, we specify the vector~$b_i$.  
If $i$ corresponds to an ordinary edge $k_1k_2$, then $b_i = (h_{k_2}-h_{k_1})/\sqrt{2}$.
If $i$ corresponds to a pair of edges $k_1k_2$ and $k_3k_4$, then $b_i = (h_{k_4} - h_{k_3} + h_{k_2} - h_{k_1})/2$.
That is, for any input vector $v_i$, except those connecting $s$ and $t$, $b_i = v_i/\sqrt{2}$.
Thus we get a factorisation of $V' = \frac{1}{2 \sqrt n} \tV$, i.e., $a_k\elem[i] b_i\elem[k] = \frac{1}{2 \sqrt n} \tV\elem[k,i] = \frac1{2\sqrt{n}} v_i\elem[k]$.  

Let us analyse the spectral gap around zero of $D(A, B) = A^* B = V'$.  The non-zero singular values of~$V'$ are the square roots of the non-zero eigenvalues of $\Delta =  V' (V')^*$.  The matrix $\Delta$ has its rows and columns labelled by the vertices of $H$, and
\[
\Delta\elem[k,k'] = \frac{1}{4n} \sum_{i\in\cI_0} v_i\elem[k] v_i\elem[k'].
\]
We compute~$\Delta$.  
%A vertex~$k$ can be represented by a tuple~$(j, \sigma) \in [\ell] \times [n]$, where $j$ specifies one of the~$\ell$ layers and~$\sigma$ specifies a vertex within the layer.  
Let $\Delta(j, j')$ be the $n \times n$ submatrix of~$\Delta$ between vertices at layers~$j$ and~$j'$.  To calculate $\Delta(j, j')$, we consider separately the contributions from all of the different layers of input vectors.  

\begin{itemize}
\item
Ordinary edges between adjacent layers~$j$ and~$j+1$ contribute $\frac{1}{4} I_n$ to $\Delta(j,j)$ and $\Delta(j+1,j+1)$, and $-\frac{1}{4n} J_n$ to $\Delta(j,j+1)$ and $\Delta(j+1,j)$.  Indeed, for the contribution to $\Delta(j,j)$, observe that any vertex $k=(j, \sigma)$ has~$n$ incident ordinary edges to the layer~$j+1$, and each incident edge $i$ contributes $\frac{1}{4n} {v_i}\elem[k]^2 = \frac{1}{4n}$ to $\Delta(j,j)\elem[\sigma,\sigma]$.  There is no edge involving vertices $(j,\sigma)$ and $(j,\sigma')$ with $\sigma \neq \sigma'$, but for any $\sigma, \sigma' \in [n]$, there is exactly one ordinary edge from $(j, \sigma)$ to $(j+1, \sigma')$, and it contributes $-\frac{1}{4n}$ to $\Delta(j, j+1)\elem[\sigma, \sigma']$.

Even though~$s$ and~$t$ are connected by two edges, the same calculations hold for the edges between their layers.

\item
Consider a set of paired edges, that go out from layer~$j$ to~$j+1$, and then return from layer~$j'$ to~$j'+1$.  
Each input vector~$v_j$ is of the form $\frac{1}{\sqrt 2} \big(- (j, \sigma) + (j+1, \sigma') - (j', \sigma') + (j'+1, \sigma)\big)$.  
The contributions of these paired edges to the sixteen blocks $\Delta(j_\alpha, j_\beta)$ are given by the $4 \times 4$ block matrix 
\begin{equation*}
\begin{pmatrix}
\;\;\;\;\; \frac{1}{8} I_n 	& \!\!\! -\frac{1}{8n} J_n 		& \!\!\! \;\;\, \frac{1}{8n} J_n 	& \!\!\! \;\, -\frac{1}{8} I_n \\[1pt]
-\frac{1}{8n} J_n 		& \!\!\! \;\;\;\;\; \frac{1}{8} I_n 	& \!\!\! \;\, -\frac{1}{8} I_n 		& \!\!\! \;\;\, \frac{1}{8n} J_n \\[1pt]
\;\;\, \frac{1}{8n} J_n 	& \!\!\! \;\, -\frac{1}{8} I_n 		& \!\!\! \;\;\;\;\; \frac{1}{8} I_n 	& \!\!\! -\frac{1}{8n} J_n \\[1pt]
\;\, -\frac{1}{8} I_n 		& \!\!\! \;\;\, \frac{1}{8n} J_n 	& \!\!\! -\frac{1}{8n} J_n 		& \!\!\! \;\;\;\;\; \frac{1}{8} I_n
\end{pmatrix}
%\place{\small $k_1$}{-290mu}{18pt}
%\place{\small $k_2$}{-290mu}{6pt}
%\place{\small $k_3$}{-290mu}{-6pt}
%\place{\small $k_4$}{-290mu}{-18pt}
%\place{\small $k_1$}{-220mu}{30pt}
%\place{\small $k_2$}{-160mu}{30pt}
%\place{\small $k_3$}{-100mu}{30pt}
%\place{\small $k_4$}{-40mu}{30pt}
 \enspace .
\end{equation*}
That can be checked similarly to the previous point.
\end{itemize}
Observe that~$\Delta$ is a constant-sized block matrix, where each block is the sum of a constant multiple of~$I_n$ and a constant multiple of~$J_n/n$.  By \ref{lem:blockmatrixspectrum}, the set of eigenvalues of $\Delta$ does not depend on~$n$.  In particular, it has an $\Omega(1)$ spectral gap from zero, as desired.  

\begin{algorithm}[ub]
\caption{Efficient implementation of path and claw detection}
\label{alg:claw}
\algbegin
\state {{\bf quprocedure} ReflectionAbout $\im(A)$ {\bf:}}
\tab
	\state {{\bf attach} register $\qreg T$, qubit $\qreg Z$}
	\state {Apply $(-1)$-phase gate}
	\state {{\bf conditioned on $\qreg V=0$ :} $\qreg Z \qgets 1$}
	\state {{\bf conditioned on $\qreg Z = 0$ :}}
	\tab
		\state {{\bf for} $j$ such that layers $j$ and $j+1$ are connected by ordinary edges {\bf:}}
		\tab
			\state{ PrepareLayer($j$,\;1,\;2)}
			\state{ perform reflection about the orthogonal complement of $\ket T|1>-\ket T|2>$ }
			\state{ PrepareLayer$^{-1}$($j$,\;1,\;2)}
		\untab
		\state {{\bf for} $j,j'$ such that layers $j, j+1, j'$ and $j'+1$ are connected by paired edges {\bf:}}
		\tab
			\state{ PrepareLayer($j$,\;1,\;2),\quad PrepareLayer($j$,\;4,\;3)}
			\state{ perform reflection about the orthogonal complement of $\ket T|1>-\ket T|2> + \ket T|3> - \ket T|4>$ }
			\state{ PrepareLayer$^{-1}$($j$,\;4,\;3) \quad PrepareLayer$^{-1}$($j$,\;1,\;2)}
		\untab
	\untab
	\state {{\bf conditioned on $\qreg Z = 1$ :}}
	\tab
		\state {reflect about the orthogonal complement of $\ket L|1>\ket U|1>\ket V|0>\ket D|0> - \ket L|\ell>\ket U|1>\ket V|0>\ket D|1>$}
	\untab
	\state {{\bf conditioned on $\qreg V=0$ :} $\qreg Z \qungets 1$}
	\state {{\bf detach} $\qreg T$, $\qreg Z$}
\untab
\state{\ }
\state {{\bf quprocedure} PrepareLayer({\bf integers} $j$, $a$, $b$) {\bf:}}
\tab
	\state {{\bf conditioned on} $\qreg L = j$ and $\qreg D = 1${\bf :} $\qreg T \qgets a$ }
	\state {{\bf conditioned on} $\qreg L = j+1$ and $\qreg D = 0$ {\bf :} $\qreg T \qgets b$ }
	\state {{\bf conditioned on} $\qreg T = a$ {\bf :} $\qreg L \qungets j$}
	\state {{\bf conditioned on} $\qreg T = b$ {\bf :}}
	\tab
		\state {Swap($\qreg U$,\; $\qreg V$),\quad $\qreg D\qungets 1$,\quad $\qreg L \qungets (j+1)$}
	\untab
\untab
\state{\ }
\state {{\bf quprocedure} ReflectionAbout $\im(B)$ {\bf:}}
\tab
	\state {PrepareVertex$^{-1}$}
	\state {{\bf conditioned on not} $(\qreg U = \qreg D = 0)$ {\bf :} apply $(-1)$-phase gate}
	\state {PrepareVertex}
\untab
\state{\ }
\state {{\bf quprocedure} PrepareVertex {\bf:}}
\tab
	\state {UniformSuperposition($\qreg D$)}
	\state {UniformSuperposition($\qreg V$)}
	\state {{\bf conditioned on} ($\qreg U=1$ and $\qreg D=0$) or ($\qreg U=\ell$ and $\qreg D=1$) {\bf :}}
	\tab
	\state {transform $\ket V |1>\mapsto \sqrt{1-1/\alpha^2} \ket V |0> + \alpha \ket V|1>$}
\untab
\algend
\end{algorithm}

We now show that both $R_A$ and $R_B$ can be implemented efficiently.  
One possible implementation is given in \rf(alg:claw).
As described earlier, the algorithm works in the space spanned by vectors $e_i \otimes h_k$, where $i$ varies over input vectors and the target vector, and $k$ varies over vertices of the graph $H$.
Moreover, only those $e_i\otimes h_k$ are used, where the edge $i$ is incident to the vertex $k$.
We represent such pairs $(i,k)$ using four registers: $\reg{LUVD}$.  
\itemstart
\item The register $\reg L$ stores the layer $k$ belongs to, it is an integer in $[\ell]$.
\item $\reg U$ stores the index of the vertex in $G$ that corresponds to $k$: an element of $[n]$.
\item $\reg V$ stores the index of the vertex in $G$ that corresponds to the second end-point of $i$.  It is again an elements of $[n]$, but we also use 0 for the second edge between $s$ and $t$.
\item $\reg D$ contains 0 or 1.  The value 0 indicates that $i$ goes to the previous layer, and $1$ indicates that $i$ goes forward.
\itemend
Thus, each edge of $H$ is represented by two elements of the computational basis: 
$\ket L|j> \ket U|\sigma > \ket V|\sigma'> \ket D|1>$ and 
$\ket L|j+1> \ket U|\sigma' > \ket V|\sigma> \ket D|0>$.
The second edge between $s$ and $t$ is represented by
$\ket L|1> \ket U|1 > \ket V|0> \ket D|0>$ and
$\ket L|\ell> \ket U|1 > \ket V|0> \ket D|1>$.
Additionally, we use a temporary register $\reg T$ that stores number $0,\dots,4$.

We start with the description of $R_A$.  For each $k = (j, \sigma) \in [\ell] \times [n]$, except $s = (1, 1)$ and $t = (\ell, 1)$, the vector $\ket X|a_k>\ket Y|k>$ from the general approach corresponds to the uniform superposition of the states
$\ket L|j> \ket U|\sigma > \ket V|\sigma'> \ket D|d>$
where $\sigma'$ ranges over $[n]$, and $d$ ranges over $\{0,1\}$.
So, the reflection is a Grover diffusion operation.
For $(j, \sigma) = (1, 1) = s$, we perform a similar operation.  

Consider the implementation of $R_B$ now.  For layers $j$ and $j+1$ with only ordinary edges between them, it suffices to apply the reflection to all pairs $\ket|j,\sigma,\sigma',1>$, and $ \ket|j+1,\sigma',\sigma,0>$. For paired layers, the reflection is performed in a four-dimensional subspace.

Finally, for the implementation of $R_\Pi$ we need to clarify the use of the random colouring.  One solution is to generate random numbers classically, and provide them in the form of an oracle mapping $\sigma \in [n]$ to the colour of vertex~$\sigma$.  This requires a QRACM array of size $\Theta(n)$.  For \refthm{star}, however, one can reduce the space complexity to $O(\log n)$, by using a $C$-uniform hash function family from~$[n]$ to~$[C]$, where $C$ is the total number of colours.  If necessary, we may assume that $n$ and~$C$ are powers of two.  $C$-wise independence is enough for the proof.  For \refthm{claws:triangle}, this does not work, though, because we need to ensure that the negative witness size is small with high probability.

Consider layer $j$ that corresponds to colour~$c$.  A vertex $(j, \sigma)$ corresponds to vertex~$\sigma$ of $G$ if and only if it has colour~$c$.  Otherwise, it is a dummy vertex.  To check whether the edge is available, the algorithm first checks whether its both end-points have correct colours.  If they do, the algorithm queries the input oracle for the availability of the edge in $G$.
If the edge is available, it does nothing.  In all other cases, it negates the phase of the state.  
\hfill\qedsymbol

\section{Summary}
In this chapter, we applied the technique of span programs to some graph problems.  Span programs provide greater flexibility compared to the dual adversary SDP.  This makes it easier to grasp the construction of the algorithm, and the algorithm itself is also easier to implement time-efficiently.

We mostly focused on minor-closed forbidden subgraph properties.  We implemented a tight quantum algorithm for subgraph detection, where the subgraph is either a path of a fixed length, or a fixed subdivision of a claw.  The algorithm is tight from all points of view: query, time and space.  Actually, its asymptotic complexity is the same as for Grover's algorithm.  The case of general minor-closed FSP is still open, however, it is likely that similar techniques can be used to solve it.

Informally, our algorithm is a variant of ``quantum dynamic programming'', as it is based on a classical algorithm that heavily uses dynamic programming.  However, our dynamic programming is more limited than its classical analogue, and requires additional tricks as the paired-edge technique, for instance.  It would be interesting to understand whether it can be used to solve other problems approachable with dynamic programming.

The method of time-efficient span program implementation in \rf(sec:efficient) is rather general and can be applied to other problems.  However, we feel that the direct application of the effective spectral gap lemma, as in the next chapter, has greater potential.

%% file: _electric.tex
In \rf(chp:walk), we mentioned two main paradigms of quantum walks: the Szegedy-type quantum walk, \rf(thm:walkSzegedy); and the MNRS quantum walk, \rf(thm:mnrs).  Both of these paradigms assume that the walk is started in the stationary distribution.  Moreover, the MNRS type quantum walks require spectral analysis of the underlying graph.  This poses some problems for potential applications.
Firstly, preparing the stationary distribution can be a limitation if the graph is complex or not given in advance.
Also, spectral analysis usually requires non-trivial tools from linear algebra.  We would like to apply more combinatorial techniques that are closer to the nature of the problem.

The main result of this chapter is the generalisation of Szegedy's algorithm to arbitrary initial distribution.  
The analysis of the resulting algorithm does not require any spectral analysis.
In order to achieve this, we add two new ingredients to the analysis of Szegedy-type quantum walks:
\descrstart
\item[Electric Networks.]  A point of view on a graph as an electric network has turned out very fruitful in the analysis of classical random walks~\cite{doyle:walksElectric, bollobas:modernGraph}.  
But it seems to be completely ignored in the analysis of quantum walks.  The analysis in \rf(chp:walk) relies solely on the spectral properties of the graph.

\item[Effective Spectral Gap Lemma. ] 
We saw the power of the effective spectral gap lemma, \rf(lem:effective), in the proofs of Theorems~\ref{thm:spanAlgorithm} and~\ref{thm:advAlgorithm}.  We show that the lemma can be also applied for general quantum walks.
\descrend

We give two examples of application of this quantum walk.  In \refsec{learning}, we show how a general learning graph from \rf(chp:cert) can be implemented as a quantum walk.  In \refsec{kdist}, we use these ideas in a time-efficient quantum algorithm for $3$-distinctness.  The last example is interesting as a quantum walk on a graph not given in advance.  This is at the very heart of classical random walks: Since only local information is required to implement a random walk, they are often used to traverse graphs whose global structure is unknown (see, e.g.,~\cite{aleliunas:connectivity, schoning:ksat}).  Quantum walks require more global information than the classical ones, and they are usually used for graphs known in advance like for the Johnson graph in \rf(sec:mnrsApplications).

This chapter is based on the pre-print
\begin{itemize}
\item[\cite{belovs:electicityQuantumWalks}]
A.~Belovs.
\newblock Quantum walks and electric networks.
\newblock 2013,  {\ttfamily arXiv:1302.3143}.
\end{itemize}
The preprint has been merged with~\cite{childs:walk3Dist} when accepted to the conference, resulting in the publication
\begin{itemize}
\item[\cite{belovs:mergedWalk3Dist}]
A.~Belovs, A.~M. Childs, S.~Jeffery, R.~Kothari, and F.~Magniez.
\newblock Time-efficient quantum walks for 3-distinctness.
\newblock In {\em Proc. of 40th ICALP, Part I}, volume 7965 of {\em LNCS},
  pages 105--122. Springer, 2013.
\end{itemize}

The chapter is organised as follows.  In \rf(sec:electric), we recall the relations between classical hitting time and electric resistance of a graph.  In \rf(sec:walk), we prove the main result, and in \rf(sec:3rdproof), we give an application to learning graphs.  In \rf(sec:kdist), we apply the new quantum walk algorithm for the 3-distinctness problem.

\section{Random Walks and Electric Networks}
\label{sec:electric}
\mycommand{st}{\sigma}
Let us recall some definitions from \rf(sec:walkClassical).
Let $G=(V,E)$ be a simple undirected graph with each edge assigned a weight $w_e\ge 0$.  
The weight of a vertex $u$ is $w_u=\sum_{uv\in E} w_{uv}$, and the total weight is $W = \sum_{e\in E} w_e = \frac12 \sum_{u\in V} w_u$.  Consider the following random walk:  If the walk is at a vertex $u\in V$, proceed to a vertex $v$ with probability $w_{uv}/w_u$.  The random walk has the {\em stationary probability distribution} $\pi=(\pi_u)$ given by $\pi_u = w_u/(2W)$.  One step of the random walk leaves $\pi$ unchanged.

Let $\st=(\st_u)$ be some {\em initial probability distribution} on the vertices of the graph, and let $M\subseteq V$ be some set of {\em marked vertices}.  We are interested in the {\em hitting time} $H_{\st, M}$ of the random walk: the expected number of steps of the random walk required to reach a vertex in $M$ when the initial vertex is sampled from $\st$.  If $\st$ is concentrated in a vertex $s\in V$, or $M$ consists of a single element $t\in V$, we often replace $\st$ by $s$ or $M$ by $t$.  
For instance, we have $H_{\st, M} = \sum_{u\in V} \st_u H_{u, M}$.
The {\em commute time} between vertices $s$ and $t$ is defined as $H_{s,t}+H_{t,s}$.
We usually assume that $G$ and $\st$ are known, and the task is to determine whether $M$ is non-empty by performing the random walk.

Assume $M$ is non-empty, and define a {\em flow} on $G$ from $\st$ to $M$ as a real-valued function $p_e$ on the (oriented) edges of the graph satisfying the following conditions.  Firstly, $p_{uv} = -p_{vu}$.  Secondly, for each non-marked vertex $u$, the flow satisfies
\begin{equation}
\label{eqn:flowCondition}
\st_u = \sum_{uv\in E} p_{uv}.
\end{equation}
That is, $\st_u$ units of the flow are injected into $u$, it traverses through the graph, and is removed in a marked vertex.  Define the {\em energy} of the flow as 
\begin{equation}
\label{eqn:flowEnergy}
\sum_{e\in E} \frac{p_e^2}{w_e}.
\end{equation}
Clearly, the value of~\rf(eqn:flowEnergy) does not depend on the orientation of each $e$.
The {\em effective resistance} $R_{\st, M}$ is the minimal possible energy of
a flow from $\st$ to $M$.   For $R$, as for $H$, we also replace $\st$ and $M$ by the corresponding singletons.

There is a nice physical description for effective resistance.  We give it for illustrative purposes, and it is not necessary for understanding the remaining part of the chapter.  In the description, we require some basic notions from the theory of electric networks.  We assume they are familiar to the reader.

Treat the graph $G$ as an electrical network where each edge $e$ is replaced by a resistor of conductance $w_e$.  Assume that, for each vertex $u$, $\st_u$ units of current are injected into it, and that the current is collected in $M$.  Then, the effective resistance $R_{\st, M}$ equals the energy dissipated by the current.
This description can be used to prove the following result:

\begin{thm}[\cite{chandra:electrical}]
\label{thm:classical}
If $G$, $w$, $W$ are as above, $s$, $t$ are two vertices of $G$, $M\subseteq V$, and $\pi$ is the stationary distribution on $G$, then
\itemstart
\item[(a)] the commute time between $s$ and $t$ equals $2WR_{s,t}$;
\item[(b)] the hitting time $H_{\pi, M}$ equals $2WR_{\pi, M}$.
\itemend
\end{thm}

\pfstart
The proof is essentially taken from~\cite{chandra:electrical}.  But since this result is not explicitly stated there, we briefly reproduce the proof here.

We start with proving (b).  Assume that $w_u$ units of current are injected into each vertex $u$, and then collected in $M$.  Let $\phi_u$ denote the potential of a vertex $u$ in this scenario.  Then, $\phi_u=0$ for $u\in M$.  Otherwise, by Kirchoff's and Ohm's laws, we have
\begin{equation}
\label{eqn:chandra1}
w_u = \sum_{uv\in E} I_{uv} = \sum_{uv\in E} w_{uv} (\phi_u-\phi_v) = w_u \phi_u - \sum_{uv\in E} w_{uv}\phi_v,
\end{equation}
where $I_{uv}$ is the current through $uv$.
Now consider the hitting time $H_{u,M}$.  Again, $H_{u,M}=0$ for $u\in M$, and, if $u\notin M$, then
\begin{equation}
\label{eqn:chandra2}
H_{u,M} = 1 + \frac{1}{w_u} \sum_{uv\in E} w_{uv}H_{v,M}.
\end{equation}
We can see that the conditions on $\phi_u$ and $H_{u,M}$,~\rf(eqn:chandra1) and~\rf(eqn:chandra2), are identical.  Moreover,~\rf(eqn:chandra2) uniquely determines $H_{u,M}$.  Thus, $H_{u,M}=\phi_u$ for all $u\in V$.
Since energy equals voltage times current, we get
\begin{equation}
\label{eqn:chandra3}
2W H_{\pi, M} = \sum_{u\in V} w_u H_{u,M} = \sum_{u\in V} w_u \phi_u = (2W)^2 R_{\pi, M}.
\end{equation}

Now let us prove (a).  
Let $\phi'_u$ and $I'_{uv}$ be the potentials in the above network for the case $M=\{t\}$, and let $\phi''_u$ and $I''_{uv}$ be defined similarly for $M=\{s\}$.  Then, by the superposition principle, $\phi_u = \phi'_u - \phi''_u$ and $I_{uv} = I'_{uv}-I''_{uv}$ give a valid electric network for the case when $2W$ units of current are injected into $s$ and removed in $t$.  The difference in potentials between $s$ and $t$ in this network is $H_{s,t}+H_{t,s}$.  Using similar calculations as in~\rf(eqn:chandra3), we obtain (a).
\pfend

%\enumstart
%\item If $\st=\pi$ is the stationary distribution of $G$, then the hitting time, i.e., the expected number of steps to reach a vertex in $M$, is exactly $WR$.
%\item If $\st$ is concentrated on one vertex $s$, then $2WR$ equals the commute time from $s$ to $M$, where all vertices in $M$ are short-circuited.  That is, consider a graph with all vertices in $M$ replaced by one vertex (preserving the edges incident to $M$).  The commute time is the expected number of steps required to go from $s$ to $M$ and back.
%\enumend
%\end{thm}

%It is possible to bind the hitting times with the expression $RW$ for other initial distributions $\st$, but it becomes less natural.

\section{Quantum Walk}
\label{sec:walk}
In this section, we construct a quantum counterpart of the random walk in \refsec{electric}.  
We obtain a quadratic improvement:  If $G$ and $\st$ are known in advance and the superposition $\sum_{u\in V} \sqrt{\st_u}\;\ket|u>$ is given, the presence of a marked vertex in $G$ can be determined in $O(\sqrt{WR_{\sigma, M}})$ steps of the quantum walk.
By combining this result with the second statement of \refthm{classical}, we obtain \rf(thm:walkSzegedy).

The quantum walk differs slightly from the quantum walk by Szegedy.  The framework of the algorithm goes back to~\cite{ambainis:formulaeEvaluation}, and \reflem{effective} is used to analyse its complexity.  We assume the notations of \refsec{electric} throughout the section.

It is customary to consider quantum walks on bipartite graphs, so we assume the graph $G=(V,E)$ is bipartite with parts $A$ and $B$.  Also, we assume the support of $\st$ is contained in $A$, i.e., $\st_u=0$ for all $u\in B$.  These are not very restrictive assumptions:  If either of them fails, consider the bipartite graph $G'$ with the vertex set $V' = V\times\{0,1\}$, the edge set $E' = \{(u,0)(v,1), (u,1)(v,0) \mid uv\in E\}$, edge weights $w'_{(u,0)(v,1)} =  w'_{(u,1)(v,0)} = w_{uv}$, the initial distribution $\st'_{(u,0)} = \st_u$, and the set of marked vertices $M' = M \times\{0,1\}$.  Then, for the new graph, $W' = 2W$, and $R'_{\st', M'} \le R_{\st, M}$.

\mycommand{qst}{\varsigma}

We assume the quantum walk starts in the state $\qst = \sum_{u\in V} \sqrt{\st_u}\;\ket|u>$ that is known in advance.  Also, we assume there is an upper bound $R$ known on the effective resistance from $\st$ to $M$ for all possible sets $M$ of marked states that might appear.

Now we define the vector space of the quantum walk.  
Let $S$ be the {\em support} of $\st$, i.e., the set of vertices $u$ such that $\st_u\ne 0$.  
The vectors $\{\ket|u> \mid u\in S\}\cup\{\ket|e>\mid e\in E\}$ form the computational basis of the vector space of the quantum walk.  Let $\cH_u$ denote the {\em local space} of $u$, i.e., the space spanned by $\ket|uv>$ for $uv\in E$ and, additionally, $\ket|u>$ if $u$ happens to be in $S$.  We have that $\bigoplus_{u\in A} \cH_u$ equals the whole space of the quantum walk, and $\bigoplus_{u\in B} \cH_u$ equals the subspace spanned by the vectors $\ket|e>$ for $e\in E$.

Let $I_S$ be the identity operator on $S$.
The {\em step of the quantum walk} is defined as $R_BR_A$, where $R_A = \bigoplus_{u\in A} D_u$ and $R_B = I_S\oplus\bigoplus_{u\in B} D_u$ are the direct sums of the {\em diffusion} operations.  Each $D_u$ is a reflection operation in $\cH_u$.  Hence, all $D_u$ in $R_A$ (or $R_B$) are performed in orthogonal subspaces, which makes them easy to implement in parallel.  They are as follows:
\itemstart
\item If a vertex $u$ is marked, then $D_u$ is the identity, i.e., the reflection about $\cH_u$;
\item If $u$ is not marked, then $D_u$ is the reflection about the orthogonal complement of $\psi_u$ in $\cH_u$, where
\begin{equation}
\label{eqn:psi}
\psi_u = \sqrt{\frac{\st_u}{C_1R}}\; \ket|u> + \sum_{uv\in E} \sqrt{w_{uv}}\;\ket|uv>
\end{equation}
for some large constant $C_1>0$ we choose later.  This definition also holds for $u\notin S$: For them, the first term in~\rf(eqn:psi) disappears.
\itemend

Similarly as in \rf(chp:walk), we have the set-up procedure that prepares the state $\qst$, and the check procedure that, given a vertex $u$ of the graph, returns 1 iff $u$ is marked.  Similarly to the proof of \rf(thm:mnrs), we ignore the data register and the update operation in our analysis, although, an actual implementation would require them.

\begin{algorithm}
\caption{The quantum walk algorithm.  Here, $C$ is some constant to be specified later.}
\label{alg:walk}
\algbegin
\state{{\bf function} QuantumWalk({\bf quprocedures} Setup, Check, {\bf reals} R, W, {\bf register} $\qreg X$) {\bf with} error $1/3$ \bf :}
\tab
\state{{\bf attach} qubit $\qreg b$} \label{line:walk:b}
\state{Setup($\qreg X$)}
\state{$\qreg b \qgets$ Check($\qreg X$)}
\state{{\bf measure} $\qreg b$}
\state{{\bf if} $\reg b = 1$ {\bf :} {\bf return } 1} \label{line:walk:endb}
\state{{\bf return} NOT(PhaseDetection($R_BR_A$, $1/(C\sqrt{RW})$, $\qreg X$)) {\bf with} precision $1/6$} \label{line:walk:phaseDetection}
\algend
\end{algorithm}

\begin{thm}
\label{thm:walk}
\refalg{walk} detects the presence of a marked vertex with probability at least $2/3$.  The algorithm uses $O(\sqrt{RW})$ steps of the quantum walk.
\end{thm}

\pfstart
The second statement follows immediately from \rf(thm:detection).  Let us prove the correctness.
The only purpose of Lines~\ref{line:walk:b}---\ref{line:walk:endb} is to ensure that $S$ is disjoint from $M$, that will be needed in our analysis. 
If a vertex in the initial distribution is marked with probability at least $2/3$, then this is detected at \rf(line:walk:endb) with the same probability, and we are done.  So, assume this probability is less than $2/3$, and the value of $\reg b$ in \rf(line:walk:endb) is 0.  Then, the state of the algorithm collapses to a state $\qst'$ with the support disjoint from $M$, and $R_{\qst', M}\le 9R$.  Thus, we may further assume that $S$ is disjoint from $M$.  

Let us consider \rf(line:walk:phaseDetection) of the algorithm.  We start with the case when $M$ is non-empty.  Let $p_e$ be a flow from $\st$ to $M$ with energy at most $R$. At first, using the Cauchy-Schwarz inequality and that $S$ is disjoint from $M$, we get
\begin{equation}
\label{eqn:RW}
RW \ge \sC[\sum_{e\in E} \frac{p_e^2}{w_e}]\sC[\sum_{e\in E} w_e] \ge \sum_{e\in E} |p_e| \ge 1.
\end{equation}
Now, we construct a 1-eigenvector
\[
\phi = \sqrt{C_1R} \sum_{u\in S} \sqrt{\st_u} \ket|u> - \sum_{e\in E} \frac{p_e}{\sqrt{w_e}}\ket|e>
\]
of $R_BR_A$ having large overlap with $\qst$  (assume the orientation of each edge $e$ is from $A$ to $B$.)
Indeed, by~\refeqn{flowCondition}, $\phi$ is orthogonal to all $\psi_u$, hence, is invariant under the action of both $R_A$ and $R_B$.  Moreover, $\|\phi\|^2 = C_1R + \sum_{e\in E} p_e^2/w_e$, and $\ip<\phi,\qst> = \sqrt{C_1R}$.  Since we assumed $R\ge \sum_{e\in E} p_e^2/w_e$, we get that the normalised vector satisfies 
\begin{equation}
\label{eqn:positive}
\ipB<\frac\phi{\|\phi\|},\;\qst>\ge \sqrt{\frac{C_1}{1+C_1}}.
\end{equation}
Now consider the case $M=\emptyset$.  Let $w'$ be defined by
\[
w' = \sqrt{C_1R}\sC[\sum_{u\in S} \sqrt{\frac{\st_u}{C_1R}}\; \ket|u> + \sum_{e\in E} \sqrt{w_e}\;\ket|e>].
\]
Let $\Pi_A$ and $\Pi_B$ be the projectors onto the invariant subspaces of $R_A$ and $R_B$, respectively.  Since $S\subseteq A$, we get that $\Pi_Aw'=0$ and $\Pi_Bw' = \qst$.  Also,
\[
\norm|w'|^2 = \sum_{u\in S} \st_u + C_1 R\sum_{e\in E} w_e = 1+C_1RW,
\]
hence, by \reflem{effective}, we have that, if
\[
\delta = \frac{1}{C_2\sqrt{1+C_1 RW}}
\]
for some constant $C_2>0$, then the overlap of $\qst$ with the eigenvectors of $R_BR_A$ with phase less than $\delta$ is at most $1/(2C_2)$.  Comparing this with~\refeqn{positive}, we get that it is enough to execute phase estimation with precision $\delta$ if $C_1$ and $C_2$ are large enough.  Also, assuming $C_1\ge 1$, we get $\delta = \Omega(1/\sqrt{RW})$ by~\refeqn{RW}.
\pfend

\section{Application: Learning Graphs}
\mycommand{btr}{\bigtriangleup}
\label{sec:3rdproof}
As an example, we give a third proof of the first half of \rf(thm:certificates).  Recall that the second proof in \rf(sec:learningProof) gives a reduction from a learning graph to a dual adversary SDP.  The latter then can be implemented using the algorithm in \rf(thm:advAlgorithm).  However, this double reduction results in a complicated quantum algorithm that is hard (if not impossible) to implement time-efficiently.  In this section, we give a direct reduction from a learning graph to a quantum walk.  The walk is quite similar to the one in \rf(prp:distinctness).

%There was one attempt of doing so in Ref.~\cite{jeffery:nestedWalks}.  The authors use an MNRS-type quantum walk~\cite{magniez:walkSearch}, and give quantum walks corresponding to a number of previously developed learning graphs.  We, however, use the Szegedy-type quantum walk, and our construction is valid for an arbitrary learning graph.

Let us remind main definitions from \rf(chp:cert).  A learning graph computes a function $f\colon[q]^\vars \supseteq \cD\to\{0,1\}$.  Vertices of the graph are subsets of $[n]$, and the edges are between vertices $S$ and $S\cup\{j\}$ for some $S\subset [n]$ and $j\in[n]\setminus S$.  
The initial distribution $\st$ is concentrated on the vertex $\emptyset$.  For each positive input $x\in f^{-1}(1)$, a vertex $S$ is marked if and only if it contains a 1-certificate for $x$, i.e., $f(z) = 1$ for all $z\in\cD$ such that $z_S = x_S$.  The {\em complexity} of the learning graph is defined as $\sqrt{WR}$ in the notations of \refsec{electric}.

The learning graph is a bipartite graph: the part $A$ contains all vertices of even cardinality, and the part $B$ contains all vertices of odd cardinality.  Also, the support of $\st$ is concentrated in $A$.  Hence, the algorithm from \refsec{walk} can be applied, and the presence of a marked vertex can be detected in $O(\sqrt{WR})$ steps of the quantum walk.  It suffices to show that one step of the quantum walk can be implemented in $O(1)$ quantum queries.

This can be done using standard techniques.  Let $\str$ be the input string.  The quantum walk has two registers: the {\em data register} $\reg D$, and the {\em coin register} $\reg C$.  The states of the first register are of the form $\ket D|S>$ for some $S\subseteq[n]$.  The register contains the description of the subset $S$, the values of $\str_j$ for $j\in S$, and some ancillary information.  Because of the interference, it is important that $\ket D|S>$ is always represented in exactly the same way that only depends on $S$ and the input string $\str$.  (Most classical data structures do not satisfy this condition.)  The second register stores an element $j\in [n]$.  A state $\ket D|S>\ket C|j>$ represents the edge of the learning graph connecting subsets $S$ and $S\btr\{j\}$, where $\btr$ stands for the symmetric difference.  Additionally, there is the state $\ket |\emptyset>$ orthogonal to all $\ket D|S>\ket C|j>$.  It corresponds to the initial vertex $\emptyset$ of the quantum walk.

The step of the quantum walk is performed as follows.  Start with a superposition of $\ket|\emptyset>$ and the states of the form $\ket D|S>\ket C|j>$ with $S$ in $A$.  At first, perform the reflection $R_A$ as described in \refsec{walk}.  It is possible to detect whether $S$ is marked by considering the values $\str_j$ stored in $\ket D|S>$, and $\psi_S$ from~\rf(eqn:psi) does not depend on the input.  Hence, this operation does not require any oracle queries.  Next, apply the {\em update operation} that maps $\ket D|S>\ket C|j>$ into $\ket D|S\btr \{j\}>\ket C|j>$.  This represents the same edge, but with the content of the data register in $B$.  The update operation requires one oracle query in order to compute or uncompute $\str_j$.  After that, perform $R_B$ similarly to $R_A$, and apply the update operation once more.  Hence, one step of the quantum walk requires $O(1)$ oracle queries, and $f(\str)$ can be computed in $O(\sqrt{WR})$ quantum queries.

\section{Application: 3-Distinctness}
\mycommand{tO}{\tilde O}
\label{sec:kdist}
In \rf(sec:3rdproof), we demonstrated that the quantum walk algorithm from \refsec{walk} can be used to implement learning graphs.  In this section, we show an example slightly beyond the scope of learning graphs.  

Recall the $k$-distinctness problem from \rf(defn:kdist).  So far, we have seen two quantum algorithms for this problem.  The first one is from \rf(cor:walkKDist) and uses $O(\vars^{k/(k+1)})$ queries.  This algorithm can be implemented time-efficiently.  In \rf(thm:kdist), we saw a quantum algorithm that requires $O(\vars^{1-2^{k-2}/(2^k-1)})$ queries.  For $k=3$, this gives a quantum $O(\vars^{5/7})$ query algorithm.  However, it seems unlikely that this algorithm can be implemented time-efficiently.

In this section, we describe a quantum algorithm for 3-distinctness having the same time complexity up to polylogarithmic factors.  This is a different algorithm, and it is based on ideas from~\cite{lee:learningKdistPrior}.  Formally, we prove the following result.

\begin{thm}
\label{thm:3dist}
The 3-distinctness problem can be solved by a quantum algorithm in time $\tO(\vars^{5/7})$ using quantum random access quantum memory (QRAQM) of size $\tO(\vars^{5/7})$.
\end{thm}

Recall that the algorithm from \rf(cor:walkKDist) consists of two phases: the set-up phase that prepares the uniform superposition, and the quantum walk itself.  Our algorithm also consists of these two phases.  Interestingly, in our case, the analysis of the quantum walk is quite simple, and can be easily generalised to any $k$.  It is the set-up phase that is hard to generalise.  The case of $k=3$ has a relatively simple {\em ad hoc} solution that we describe in \rf(sec:setup).

We note that there exists an alternative time-efficient quantum algorithm for the 3-distinctness problem by Childs, Jeffery, Kothari and Magniez~\cite{childs:walk3Dist}.  It is based on an MNRS-type quantum walk.  The algorithm has a similar set-up phase as ours, but a more complicated quantum walk phase, that is hard to generalise to arbitrary $k$.

\subsection{Technicalities}
\label{sec:3disttech}
We start the section with some notations and algorithmic primitives we need for our algorithm.  For more detail on the implementation of these primitives, refer to the paper by Ambainis~\cite{ambainis:distinctness}.  Although this paper does not exactly give the primitives we need, it is straightforward to apply the necessary modifications, so we don't go into detail.

Recall the settings in the $k$-distinctness problem.  We are given a string $\str\in [q]^\vars$.  A subset $J\subseteq[\vars]$ of size $\ell$ is called an {\em $\ell$-collision} iff $\str_i=\str_j$ for all $i,j\in J$.  In the $k$-distinctness problem, the task is to determine whether the given input string contains a $k$-collision.  Inputs with a $k$-collision are called {\em positives}, the remaining ones are called {\em negative}.

Again, we may assume that any positive input contains exactly one $k$-collision.  Also, we may assume there are $\Omega(\vars)$ $(k-1)$-collisions in any input.
For a subset $S\subseteq[\vars]$ and $i\in[k]$, let $S_i$ denote the set of $j\in S$ such that $|\{j'\in S \mid \str_{j'} = \str_j\}| = i$.  Denote $r_i = |S_i|/i$, and call $\tau = (r_1,\dots,r_k)$ the {\em type} of $S$.

Our main technical tool is a dynamical quantum data structure that maintains a subset $S\subseteq[\vars]$ and the values $\str_j$ for $j\in S$.  We use notation $\ket D|S>$ to denote a register containing the data structure for a particular choice of $S\subseteq [\vars]$.

The data structure is capable of performing a number of operations in polylogarithmic time.  The initial state of the data structure is $\ket D|\emptyset>$.  The update operation adds or removes an element:  $\ket D|S>\ket |j>\ket |\str_j>\mapsto \ket D|S\btr \{j\}> \ket |j>\ket |0>$.  Recall that $\btr$ stands for the symmetric difference.
There is a number of query operations to the data structure.  It is able to give the type $\tau$ of $S$.  For integers $i\in[k]$ and $\ell\in [|S_i|]$, it returns the $\ell$th element of $S_i$ according to some internal ordering.  Given an element $j\in[\vars]$, it detects whether it is in $S$, and if it is, returns the tuple $(i,\ell)$ such that $j$ is the $\ell$th element of $S_i$.  Given $a\in[q]$, it returns $i\in[k]$ such that $a$ equals to a value in $S_i$ or says there is no such $i$.

The data structure is coherence-friendly, i.e., a subset $S$ will have the same representation $\ket D|S>$ independently of the sequence of update operations that results in this subset.  Next, it has an exponentially small error probability of failing that can be ignored.  Finally, the implementation of this data structure requires QRAQM.

\subsection{Quantum Walk}
\label{sec:3distwalk}
In this section, we describe the quantum walk part of the algorithm.  Formally, it is as follows.

\begin{lem}
\label{lem:kdist}
Let $r_1,\dots,r_{k-1} = o(\vars)$ be positive integers, $\str\in[q]^\vars$ be an input for the $k$-distinctness problem, and $V_0$ be the set of $S\subseteq[\vars]$ having type $(r_1,\dots,r_{k-1},0)$.
Given the uniform superposition $\qst = \frac{1}{\sqrt{|V_0|}} \sum_{S\in V_0} \ket|S>$, it is possible to solve the $k$-distinctness problem in $\tO(\vars/\sqrt{\min\{r_1,\dots,r_{k-1}\}})$ quantum time.
\end{lem}

\pfstart
As mentioned in \rf(sec:3disttech), we may assume that any input contains at most one $k$-collision and $\Omega(\vars)$ $(k-1)$-collisions.  Define $r_k=0$, and the type $\tau_i$ as $(r_1,\dots,r_{i-1}, r_i+1, r_{i+1},\dots,r_k)$ for $i\in\{0,1,\dots,k\}$.  Let $V_i$ be the set of all $S\subseteq[\vars]$ having type $\tau_i$.  It is consistent with our previous notation for $V_0$.  Denote $V=\bigcup_i V_i$.  Also, for $i\in[k]$, define the set $Z_i$ of {\em dead-ends}  consisting of vertices of the form $(S,j)$ for $S\in V_{i-1}$ and $j\in [\vars]$ such that $S\btr \{j\}\notin V$.  Again, $Z=\bigcup_i Z_i$.

The vertex set of $G$ is $V\cup Z$.  Each $S\in V\setminus V_k$ is connected to $\vars$ vertices: one for each $j\in [\vars]$.  If $S\btr\{j\}\in V$, it is the vertex $S\btr\{j\}$, otherwise, it is $(S,j)\in Z$.  A vertex $S\in V_k$ is connected to $k$ vertices in $V_{k-1}$ differing from $S$ in one element.   Each $(S,j)\in Z$ is only connected to $S$.  
The weight of each edge is 1.
A vertex is marked if and only if it is contained in $V_k$.

\refalg{walk} is not directly applicable here because we do not know the graph in advance (it depends on the input), nor do we know the amplitudes in the initial distribution $\qst$.  However, we know the graph locally, and our ignorance in the amplitudes of $\qst$ conveniently cancels out with our ignorance in the size of $G$.

Let us briefly describe the implementation of the quantum walk on $G$ following \refsec{walk}.  Let $G=(V\cup Z,E)$ be the graph described above. It is bipartite: The part $A$ contains all $V_i$ and $Z_i$ for $i$ even, and $B$ contains all $V_i$ and $Z_i$ for $i$ odd.  The support of $\qst$ is contained in $A$.  The reflections $R_A$ and $R_B$ are the direct sums of local reflections $D_u$ over all $u$ in $A$ and $B$, respectively.  They are as follows:
\itemstart
\item If $u\in V_k$, then $D_u$ is the identity in $\cH_u$.
\item If $u\in Z_i$, then $D_u$ negates the amplitude of the only edge incident to $u$.
\item If $u\in V_i$ for $i<k$, then $D_u$ is the reflection about the orthogonal complement of $\psi_u$ in $\cH_u$.  If $u\in V_0$, or $u\in V_i$ with $i>0$, then $\psi_u$ is defined as
\[
\psi_u = \frac{1}{\sqrt{C_1}} \ket |u> + \sum_{uv\in E} \ket|uv>,\qquad\text{or}\qquad
\psi_u = \sum_{uv\in E} \ket|uv>,
\]
respectively.  Here, $C_1$ is a constant.%, and we assume that $\ket|uv> = -\ket|vu>$.
\itemend

The space of the algorithm consists of three registers: $\reg D$, $\reg C$ and $\reg Z$.  The data register $\reg D$ contains the data structure for $S\subseteq[\vars]$.  The coin register $\reg C$ contains an integer in $\{0,1,\dots,\vars\}$, and the qubit $\reg Z$ indicates whether the vertex is an element of $Z$.  A combination $\ket D|S>\ket C|0>\ket Z|0>$ with $S\in V_0$ indicates a vertex in $V_0$ that is used in $\qst$.  A combination $\ket D|S> \ket C|j>\ket Z|0>$ with $j>0$ indicates the edge between $S$ and $S\btr \{j\}$ or $(S,j)\in Z$.  Finally, a combination $\ket D|S>\ket C|j>\ket Z|1>$ indicates the edge between $(S,j)\in Z$ and $S\in V$.

Similarly to \rf(sec:3rdproof), the reflections $R_A$ and $R_B$ are broken down into the diffuse and update operations.  The diffuse operations perform the local reflections in the list above. 
For the first one, do nothing conditioned on $\ket D|S>$ being marked.  For the second one, negate the phase conditioned on $\reg Z$ containing 1.  The third reflection is the standard Grover diffusion with one special element if $S\in V_0$.  Similarly to \rf(alg:walk), the orientation of the edges may be ignored because the graph is bipartite.

The update operation can be performed using the primitives from \rf(sec:3disttech).  Given $\ket D|S> \ket C|j>\ket Z|b>$, calculate whether $S\btr\{j\}\in V$ in a fresh qubit $\reg Y$.  Conditioned on $\reg Y$, query the value of $\str_j$ and perform the update operation for the data structure.  Conditioned on $\reg Y$ not being set, flip the value of $\reg Z$.  Finally, uncompute the value in $\reg Y$.  On the last step, we use that $\ket D|S>\ket C|j>$ represents an edge between vertices in $V$ if and only if $\ket D|S\btr\{j\}> \ket C|j>$ does the same.

After we showed how to implement the step of the quantum walk efficiently, let us estimate the required number of steps.  The argument is very similar to the one in \rf(thm:walk).  Let us start with the positive case.  Assume $\{a_1,\dots,a_k\}$ is the unique $k$-collision.
Let $V_0'$ denote the set of $S\in V_0$ that are disjoint from $\{a_1,\dots,a_k\}$, and $\st'$ be the uniform probability distribution on $V_0'$.  
Define the flow $p$ from $\st'$ to $V_k$ as follows.  For each $S\in V_i$ such that $i<k$ and $S\cap M = \{a_1,\dots,a_i\}$, define flow $p_e = 1/|V_0'|$ on the edge $e$ from $S$ to $S\cup\{a_{i+1}\}\in V_{i+1}$.  Define $p_e = 0$ for all other edges $e$.
Let
\[
\phi = \sqrt{C_1} \sum_{S\in V_0'}\frac{1}{|V_0'|} \ket|S> - \sum_{e\in E} p_e\ket|e>.
\]
This vector is orthogonal to all $\psi_u$, hence, is invariant under the action of $R_BR_A$.
Also, $\|\phi\|^2 = (k+C_1)/|V_0'|$, and $\ip<\phi, \qst> = \sqrt{C_1/|V_0|}$.  Hence,
\[
\ipB<\frac\phi{\|\phi\|},\;\qst> = \sqrt{\frac{C_1|V_0'|}{(k+C_1)|V_0|}} \sim \sqrt{\frac{C_1}{k+C_1}}
\]
where $\sim$ stands for the asymptotic equivalence as $\vars\to\infty$.

In the negative case, define
\[
w' = \sqrt{\frac{C_1}{|V_0|}} 
\sC[\sum_{S\in V_0} \frac{1}{\sqrt{C_1}} \ket|S> + \sum_{e\in E} \ket|e>].
\]
Similarly to the proof of \refthm{walk}, we have that $\Pi_Aw'=0$ and $\Pi_Bw' = \qst$.

Let us estimate $\|w'\|$.  The number of edges in $E$ is at most $\vars$ times the number of vertices in $V_0\cup\cdots\cup V_{k-1}$.  Thus, we have to estimate $|V_i|$ for $i\in[k-1]$.  Consider the relation between $V_0$ and $V_i$ where $S\in V_0$ and $S'\in V_i$ are in the relation iff $S'\setminus S$ consists of $i$ equal elements. Each element of $V_0$ has at most $\vars{k-1\choose i} = O(\vars)$ images in $V_i$ because there are at most $\vars$ maximal collisions in the input, and for each of them, there are at most ${k-1\choose i}$ variants to extend $S$ with.  On the other hand, each element in $V_i$ has exactly $r_i+1$ preimages in $V_0$.  Thus, $|V_i| =O( \vars |V_0|/r_i )$.  Thus,
\[
\|w'\| = O\s[ \sqrt{1+\vars/r_1+\vars/r_2+\cdots+\vars/r_{k-1}} ] =
O\s[ {\vars}/{\sqrt{\min\{r_1,\dots, r_{k-1}\}}} ].
\]
By \reflem{effective}, we have that if
$\delta = \Omega(1/\|w'\|)$,
then the overlap of $\qst$ with the eigenvectors of $R_BR_A$ with phase less than $\delta$ can be made at most $1/C_2$ for any constant $C_2>0$.  Thus, it is enough to execute the phase estimation with precision $\delta$ if $C_1$ and $C_2$ are large enough.  By \rf(thm:detection), this requires $O( {\vars}/{\sqrt{\min\{r_1,\dots, r_{k-1}\}}} )$ iterations of the quantum walk.
\pfend

\subsection{Preparation of the Initial State}
\label{sec:setup}
Now we describe how to generate the uniform superposition $\qst$ over all elements in $V_0$ from the formulation of \rf(lem:kdist) efficiently in the special case of $k=3$.  Let us denote $r_1 = \vars^{5/7}$ and $r_2 = \vars^{4/7}$.  We start under the assumption the input is negative.

Prepare the state ${\vars\choose r_1}^{-1/2} \sum_{S: |S| = r_1} \ket D|S>$ in time $\tO(r_1)$.  This is very similar to the algorithm by Ambainis, and we omit the details.  
  Measure the type of $S$.  
  The state of the algorithm collapses to the uniform superposition of the subsets of some type $\tau = (t_1,t_2)$.  Unfortunately, with high probability, $t_2$ will be of order $r_1^2/\vars$ that is much smaller than the required size $r_2$.

We enlarge the size of $S_2$ by using the Grover search repeatedly.  For each $S$ in the superposition, apply the Grover search over $[\vars]$.  An element $j\in[\vars]$ is marked iff $j\notin S$ and $\str_j$ is equal to an element in $S_1$.  This can be tested using the primitives from \rf(sec:3disttech). If the Grover search fails, repeat it from the current state.  If the search succeeds, the state is a superpositon of states of the form $\ket D|S> \ket |j>$.  Query the value of $\str_j$, and update the data structure.  This gives a superposition over $\ket D|S\cup\{j\}> \ket |j>$.  Let $S' = S\cup\{j\}$.  Apply the primitive that transforms $j$ into its number in $S'_2$.  This gives a superposition over $\ket D|S'>\ket |i>$ where $i\in [|S'_2|]$.  We show in a moment that, for a fixed $S'$, all states $\ket D|S'>\ket |i>$ have the same amplitude, hence, the second register can be detached in the sense of \rf(sec:RAM).

A typical subset has $\Omega(r_1)$ elements in $S_1$ that can be extended to a 2-collision, hence, the Grover search requires $O(\sqrt{\vars/r_1})$ iterations.  As we load $O(r_2)$ additional elements, the time spent during the Grover search is $\tO(r_2\sqrt{\vars/r_1})$.

Now assume each $S$ contains $r_2$ 2-collisions.  Unfortunately, the state is not the uniform superposition we require for the quantum walk in \rf(lem:kdist).  But due to symmetry, at any place in the algorithm, the amplitude of a subset $S$ only depends on the number of elements in $S_1$ that can be extended to a 2-collision.  This shows that, indeed, the second register can be detached after the Grover search.  Moreover, this gives us a way to generate the uniform superposition we require.

We measure the content of $S_1$.  Let $B$ be the outcome.  The state collapses to the uniform superposition over subsets $S_2$ consisting of $r_2$ 2-collisions not using the values in $B$.  Then, we repeat the first step, i.e., for each $S$, we construct the uniform superposition over subsets of size $r_1$ consisting of elements outside $S_2$ and having values different from the ones in $B$.  After that, we measure the type of the subset.  This results in the uniform superposition over states in $V_0$ of type $(r_2', r_1')$ with $r_2'>r_2$ and $r_1'=\Theta(r_1)$ and avoiding elements with values in $B$.

In the positive case, due to a similar argument, the state can be written as $\alpha\qst' + \sqrt{1-\alpha^2}\qst''$ where $\qst'$ is the uniform superposition over $V_0'$ as defined in the proof of \rf(lem:kdist), and $\qst''$ is some superposition over $\ket D|S>$ where $S$ intersects $\{a_1,a_2,a_3\}$.  One can show that $\alpha$ is close to 1, hence, the initial state has large overlap with 1-eigenspace of $R_BR_A$.%
\footnote{One can modify the algorithm so that it does not require this observation.  With probability $1/2$, continue with the old algorithm, and with probability $1/2$, measure the content of $S$ and search for a 2-distinctness outside $S$ having a value equal to a value in $S$.  This can be done using the standard algorithm for 2-distinctness with minor modifications.}

Then, we can apply the algorithm from \reflem{kdist} with additional modification that a vertex $(S,j)$ is declared a dead-end also if $\str_j$ has a value in $B$.  This finds a 3-collision in time $\tO(\vars/\sqrt{r_2})$ if its value is different from a value in $B$.  For the values in $B$, we search for a 2-collision outside $B$ but having a value equal to a value in $B$.  This can be implemented in time $\tO(\vars^{2/3})$ using the standard algorithm for 2-distinctness with minor modifications.

Thus, up to polylogarithmic factors, the time complexity of the algorithm is
\(
r_1 + r_2\sqrt{\vars/r_1} + \vars/\sqrt{r_2}.
\)
This attains optimal value of $\tO(\vars^{5/7})$ for $r_1 = \vars^{5/7}$ and $r_2 = \vars^{4/7}$.  This finishes the proof of \rf(thm:3dist).

\section{Summary}
In this chapter, we constructed a new variant of Szegedy-type quantum walks.  The main innovations are that the walk need not start in the stationary distribution, and the expected number of steps of the walk is expressed in terms of electric resistance, not the classical hitting time.  This provides greater flexibility, since one is not interested in preparing the stationary distribution state.  Also, the effective resistance is easier to estimate than the spectral gap of a classical random walk.

We used this type of quantum walks in an alternative realisation of learning graphs, and we also applied it to the $k$-distinctness problem.  The application to the $k$-distinctness problem is different from the learning-graph-based algorithm from \rf(chp:kdist).  We split the algorithm into two parts: the set-up and the walk phases, akin to the original algorithm by Ambainis, \rf(prp:distinctness).  The walk phase can be implemented for any $k$, but we managed to implement the set-up phase only for $k=3$.  It is an open problem to implement the set-up phase, or to come up with a different time-efficient implementation of $k$-distinctness, for $k>3$.

We are interested in further applications of \rf(thm:walk).  It could be based on some classical algorithm, since many of them start in a state far from the stationary distribution.  
However, the new quantum walk still does not match the flexibility of classical random walks: the initial state has to be hard-wired into the algorithm, as well as the estimate $R$ on the effective resistance.  This makes it hard to apply this result for graphs of unknown structure, and limits the scope of possible applications.

%% file: _conclusion.tex
In the thesis, we developed a number of tools for the development of query- and time-efficient quantum algorithms, as well as for proving lower bounds.  We improved quantum query complexity of such problems as triangle detection, $k$-distinctness, associativity testing.  We proved lower bounds on the quantum query complexity of the $k$-sum and the triangle sum problems.  We developed time-efficient algorithms for path- and claw-detection, and for the 3-distinctness problems.

In the thesis, we have already mentioned some interesting open problems that can be attacked using these techniques.  This includes determining the quantum query complexity of the graph collision problem, proving adversary lower bounds for the collision, the set equality and the $k$-distinctness problems, developing new time-efficient algorithms using new variants of the quantum walk.

However, there are more open problems.
All the algorithms in our thesis used that the computed function has small certificates.  Also, they all were for decision problems with Boolean output.  However, we know that the adversary bound is tight also for functions beyond this class.  It would be interesting to construct dual adversary solutions for such functions.  

Also, all speed-ups we attain in the thesis are polynomial (actually, at most quadratic).
It is an interesting open problem to develop quantum walks that obtain superpolynomial separations either in the query complexity or in the time complexity.  Some examples are known for quite a while~\cite{kempe:walkHitFaster, childs:walkExponentialSeparation}, and we hope that some of the new ideas may help in more practical problems, like the Hidden Subgroup Problem or the $\mathsf{BQP}$ versus $\mathsf{PH}$ problem.

%% file: appendix.tex
In this chapter, we briefly describe some technical results we use throughout the thesis.

\mycommand{semidef}{\mathbf{S}}

\section{Linear Algebra}
\label{sec:linearAlgebra}
Most of the results mentioned in the appendix can be found in the book by Horn and Johnson~\cite{horn:matrix}.  If the following, we implicitly assume that all involved sums and products of matrices/vectors are defined.

We work with finite-dimensional complex or real vectors spaces.  We use $\C^A$ or $\R^A$ to denote vector space over the field of complex or real number, respectively, having elements of $A$ as its orthonormal basis.  The element of the basis corresponding to an element $a\in A$ will be usually denoted by $e_a$.  The most common case is $\C^n$ or $\R^n$, in which case the elements of the basis are labelled by integers from 1 to $n$.
We represent vectors by column-matrices.  If $v\in \C^A$, we use $v\elem[a]$ to denote the $a$th component of the vector.  Thus, $v = \sum_{a\in A} v\elem[a]e_a$.

Also, we work with complex and real matrices.  
An $X\times Y$ matrix is a matrix with its rows labelled by the elements of $X$ and its columns labelled by the elements of $Y$.
By $A\elem[i,j]$ we denote the $(i,j)$th entry of the matrix $A$.  For two matrices $A$ and $B$, we define their ordinary product $AB$, and Hadamard (entry-wise) product $A\circ B$.  
If $A$ is an $X\times Y$ matrix, and $B$ is an $X'\times Y'$ matrix, then their tensor (Kronecker) product $A\otimes B$ is the $(X\times X')\times (Y\times Y')$ matrix defined by $(A\otimes B)\elem[(x,x'), (y,y')] = A\elem[x,y]B\elem[x',y']$.
It satisfies $(A\otimes B)(C\otimes D) = (AC)\otimes (BD)$.
Also, we use direct sums of vectors, as well as of matrices that we denote by $\oplus$.  We denote the trace (the sum of the diagonal elements) by $\tr A$.

The inner product $\ip<u,v>$ of two vectors in $\C^n$ is defined by $\sum_{i=1}^n (u\elem[i])^* v\elem[i]$, where ${}^*$ stands for the complex conjugate.  By $A^*$, where $A$ is a matrix, we denote the conjugate transpose of $A$.  It is the only matrix that satisfies $\ip<u,Av> = \ip<A^*u,v>$ for all vectors $u$ and $v$ of the appropriate sizes.  

The norm of the vector $u$ is defined as $\|u\| = \sqrt{\ip<u,u>}$.  The unit vector is a vector of norm 1.  The norm satisfies the triangle inequality $\|u+v\| \le \|u\| + \|v\|$ and the Cauchy-Schwarz inequality $|\ip<u,v>|\le \|u\|\|v\|$.

Similarly, we define the inner product of two $m\times n$ matrices $A$ and $B$ as $\ip<A,B> = \tr (A^*B)$.  It is equivalent to the inner product on vectors, if the matrices are treated as vectors in $mn$-dimensional vector space.  The corresponding norm is the Frobenius norm $\normFrob|A| = \sqrt{\ip<A,A>}$.  It also satisfies the triangle and the Cauchy-Schwarz inequalities.

The spectral norm of an $m\times n$ matrix $A$ is defined as $\|A\| = \max_{u,v} {u^*Av}$ where $u$ and $v$ range over unit vectors in $\C^n$ and $C^m$, respectively.  Thus, $\|Av\|\le \|A\|\|v||$ and $\normFrob|AV|, \normFrob|VA|\le \|A\|\normFrob|V|$.  The spectral norm satisfies the triangle inequality, and also $\|AB\|\le \|A\|\|B\|$.

Assume $A$ is $n\times n$ matrix.  A number $\lambda\in\C$ and vector $v\in\C^n$ are {\em eigenvalue} and {\em eigenvector} of $A$, respectively, if $Av = \lambda v$.  
An $\lambda$-eigenvector is an eigenvector with eigenvalue $\lambda$.  The $\lambda$-eigenspace is the linear subspace of all $\lambda$-eigenvectors.
The largest absolute value of an eigenvalue of $A$ is called the spectral radius of $A$ and is denoted by $\rho(A)$.  For any two matrices $A$ and $B$, the matrices $AB$ and $BA$ have the same non-zero eigenvalues, counted with multiplicity.

A square matrix $U$ is called unitary if $U^*U = I$.  It is equivalent to $\ip<Uu,Uv> = \ip<u,v>$ for all $u$ and $v$.  
A square matrix $A$ can be transformed into some canonical form $\Lambda$ using a unitary matrix $U$: $A = U^* \Lambda U$.  The form of $\Lambda$ depends on the conditions the matrix $A$ satisfies, and is summarised in \rf(tbl:schur).  The result of the first row of the table is known as Schur's theorem.

\begin{table}[htb]
\centering
\begin{tabular}{|lcl|}
\multicolumn{1}{c}{Name} &
\multicolumn{1}{c}{Condition} &
\multicolumn{1}{c}{Form of $\Lambda$} \\
\hline
General matrix & & Upper triangular, eigenvalues on the diagonal\\
Normal & $A^*A = AA^*$ & Diagonal \\
Hermitian & $A^* = A$ & Diagonal with real entries \\
Positive semi-definite & $\ip<Av,v>\ge 0$ for all $v$ & Diagonal with non-negative entries\\
Positive definite & $\ip<Av,v> > 0$ for all $v\ne 0$ & Diagonal with positive entries\\
\hline
\end{tabular}
\caption{Main types of matrices}
\label{tbl:schur}
\end{table}

We use notation $A\succeq 0$ and $A\succ 0$ to denote that $A$ is positive semi-definite or positive definite, respectively.  The set of $n\times n$ real positive semi-definite matrices is a topologically closed convex cone: if $A, B\succeq 0$ and $c,d\ge 0$, then $cA+dB\succeq 0$.  
This induces a partial ordering on the set of all real symmetric matrices.  We say that $A\succeq B$ if $A-B\succeq 0$.  The cone of real positive semi-definite matrices is self-dual:  $\ip<A,B> \ge 0$ for all $A, B\succeq 0$.  Also, if $\ip<A,V>\ge 0$ for all $A\succeq 0$, then $V$ is positive semi-definite.
Finally, $A\otimes B\succeq 0$, and $A\circ B\succeq 0$ for all positive semi-definite $A$ and $B$.

From the second row of \rf(tbl:schur), it follows that each normal $n\times n$ matrix $A$ can be decomposed as $A = \sum_{i=1}^r \lambda_i v_i^* v_i$, where $r$ is the rank of $A$, $\lambda_i\in \C$, and $v_i$ are pairwise orthogonal unit vectors in $\C^n$.  For each $i$, $\lambda_i$ and $v_i$ are an eigenvalue and the corresponding eigenvector of $A$.  This decomposition is called the eigenvalue decomposition of $A$.  If $A$ is real and Hermitian (also called symmetric), then $\lambda_i$ and all the entries of $v_i$ can be taken real.

An arbitrary $m\times n$ matrix $A$ admits a singular value decomposition: $A = \sum_{i=1}^r \sigma_i u_i^* v_i$.  Here, $r$ is the rank of $A$, and $\sigma_i$ are positive real numbers.  The vectors $u_i$ are pairwise orthogonal unit vector of $\C^n$, and $v_i$ are pairwise orthogonal unit vectors in $\C^m$.  The number $\sigma_i$ is called a singular value of $A$, and $u_i$ and $v_i$ are the corresponding left and right singular vectors of $A$.  The set of $\sigma_i$ coincides with the set of square roots of the non-zero eigenvalues of $A^*A$ counted with multiplicity (alternatively, one may take $AA^*$).  The vectors $v_i$ and $u_i$ are the corresponding eigenvectors of $A^*A$ and $AA^*$, respectively.  The spectral norm of $A$ equals the largest of $\sigma_i$.

A square matrix $A$ is called non-negative (resp., positive), if all its entries are non-negative (resp., positive).  A non-negative matrix $A$ is called primitive if $A^k$ is positive for some positive integer $k$.  The Perron-Frobenius Theorem asserts that if $A$ is a primitive non-negative matrix, then the spectral radius $\rho(A)$ is an eigenvalue of $A$ with multiplicity 1, the corresponding eigenvector has only positive entries, and all the remaining eigenvalues of $A$ are strictly less than $\rho(A)$ in absolute value.

\section{Convex Optimisation}
\label{sec:convexOptimisation}
In this section, we briefly describe convex optimisation with emphasis on inequalities involving real semi-definite matrices.  This type of optimisation is known as semi-definite optimisation.  We explicitly write out the Lagrangian, as we find it a more intuitive approach compared to usage of black-box standard forms.  For more detail, the reader may refer to the book by Boyd and Vandenberghe~\cite{boyd:convex}.

Let $V$ be a vector space.  A subset $\cD\subseteq V$ is called convex if $x,y\in \cD$ implies that $\theta x + (1-\theta) y \in \cD$ for all real $\theta$ between 0 and 1.  A function $f\colon \cD\to\R$, defined on a convex subset, is called {\em convex} if $f(\theta x+(1-\theta)y)\le \theta f(x) + (1-\theta) f(y)$ for all $x,y\in \cD$ and all real $\theta$ between 0 and 1.  For example, all linear and affine functions are convex.

Similarly, let $F\colon \cD\to \R_{d,d}$ be a function to the linear space of all real symmetric $d\times d$ matrices.  We say that the function $F$ is convex (with respect to the cone of semi-definite matrices) if $F(\theta x + (1-\theta) y) \preceq \theta F(x) + (1-\theta) F(y)$ for all $x,y\in\cD$ and $0\le\theta\le 1$.

A convex optimisation problem is a problem of the form
\begin{subequations}
\label{eqn:primal}
\begin{alignat}{3}
 &{\mbox{\rm minimise }} &\qquad& f_0(x) \label{eqn:primalObjective} \\
 &{\mbox{\rm subject to }} && f_i(x)\le 0 &\qquad& \text{for all $i$;}\label{eqn:primalCondition1}\\
 &&& G_j(x)\preceq 0 && \text{for all $j$;} \label{eqn:primalCondition2}\\
 &&& h_k(x) = 0 && \text{for all $k$.} \label{eqn:primalCondition3}
\end{alignat}
\end{subequations}
where $f_i$ and $G_j$ are convex functions with the values in $\R$ and the space of square matrices, respectively; and $h_k\colon V\to \R$ are linear functions.  The function $f_0$ is known as the {\em objective function}, and equations~~\rf(eqn:primalCondition1), \rf(eqn:primalCondition2) and ~\rf(eqn:primalCondition3) are known as the {\em constraints}.

Let $\cD$ denote the intersection of the domains of all the functions in~\rf(eqn:primal).  It is a convex set.
A point $x\in\cD$ is called {\em feasible} if it satisfies the constraints~\rf(eqn:primalCondition1), \rf(eqn:primalCondition2), and~\rf(eqn:primalCondition3).  
A feasible point $x^*$ is called {\em optimal} if the value of $f_0(x^*)$ is minimal possible among all feasible points.

We write out the {\em Lagrangian} of this optimisation problem, that is a function of the form
\[
L(x, \lambda, M, \nu) = f_0(x) + \sum_i \lambda_i f_i(x) + \sum_j \tr (M_j G_j(x)) + \sum_k \nu_k h_k(x)
\]
where $\lambda = (\lambda_i)$, $M = (M_j)$ and $\nu = (\nu_k)$ satisfy $\lambda_i\ge 0$, $M_j\succeq 0$, and $h_k\in\R$ for all $i,j$ and $k$.  They are known as the {\em Lagrange multipliers}.  Let us denote
\[
g(\lambda, M, \nu) = \inf_{x\in\cD} L(x, \lambda, M, \nu).
\]
Let $p^*$ be the optimal value of \rf(eqn:primal) that is attained for the value $x^*$ of the variable.  Then, for any choice of $\lambda_i$, $M_j$ and $\nu_k$, we get that
\[
g(\lambda, M, \nu) \le L(x^*,\lambda, M, \nu) = f_0(x^*) + \sum_i \lambda_i f_i(x^*) + \sum_j \tr(M_j G(x^*)) + \sum_k \nu_k h_k(x^*) \le f_0(x^*) = p^*,
\]
since $x^*$ is a feasible solution.  (Here, we used that $\tr(AB)\ge 0$ for all $A,B\succeq 0$.) Thus, $g(\lambda, M, \nu)$ provides a lower bound on the value of $p^*$.  If we search for the best lower bound obtained in this way, we arrive to the following {\em dual} optimisation problem:
\begin{subequations}
\label{eqn:dual}
\begin{alignat}{3}
 &{\mbox{\rm maximise}} &\qquad& g(\lambda, M, \nu) \label{eqn:dualObjective} \\
 &{\mbox{\rm subject to }} && \lambda_i \ge 0 &\qquad& \text{for all $i$;}\\
 &&& M_j \succeq 0 && \text{for all $j$.}
\end{alignat}
\end{subequations}

Thus, the optimal value $d^*$ of~\rf(eqn:dual) never exceeds the optimal value of~\rf(eqn:primal).  This fact is known as {\em weak duality}.  In actual applications, we not always use exactly the same form of the optimisation problem, but rather follow a closely related transformation.

In many cases, the equality $p^* = d^*$ holds.  This is known as {\em strong duality}.  The optimisation problem~\rf(eqn:primal) is called {\em strictly feasible}, if there exists $x$ such that $f_i(x)<0$ and $G_j(x)\prec 0$ for all $i$ and $j$.  {\em Slater's condition} implies strong duality if~\rf(eqn:primal) is convex and strictly feasible.